\documentclass[
    a4paper,
    11pt, % fontsize options: 10pt, 11pt, 12pt
    twoside, % layout one- or twosided: oneside, twoside
    openright % for twosided layouts: start chapters on right side
]{memoir}
\chapterstyle{veelo}

\usepackage{ifxetex,ifluatex}

\ifxetex
	\usepackage[ttscale=.85]{libertine}
\fi
\ifluatex
	\setromanfont[Ligatures=TeX]{Linux Libertine O}
\fi
\ifluatex
	\usepackage[german,english]{selnolig} 
\fi
\usepackage{microtype}

\usepackage{amssymb}
\usepackage{amsmath}
\usepackage{amsthm}
\usepackage{nccmath}
\usepackage{mathtools}
\usepackage{thmtools} 
\usepackage{graphicx}
\usepackage{todonotes}
\usepackage{caption}
\usepackage{subcaption}
\usepackage{tikz,pgfplots}
\usepackage[T1]{fontenc}
\usetikzlibrary{patterns}

\usepackage{csquotes}
\usepackage{numprint}
\usepackage{microtype}

\usepackage[backend=biber,style=numeric-comp,maxnames=10,giveninits=true,doi=false,url=false,isbn=false]{biblatex}

\usepackage{algorithm}
\usepackage{algpseudocode}

\algnewcommand\algorithmicinput{\textbf{Input:}}
\algnewcommand\INPUT{\item[\algorithmicinput]}
\algnewcommand\algorithmicoutput{\textbf{Output:}}
\algnewcommand\OUTPUT{\item[\algorithmicoutput]}

\usepackage{xspace}
\usepackage{hyphenat} % allows hyphenation of already hyphenated words with comman \hyp{}

\usepackage[NoDate]{currvita}

\title{Algorithm Engineering for Cut Problems}
\author{Alexander Noe, MSc BSc}

\setsecnumdepth{subsection}
\settocdepth{subsection}

\allowdisplaybreaks[1]

\ifluatex
	\setlibertinemath
\fi

\declaretheorem[numberwithin=section]{theorem}
\declaretheorem[numberlike=theorem]{lemma}

\declaretheorem[numberlike=theorem]{claim}

\def\MdN{\ensuremath{\mathbb{N}}}
\newcommand{\Oh}[1]{\mathcal{O}\!\left( #1\right)}
\newcommand{\Otilde}[1]{\tilde{\mathcal{O}}\!\left( #1\right)}
\newcommand{\Is}{:=}
\newcommand{\etal}{et~al.\ }
\newcommand{\eg}{e.g.\ }
\newcommand{\ie}{i.e.\ }
\newcommand{\CC}{C\texttt{++}}

\newcommand{\cut}{\mathcal{C}}
\newcommand{\bestcut}{\widehat{\cut}}
\newcommand{\wgt}{\mathcal{W}}
\newcommand{\bestwgt}{\widehat{\wgt}}
\newcommand{\vopt}{\mathcal{V}}
\newcommand{\queue}{\mathcal{Q}}
\newcommand*{\shorterDots}{.\kern-0.06em.\kern-0.06em.} % touching at \kern-0.1725em

\newcommand{\prOne}{\textttA{HeavyEdge}}
\newcommand{\prTwo}{\textttA{ImbalancedVertex}}
\newcommand{\prThree}{\textttA{ImbalancedTriangle}}
\newcommand{\prFour}{\textttA{HeavyNeighborhood}}

\newcommand{\optZero}{\textttA{BasicCactus}}
\newcommand{\optOne}{\textttA{+Connectivity}}
\newcommand{\optTwo}{\textttA{+LocalContract}}
\newcommand{\optThree}{\textttA{+DegreeOne}}
\newcommand{\optFour}{\textttA{+C\&LInCactus}}
\newcommand{\optFive}{\textttA{+D1InCactus}}
\newcommand{\optSix}{\textttA{FullAlgorithm}}

\def\theLetterSpace{-0.5pt}
\def\extraWordSpace{-0.5pt}
\newcommand\spaceout[2][\theLetterSpace]{%
  \def\LocalLetterSpace{#1}\expandafter\spaceouthelpA#2 \relax\relax}
\def\spaceouthelpA#1 #2\relax{%
  \spaceouthelpB#1\relax\relax%
  \ifx\relax#2\else\kern\extraWordSpace\ \kern\LocalLetterSpace\spaceouthelpA#2\relax\fi
}
\def\spaceouthelpB#1#2\relax{%
  #1%
  \ifx\relax#2\else
    \kern\LocalLetterSpace\spaceouthelpB#2\relax%
  \fi
}
\def\textttA#1{\texttt{\spaceout{#1}}}

\newcommand{\vcbase}{\texttt{VieCut-MTC}}
\newcommand{\exact}{\texttt{Exact-MTC}}
\newcommand{\inexact}{\texttt{Inexact-MTC}}

\npdecimalsign{.}

\def\reflist#1#2{
   \list
   {\theenumi.}{\settowidth\labelwidth{[#1]}\leftmargin\labelwidth
   \advance\leftmargin\labelsep
   \usecounter{enumi}\listparindent-\leftmargin}
   \renewcommand{\theenumi}{#2-\arabic{enumi}}
   \def\newblock{\hskip .11em plus .33em minus -.07em}
   \sloppy
   \sfcode`\.=1000\relax}

\usepackage[hidelinks,unicode,psdextra]{hyperref}
\usepackage{cleveref} % must be last (after hyperref!) so that cref works

\hypersetup{
    pdftitle = {Algorithm Engineering for Cut Problems},
    pdfauthor = {Alexander Noe},
    pdfsubject = {Dissertation},
    pdfkeywords = {theoretical computer science, efficient algorithms}
}

\begin{document}

\frontmatter

%\begin{german}
\makepagestyle{titlepage}
\makeoddhead{titlepage}{}{}{\includegraphics{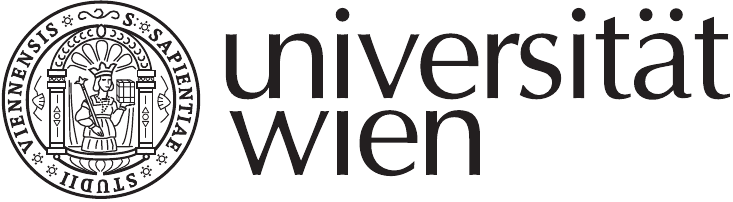}}

\calccentering{\unitlength}                         % Calculate center length and stores in unitlength
\begin{adjustwidth*}{\unitlength}{-\unitlength}     % Adjust center
\begin{adjustwidth}{-1cm}{-1cm}

\thispagestyle{titlepage}
{\centering

\sffamily

~
\vfill

\vfill

\HUGE \textbf{\textsc{Dissertation / Doctoral Thesis}}\\

\vfill

\normalsize Titel der Dissertation / Title of the Doctoral Thesis \\

\huge \textbf{\thetitle}

\vfill

\normalsize verfasst von / submitted by \\
\Large \theauthor

\vfill

\normalsize angestrebter akademischer Grad / in partial fulfillment of the
requirement for the degree of \\
\Large Doktor der Technischen Wissenschaften (Dr.\,techn.)

\vfill

\normalsize

\begin{tabbing}
% this line defines the width
field of study as it appears on
the student record sheet: \hspace{0.5em} \= Informatik \kill
Wien, 2021 / Vienna, 2021 \\
\\
Studienkennzahl lt. Studienblatt: / \\
degree programme code as it appears on
the student\\ record sheet: \> A 786 880 \\
Dissertationsgebiet lt. Studienblatt: /\\ field of study as it appears on
the student record sheet: \> Informatik \\
Betreuerin: / Supervisor: \> Univ.-Prof.\ Dr.\ Monika Henzinger \\
%Zweitbetreuer: \> Univ.-Prof.\ Dr.\ Christian Schulz \todo{add or remove?}
\end{tabbing}

\vspace{-7ex}

~
} % \centering

\end{adjustwidth}
\end{adjustwidth*}

% reset font
\normalfont

%\end{german}

\cleardoublepage
\begin{abstract}
Graphs are a natural representation of data from various contexts, such as social
connections, the web, road networks, and many more. In the last decades, many of
these networks have become enormous, requiring efficient
algorithms to cut networks into smaller, more readily comprehensible blocks. In this
work, we aim to partition the vertices of a graph into
multiple blocks while minimizing the number of edges that connect different
blocks. There is a multitude of cut or partitioning problems that have been the
focus of research for multiple decades. This work develops highly-efficient
algorithms for the \emph{(global) minimum cut problem}, the \emph{balanced graph
partitioning problem} and the \emph{multiterminal cut problem}. All of these
algorithms are efficient in practice and freely available for use\footnote{\url{https://github.com/VieCut/VieCut}}. In
particular, we obtain the following results and algorithms:
\begin{itemize}
\item Fast heuristic and exact shared-memory parallel algorithms for the
\emph{(global) minimum cut problem}. We present efficient implementations of
existing techniques and combine them with novel approaches to give algorithms
that find a minimum cut in huge networks significantly faster than
state-of-the-art algorithms. Our
heuristic algorithm has a lower empirically observed error rate than existing
inexact algorithms for the problem.
\item The first engineered algorithm that finds \emph{all} (global) minimum cuts and
returns a compact \emph{cactus graph} data structure which represents all of them in graphs
with billions of edges in a few minutes. With a multitude of data reduction
techniques, we improve the running time of state-of-the-art algorithms by up to
multiple orders of magnitude. Based on the representation of all minimum cuts,
we are able to find the most balanced minimum cut in time linear to the size of
the cactus graph.
\item A \emph{fully-dynamic minimum cut} algorithm that efficiently
maintains the minimum cut on a graph under edge insertions and deletions. While
there is theoretical work, our algorithm is the first
implementation of a fully-dynamic algorithm for the problem. Our
algorithm uses the theoretical foundation and builds on it with efficient and
finely-tuned implementations to give an algorithm that gives up to multiple
orders of magnitude speedup to static recomputation. 
\item An integer linear programming (ILP) based meta-heuristic for the \emph{balanced
graph partitioning problem}. As ILPs do not scale to large inputs, we define a
much smaller model that allows us to use symmetry breaking and make the approach
more scalable. This gives a powerful local search meta-heuristic that can
improve given high-quality partitionings even further. We incorporate this
meta-heuristic into an existing evolutionary algorithm to give an algorithm that
computes state-of-the-art partitionings from scratch.
\item A shared-memory parallel exact branch-and-reduce algorithm for the
\emph{multiterminal cut problem}. For this algorithm, we develop and engineer
highly-efficient data reduction rules to transform a problem into a much smaller
equivalent problem. Additionally we give an inexact algorithm that gives
high-quality solutions for very hard problems in reasonable time.
\end{itemize}
\end{abstract}

\cleardoublepage
%\begin{german}
\begin{abstract}
Graphen sind eine natürliche Representation von Daten aus zahlreichen Kontexten,
zum Beispiel Verbindungen in sozialen Netzwerken, Web-Netzwerken,
Straßennetzwerken und vielen weiteren. In den letzten Jahrzehnten sind viele
dieser Netzwerke zu enormer Größe gewachsen, was effiziente Algorithmen zu ihrer
Partitionierung in kleinere, eher begreifliche Teile erforderlich macht. In
dieser Arbeit versuchen wir, die Knoten von Graphen in mehrere Blöcke zu
partitionieren, so dass die Anzahl von Kanten, welche Blockgrenzen schneiden,
minimiert wird. Es gibt eine Vielzahl von Schnitt- und
Partitionierungsproblemen auf Graphen, welche Bereits seit Jahrzehnten erforscht
werden. Diese Arbeit entwickelt hocheffiziente Algorithmen für das \emph{(Global)
Minimum Cut Problem}, das \emph{Balanced Graph Partitioning Problem} und das
\emph{Multiterminal Cut Problem}. Alle hierbei entwickelten Algorithmen sind
effizient in der Praxis und frei
nutzbar\footnote{\url{https://github.com/VieCut/VieCut}}. Im Einzelnen haben
wir die folgenden Ergebnisse erzielt und Algorithmen entwickelt:

\begin{itemize}
\item Schnelle \emph{heuristische} und \emph{exakte} shared-memory parallele Algorithmen für
das \emph{(Global) Minimum Cut Problem}. Wir präsentieren effiziente
Implementierungen bestehender Methoden und kombinieren diese mit neuartigen
Verfahren, um Algorithmen zu entwickeln, die einen minimalen Schnitt signifikant
schneller finden können als der bisherige Stand der Forschung. Die heuristische Variante
unseres Algorithmus hat hierbei auch eine deutlich niedrigere empirisch
beobachtete Fehlerrate als bestehende inexakte Algorithmen für das Problem.
\item Der erste praktisch effiziente Algorithmus, welcher \emph{alle} global
minimalen Schnitte eines Graphen findet und eine kompakte \emph{Cactus Graph
Datenstruktur} bildet, welche diese Schnitte repräsentiert. Unser Algorithmus
findet alle minimalen Schnitte in Graphen mit bis zu mehreren Milliarden Kanten
und mehreren Millionen minimalen Schnitten in wenigen Minuten. Mithilfe einer
Vielzahl von Datenreduktionstechniken verbessern wir die Laufzeit von
bestehenden Algorithmen um bis zu mehreren Größenordnungen. Ausgehend von der
Cactus Graph Repräsentation sind wir auch in der Lage, den \emph{Most Balanced
Minimum Cut} in Laufzeit linear zur Größe des Kaktusgraphen zu finden.
\item Ein \emph{fully-dynamic Minimum Cut} Algorithmus, welcher effizient einen
minimalen Schnitt eines Graphen unter Kanteneinfügungen und -löschungen aufrecht
erhält. Während es bereits theoretische Forschung zu diesem Problem gibt, ist
unser Algorithmus der erste implementierte fully-dynamic Algorithmus für das
Problem. Unsere Arbeit nutzt die bestehenden theoretischen Grundlagen und
kombiniert sie mit effizienten und fein abgestimmten Implementierungen, um zu einem
Algorithmus zu gelangen, welcher um bis zu mehrere Größenordnungen schneller ist
als Neuberechnung mit statischen Algorithmen.
\item Eine Metaheuristik auf Basis von ganzzahliger linearer Optimierung für das
\emph{Balanced Graph Partitioning Problem}. Da ganzzahlige lineare Programme
nicht für große Eingaben skalieren, definieren wir ein deutlich kleineres
Modell, auf welchem wir das Problem unter Zuhilfenahme von Symmetry Breaking
skalierbar machen. Dies resultiert in einer mächtigen Metaheuristik zur lokalen
Suche, welche existierende hochqualitative Partitionierungen noch weiter
verbessern kann. Wir binden diese Metaheuristik in einen existierenden
evolutionären Algorithmus ein und erhalten so einen Algorithmus, der
Partitionierung hoher Qualität selbst erzeugen kann.
\item Ein shared-memory paralleler exakter Algorithmus für das
\emph{Multiterminal Cut Problem}. Für diesen branch-and-reduce Algorithmus entwickeln wir
hocheffiziente Datenreduktionsregeln, um ein Problem in ein viel kleineres
äquivalentes Problem umzuwandeln. Außerdem präsentieren wir einen inexakten
Algorithmus, welcher hochqualitative Lösungen für extrem schwere Instanzen in
annehmbarer Zeit liefert.
\end{itemize}
\end{abstract}
%\end{german}

\cleardoublepage
\renewcommand{\abstractname}{Acknowledgments}
\begin{abstract}
Thank you to everyone who made these last four years a very enjoyable time - I learned a lot and got to explore very interesting problems! 

First and foremost, I would like to thank my advisors Monika Henzinger and
Christian Schulz for their support and guidance with my research. You gave me many
opportunities to learn new things and approach fascinating problems. Thank you
for being so incredibly generous with your time and expertise!

Also, I would like to thank Darren Strash for great
collaboration on the papers we wrote together! Thank you also to Ulrich
Meyer and Ulrik Brandes who agreed to review this thesis. I
could not have wished for a better thesis committee.

I am deeply thankful to Andrew Goldberg for providing me with the incredible 
opportunity to do an internship at Amazon. On this note, I am
very thankful to Quico Spaen, Nhat Le, Larissa Petroianu, Mauricio Resende, Tim Jacobs, and many
others for making this internship memorable and enjoyable. I learned
so much from all of you!

Being a part of the TAA research group over the past few years has been a great
experience! I want to thank Stefan Neumann, Gramoz Goranci, Bernhard Schuster,
Alexander Svozil, Marcelo Fonseca Faraj, Wolfgang Ost, and Richard Paul for being great office
mates! Thank you to Stefan, Alexander, and Gramoz for helping me find my way in the
wonderful city of Vienna; and to Marcelo and Wolfgang for many fruitful
discussions. I would also like to thank Sebastian Forster, Kathrin Hanauer,
Rudolf Hürner, Sagar Kale, Shahbaz Khan, Ami Paz, Pan Peng, Xiaowei Wu, Vaidehi Srinivas, Ulrike
Frolik-Steffan, Iris Gundacker, and Christina Licayan for being wonderful
colleagues and making these years very enjoyable!

I am very grateful to my friends and family for their endless support during my
years of university and graduate school. I thank my parents Birgitt and Wolfgang
for always being there for me.

Finally, I would like to wholeheartedly thank my partner Anique-Marie Cabardos
for her love and support. Thank you for keeping me happy and motivated and thank
you for your valuable help with proofreading manuscripts and parts of this thesis!

\newpage

The research leading to these results has received funding from the European
Research Council under the European Community's Seventh Framework Programme
(FP7/2007-2013) /ERC grant agreement No. 340506.

Partially supported by DFG grant SCHU 2567/1-2.

Moreover, we gratefully  acknowledge  the  Gauss  Centre for
Supercomputing e.V.  (www.gauss-centre.eu) for funding this project by providing
computing time on the GCS Supercomputer SuperMUC at Leibniz Supercomputing
Centre (www.lrz.de).

We further thank the Vienna Scientific Cluster (VSC) for providing high
performance computing resources.

\end{abstract}

\cleardoublepage
\renewcommand{\abstractname}{Bibliographic Note}
\begin{abstract}
Several results in this thesis were already published in conference and journal
papers and thus the chapters of this thesis are based on the following papers:
\begin{itemize}
    \item \textbf{Chapter 3:} Monika Henzinger, Alexander Noe, Christian Schulz
    and Darren Strash. \emph{``Practical Minimum Cut Algorithms''}. In: \emph{ALENEX.},
    2018, pp. 48--61 \\ \url{https://arxiv.org/abs/1708.06127} \\
    Monika Henzinger, Alexander Noe, Christian Schulz
    and Darren Strash \emph{``Practical Minimum Cut Algorithms''}. In: \emph{ACM
    JEA.}, 2018, Vol. 23, Article 1.8 pp. 1-22 \\ \url{https://doi.org/10.1145/3274662}
    \item \textbf{Chapter 4:} Monika Henzinger, Alexander Noe and Christian
    Schulz. \emph{``Shared-memory Exact Minimum Cuts''}. In: \emph{IPDPS.},
    2019., pp. 13--22 \\ \url{https://arxiv.org/abs/1808.05458}
    \item \textbf{Chapter 5:} Monika Henzinger, Alexander Noe, Christian Schulz
    and Darren Strash. \emph{``Finding All Global Minimum Cuts in Practice''}.
    In \emph{ESA.}, 2020., Article 59, pp. 1--20 \\ \url{https://arxiv.org/abs/2002.06948}
    \item \textbf{Chapter 6:} Monika Henzinger, Alexander Noe and Christian
    Schulz. \emph{``Practical Fully Dynamic Minimum Cut Algorithms''}.
    Manuscript., 2021. \\ \url{https://arxiv.org/abs/2101.05033}    
    \item \textbf{Chapter 7:} Alexandra Henzinger, Alexander Noe and Christian Schulz.
    \emph{``ILP-based Local Search for Graph Partitioning''}. In: \emph{SEA.},
    2018., Article 4, pp. 1--15 \\ \url{https://arxiv.org/abs/1802.07144} \\
    Alexandra Henzinger, Alexander Noe and Christian Schulz.
    \emph{``ILP-based Local Search for Graph Partitioning''}. In: \emph{ACM
    JEA.}, 2020., Vol. 25, Article 9, pp. 1--26 \\  \url{https://doi.org/10.1145/3398634}
    \item \textbf{Chapter 8:} Monika Henzinger, Alexander Noe and Christian
    Schulz. \emph{``Shared-memory Branch-and-reduce for Multiterminal Cuts''}.
    In: \emph{ALENEX.}, 2020., pp. 42--55 \\ \url{https://arxiv.org/abs/1908.04141} \\
    Monika Henzinger, Alexander Noe and Christian Schulz. \emph{``Faster
    Parallel Multiterminal Cuts''}, Manuscript., 2020. \\ \url{https://arxiv.org/abs/2004.11666}
\end{itemize}

Authors appear in alphabetical order in all listed publications.

\end{abstract}

\cleardoublepage
\tableofcontents*

\mainmatter

\chapter{Introduction}

\section{Motivation}

In the last few decades, world-spanning networks have created a plethora of
structured and unstructured data. One very prominent example is the internet,
which has seen the creation and growth of many networks, some of them to immense
scale. This immense scale makes extracting information from the networks a hard
task and necessitates the \emph{partitioning} of networks into smaller, more
readily comprehensible blocks. \emph{Graphs} are a good abstraction to constitute such
networks in a way that is understandable both for humans and machines. In a
graph, we have a set of \emph{vertices}, where each vertex represents an entity, such
as a person, street address or a work package in a computer program. If two
vertices are linked, such as friends in a social network or street
addresses that are connected by a road, they are connected by an \emph{edge}.
This work focuses on \emph{undirected graphs}, \ie edges do not have a direction
and a connection from A to B implies that B is also connected to A. In some
graphs, vertices and edges have \emph{weights}, for example if we have a graph
that depicts a complex program where vertices are subprograms and connections
represent communication, vertex weights indicate the computational complexity of
a subprogram and edge weights indicate communication volume.

\emph{Graph algorithms} aim to solve problems on such a graph. In this work, we
look at various \emph{cut problems} or \emph{partitioning problems}, problems in which we want to partition the set of vertices into two or more subsets. Due to the
large scale of global connections we want to be able to partition them into more
manageable subgraphs. In all of the problems discussed in this dissertation, we
aim to partition the set of vertices in such a way that the total weight of \emph{cut
edges}, \ie edges that connect vertices in different blocks, or number of cut
edges in graphs without edge weights, is minimized. We call the weight sum of
cut edges the \emph{cut size}. This allows the partitioning
of networks in such a way that communication over block boundaries in computing
networks or separated relationships in social networks is as small as possible.

In this dissertation, we look at three important cut problems. In
Part~\ref{p:mincut} we look at the \emph{minimum cut problem} or \emph{global
minimum cut problem} where the aim is to find the smallest cut between two
non-empty blocks of vertices without making any restrictions on the size of
either block. In Part~\ref{p:gp}, we look at the \emph{balanced graph
partitioning problem}. In this problem we aim to partition the vertex set into
$k$ blocks of roughly equal size so that the cut size is minimal.
Part~\ref{p:mtc} deals with the \emph{multiterminal cut problem}, where, given a
set of $k$ vertices called \emph{terminals}, we want to find the smallest cut
that pairwisely separates all terminals. The three parts of the dissertation are
mostly independent; however, some techniques and ideas are shared between
algorithms for different problems. We then give a brief re-introduction in the
latter part and also cross-reference to the previous usage for further details.

We use the methodology and techniques of \emph{algorithm
engineering}~\cite{sanders2009algorithm} to give algorithms which give fast and strong solutions on a wide
variety of different real-world instances but also stand on a sound theoretical
base. In the methodology of algorithm engineering, algorithms are designed and
analyzed using realistic machine models. In contrast to algorithm theory, these
algorithms are then implemented and evaluated using experiments on data from
real-world applications. Based on these experiments, we amend our design and repeat
this inductive cycle until our algorithm is satisfactory. One important
aspect is that the results of the implementation can be published as
algorithm libraries so that other people can use them. As we develop algorithms
for fundamental graph problems in this dissertation, we publish all of our
algorithms under the permissive MIT license so that they can be used as building
blocks for complex systems. The implementations in
Parts~\ref{p:mincut}~and~\ref{p:mtc} are available as the \emph{VieCut} (Vienna
Minimum Cuts) library~\footnote{\url{https://github.com/VieCut/VieCut}}, the
implementations in Part~\ref{p:gp} are integrated into the
KaHIP~\footnote{\url{https://github.com/KaHIP/KaHIP}} graph partitioning
framework~\cite{kabapeE,kaffpaE}. For a detailed description of the methodology
of algorithm engineering we refer the reader to~\cite{sanders2009algorithm}.

For the minimum cut problem and the multiterminal cut problem, we develop and use
a multitude of \emph{local reduction rules} or \emph{kernelization rules}. These
reduction rules are related to the concept of fixed-parameter tractable (FPT)
algorithms, where a hard problem can be solved efficiently as long as some
problem parameter is not too large. FPT algorithms have long been a
well-established field in algorithm theory, however only few of the techniques
are implemented and tested on real datasets, and their practical potential is
far from understood. More recently, the engineering aspect has gained some
momentum. There are several experimental studies in this area that take up ideas
from FPT or kernelization theory, e.g.~for independent sets (or equivalently
vertex
cover)~\cite{akiba-tcs-2016,DBLP:conf/sigmod/ChangLZ17,dahlum2016accelerating,DBLP:conf/alenex/Lamm0SWZ19,DBLP:journals/corr/abs-1908-06795,DBLP:conf/alenex/Hespe0S18},
for cut tree construction\cite{DBLP:conf/icdm/AkibaISMY16}, for treewidth
computations
\cite{bannach_et_al:LIPIcs:2018:9469,DBLP:conf/esa/Tamaki17,koster2001treewidth},
for the feedback vertex set problem
\cite{DBLP:conf/wea/KiljanP18,DBLP:conf/esa/FleischerWY09}, for the dominating
set problem~\cite{10.1007/978-3-319-55911-7_5},  for the maximum cut
problem~\cite{DBLP:journals/corr/abs-1905-10902}, for the cluster editing
problem~\cite{Boecker2011}, and the matching problem~\cite{DBLP:conf/esa/KorenweinNNZ18}. In this dissertation, we make heavy use of
data reduction techniques to improve the performance of algorithms for the
minimum cut problem and the multiterminal cut problem. A recent survey on data
reduction rules in practice is given in~\cite{abu2020recent}. This survey covers
data reduction for the global minimum cut problem and the multiterminal cut
problem, as well as a multitude of other problems. 

\section{Main Contributions and Outline}

This thesis consists of three individual parts, each addressing a fundamental
cut problem. In this section, we give a brief overview where we briefly
introduce the problems and then give the main contributions in this
dissertation. In the introductory sections or chapters of each part we will give
a more detailed outline.

\subsection{Part I: Minimum Cut}

In the first part of this dissertation we study the \emph{(global) minimum cut problem}. This
part is larger than the others, as we give inexact and exact shared-memory
parallel algorithms for the problem, as well as an algorithm that finds all
minimum cuts and an algorithm that maintains a minimum cut on a dynamically
changing graph in which edges are inserted and deleted in arbitrary order. The
minimum cut problem on a graph is to partition the vertices into two non-empty
sets so that the sum of edge weights between the two sets is
minimized. The minimum cut problem is one of the most fundamental graph problems
and has seen a large amount of research. In Chapter~\ref{c:mincut}, we give a
brief overview over this research and introduce in a bit more detail algorithms that we use in the
following chapters. We first give a practical shared-memory
parallel heuristic algorithm in Chapter~\ref{c:viecut}. This algorithm
repeatedly reduces the input graph size with both heuristic and exact techniques
by identifying and contracting edges that are likely or provably not part of a
minimum cut. It is significantly faster than existing algorithms and has a lower
empirically observed error rate than other inexact algorithms. Based on this
inexact algorithm and practically efficient parallelization of an existing
sequential algorithm, in Chapter~\ref{c:exmc}, we then give a shared-memory
parallel \emph{exact} algorithm that provably finds a minimum cut for large
graphs. Using $12$ cores, this algorithm outperforms the state-of-the-art for exact
minimum cut algorithms by a factor of up to $12.9$ on some graphs.

In Chapter~\ref{c:allmc} we follow that up with an exact shared-memory parallel
algorithm that finds \emph{all} global minimum cuts in a graph and returns a
compact \emph{cactus graph} data structure that represents them all. This
algorithm is able to solve instances with more than a billion edges and millions
of minimum cuts in a few minutes on a single shared-memory parallel machine. We
also give a new linear-time algorithm that, given a cactus graph data structure
that represents all minimum cuts, gives the most balanced minimum cut.

Chapter~\ref{c:dynmc} then details our algorithm
that maintains a global minimum cut on a dynamically changing graph under edge
insertions and deletions. As an edge insertion increases the value of some cuts
but leaves most cuts untouched, it is useful to have a data structure with all
minimum cuts, so that we only remove the minimum cuts whose value changed and
retain all others without expensive recomputation. Our dynamic algorithm
outperforms existing static algorithms by up to multiple orders of magnitude. 
While there have been various theoretical algorithms for finding all minimum
cuts in a graph as well as for maintaining the minimum cut on a dynamically
changing graph, to the best of our knowledge, our algorithms are the first
publically available implementations for these problems.

\subsection{Part II: Balanced Graph Partitioning}

In the second part of this dissertation, we study the \emph{balanced graph partitioning
problem}. The balanced graph partitioning problem on an undirected graph with
positive vertex and edge weights is to partition the vertex set into $k \geq 2$
blocks so that every block has roughly the same sum of contained node weights.
More precisely, every block has a weight limit of $(1 + \epsilon)$ times the
average block weight, \ie the sum of all node weights in the graph divided by
the number of blocks, for a given $\epsilon \geq 0$. In this dissertation, we
present a novel meta-heuristic for the balanced graph partitioning problem. Our
approach is based on integer linear programs that solve the partitioning problem
to optimality. However, since those programs typically do not scale to large
inputs, we adapt them to heuristically improve a given partition. We do so by
defining a much smaller model that allows us to use symmetry breaking and other
techniques that make the approach scalable. For example, in Walshaw’s well-known
benchmark tables~\cite{wswebsite}, we are able to improve roughly half of all
entries when the number of blocks is high. Additionally, we include our
techniques in a memetic framework~\cite{kabapeE} and develop a crossover operation based on
the proposed techniques. This extended evolutionary algorithm produces
high-quality partitions from scratch. For half of the hard problems from
Walshaw's graph partitioning benchmark, the result of our algorithm is at least
as good as the previous best result. For $17\%$, the solution given is better
than the previous best solution.

\subsection{Part III: Multiterminal Cut Problem}

In the third and final part of this dissertation we study the \emph{multiterminal
cut problem}. The multiterminal cut problem, given an undirected graph with positive
edge weights and a set of $k$ terminal vertices, is to
partition the vertex set into $k$ blocks so that each block contains exactly one
terminal vertex. We present a fast shared-memory parallel exact algorithm for
the multiterminal cut problem. In particular, we engineer existing as well as
new efficient data reduction rules to transform the graph into a smaller
equivalent instance. We use these reduction rules within a branch-and-reduce
framework and combine this framework with an integer linear programming solver
to give an algorithm that can solve a wide variety of large instances.
Additionally, we present an inexact heuristic algorithm that gives high-quality
solutions for very hard instances in reasonable time. Among other techniques, we
use local search to significantly improve a given solution to the problem. Our
algorithms achieve improvements in running time of up to multiple orders of
magnitude over the ILP formulation without data reductions.

\part{The (Global) Minimum Cut Problem}
\label{p:mincut}

\chapter{Minimum Cut}
\label{c:mincut}
% Introduction

\section{Introduction} 

Given an undirected graph with non-negative edge weights, the \emph{minimum cut
problem} is to partition the vertices into two sets so that the sum of edge
weights between the two sets is minimized. An edge that crosses the partition
boundary is called a \emph{cut edge}. A cut that minimizes the weight sum
of cut edges for all possible cuts is called the \emph{minimum cut} or
\emph{global minimum cut} of the graph. In graphs where each edge has unit
weight, a minimum cut is often also referred to as the \textit{edge
connectivity} of a graph~\cite{nagamochi1992computing,henzinger2017local}. A
variant of the minimum cut problem is the problem of finding \emph{all global
minimum cuts} in a graph.

The minimum cut problem has
applications in many fields. In particular, for network
reliability~\cite{karger2001randomized,ramanathan1987counting}, assuming equal
failure chance on edges, the smallest edge cut in the network has the highest
chance to disconnect the network; in VLSI
design~\cite{krishnamurthy1984improved}, a minimum cut can be used to minimize
the number of connections between microprocessor blocks; and it is further used
as a subproblem in the branch-and-cut algorithm for solving the Traveling
Salesman Problem and other combinatorial problems~\cite{padberg1991branch}.
Minimum cuts in similarity graphs can be used to find
clusters~\cite{wu1993optimal,hartuv2000clustering}. In community detection, the
absence of a small cut inside a cluster can indicate a likely community in a
social network~\cite{cai2005mining}. In graph drawing~\cite{kant1993algorithms},
minimum cuts are used to separate the network. Finding all minimum cuts is an
important subproblem for edge-connectivity augmentation
algorithms~\cite{gabow1991applications,naor1997fast}.

Part~\ref{p:mincut} of this dissertation is based on our papers on the global
minimum cut problem. This chapter gives a brief overview of preliminaries
and related work. In Section~\ref{p:mincut:s:preliminaries}, we will 
introduce the notation and preliminaries used throughout this part of the
dissertation. We give an overview of related work on the minimum cut problem
and related problems in Section~\ref{p:mincut:s:related}. We aim to give
a general overview of algorithms and research and give some more detail about some
of the algorithms and techniques used in later chapters of this part. We then
give a fast heuristic shared-memory parallel algorithm for the global minimum
cut problem in Chapter~\ref{c:viecut} and based on this work, an exact
shared-memory parallel algorithm in Chapter~\ref{c:exmc}. In
Chapter~\ref{c:allmc}, we give an algorithm that finds all minimum cuts in a
graph and gives their compact cactus graph representation. We use this cactus
graph representation to maintain the global minimum cut in a dynamic graph, \ie
a graph in which edges are deleted and inserted over time. This dynamic
algorithm is given in Chapter~\ref{c:dynmc}.

\section{Preliminaries}
\label{p:mincut:s:preliminaries}

Let $G = (V, E, c)$ be a weighted undirected simple graph with vertex set $V$,
edge set $E \subset V \times V$ and non-negative edge weights $c: E \rightarrow
\MdN$. We extend $c$ to a set of edges $E' \subseteq E$ by summing the weights
of the edges; that is, let $c(E')\Is \sum_{e=(u,v)\in E'}c(u,v)$ and let $c(u)$
denote the sum of weights of all edges incident to vertex $v$. Let $n = |V|$ be
the number of vertices and $m = |E|$ be the number of edges in $G$. The
\emph{neighborhood} $N(v)$ of a vertex $v$ is the set of vertices adjacent to
$v$. The \emph{weighted degree} of a vertex is the sum of the weights of its
incident edges. For brevity, we simply call this the \emph{degree} of the
vertex. For a set of vertices $A\subseteq V$, we denote by $E[A]\Is \{(u,v)\in E
\mid u\in A, v\in V\backslash A\}$; that is, the set of edges in $E$ that start
in $A$ and end in its complement. A cut $(A, V \backslash A)$ is a partitioning
of the vertex set $V$ into two non-empty \emph{partitions} $A$ and $V \backslash
A$, each being called a \emph{side} of the cut. The \emph{capacity} or
\emph{weight} of a cut $(A, V \backslash A)$ is $c(A) = \sum_{(u,v) \in E[A]}
c(u,v)$. A \emph{minimum cut} is a cut $(A, V \backslash A)$ that has smallest
capacity $c(A)$ among all cuts in $G$. For two non-overlapping vertex sets $A
\subset V$ and $B \subset V$, the capacity of the cut $c(A,B) = \sum_{(u,v) \in
E, u \in A, v \in B} c(u,v)$ is the weight of all edges that connect vertices in
$A$ with vertices in $B$.

We use $\lambda(G)$ (or simply $\lambda$, when its meaning is clear) to denote
the value of the minimum cut over all non-empty $A \subset V$. For two vertices
$s$ and $t$, we denote $\lambda(G,s,t)$ as the capacity of the smallest cut of
$G$, where $s$ and $t$ are on different sides of the cut. $\lambda(G,s,t)$ is
also known as the \emph{minimum s-t-cut} of the graph. $\lambda(G,s,t)$ is also
called the \emph{connectivity} of vertices $s$ and $t$. The connectivity
$\lambda(G,e)$ of an edge $e=(s,t)$ is defined as $\lambda(G,s,t)$, the
connectivity of its incident vertices. At any point in the execution of a
minimum cut algorithm, $\hat\lambda(G)$ (or simply $\hat\lambda$) denotes the
smallest upper bound of the minimum cut that the algorithm discovered up to that
point. For a vertex $u \in V$, the size of the \emph{trivial cut} $(\{u\},
V\backslash \{u\})$ is equal to the vertex degree of $u$. For most minimum cut
algorithms, $\hat\lambda(G)$ is initially set to the value of the minimum degree in $G$,
as this is the weight of the trivial cut which separates the minimum degree
vertex from the rest of the vertex set. When \emph{clustering} a graph, we are
looking for \emph{blocks} of nodes $V_1$,\ldots,$V_k$ that partition $V$, that
is, $V_1\cup\cdots\cup V_k=V$ and $V_i\cap V_j=\emptyset$ for $i\neq j$. The
parameter $k$ is usually not given in advance.

Many algorithms for the minimum cut problem use \emph{graph contraction}.
Given an edge $e = (u, v) \in E$, we define $G/(u, v)$ (or $G/e$) to be the
graph after \emph{contracting edge} $(u, v)$. In the contracted graph, we delete
vertex $v$ and all edges incident to this vertex. For each edge $(v, w) \in E$,
we add an edge $(u, w)$ with $c(u, w) = c(v, w)$ to~$G$ or, if the edge already
exists, we give it the edge weight $c(u,w) + c(v,w)$. Given an edge $e \in (V
\times V) \backslash E$, we define $G+e$ to be the graph after \emph{inserting}
edge $e$ and given an edge $e \in E$ we define $G-e$ to be the graph after
\emph{deleting} edge $e$.

A graph with $n$ vertices can have up to $\Omega(n^2)$ minimum
cuts~\cite{karger2000minimum}. To see that this bound is tight, consider an
unweighted cycle with $n$ vertices. Each set of $2$ edges in this cycle is a
minimum cut of $G$. This yields a total of $\binom{n}{2}$ minimum cuts. However,
all minimum cuts of an arbitrary graph $G$ can be represented by a cactus graph
$C_G$ with up to $2n$ vertices and $\Oh{n}$ edges~\cite{nagamochi2000fast}. A
cactus graph is a connected graph in which any two simple cycles have at most
one vertex in common. In a cactus graph, each edge belongs to at most one simple
cycle. 

To represent all minimum cuts of a graph $G$ in an edge-weighted cactus graph
$C_G = (V(C_G), E(C_G))$, each vertex of $C_G$ represents a possibly empty set
of vertices of $G$ and each vertex in $G$ belongs to the set of one vertex in
$C_G$. Let $\Pi$ be a function that assigns to each vertex of $C_G$ a set of
vertices of $G$. Then every cut $(S, V(C_G) \backslash S)$ corresponds to a
minimum cut $(A, V \backslash A)$ in $G$ where $A=\cup_{x\in S} \Pi(x)$. In
$C_G$, all edges that do not belong to a cycle have weight $\lambda$ and all
cycle edges have weight $\frac{\lambda}{2}$. A minimum cut in $C_G$ consists of
either one tree edge or two edges of the same cycle. We denote by $n^*$ the
number of vertices in $C_G$ and $m^*$ the number of edges in $C_G$. The weight
$c(v)$ of a vertex $v \in C_G$ is equal to the number of vertices in $G$ that
are assigned to $v$.

\section{Related Work}
\label{p:mincut:s:related}

We now review algorithms for the global minimum cut and related problems. A
closely related problem is the \textit{minimum s-t-cut} problem, which asks for
a minimum cut with nodes $s$ and $t$ in different partitions. Ford and
Fulkerson~\cite{ford1956maximal} proved that minimum $s$-$t$-cut is equal to
maximum $s$-$t$-flow. Gomory and Hu~\cite{gomory1961multi} observed that the
(global) minimum cut can be computed with $n-1$ minimum $s$-$t$-cut
computations. For the following decades, this result by Gomory and Hu was used
to find better algorithms for global minimum cut using improved maximum flow
algorithms~\cite{karger1996new}. One of the fastest known maximum flow
algorithms is the push-relabel algorithm~\cite{goldberg1988new} by
Goldberg and Tarjan, which computes a maximum $s$-$t$-flow in
$\Oh{mn\log{\frac{n^2}{m}}}$. Using their algorithm to find maximum
$s$-$t$-flows, the algorithm of Gomory and Hu finds a global minimum cut in
$\Oh{mn^2\log{\frac{n^2}{m}}}$.

Hao and Orlin~\cite{hao1992faster} adapt the push-relabel algorithm to pass
information to future flow computations. When an iteration of the push-relabel algorithm is
finished, they implicitly merge the source and sink vertices to form a new sink and find
a new source vertex. Vertex heights are maintained over multiple iterations of
push-relabel. With these techniques they achieve a total running time of
$\Oh{mn\log{\frac{n^2}{m}}}$ for a graph with $n$ vertices and $m$ edges, which is
asymptotically equal to a single run of the push-relabel algorithm.

Padberg and Rinaldi~\cite{padberg1990efficient} give a set of heuristics to find
edges which can be contracted without affecting the minimum cut.
Chekuri~\etal\cite{Chekuri:1997:ESM:314161.314315} give an implementation of
these heuristics that can be performed in time linear in the graph size. Using
these heuristics it is possible to sparsify a graph while preserving at least
one minimum cut in the graph. In Section~\ref{p:mincut:ss:pr} we outline their
results, as our algorithms for the minimum cut problem make use of them.

Nagamochi \etal\cite{nagamochi1992computing,nagamochi1994implementing} give a
minimum cut algorithm which does not use any flow computations. Instead, their
algorithm uses maximum spanning forests to find a non-empty set of contractible
edges. This contraction algorithm is run until the graph is contracted into a
single node. The algorithm has a running time of $\Oh{mn+n^2\log{n}}$. As our
exact algorithm is partially based on their contraction routine, we summarize their
results in Section~\ref{p:mincut:ss:noi}. Wagner and
Stoer~\cite{stoer1997simple} give a simpler variant of the algorithm of
Nagamochi, Ono and Ibaraki~\cite{nagamochi1994implementing}, which has the
same asymptotic time complexity. The performance of this algorithm on real-world
instances, however, is significantly worse than the performance of the
algorithms of Nagamochi, Ono and Ibaraki or Hao and Orlin, as shown
independently in experiments conducted by J\"unger
\etal\cite{junger2000practical} and
Chekuri~\etal\cite{Chekuri:1997:ESM:314161.314315}. In fact, both the
algorithms of Hao and Orlin or Nagamochi, Ono and Ibaraki achieve close to
linear running time on most benchmark
instances~\cite{junger2000practical,Chekuri:1997:ESM:314161.314315}. Based on
the algorithm of Nagamochi, Ono and Ibaraki, Matula~\cite{matula1993linear}
gives a $(2+\varepsilon)$-approximation algorithm for the minimum cut problem.
The algorithm contracts more edges than the algorithm of Nagamochi, Ono and
Ibaraki to guarantee a linear time complexity while still guaranteeing a
$(2+\varepsilon)$-approximation~factor. 

Based on the observations that the contraction of an edge not in a minimum cut
does not affect the value of said cut and that a minimum cut contains by
definition only a small fraction of the edge set, Karger~\cite{karger1993global}
gives a simple algorithm that contracts random edges until the graph has only
two vertices left and then evaluates the cut value between them. They prove that
by repeating this process $\Oh{n^2 \log{n}}$ times, the contraction algorithm
finds a minimum cut with high probability. Thus, one can find a minimum cut in
$\Oh{mn^2 \log{n}}$ in unweighted and $\Oh{mn^2 \log^3{n}}$ in weighted graphs
with high probability. Karger and Stein~\cite{karger1996new} show that minimum
cut edges are contracted more often near the end of the contraction routine when
the graph has only few vertices left. Their random contraction algorithm
contracts a small set of edges, recurses twice and continues the contraction in
both subproblems. Therefore the later stages are performed more often and the
recursive contraction process only needs to be performed $\Oh{\log^2{n}}$ times
to find a minimum cut with high probability. This algorithm finds a minimum cut
with high probability in $\Oh{n^2 \log^3{n}}$ and was the first algorithm to
break the $\Otilde{mn}$ barrier. The $\Otilde{}$ notation ignores logarithmic
factors. Gianinazzi~\etal~\cite{gianinazzi2018communication} give a parallel
implementation of the algorithm of Karger and Stein. Other than that, there are
no parallel implementation of either algorithm known to us. More recently, the
randomized contraction-based algorithm of
Ghaffari~\etal\cite{ghaffari2020faster} solves the minimum cut problem on
unweighted graphs in $\Oh{m \log{n}}$ or $\Oh{m+n \log^3{n}}$.

Kawarabayashi and Thorup~\cite{kawarabayashi2015deterministic} give a
deterministic near-linear time algorithm for the minimum cut problem on
unweighted graphs, which runs in $\Oh{m \log^{12}{n}}$. Their algorithm works by
growing contractible regions using a variant of
PageRank~\cite{page1999pagerank}. It was improved by
Henzinger~\etal\cite{henzinger2017local} to run in $\Oh{m \log^2{n} \log \log^2
n}$ time, which is the currently fastest deterministic algorithm on unweighted
graphs. Li and Panigrahi~\cite{li2020maxflows} give a deterministic algorithm that finds a
global minimum cut on weighted graphs in $\Oh{m^{1+\epsilon}}$ plus
poly-logarithmic maximum flows for any constant $\epsilon > 0$.

Another approach to the global minimum cut problem is \emph{tree packing}.
Nash-Williams~\cite{nash1961edge} proves that any graph with minimum cut
$\lambda$ contains a set of $\lambda / 2$ edge-disjoint spanning trees. Such a
tree packing can be found using Gabow's algorithm~\cite{gabow1995matroid} in
$\Oh{m\lambda \log{n}}$. Karger~\cite{karger2001randomized} introduces the
concept of \emph{$k$-respecting cuts}, where a cut $k$-respects a tree if it
only cuts up to $k$ tree edges. In his algorithm,
Karger~\cite{karger2001randomized} finds a set of $\Oh{\log n}$ spanning trees
so that the minimum cut $1$- or $2$-respects any of them with high probability.
For each of the spanning trees, the algorithm computes the minimum cut that $1$-
or $2$-respects it. This algorithm finds a minimum cut with high
probability in $\Oh{m \log^3{n}}$.
Gawrychowski~\etal\cite{gawrychowski2020minimum} improve the running time of
this algorithm to $\Oh{m \log^2{n}}$, which is the currently fastest algorithm
for the global minimum cut problem on weighted graphs.
Bhardwaj~\etal\cite{DBLP:conf/swat/BhardwajLS20} give a simpler
tree-packing-based algorithm with a running time of $\Oh{m \log^3{n}}$ -- matching the
algorithm of Karger~\cite{karger2001randomized} -- and implement a version with
a running time of $\Oh{m \log^4{n}}$. This implementation compares favorably
against the algorithms of Karger and Stein~\cite{karger2000minimum} and Stoer
and Wagner~\cite{stoer1997simple}, however they do not compare their algorithm
to algorithms that outperformed these by up to multiple orders of magnitudes in
other experimental
evaluations~\cite{junger2000practical,Chekuri:1997:ESM:314161.314315}, such as
the algorithms of Nagamochi~\etal\cite{nagamochi1994implementing} or the
algorithm of Hao and Orlin~\cite{hao1992faster}. Mukhopadhyay and
Nanongkai~\cite{DBLP:conf/stoc/MukhopadhyayN20} give an algorithm to find a
minimum $2$-respecting cut in $\Oh{m \log n + n \log^4{n}}$. They also give a
streaming variant of their algorithm that requires $\Otilde{n}$ space and
$\Oh{\log n}$ passes to compute the global minimum cut. Recently,
Li~\cite{li2020deterministic} gave a deterministic algorithm using the
techniques of Karger that finds a minimum cut in weighted graphs
in $\Oh{m^{1+o(1)}}$.

Recently, Georgiadis~\etal\cite{georgiadis2021experimental} carried out an experimental study of
global minimum cut algorithms on directed graphs. Their experimental study shows
that the directed version of Gabow's algorithm~\cite{gabow1995matroid} performs
well in practice; and for graphs with a low minimum cut value $\lambda$, local
search based algorithms~\cite{chechik2017faster,forster2020computing} also
perform well.

\subsection{Finding \emph{All} Global Minimum Cuts}

Even though a graph can have up to $\binom{n}{2}$ minimum
cuts~\cite{karger2000minimum}, there is a compact representation of all minimum
cuts of a graph called \emph{cactus graph} with $\Oh{n}$ vertices and edges, as
described earlier in Section~\ref{p:mincut:s:preliminaries}.
Karzanov and Timofeev~\cite{karzanov1986efficient} give the first polynomial
time algorithm to construct the cactus representation for all minimum cuts.
Picard and Queyranne~\cite{picard1980structure} show that all minimum cuts
separating two specified vertices can be found from a single maximum flow between them.
Thus, similar to the classical algorithm of Gomory and Hu~\cite{gomory1961multi}
for the minimum cut problem, we can find all minimum cuts in $n-1$ maximum flow
computations. The algorithm of Karzanov and
Timofeev~\cite{karzanov1986efficient} combines all those minimum cuts into a
cactus graph representing all minimum cuts. Nagamochi and
Kameda~\cite{nagamochi1994canonical} give a representation of all minimum cuts
separating two vertices $s$ and $t$ in a so-called $(s,t)$-cactus
representation. Based on this $(s,t)$-cactus representation,
Nagamochi~\etal\cite{nagamochi2000fast} give an algorithm that finds all minimum
cuts and gives the minimum cut cactus in $\Oh{nm + n^2 \log{n} + n^*m\log{n}}$,
where $n^*$ is the number of vertices in the cactus.
Fleischer~\cite{fleischer1999building} gives an algorithm based on the flow
algorithm of Hao and Orlin that gives the cactus representation of all minimum
cuts in a graph in the same asymptotic running time, $\Oh{mn\log{\frac{n^2}{m}}}$.

The aforementioned recursive contraction algorithm of Karger and Stein~\cite{karger1996new}
above not only finds a single minimum cut, but is able to find all minimum
cuts of a graph in $\Oh{n^2 \log^3{n}}$ with high probability. Based on
the algorithm of Karzanov and Timofeev~\cite{karzanov1986efficient} and its
parallel variant given by Naor and Vazirani~\cite{naor1991representing}, they
show how to give the cactus representation of the graph in the same asymptotic
time. Likewise, the recent algorithm of Ghaffari~\etal\cite{ghaffari2020faster}
finds all \emph{non-trivial minimum cuts} (\ie minimum cuts where each side
contains at least two vertices) of a simple unweighted graph in $\Oh{m
\log^2{n}}$ time. Using the techniques of Karger and Stein, the algorithm can
trivially give the cactus representation of all minimum cuts in $\Oh{n^2
\log{n}}$. 

While there are implementations of the algorithm of Karger and
Stein~\cite{Chekuri:1997:ESM:314161.314315,gianinazzi2018communication}
for the minimum cut problem, to the best of our knowledge there are no published
implementations of either of the algorithms to find the cactus graph
representing all minimum cuts.

A closely related problem is the cut tree problem (or Gomory-Hu tree problem),
which aims to find a tree $T = (V,E_T,c_T)$, such that for each two vertices
$u,v \in V$, the weight of the minimum $u$-$v$-cut is equal to the lightest edge
weight on the unique path from $u$ to $v$ on $T$. This problem was first solved
by Gomory and Hu~\cite{gomory1961multi} using $n-1$ minimum $s$-$t$-cuts and has
been studied experimentally by Goldberg and
Tsioutsiouliklis~\cite{goldberg2001cut} and Akiba~\etal\cite{akiba2016cut}, who
solve the cut tree problem for graphs with millions of vertices and up to one
billion edges in a few hours. Hartmann and
Wagner~\cite{DBLP:conf/isaac/HartmannW12} give a fully-dynamic algorithm to
construct and maintain a cut tree under edge insertions, deletions, and weight
changes.

\subsection{Dynamic Minimum Cut}

The field of dynamic graph algorithms~\cite{eppstein1999dynamic} gives
algorithms that maintain a solution to a graph problem on \emph{dynamic graphs},
\ie graphs that are undergoing updates such as the insertion or deletion of
edges in the graph. A dynamic algorithm allows an efficient update of the
solution instead of recomputing the solution from scratch. An algorithm performs
an \emph{update} when an edge is inserted or deleted and a \emph{query} when we
ask for a solution, \eg the value of the minimum cut on the graph. A dynamic
graph algorithm is called \emph{incremental} if edges are only inserted and
\emph{decremental} if edges are only deleted. If edges are both inserted and
deleted, we call the algorithm \emph{fully dynamic}.

Henzinger~\cite{henzinger1995approximating} gives the first incremental minimum
cut algorithm, which maintains the exact minimum cut with an amortized update
time of $\Oh{\lambda \log{n}}$ per edge insertion and query time of $\Oh{1}$.
The algorithm of Henzinger maintains the cactus graph of all minimum cuts and
invalidates minimum cuts whose weight was increased due to an edge insertion. If
there are not remaining minimum cuts, the algorithm recomputes all minimum cuts
from scratch. Goranci~\etal\cite{goranci2018incremental} manage to remove the
dependence on $\lambda$ from the update time and give an incremental algorithm
with $\Oh{\log^3{n} \log \log^2{n}}$ amortized time per edge insertion and
$\Oh{1}$ query time. They combine techniques of the incremental minimum cut
algorithm of Henzinger with the quasi-linear static minimum cut algorithms of
Kawarabayashi and Thorup~\cite{kawarabayashi2015deterministic} and
Henzinger~\etal\cite{henzinger2017local}.

For minimum cut values up to polylogarithmic size, Thorup~\cite{thorup2007fully}
gives a fully dynamic algorithm with $\Otilde{\sqrt{n}}$ worst-case time per
edge update. The algorithm of Thorup uses tree packing similar to the static
algorithm of Karger~\cite{karger2001randomized}. Note that all of these
algorithms are limited to unweighted graphs. For \emph{planar} graphs with
arbitrary edge weights, {\L}{\k{a}}cki and Sankowski~\cite{lkacki2011min} give a
fully-dynamic algorithm with $\Oh{n^{5/6}\log^{5/3}{n}}$ time per update and
query. To the best of our knowledge, there exists no implementation of any of
these algorithms.

\section[Further Detail on Some Algorithms]{Further Detail on Some Algorithms for the Minimum Cut Problem}
\label{p:mincut:s:algs}

\begin{figure}
  \includegraphics[width=\textwidth]{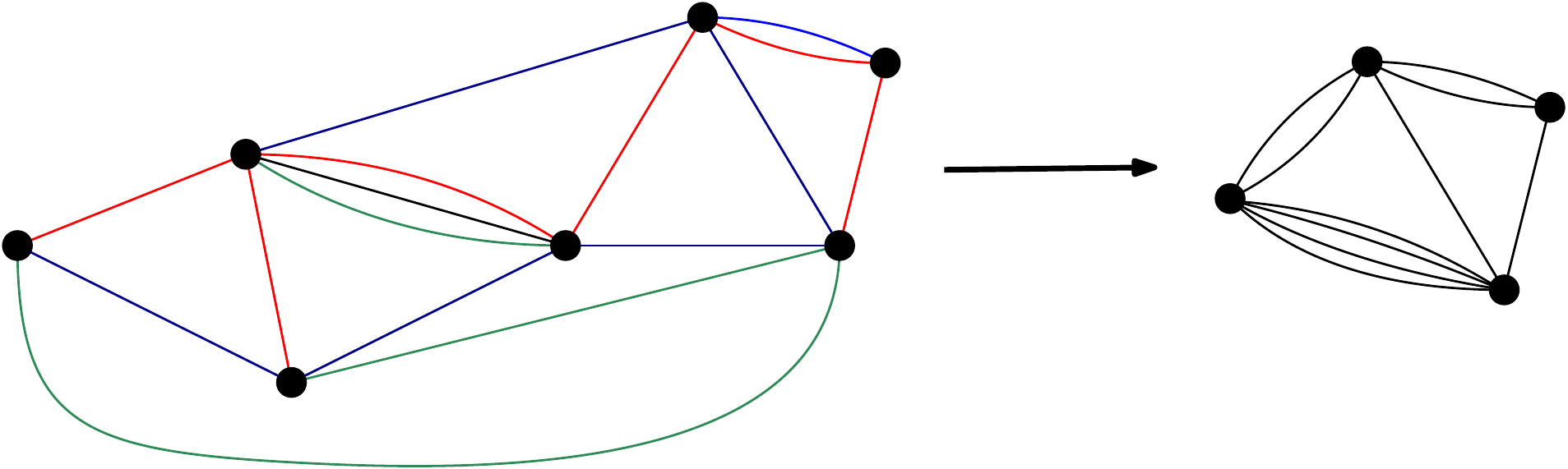}
  \caption{Left: Partition of edge-set into edge-disjoint spanning forests.
  Right: resulting graph after contracting black and green spanning forests.\label{fig:capforest}}
\end{figure}

\subsection{Algorithm of Nagamochi, Ono and Ibaraki}
\label{p:mincut:ss:noi}

We discuss the algorithm by Nagamochi, Ono and
Ibaraki~\cite{nagamochi1992computing,nagamochi1994implementing} in greater
detail since our work relies heavily on their results. The minimum cut algorithm
of Nagamochi~\etal works on graphs with positive integer weights. The intuition
behind the algorithm is as follows: imagine you have an unweighted graph with
minimum cut value exactly one. Then any spanning tree must contain at least one
edge of each of the minimum cuts. Hence, after computing a spanning tree, every
remaining edge can be contracted without losing the minimum cut. Nagamochi~\etal
extend this idea to the case where the graph can have edges with positive weight
as well as the case in which the minimum cut is bounded by $\hat \lambda$. The
first observation is the following: assume that you already found a cut in the
current graph of size~$\hat \lambda$ and you want to find out whether there is a
cut of size~$< \hat \lambda$. Then the contraction process only needs to ensure
that the contracted graph contains all cuts having a value \emph{strictly}
smaller than~$\hat \lambda$. To do so, Nagamochi~\etal build \emph{edge-disjoint
maximum spanning forests} and contract all edges that are not in one of
the~$\hat\lambda - 1$ first spanning forests, as those connect vertices that
have connectivity of at least $\hat\lambda$. Note that the edge-disjoint maximum
spanning forest \emph{certifies} for any edge $e=(u,v)$ that is not in the
forest that the minimum cut between $u$ and $v$ is at least $\hat \lambda$.
Hence, the edge can be ``safely'' contracted. As weights are integral, this
guarantees that the contracted graph still contains all cuts that are
\emph{strictly} smaller than~$\hat\lambda$. Figure~\ref{fig:capforest} shows a
small graph where the edge set is partitioned into edge-disjoint maximum
spanning forests. For this, an edge $e$ of weight $c(e)$ is replaced with $c(e)$
unweighted edges. The first two spanning forests (red, blue) are trees, the
subsequent ones (green, black) are not. As the minimum vertex degree is $3$, the upper bound
for the minimum cut $\hat\lambda = 3$. For each green and black edge $e =
(u,v)$, we can find a path from $u$ to $v$ that only consists of red edges and
one that only consists of blue edges. Thus, the connectivity $\lambda(u,v) \geq
3$ and no cut of value $< 3$ can separate $u$ and $v$. Using this information,
we can contract all green and black edges and repeat this process on the
resulting graph until there are only two vertices left.

Since it would be inefficient to directly compute $\hat \lambda - 1$ edge
disjoint maximum spanning trees and the running time would then depend on the
value of the minimum cut $\lambda$, the authors give a modified algorithm
CAPFOREST to be able to detect contractable edges faster. This is done by
computing a lower bound for the connectivity of the endpoints of an edge which
serves as a certificate for an edge to be contractable. If the lower bound for an
edge $e$ is $\geq \hat\lambda$, then $e$ can be contracted, as no cut smaller
than $\hat\lambda$ contains it. The minimum cut algorithm
of Nagamochi~\etal has a worst
case running time of $\Oh{mn+n^2\log n}$. In experimental
evaluations~\cite{Chekuri:1997:ESM:314161.314315,junger2000practical,henzinger2018practical},
it is one of the fastest exact minimum cut algorithms, both on real-world and
generated instances.

We now take a closer look at details of the algorithm. To find contractable
edges, the algorithm uses a modified \emph{breadth-first graph traversal} (BFS)
algorithm CAPFOREST. The CAPFOREST algorithm starts at an arbitrary vertex. In
each step, the algorithm visits (scans) the vertex $v$ that is most strongly
connected to the already visited vertices. For this purpose, a priority queue
$\mathcal{Q}$ is used, in which the connectivity strength of each vertex $r: V
\to \MdN$ to the already discovered vertices is used as a key. When scanning a
vertex $v$, the value $r(w)$ is kept up to date for every unscanned neighbor $w$
of $v$ by setting \ie $r(w) := r(w) + c(e)$. Moreover, for each edge $e =
(v,w)$, the algorithm computes a lower bound~$q(e)$ for the connectivity, \ie
the smallest cut $\lambda(G, v, w)$, which places $v$ and $w$ on different sides
of the cut. To be precise, it is set to the connectivity strength of $w$ to the
already scanned vertices $q(e) := r(w)$. The vertices are
scanned in an order such that the next scanned vertex is the unscanned vertex
with the highest connection strength value $r$ (the order used by the
algorithm). Using this order,
Nagamochi~\etal\cite{nagamochi1994implementing,nagamochi1992computing} show that
$r(w)$ is a lower bound on $\lambda(G,v,w)$. The order in which
the vertices are scanned is important for the correctness of the algorithm.

For an edge that has connectivity $\lambda(G,v,w) \geq \hat\lambda$, we know
that there is no cut smaller than $\hat\lambda$ that places $v$ and $w$ in
different partitions. If an edge $e$ is not in a given cut $(A,V \backslash A)$,
it can be contracted without affecting the cut. Thus, we can contract edges with
connectivity of at least $\hat\lambda$ without losing any cuts smaller than
$\hat\lambda$. As $q(e) \leq \lambda(G,u,v)$ (lower bound), all edges with $q(e)
\geq \hat\lambda$ are contracted.

Afterwards, the algorithm continues on the contracted graph. A single iteration
of the subroutine can be performed in $\Oh{m+n \log n}$ time. The authors show that
in each BFS run, at least one edge of the graph can be
contracted~\cite{nagamochi1992computing}. This yields a total running time of
$O(mn+n^2 \log n)$. However, in practice the number of iterations is typically
much less than $n-1$, rather it is often proportional to $\log n$.

\subsection{Exact Reductions by Padberg and Rinaldi}
\label{p:mincut:ss:pr}
Padberg and Rinaldi~\cite{padberg1990efficient} give conditions that allow for
shrinking the size of a graph. They prove the following lemma which allows the
contraction of an edge $e = (u,v)$.

\begin{lemma}~[Padberg and Rinaldi~\cite{padberg1990efficient}, Corollary 2.2] \label{lem:pr-nph}  
  Let $u \neq v \in V$. If there exists $Y \subseteq N(u) \cap N(v)$ so that \\
  (a) $c(u) \leq 2c(\{u\},Y+\{v\}-T)$, or \\
  (b) $c(v) \leq 2c(\{v\},T+\{u\})$ \\
  holds for all $T \subseteq Y$, then either $c(u)$ or $c(v)$ is a minimum cut
  or there exists a minimum cut $X$,$Y$ such that both $u \in X$ and $v \in X$.
\end{lemma}

If Lemma~\ref{lem:pr-nph} holds for an edge $e=(u,v)$, it can be contracted since
the trivial cuts $c(u)$ and $c(v)$ were already evaluated and an edge that is
not part of a minimum cut can be contracted without affecting the value of said
minimum cut. Unfortunately, checking Lemma~\ref{lem:pr-nph} is NP-complete in
general (\cite{padberg1990efficient}, Remark 2.3) as the knapsack problem can be
reduced to the problem. It is thus not feasible to check Lemma~\ref{lem:pr-nph}
for every edge, especially not for edges whose incident vertices have a large
shared neighborhood. In their work, Padberg and Rinaldi give a set of conditions
that follow from Lemma~\ref{lem:pr-nph} and can be checked faster. These
conditions are given in Lemma~\ref{lem:pr14}~and~Figure~\ref{fig:local}.

\begin{figure*}[t!]
  \centering
  \includegraphics[width=\textwidth]{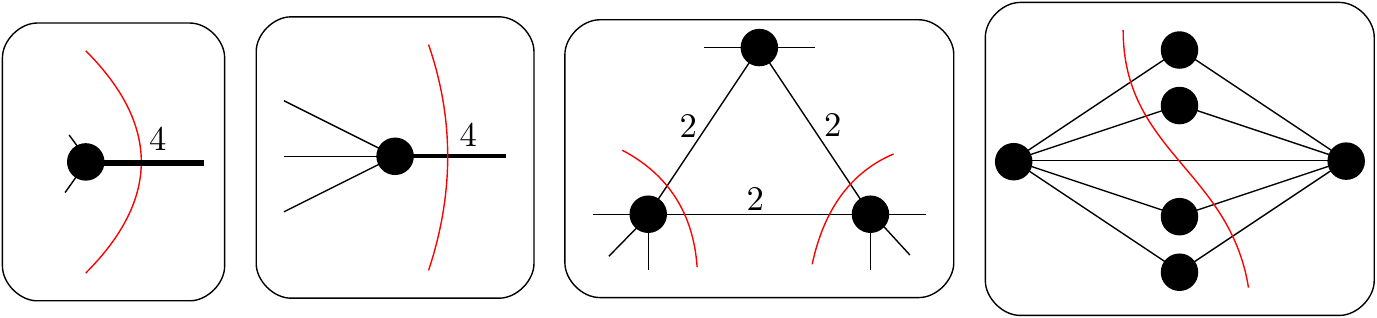}
  \caption{Reductions of Padberg and Rinaldi~\cite{padberg1990efficient}}
  \label{fig:local}
\end{figure*}

\begin{lemma}~[Padberg and Rinaldi~\cite{padberg1990efficient}]
  If two vertices $v,w \in V$ with an edge $(v,w) \in E$ satisfy at least one of
  the following four conditions and $(v,w)$ is not the only edge adjacent to
  either $v$ or $w$, then they can be contracted without increasing the value of
  the minimum cut:
  \begin{enumerate}
  \item $c(v,w) \geq \hat\lambda$,
  \item $c(v) \leq 2 c(v, w)$ or $c(w) \leq 2 c(v, w)$,
  \item $\exists u \in V$ such that $c(v) \leq 2 \{c(v,w) + c(v,u)\}$ and $c(w)
\leq 2\{c(v,w)+c(w,u)\}$, or
  \item $c(v, w) + \sum_{u \in V}\min\{c(v,u), c(w,u)\} \geq \hat\lambda$.
  \end{enumerate}
  \label{lem:pr14}
\end{lemma}

Condition $1$ contracts every edge $e$ whose weight is $\geq \hat\lambda$. By
definition of a cut, we know that no cut that contains edge $e$ can have weight
$< \hat\lambda$. It is therefore safe to contract edge $e$ without losing any
cuts smaller than the smallest cut already found.

Condition $2$ contracts an edge $e=(v,w)$, if its weight is at least half the
degree of one of its incident vertices. In other words, $e$ is at least as heavy
as all other edges incident to one of $v$ or $w$. Without loss of generality, let $v$ be that
vertex. For every cut that contains $e$, we can find another cut that replaces it
with all other edges incident to $v$. As $c(v) \leq 2c(v,w)$, this cut is at
most as heavy as the original cut. Edge $e$ can therefore be contracted, as
there is at least one minimum cut that does not contain it. Condition $3$ is
closely related but additionally uses information from the shared neighborhood
of $v$ and $w$. If there is a vertex $u$ in the shared neighborhood of $v$ and
$w$ (\ie $u$, $v$, and $w$ form a triangle), so that the two triangle edges
incident to $v$ and $w$ respectively each have weight of at least half of their
respective vertex degree, every cut that separates $v$ and $w$ can be replaced
with one of smaller or equal weight that does not separate them. Note that the
condition does not require that $u$ is in the shared neighborhood of $v$ and
$w$; however, if it is not, condition $2$ already detects every contractible edge
that condition $3$ does.

Condition $4$ uses the whole shared neighborhood of $v$ and $w$. Each cut that
separates vertices $v$ and $w$ has to contain $(v,w)$ and for every shared
neighbor $u \in N(v) \cap N(w)$, either edge $(u,v)$ or $(u,w)$. Thus, we can
sum up over the lighter edge for each shared neighbor and find a lower bound for
the connectivity $\lambda(u,v)$. If this bound is already $\geq \hat\lambda$,
$(u,v)$ can be contracted.

In their experimental evaluation of various algorithms for the minimum cut
problem, Chekuri~\etal\cite{Chekuri:1997:ESM:314161.314315} use these reductions
to improve the performance of the algorithms by contracting edges that fulfill
either of the criteria in Lemma~\ref{lem:pr14}. Conditions $1$ and $2$ can be
exhaustively checked in linear time.

In order to check conditions $3$ and $4$ exhaustively, potentially all triangles
need to be checked. As an arbitrary graph can have up to
$\Theta(m^{\frac{3}{2}})$ triangles~\cite{schank2005finding}, an exhaustive
check introduces excessive running time penalties.
Chekuri~\etal\cite{Chekuri:1997:ESM:314161.314315} thus perform linear-time
passes that check these conditions on a subset of vertex sets as follows. In the
beginning of a pass, their algorithm marks each vertex as unscanned and then
scans vertices in order. When scanning vertex $v$, their algorithm checks
conditions $3$ and $4$ for each unscanned neighbor $w$ of $v$. In this check of
$v$ and $w$, they test condition $3$ for all vertices $u$ in the common
neighborhood $N(v) \cap N(w)$. When iterating over all vertices in the common
neighborhood, they compute the sum in condition $4$ by adding up the smaller of
the two edge weights for each vertex in the common neighborhood. Afterwards they
mark both $v$ and $w$ as scanned. This ensures a time complexity of $\Oh{n+m}$,
as each edge is processed at most twice. However, not all possible edges $(v,w)$
are tested to see whether the incident vertices $v$ and $w$ can be contracted.

\begin{table}[!ht]
  \setlength\intextsep{0pt}
  \centering
  \resizebox*{.46\textwidth}{!}{
  \begin{tabular}{l||r|r|r|r|r}
         \hline
  \multicolumn{6}{c}{Graph Family A}\\\hline
    Graph & $n$ & $m$ & $\lambda$ & $\delta$ & $n^*$ \\\hline\hline
    \texttt{com-orkut} & $2.4M$ & $112M$ & \numprint{14} & \numprint{16} & \numprint{2} \\
    & \numprint{114190} & $18M$ & \numprint{89} & \numprint{95} & \numprint{2} \\
    & \numprint{107486} & $17M$ & \numprint{76} & \numprint{98} & \numprint{2} \\
    & \numprint{103911} & $17M$ & \numprint{70} & \numprint{100} & \numprint{2} \\\hline
    \texttt{eu-2005} & \numprint{605264} & $15M$ & \numprint{1} & \numprint{10} & \numprint{63} \\
    & \numprint{271497} & $10M$ & \numprint{2} & \numprint{25} & \numprint{3} \\
    & \numprint{58829} & $3.7M$ & \numprint{29} & \numprint{60} & \numprint{2} \\
    & \numprint{5289} & \numprint{464821} & \numprint{19} & \numprint{100} & \numprint{2}\\\hline
    \texttt{gsh-2015-host}  & $25M$ & $1.3B$ & \numprint{1} & \numprint{10} & \numprint{175}\\
     & $5.3M$ & $944M$ & \numprint{1} & \numprint{50} & \numprint{32}\\
          & $2.6M$ & $778M$ & \numprint{1} & \numprint{100} & \numprint{16}\\
         & \numprint{98275} & $188M$ & \numprint{1} & \numprint{1000} & \numprint{3} \\\hline
  \texttt{hollywood-2011} & $1.3M$ & $109M$ & \numprint{1} & \numprint{20} & \numprint{13}\\
   & \numprint{576111} & $87M$ & \numprint{6} & \numprint{60} & \numprint{2} \\
          & \numprint{328631} & $71M$ & \numprint{77} & \numprint{100} & \numprint{2}\\
          & \numprint{138536} & $47M$ & \numprint{27} & \numprint{200} & \numprint{2}\\\hline
          \texttt{twitter-2010}  & $13M$ & $958M$ & \numprint{1} & \numprint{25} & \numprint{2} \\
            & $10M$ & $884M$ & \numprint{1} & \numprint{30} & \numprint{3} \\
               & $4.3M$ & $672M$ & \numprint{3} & \numprint{50} & \numprint{3}\\
                & $3.5M$ & $625M$ & \numprint{3} & \numprint{60} & \numprint{2}\\\hline
    \texttt{uk-2002} & $9M$ & $226M$ & \numprint{1} & \numprint{10} & \numprint{1940} \\
     & $2.5M$ & $115M$ & \numprint{1} & \numprint{30} & \numprint{347}\\
          & \numprint{783316} & $51M$ & \numprint{1} & \numprint{50} & \numprint{138}\\
          & \numprint{98275} & $11M$ & \numprint{1} & \numprint{100} & \numprint{20} \\\hline
    \texttt{uk-2007-05} & $68M$ & $3.1B$ & \numprint{1} & \numprint{10} & \numprint{3202}\\
    & $16M$ & $1.7B$ & \numprint{1} & \numprint{50} & \numprint{387} \\
          & $3.9M$ & $862M$ & \numprint{1} & \numprint{100} & \numprint{134}\\
          & \numprint{223416} & $183M$ & \numprint{1} & \numprint{1000} & \numprint{2}\\\hline

         \hline\hline
         \multicolumn{6}{c}{Graph Family B}\\\hline
\texttt{amazon} & \numprint{64813} & \numprint{153973} & \numprint{1}& \numprint{1} & \numprint{10068} \\\hline
\texttt{auto} & \numprint{448695} & $3.31M$ & \numprint{4} & \numprint{4} & \numprint{43} \\
& \numprint{448529} & $3.31M$ & \numprint{5} & \numprint{5} & \numprint{102} \\
& \numprint{448037} & $3.31M$ & \numprint{6} & \numprint{6} & \numprint{557} \\
& \numprint{444947} & $3.29M$ & \numprint{7} & \numprint{7} & \numprint{1128} \\
& \numprint{437975} & $3.24M$ & \numprint{8} & \numprint{8} & \numprint{2792} \\
& \numprint{418547} & $3.10M$ & \numprint{9} & \numprint{9} & \numprint{5814} \\ \hline
\texttt{caidaRouterLevel} & \numprint{190914} & \numprint{607610} & \numprint{1} & \numprint{1} & \numprint{49940} \\\hline
\texttt{cfd2} & \numprint{123440} & $1.48M$ & \numprint{7} & \numprint{7} & \numprint{15} \\\hline
\texttt{citationCiteseer} & \numprint{268495} & $1.16M$ & \numprint{1} & \numprint{1} & \numprint{43031} \\
& \numprint{223587} & $1.11M$ & \numprint{2} & \numprint{2} & \numprint{33423} \\
& \numprint{162464} & \numprint{862237} & \numprint{3} & \numprint{3} & \numprint{23373} \\
& \numprint{109522} & \numprint{435571} & \numprint{4} & \numprint{4} & \numprint{16670} \\
& \numprint{73595} & \numprint{225089} & \numprint{5} & \numprint{5} & \numprint{11878} \\
& \numprint{50145} & \numprint{125580} & \numprint{6} & \numprint{6} & \numprint{8770} \\\hline
\texttt{cnr-2000} & \numprint{325557} & $2.74M$ & \numprint{1} & \numprint{1} & \numprint{87720} \\
& \numprint{192573} & $2.25M$ & \numprint{2} & \numprint{2} & \numprint{33745} \\
& \numprint{130710} & $1.94M$ & \numprint{3} & \numprint{3} & \numprint{11604} \\
& \numprint{110109} & $1.83M$ & \numprint{4} & \numprint{4} & \numprint{9256} \\
& \numprint{94664} & $1.77M$ & \numprint{5} & \numprint{5} & \numprint{4262} \\
& \numprint{87113} & $1.70M$ & \numprint{6} & \numprint{6} & \numprint{5796} \\ 
& \numprint{78142} & $1.62M$ & \numprint{7} & \numprint{7} & \numprint{3213} \\
& \numprint{73070} & $1.57M$ & \numprint{8} & \numprint{8} & \numprint{2449} \\\hline
\texttt{coAuthorsDBLP} & \numprint{299067} & \numprint{977676} & \numprint{1} & \numprint{1} & \numprint{45242} \\\hline
\texttt{cs4} & \numprint{22499} & \numprint{43858} & \numprint{2} & \numprint{2} & \numprint{2} \\\hline
\texttt{delaunay\_n17} & \numprint{131072} & \numprint{393176} & \numprint{3} & \numprint{3} & \numprint{1484} \\ \hline
\texttt{fe\_ocean} & \numprint{143437} & \numprint{409593} & \numprint{1} & \numprint{1} & \numprint{40} \\\hline
\texttt{kron-logn16} & \numprint{55319} & $2.46M$ & \numprint{1} & \numprint{1} & \numprint{6325} \\\hline
\texttt{luxembourg} & \numprint{114599} & \numprint{239332} & \numprint{1} & \numprint{1} & \numprint{23077} \\\hline
\texttt{vibrobox} & \numprint{12328} & \numprint{165250} & \numprint{8} & \numprint{8} & \numprint{625} \\\hline
\texttt{wikipedia} & \numprint{35579} & \numprint{495357} & \numprint{1} & \numprint{1} & \numprint{2172} \\\hline
  \end{tabular}
  }
  \quad
  \resizebox*{.50\textwidth}{!}{
  \begin{tabular}{l||r|r|r|r|r}
  \hline
    \multicolumn{6}{c}{Graph Family B (continued)}\\\hline
    Graph & $n$ & $m$ & $\lambda$ & $\delta$ & $n^*$ \\\hline\hline
    \texttt{amazon-2008} & \numprint{735323} & $3.52M$ & \numprint{1} & \numprint{1} & \numprint{82520}\\
    & \numprint{649187} & $3.42M$ & \numprint{2} & \numprint{2} & \numprint{50611}\\
    & \numprint{551882} & $3.18M$ & \numprint{3} & \numprint{3} & \numprint{35752}\\
    & \numprint{373622} & $2.12M$ & \numprint{5} & \numprint{5} & \numprint{19813}\\
    & \numprint{145625} & \numprint{582314} & \numprint{10} & \numprint{10} & \numprint{64657}\\\hline
    \texttt{coPapersCiteseer} & \numprint{434102} & $16.0M$ & \numprint{1} & \numprint{1} & \numprint{6372} \\
    & \numprint{424213} & $16.0M$ & \numprint{2} & \numprint{2} & \numprint{7529} \\
    & \numprint{409647} & $15.9M$ & \numprint{3} & \numprint{3} & \numprint{7495} \\
    & \numprint{379723} & $15.5M$ & \numprint{5} & \numprint{5} & \numprint{6515} \\
    & \numprint{310496} & $13.9M$ & \numprint{10} & \numprint{10} & \numprint{4579} \\\hline
    \texttt{eu-2005} & \numprint{862664} & $16.1M$ & \numprint{1} & \numprint{1} & \numprint{52232} \\
    & \numprint{806896} & $16.1M$ & \numprint{2} & \numprint{2} & \numprint{42151} \\
    & \numprint{738453} & $15.7M$ & \numprint{3} & \numprint{3} & \numprint{21265} \\
    & \numprint{671434} & $13.9M$ & \numprint{5} & \numprint{5} & \numprint{18722} \\
    & \numprint{552566} & $11.0M$ & \numprint{10} & \numprint{10} & \numprint{23798} \\\hline
    \texttt{hollywood-2009} & $1.07M$ & $56.3M$ & \numprint{1} & \numprint{1} & \numprint{11923} \\
    & $1.06M$ & $56.2M$ & \numprint{2} & \numprint{2} & \numprint{17386} \\
    & $1.03M$ & $55.9M$ & \numprint{3} & \numprint{3} & \numprint{21890} \\
    & \numprint{942687} & $49.2M$ & \numprint{5} & \numprint{5} & \numprint{22199} \\
    & \numprint{700630} & $16.8M$ & \numprint{10} & \numprint{10} & \numprint{19265} \\\hline
    \texttt{in-2004} & $1.35M$ & $13.1M$ & \numprint{1} & \numprint{1} & \numprint{278092} \\
    & \numprint{909203} & $11.7M$ & \numprint{2} & \numprint{2} & \numprint{89895} \\
    & \numprint{720446} & $9.2M$ & \numprint{3} & \numprint{3} & \numprint{45289} \\
    & \numprint{564109} & $7.7M$ & \numprint{5} & \numprint{5} & \numprint{33428} \\
    & \numprint{289715} & $5.1M$ & \numprint{10} & \numprint{10} & \numprint{12947} \\\hline
    \texttt{uk-2002} & $18.4M$ & $261.6M$ & \numprint{1} & \numprint{1} & $2.5M$ \\
    & $15.4M$ & $254.0M$ & \numprint{2} & \numprint{2} & $1.4M$ \\
    & $13.1M$ & $236.3M$ & \numprint{3} & \numprint{3} & \numprint{938319} \\
    & $10.6M$ & $207.6M$ & \numprint{5} & \numprint{5} & \numprint{431140} \\
    & $7.6M$ & $162.1M$ & \numprint{10} & \numprint{10} & \numprint{298716} \\
    & \numprint{657247} & $26.2M$ & \numprint{50} & \numprint{50} & \numprint{24139} \\
    & \numprint{124816} & $8.2M$ & \numprint{100} & \numprint{100} & \numprint{3863} \\
    \hline\hline
    \multicolumn{6}{c}{Graph Family C}\\\hline
    Dynamic Graph & $n$ & Insertions & Deletions & Batches & $\lambda$ \\ \hline\hline
    \texttt{aves-weaver-social} & \numprint{445} & \numprint{1423} & \numprint{0} & \numprint{23} & \numprint{0} \\
    \texttt{ca-cit-HepPh} & \numprint{28093} & $4.60M$ & \numprint{0} & \numprint{2337} & \numprint{0} \\
    \texttt{ca-cit-HepTh} & \numprint{22908} & $2.67M$ & \numprint{0} & \numprint{219} & \numprint{0} \\
    \texttt{comm-linux-kernel-r} & \numprint{63399} & $1.03M$ & \numprint{0} & \numprint{839643} & \numprint{0} \\
    \texttt{copresence-InVS13} & \numprint{987} & \numprint{394247} & \numprint{0}& \numprint{20129} & \numprint{0} \\
    \texttt{copresence-InVS15} & \numprint{1870} & $1.28M$ & \numprint{0} & \numprint{21536} & \numprint{0} \\
    \texttt{copresence-LyonS} & \numprint{1922} & $6.59M$ & \numprint{0} & \numprint{3124} & \numprint{0} \\
    \texttt{copresence-SFHH} & \numprint{1924} & $1.42M$ & \numprint{0} & \numprint{3149} & \numprint{0} \\
    \texttt{copresence-Thiers} & \numprint{1894} & $18.6M$ & \numprint{0} & \numprint{8938} & \numprint{0} \\
    \texttt{digg-friends} & \numprint{279630} & $1.73M$ & \numprint{0} & $1.64M$ & \numprint{0} \\
    \texttt{edit-enwikibooks} & \numprint{134942} & $1.16M$ & \numprint{0} & $1.13M$ & \numprint{0} \\
    \texttt{fb-wosn-friends} & \numprint{63731} & $1.27M$ & \numprint{0} & \numprint{736675} & \numprint{0} \\
    \texttt{ia-contacts\_dublin} & \numprint{10972} & \numprint{415912} & \numprint{0} & \numprint{76944} & \numprint{0} \\
    \texttt{ia-enron-email-all} & \numprint{87273} & $1.13M$ & \numprint{0} & \numprint{214908} & \numprint{0} \\
    \texttt{ia-facebook-wall} & \numprint{46952} & \numprint{855542} & \numprint{0} & \numprint{847020} & \numprint{0} \\
    \texttt{ia-online-ads-c} & $15.3M$ & \numprint{133904} & \numprint{0} & \numprint{56565} & \numprint{0} \\
    \texttt{ia-prosper-loans} & \numprint{89269} & $3.39M$ & \numprint{0} & \numprint{1259} & \numprint{0} \\
    \texttt{ia-stackexch-user} & \numprint{545196} & $1.30M$ & \numprint{0} & \numprint{1154} & \numprint{1} \\
    \texttt{ia-sx-askubuntu-a2q} & \numprint{515273} & \numprint{257305} & \numprint{0} & \numprint{257096} & \numprint{0} \\
    \texttt{ia-sx-mathoverflow} & \numprint{88580} & \numprint{390441} & \numprint{0} & \numprint{390051} & \numprint{0} \\
    \texttt{ia-sx-superuser} & \numprint{567315} & $1.11M$ & \numprint{0} & $1.10M$ & \numprint{0} \\
    \texttt{ia-workplace-cts} & \numprint{987} & \numprint{9827} & \numprint{0} & \numprint{7104} & \numprint{0} \\
    \texttt{imdb} & \numprint{150545} & \numprint{296188} & \numprint{0} & \numprint{7104} & \numprint{0} \\
    \texttt{insecta-ant-colony1} & \numprint{113} & \numprint{111578} & \numprint{0} & \numprint{41} & \numprint{4285} \\
    \texttt{insecta-ant-colony2} & \numprint{131} & \numprint{139925} & \numprint{0} & \numprint{41} & \numprint{3742} \\
    \texttt{insecta-ant-colony3} & \numprint{160} & \numprint{241280} & \numprint{0} & \numprint{41} & \numprint{1539} \\
    \texttt{insecta-ant-colony4} & \numprint{102} & \numprint{81599} & \numprint{0} & \numprint{41} & \numprint{1838} \\
    \texttt{insecta-ant-colony5} & \numprint{152} & \numprint{194317} & \numprint{0} & \numprint{41} & \numprint{6671} \\
    \texttt{insecta-ant-colony6} & \numprint{164} & \numprint{247214} & \numprint{0} & \numprint{39} & \numprint{2177} \\
    \texttt{mammalia-voles-kcs} & \numprint{1218} & \numprint{4258} & \numprint{0} & \numprint{64} & \numprint{0} \\
    \texttt{SFHH-conf-sensor} & \numprint{1924} & \numprint{70261} & \numprint{0} & \numprint{3509} & \numprint{0} \\
    \texttt{soc-epinions-trust} & \numprint{131828} & \numprint{717129} & \numprint{123670} & \numprint{939} & \numprint{0} \\
    \texttt{soc-flickr-growth} & $2.30M$ & $33.1M$ & \numprint{0} & \numprint{134}& \numprint{0} \\
    \texttt{soc-wiki-elec} & \numprint{8297} & \numprint{83920} & \numprint{23093} & \numprint{101014} & \numprint{0} \\
    \texttt{soc-youtube-growth} & $3.22M$ & $12.2M$ & \numprint{0} & \numprint{203} & \numprint{0} \\
    \texttt{sx-stackoverflow} & $2.58M$ & \numprint{392515} & \numprint{0} & \numprint{384680} & \numprint{0} \\
    \hline
  \end{tabular}
  }
  \caption{Statistics of the static and dynamic graphs used in experiments.\label{p:mincut:table:graphs}}
\end{table}

\section{Graph Instances}
\label{rwgraphs}

In our experiments, we use a wide variety of large static and dynamic graph
instances. These are social graphs, web graphs, co-purchase matrices,
cooperation networks and some generated instances. These large graphs
from~\cite{bader2013graph,BRSLLP,BoVWFI,davis2011university,nr-aaai15,nr-sigkdd16}
are detailed in Table~\ref{p:mincut:table:graphs}. All instances are undirected.
If the original graph is directed, we generate an undirected graph by removing
edge directions and then removing duplicate edges. In our experiments, we use
three families of graphs for different subproblems. In
Table~\ref{p:mincut:table:graphs}, we show the number of vertices $n$ and edges $m$ for
each graph, the minimum cut $\lambda$ and the minimum degree $\delta$.
Additionally, we also show the number of vertices in the cactus graph $n^*$ for all
minimum cuts. This number is an indication of how many minimum cuts exist. A
value $n^* = 2$ indicates that there is a single minimum cut that
separates two sides. A larger value indicates that there are multiple minimum
cuts in the graph.

Graph family A consists of problems for finding \emph{some} minimum cut. These
graphs generally have multiple connected components and contain vertices with
very low degree. To create instances with $\lambda > 0$, we use the
largest connected component. As we want to find \emph{some} minimum cut,
instances in which the minimum cut is equal to the minimum degree are trivial to
solve. Thus, we use a $k$-core
decomposition~\cite{seidman1983network,batagelj2003m} to generate versions of
the graphs with a minimum degree of $k$ and use versions where $k > \lambda$,
\ie there exists at least one cut strictly smaller than the minimum degree and
the problem is therefore not trivial to solve. Generally these instances
have very few minimum cuts, and in many cases, there is only a single minimum cut.

The $k$-core of a graph $G = (V, E)$ is the largest subgraph $G' = (V',E')$ with
$V' \subseteq V$ and $E' \subseteq E$, which fulfills the condition that every
vertex in $G'$ has a degree of at least $k$. We perform our experiments on the
largest connected component of $G'$. For every real-world graph we use, we
compute a set of $4$ different $k$-cores, in which the minimum cut is not equal
to the minimum degree.

We generate a diverse set of graphs with different sizes. For the large graphs
\textttA{gsh-2015-host} and \textttA{uk-2007-05}, we use cores with $k$ in 10,
50, 100, and 1000. In the smaller graphs we use cores with $k$ in 10, 30, 50, and
100. \textttA{twitter-2010} and \textttA{com-orkut} had only a few cores for which the
minimum cut is not equal to the minimum degree. Therefore we used those cores.
As \textttA{hollywood-2011} is very dense, we used $k = 20,60,100,200$.

Graph family B consists of problems for finding \emph{all} minimum cuts. Thus,
the problem does not become trivial when the minimum cut is equal to the minimum
degree. We therefore do not compute $k$-cores of the graphs and instead run the
algorithms on the largest connected component of the source graph. However, as most
large real-world networks have cuts of size 1, finding all minimum cuts
becomes essentially the same as finding all bridges, which can be solved in
linear time using depth-first search~\cite{tarjan1972depth}. Usually
there is one huge block that is connected by minimum cuts to a set of
small and medium size blocks. Thus, we use our minimum cut algorithms to generate a more
balanced set of instances. We find all minimum cuts and contract each edge that
does not connect two vertices of the largest block. Thus, the remaining graph
only contains the huge block and is guaranteed to have a minimum cut value
$>\lambda$. We use this method to generate multiple graphs with different
minimum cuts for each instance. These graphs usually have $\lambda = \delta$,
\ie the value of the minimum cut is equal to the minimum degree, and have a large
set of minimum cuts. Thus, finding some minimum cut on these graphs is very
easy, but finding all of them is a significantly harder problem.

Graph family C consists of a set of $36$ dynamic graphs from Network
Repository~\cite{nr-aaai15,nr-sigkdd16}. These graphs consist of a sequence of
edge insertions and deletions. While edges are inserted and deleted, all
vertices are static and remain in the graph for the whole time. Each edge update
has an associated timestamp and a set of updates with the same timestamp is called
a \emph{batch}. Most of the graphs in this dataset have multiple connected
components, \ie their minimum cut $\lambda$ is $0$.

\chapter[Shared-memory Parallel Heuristic Minimum Cut]{VieCut: Shared-memory Parallel Heuristic Minimum Cut}
\label{c:viecut}

%% ALENEX'18 Paper, JEA Mincut Paper, IPDPS'19 Paper

In this chapter, we give a \emph{practical} shared-memory parallel
algorithm for the minimum cut problem. Our algorithm is heuristic (i.e., there
are no guarantees on solution quality), randomized, and has a running time of
$\Oh{n+m}$ when run sequentially. The algorithm works in a multilevel fashion:
we repeatedly reduce the input graph size with both heuristic and exact
techniques, and then solve the smaller remaining problem with exact methods. Our
heuristic technique identifies edges that are unlikely to be in a minimum cut
using the label propagation technique introduced by
Raghavan~\etal\cite{raghavan2007near} and contracts them in bulk. We further
combine this technique with exact reduction routines by Padberg and
Rinaldi~\cite{padberg1990efficient}, as discussed in Section~\ref{p:mincut:ss:pr}. We perform extensive experiments comparing
our algorithm with other heuristic algorithms as well as exact algorithms on
real-world and generated instances, which include graphs of up to \numprint{70}
million vertices and \numprint{5} billion edges. Results indicate that our algorithm
finds optimal cuts for almost all instances and also that the empirically
observed error rate is lower than for competing approximation algorithms (i.e., that
come with guarantees on the solution quality). At the same time, even when run
sequentially, our algorithm is significantly faster (up to a factor of $4.85$)
than other state-of-the-art algorithms. To further speedup computations, we also
give a version of our algorithm that performs random edge contractions as
preprocessing. This version achieves a lower running time and has better parallel
scalability at the expense of a higher error rate. 

The content of this chapter is based
on~\cite{DBLP:conf/alenex/HenzingerN0S18}~and~\cite{henzinger2018practical}.

\section{VieCut: A Parallel Heuristic Minimum-Cut Algorithm}
\label{c:vc:s:intro}
In this section we introduce our new approach to the minimum cut problem. Our
algorithm is based on edge contractions: we find densely connected vertices in
the graph and contract those into single vertices.  Due to the way contractions
are defined, we ensure that a minimum cut of the contracted graph corresponds to
a minimum cut of the input graph. Once the graph is contracted, we apply exact
reductions. These two contraction steps are repeated until the graph has a
constant number of vertices. We apply an exact minimum cut algorithm to find the
optimal cut in~the~contracted~graph.

Throughout our algorithm we maintain a variable $\hat\lambda$, which denotes the
current lowest upper bound for the minimum cut. In the beginning, $\hat\lambda$
equals the minimum node degree of $G$. After every contraction, if the minimum
node degree in the contracted graph is smaller than $\hat\lambda$, we set
$\hat\lambda$ to the minimum node degree of the contracted graph. As we only
perform contractions and therefore do not introduce any new cuts we can
guarantee that our algorithm will never output a value that is lower than the
minimum~cut.

This chapter is organized as follows. In Section~\ref{c:vc:s:pr}, we give a
general overview of our algorithm \texttt{VieCut} for the global minimum cut
problem and discuss in detail the parts that form the algorithm. Additionally we
give insight into parallelization and implementation details. In
Section~\ref{c:vc:s:rec}, we discuss a variant which combines \texttt{VieCut} with
random edge contraction to achieve an even lower running time at the expensive
of a higher error rate. We then show experiments and results in
Section~\ref{c:vc:s:experiments} before we conclude in Section~\ref{c:vc:s:conclusions}.

\section{Fast Minimum Cuts}
\label{c:vc:s:pr}

The algorithm of Karger and Stein~\cite{karger1996new} spends a large amount of
time computing graph contractions recursively.  One idea to speed up their
algorithm therefore is to increase the number of contracted edges per level.
However, this strategy is undesirable: it increases the error both in theory and
in practice, as their algorithm selects edges for contraction at random.  We
solve this problem by introducing an aggressive coarsening strategy that
contracts a large number of edges that are unlikely to be in a~minimum~cut.

We first give a high level overview before diving into the details of the
algorithm. Our algorithm starts by using the label propagation
algorithm~\cite{raghavan2007near} to cluster the vertices into densely connected
clusters. We then use a correcting algorithm to find misplaced vertices that
should form a singleton cluster. Finally, we contract the graph and apply the
exact reductions of Padberg and Rinaldi~\cite{padberg1990efficient}, as
discussed in Section~\ref{p:mincut:ss:pr}. We repeat these contraction steps
until the graph has at most a constant number $n_0$ of vertices. When the
contraction step is finished we apply the algorithm of Nagamochi, Ono and
Ibaraki~\cite{nagamochi1994implementing}, as discussed in
Section~\ref{p:mincut:ss:noi}, to find the minimum cut of the contracted graph.
Finally, we transfer the resulting cut into a cut in the original graph.
Overview pseudocode can be found in Algorithm~\ref{alg:generalalgorithm}.

\begin{algorithm}[h!]
\caption{\textttA{VieCut} \label{alg:generalalgorithm}}
\begin{algorithmic}[1]
  \INPUT $G=(V,E,c: V \to \MdN_{>{0}}), n_0 :$ bound for exact algorithm,
  \State $\mathcal{G} \leftarrow G$
  \While{$|V_{\mathcal{G}}| > n_0$} \Comment{compute inexact kernel}
  \State $\mathcal{C} \leftarrow$ computeClustering($\mathcal{G}$)
  \Comment{label propagation clustering}
  \State $\mathcal{C}' \leftarrow$ fixMisplacedVertices($\mathcal{G}$,
  $\mathcal{C}$)
  \State $G_{\mathcal{C}'} \leftarrow$ contractClustering($\mathcal{G}$,
  $\mathcal{C}'$)
  \State $\mathcal{E} \leftarrow$ findContractableEdges($G_{\mathcal{C}'}$)
  \Comment{further apply exact reductions}
  \State $\mathcal{G} \leftarrow$ contractEdges($G_{\mathcal{C}'}$,
  $\mathcal{E}$)
  \EndWhile
  \State $(A,B) \leftarrow $ NagamochiOnoIbaraki($\mathcal{G}$) 
  \Comment{solve minimum cut problem on final kernel}
  \State $(A',B') \leftarrow $ solutionTransfer($A,B$)
  \Comment{transfer solution to input network}
  \State \textbf{return} $(A',B')$
\end{algorithmic}
\end{algorithm}

The \emph{label propagation algorithm} (LPA) was proposed by Raghavan
\etal\cite{raghavan2007near} for graph clustering. It is a fast algorithm that
locally minimizes the number of edges cut. We outline the algorithm briefly.
Initially, each node is in its own cluster/block, \ie the initial block ID of a
node is set to its node ID.  The algorithm then works in rounds. In each round,
the nodes of the graph are traversed in a random order.  When a node $v$ is
visited, it is \emph{moved} to the block that has the strongest connection to
$v$, \ie it is moved to the cluster $C$ that maximizes $c(\{(v, u) \mid u \in
N(v) \cap C \})$. Ties are broken uniformly at random. The block IDs of
round~$i$ are used as initial block IDs of round $i+1$.

In the original
formulation~\cite{raghavan2007near}, the process is repeated until the process
converges and no vertices change their labels in a round. Kothapalli~\etal\cite{kothapalli2013analysis} show that label propagation finds all
clusters in few iterations with high probability, when the graph has a distinct
cluster structure. Hence, we perform at most $\ell$ iterations of the algorithm,
where $\ell$ is a tuning parameter.  One LPA round can be implemented to run in
$\Oh{n+m}$ time. As we only perform $\ell$ iterations, the algorithm runs in
$\Oh{n+m}$ time as long as $\ell$ is constant. In this formulation the
algorithm has no bound on the number of clusters. However, we can modify the
first iteration of the algorithm, so that a vertex $i$ is not allowed to change
its label when another vertex already moved to block $i$. In a connected graph
this guarantees that each cluster has at least two vertices and the contracted
graph has at most $\frac{|V|}{2}$ vertices. The only exceptions are connected
components consisting of only a single vertex (\emph{isolated} vertices with a
degree of $0$) which can not be contracted by the label propagation algorithm.
However, when such a vertex is detected, our minimum cut algorithm terminates
immediately and returns a cut of size $0$. In practice we do not use the
modification, as label propagation usually returns far fewer than
$\frac{|V|}{2}$~clusters.

\begin{figure}[t]
    \centering
    \includegraphics[width=.75\textwidth]{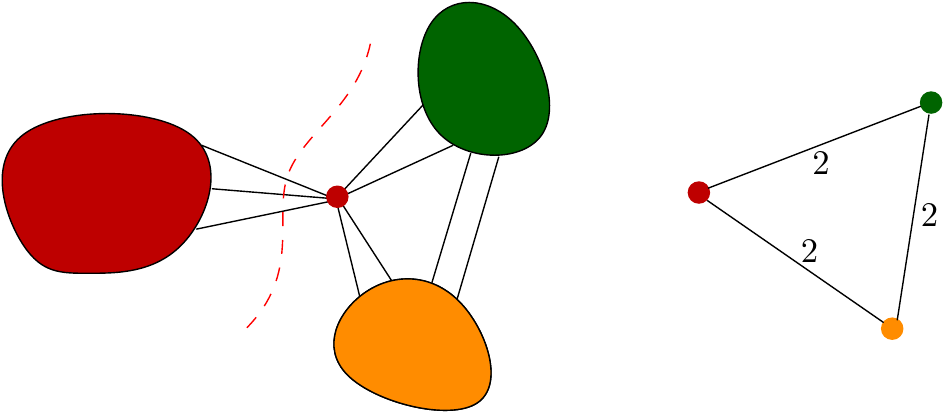}
    \includegraphics[width=.75\textwidth]{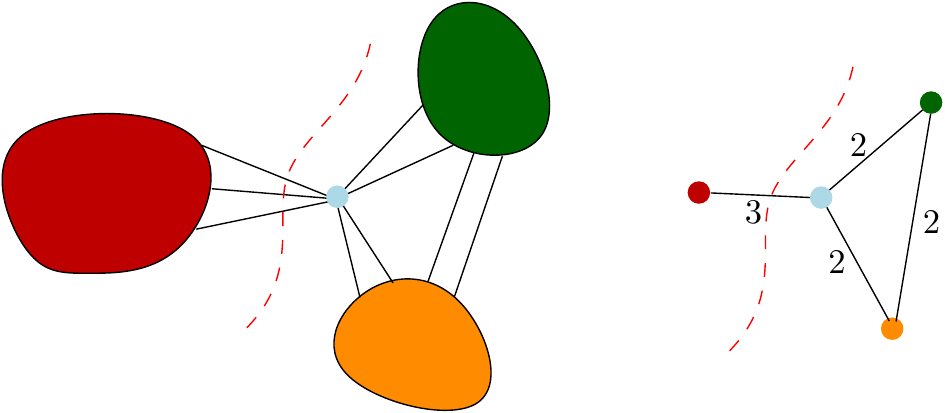}
    \caption{A case in which label propagation misplaces vertices. Top: label propagation assigns the centered vertex correctly to the left (red) cluster. However, this results in a situation in which the contracted graph no longer contains the minimum cut. Setting the centered vertex to be a singleton fixes this problem.}
    \label{fig:misplaced}
\end{figure}

Once we have computed the clustering with label propagation, we search for
single misplaced vertices using a \emph{correcting algorithm}. A misplaced
vertex is a vertex, whose removal from its cluster improves the minimum weighted
degree of the contracted graph. Figure~\ref{fig:misplaced} gives an example in
which the clustering misplaces a vertex. To find misplaced vertices, we sweep
over all vertices and check for each vertex whether it is misplaced. We only
perform this correcting algorithm on small clusters, which have a size of up to
$\log_2(n)$ vertices, as it is likely that large clusters would have more than a
single node misplaced at a time. In general, one can enhance this algorithm by
starting at any node whose removal would lower the cluster degree and greedily
adding neighbors whose removal further lowers the remaining cluster degree.
However, even when performing this greedy search on all clusters, this did not
yield further improvement over the single vertex version on small clusters. This
correcting step never makes a solution worse and on several instances it
improved the value of the final result.

After we computed the final clustering, we contract it to obtain a
coarser graph.  Contracting the clustering works as follows: each block of the
clustering is contracted into a single node. There is an edge between two nodes
$u$ and $v$ in the contracted graph if the two corresponding blocks in the
clustering are adjacent to each other in $G$, \ie block $u$ and block $v$ are
connected by at least one edge.  The weight of an edge $(A,B)$ is set to the sum
of the weight of edges that run between block $A$ and block $B$ of the
clustering.  Our contractions ensure that a minimum cut of the coarse graph
corresponds to a cut of the finer graph with the same value, but not vice versa:
we can not guarantee that a minimum cut of the contracted graph is equal to a
minimum cut of the original graph. It is possible that a single cluster contains
nodes from both sides of the cut. In this case, contracting the cluster
eliminates this minimum cut. If all minimum cuts are eliminated,
$\lambda(G_{\mathcal{C}}) > \lambda(G)$. Thus our newly introduced reduction for
the minimum cut problem is \emph{inexact}. However,~the~following~lemma~holds:

\begin{lemma}
  If there exist a minimum cut of $G$ such that each cluster of the clustering
$\mathcal{C}$ is completely contained in one side of the minimum cut of $G$ and
$|V_{\mathcal{C}}| > 1$, then $\lambda(G) = \lambda(G_{\mathcal{C}})$.
\label{lemma:cluster}
\end{lemma}

\begin{proof} As node contraction removes cuts but does not add any new cuts,
$\lambda(G_{\mathcal{C}}) \geq \lambda(G)$ for each contraction with
$|V_{\mathcal{C}}| > 1$. For an edge $e$ in $G$, which is not part of some minimum
cut of $G$, $\lambda(G) = \lambda(G/e)$~\cite{karger1996new}. Contraction of a
cluster $C$ in $G$ can also be represented as the contraction of all edges in
any spanning tree of $C$. If the cluster $C$ is on one side of the minimum cut,
none of the spanning edges are part of the minimum cut. Thus we can contract
each of the edges without affecting the minimum cut of $G$. We can perform this
contraction process on each of the clusters and $\lambda(G_{\mathcal{C}}) =
\lambda(G)$.
\end{proof}

\subsubsection{Exact Reductions by Padberg and Rinaldi}
We use the Padberg-Rinaldi reductions to further shrink the size of the graph. These are exact reductions, which do not modify the size of the minimum cut. Our
algorithm contracts all edges which are marked by the Padberg-Rinaldi
heuristics. In our experiments, we also tried to run the exact reductions first
and cluster contraction last. However, this resulted in a slower algorithm since
not many exact reductions could be applied on the initial unweighted network.
These reductions are described in Section~\ref{p:mincut:ss:pr}. Conditions $1$
and $2$ contract individual heavy edges and conditions $3$ and $4$ use the
shared neighborhood of the incident vertices to certify whether an edge can be
contracted.

We iterate over all edges of $G$ and check conditions $1$ and $2$. Whenever we
encounter an edge $(u,v)$ that satisfies either condition $1$ or $2$ we mark it
as contractible. After finishing the pass, we build the contracted graph. More
precisely, we perform contraction in linear time by deleting all unmarked edges,
contracting connected components and then re-adding the deleted edges as defined
in the contraction process. In practice, we achieve better performance using a
union-find data structure~\cite{galler1964improved}, which results in a running
time of $\Oh{n\alpha(n)+m}$

It is not possible to perform an exhaustive check for conditions $3$ and $4$ in
all triangles in an arbitrary graph $G$ in linear time, as the graph might have
as many as $\Theta(m^{\frac{3}{2}})$ triangles~\cite{schank2005finding}. We
therefore perform linear time passes similar to the implementation of
Chekuri~\etal\cite{Chekuri:1997:ESM:314161.314315}, as discussed in
Section~\ref{p:mincut:ss:pr}.

\subsubsection{Final Step: Exact Minimum Cut Algorithm.} 
To find the minimum cut of the final problem kernel, we use the minimum cut
algorithm of Nagamochi, Ono and Ibaraki, as discussed in
Section~\ref{p:mincut:ss:noi}.

\begin{lemma}
  The algorithm \textttA{VieCut} has a running time complexity of $\Oh{n+m}$.
\end{lemma}

\begin{proof}
  One round of all reduction and contraction steps
  (Algorithm~\ref{alg:generalalgorithm}, lines 2-8) can be performed in
  $\Oh{n+m}$. The label propagation step contracts the graph by at least a
  factor of $2$, which yields geometrically shrinking graph size and thus a
  total running time of $\Oh{n+m}$.  We break this loop when the contracted
  graph has less than some constant $n_0$ number of vertices. The exact minimum
  cut of this graph with constant size can therefore be found in constant time.
  The solution transfer can be performed in linear time by performing the
  coarsening in reverse and pushing the two cut sides from each graph to the
  next~finer~graph.

If the graph is not connected, throughout the algorithm one of the contracted
graphs can contain isolated vertices, which our algorithm does not contract.
However, when we discover an isolated vertex, there exists a cut of size $0$
that separates the connected components. As no cut can be smaller than $0$, this
cut is minimum and our algorithm terminates and reports it.
\end{proof}

\subsection{Parallelization}

We describe how to parallelize \textttA{VieCut}. We parallelize each part of
the algorithm, except the final invocation of the algorithm of Nagamochi, Ono,
and Ibaraki.

\subsubsection{Parallel Label Propagation}
To perform the label update for vertex $v$, we only need to consider vertices in
the neighborhood $N(v)$. Therefore the label propagation algorithm can be
implemented in parallel on shared-memory machines~\cite{staudt2013engineering}
using the \textttA{parallel for} directive from the OpenMP~\cite{dagum1998openmp}
API. We store the cluster affiliation for all vertices in an array of size $n$,
where position $i$ denotes the cluster affiliation of vertex $i$. We explicitly
do not perform label updates in a \textttA{critical} section, as each vertex is
only traversed once and the race conditions are not critical but instead
introduce another source of randomness.

\subsubsection{Parallel Correcting Step}
As the clusters are independent of each other for this correcting step, we
parallelize it on a cluster level, that is, a cluster is checked by a single
thread but each thread can check a different cluster without the need for locks
or mutexes.

\subsubsection{Parallel Graph Contraction}
\label{contract}
After label propagation has partitioned the graph into $c$ clusters, we build
the cluster graph. As the time to build this contracted graph is not negligible,
we parallelize graph contraction as well. One of the $p$ threads performs the
memory allocations to store the contracted graph, while the other $p-1$ threads
prepare the data for this contracted graph. When $c^2 > n$, we parallelize the
graph on a cluster level. To build the contracted vertex for cluster $C$, we
iterate over all outgoing edges $e = (u,v)$ for all vertices $u \in C$. If $v
\in C$ then $e$ is an intra-cluster edge and we discard the edge, otherwise we
add $c(u,v)$ to the edge weight between $C$ and the cluster of vertex $v$. When
$c^2 < n$, we achieve lower running time and better scaling when using a
shared-memory parallel hash table~\cite{akhremtsev2020high,maier2016concurrent}.
We generate the contracted graph $G_C = (V_C,E_C)$, in which each block is
represented by a single vertex - first we assign each block a vertex ID in the
contracted graph in $[0,|V_C|)$. For each edge~$e=(u,v)$, we compute a hash of
the block IDs of~$u$ and $v$ to uniquely identify the edge in $E_C$. We use this
identifier to compute the weights of all edges between blocks. Every thread
iterates over a distinct block of edges and we use the parallel hash table to
sum up the edge weights between vertices in the contracted graph. If the
contracted graph contains two extremely heavy vertices, \ie two vertices that
each encompass at least $20\%$ of the vertices of the original graph, we noticed
slowdown due to the many accesses to the same hash table entry. We therefore
compute the edge weight between those two blocks separately on each processor
and sum up these local values at the end.

\subsubsection{Parallel Padberg-Rinaldi Reductions}
In parallel, we run the Padberg-Rinaldi reductions on the
contracted graph. As these criteria are local and independent, they can be
parallelized trivially. We use a parallel wait-free union-find data
structure~\cite{anderson1991wait} to avoid locking. Reductions $2$ and $3$ use
the weighted vertex degree which changes when edges are contracted. Updating the
vertex degrees before performing the actual bulk contraction would entail
additional locks. These reductions are therefore only performed on edges where
both incident vertices where not yet affected by a contraction. We use a
compare-and-swap mechanism to make sure this holds in parallel.

\subsection{Further Implementation Details}
\label{c:vc:s:shuffle}
The label propagation algorithm by Raghavan et al.~\cite{raghavan2007near}
traverses the graph vertices in random order. Other implementations of the
algorithm~\cite{staudt2013engineering} omit this explicit randomization and rely
on implicit randomization through parallelism, as the vertex processing order in
parallel label propagation is non-deterministic. Our implementation to find the
new label of a vertex $v$ in uses an array, in which we sum up the weights for
all clusters in the neighborhood $N(v)$. Therefore randomizing the vertex
traversal order would destroy any graph locality, leading to many random reads
in the large array, which is very cache inefficient. Thus we trade off
randomness and graph locality by randomly shuffling small blocks of vertex ids
but traversing each of these shuffled blocks successively.

Using a time-forward processing technique~\cite{zeh2002efficient} the label
propagation as well as the contraction algorithm can be implemented in external
memory~\cite{akhremtsev2014semi} using Sort($|E|$) I/Os overall. Hence, if we
only use the label propagation contraction technique in external memory and use
the whole algorithm as soon as the graph fits into internal memory, we directly
obtain an external memory algorithm for the minimum cut problem. We do not
further investigate this variant of the algorithm as our focus is on fast
internal memory algorithms for the problem. 

\section{Random Edge Contraction}
\label{c:vc:s:rec}
We now propose an additional variant of our algorithm, which aims to achieve a
lower running time at the expense of a higher error rate. Similar to the
algorithm of Karger and Stein~\cite{karger1996new}, we shrink the graph by contracting random edges
and then perform the \textttA{VieCut} algorithm on the contracted graph.

In contrast to Karger and Stein's original algorithm~\cite{karger1996new}, our
implementation of random contraction does not perform the edge contractions
independently. Instead, we use a wait-free parallel union-find data
structure~\cite{anderson1991wait} to mark contracted blocks and perform
contractions in bulk, as discussed in the previous section.

In detail, the process works as follows: we draw a random integer  $i \in
[0,\dots,m)$ and use the union-find data structure to check whether the vertices
$u$ and $v$ incident to edge $i$ are in the same block. If they are not, we
unite the blocks containing $u$ and $v$ and decrement the number of blocks. We
repeat this process until the number of blocks is smaller than the number of
vertices multiplied by a given contraction factor $\alpha \in (0,1)$. We then
perform all contractions in a single operation, similar to the contraction in
Section~\ref{c:vc:s:pr}. This implementation of the random contraction algorithm
for the minimum cut problem was first employed by Chekuri
\etal~\cite{Chekuri:1997:ESM:314161.314315}.

For edge-weighted graphs we draw each edge with probability proportional to its weight.
To do this efficiently, we build the prefix sum of all edge weights. This prefix
sum $p_e$ of an edge $e$ is defined as the weight of all previous edges as given
by the edge order in the graph data structure, more formally defined in
Equation~\ref{eq:pre}.
\begin{equation} p_{e} = \sum\limits_{i=0}^{e-1} c(i)
\label{eq:pre}
\end{equation}
In weighted graphs we can then draw edges from the range $i \in
[0,\dots,\sum_{e \in E}c(e))$. If $p_{j} < i \leq p_{j+1}$, we contract edge $j$
similar to the unweighted case. This can be implemented in $\Oh{\log n}$ time
using binary search on the array of the prefix sums.

We also tested other techniques to achieve a speedup by contracting the graph:
in expectation, random edge sampling approximately preserves the minimum cut
with high probability~\cite{hu2013survey}. However, in order to actually achieve
a speedup, we need very low sampling rates and the approximation factor
deteriorates both in theory and practice. Removing high-degree vertices and
their incident edges often disconnects the graph. Greedily re-adding the removed
vertices to the partition with stronger connection does not result in cuts with
low weight. Hence, we omit further investigation of those techniques here.

\section{Experiments}\label{c:vc:s:experiments}

In this section we compare our algorithm \textttA{VieCut} with existing
algorithms for the minimum cut problem on real-world and synthetic graphs. We
compare the sequential variant of our algorithm to efficient implementations of
existing algorithms and show how our algorithm scales on a shared-memory
machine.

\subsection{Experimental Setup and Methodology}

We implemented the algorithms using \CC-17. Our experiments are conducted on two
machines: Machine A, which is used for nearly all experiments, has two Intel
Xeon E5-2643 v4 with 3.4GHz with 6 CPU cores each and 1.5 TB RAM in total. On
this machine we compiled our code using g++-7.1.0 with full optimization
(\textttA{-O3}). Machine B contains 4 Intel Xeon E7-8677 v3 with 2.5GHz with $16$
cores each. It has $1$TB of RAM in total. This machine is used for the parallel
experiments in Section~\ref{exp:rec} with up to $128$ threads. On this machine,
we compiled all code with g++-6.3.0 with full optimization  (\textttA{-O3}). In
general, we perform five repetitions per instance and report the average running
time as well as the cut size.

\subsection{Algorithms}  We compare our algorithm with our implementations
of the algorithm of Nagamochi, Ono and
Ibaraki~(\textttA{NOI})~\cite{nagamochi1994implementing} and the
$(2+\varepsilon)$-approximation algorithm of Matula
(\textttA{Matula})~\cite{matula1993linear}. In addition, we compare against the
preflow-based algorithm of Hao and Orlin~(\textttA{HO})~\cite{hao1992faster} by
using the implementation of Chekuri et
al.~\cite{Chekuri:1997:ESM:314161.314315}. We also performed experiments with
Chekuri et al.'s implementations of \textttA{NOI}, but our implementation is
generally faster. For \textttA{HO}, Chekuri~\etal give variants with and without
Padberg-Rinaldi tests and with an excess detection
heuristic~\cite{Chekuri:1997:ESM:314161.314315}, which contracts nodes with
large preflow excess. We use three variants of the algorithm of Hao and Orlin in
our experiments: \textttA{HO\_A} uses Padberg-Rinaldi tests, \textttA{HO\_B} uses
excess detection and \textttA{HO\_C} uses both. We also use their implementation
of the algorithm of Karger and
Stein~\cite{karger1996new,code,Chekuri:1997:ESM:314161.314315} (\textttA{KS})
without Padberg-Rinaldi tests. The variant of Karger-Stein with Padberg-Rinaldi
tests decomposed most graphs in preprocessing with repeated Padberg-Rinaldi
tests. It therefore performed very similar to \textttA{HO\_A} and \textttA{HO\_C}
and was omitted. We only perform a single iteration of the Karger-Stein
algorithm, as this is already slower than all other algorithms. Note that
performing more iterations yields a smaller error probability, but also makes
the algorithm even slower. The implementation crashes on very large instances
due to overflows in the graph data structure used for edge contractions. We do
not include the algorithm by Stoer and Wagner~\cite{stoer1997simple}, as it is
far slower than \textttA{NOI} and \textttA{HO} in the experiments of Chekuri
\etal\cite{Chekuri:1997:ESM:314161.314315} and Jünger
\etal\cite{junger2000practical} and was also slower in preliminary experiments
we conducted. We also do not include the near-linear algorithm of Henzinger
\etal~\cite{henzinger2017local}, as the other algorithms are quasi linear in
most instances examined and the algorithm of Henzinger \etal has large constant
factors in the running time. We performed, however, preliminary experiments with
the core of the algorithm, which indicate that the algorithm is slower in
practice. We also performed preliminary experiments with an ILP formulation
using Gurobi 8.0.0. On an RHG graph with $n=2^{15}$ and an average density of
$2^5$ that was solved exactly in $0.04$ seconds using \textttA{HO\_A}, the ILP
was solved in \numprint{3500} seconds. We therefore did not further investigate
using ILP formulations to solve the minimum cut problem. Finally, we note that
the MPI-parallel implementation of \textttA{KS} by
Gianinazzi~\etal~\cite{gianinazzi2018communication} finds the minimum cut of
RMAT graphs with $n=$\numprint{16000} and an average degree of \numprint{4000}
in $5$ seconds using \numprint{1536} cores~\cite{gianinazzi2018communication}.
This is significantly slower than our \textttA{VieCut} algorithm, which finds
the minimum cut on a similar-sized RMAT graph~\cite{wsvr} in $0.2$ seconds using
just $24$ threads. Given this stark difference in running time, we exclude their
algorithm from our experiments.

\subsection{Instances}
We perform experiments on clustered Erd\H{o}s-Rényi graphs that are generated
using the generator from Chekuri~\etal\cite{Chekuri:1997:ESM:314161.314315},
which are commonly used in the
literature~\cite{nagamochi1994implementing,junger2000practical,Chekuri:1997:ESM:314161.314315,padberg1990efficient}.
We also perform experiments on random hyperbolic
graphs~\cite{krioukov2010hyperbolic, von2015generating} and on large undirected
real-world graphs taken from the 10th DIMACS Implementation
Challenge~\cite{bader2013graph} and from the Laboratory for Web
Algorithmics~\cite{BRSLLP,BoVWFI}. As these graphs contain vertices with low
degree (and therefore trivial cuts), we use the $k$-core
decomposition~\cite{batagelj2003m}, which gives the largest subgraph, in which
each vertex has a degree of at least~$k$, to generate input graphs. We use the
largest connected components of these core graphs to generate graphs in which
the minimum cut is not trivial. For every real-world graph, we use $k$-cores for
four different values of $k$. In Section~\ref{rwgraphs} we show the instances in
further detail and in Table~\ref{p:mincut:table:graphs} (Graph Family A) we give sizes and cut
values for each instance used.

The graphs used in our experiments have up to $70$ million vertices
(\textttA{uk-2007-05}, $k=10$) and up to $5$ billion edges (Clustered
Erd\H{o}s-Rényi, $n=100$K, $d=100\%$). To the best of our knowledge, these
graphs are the largest instances reported in literature to be used for
experiments on global minimum cuts.

\subsubsection{Clustered Erd\H{o}s-Rényi Graphs}

Many prior experimental studies of minimum cut algorithms used a family of
clustered Erd\H{o}s-Rényi graphs with $m =
O(n^2)$~\cite{nagamochi1994implementing,junger2000practical,Chekuri:1997:ESM:314161.314315,padberg1990efficient}.
This family of graphs is specified by the following parameters: number of
vertices $n = |V|$, $d$ the graph density as a percentage where $m = |E| =
\frac{n\cdot(n-1)}{2} \cdot \frac{d}{100}$ and the number of clusters $k$. For
each edge $(u, v)$, the integral edge weight $c(u,v)$ is generated independently
and uniformly in the interval $[1,100]$. When the vertices $u$ and $v$ are in
the same cluster, the edge weight is multiplied by $n$, resulting in edge
weights in the interval $[n, 100n]$. Therefore the minimum cut can be found
between two clusters with high probability. We performed three experiments on
this family of graphs. In each of these experiments we varied one of the graph
parameters and fixed the other two parameters. These experiments are similar to
older
experiments~\cite{nagamochi1994implementing,junger2000practical,Chekuri:1997:ESM:314161.314315,padberg1990efficient}
but scaled to larger graphs to account for improvements in machine hardware. We
use the generator \emph{noigen} of Andrew Goldberg~\cite{code} to generate the
clustered Erd\H{o}s-Rényi graphs for these experiments. This generator was also
used in the study conducted by Chekuri
\etal\cite{Chekuri:1997:ESM:314161.314315}. As our code uses the
METIS~\cite{karypis1998metis} graph format, we use a script to translate the
graph format. All experiments exclude I/O times.

\subsubsection{Random Hyperbolic Graphs (RHG) \cite{krioukov2010hyperbolic}}
\label{rhggraphs}

Random hyperbolic graphs replicate many features of real-world
networks~\cite{chakrabarti2006graph}: the degree distribution follows a power
law, they often exhibit a community structure and have a small diameter. In
denser hyperbolic graphs, the minimum cut is often equal to the minimum degree,
which results in a trivial minimum cut. In order to prevent trivial minimum
cuts, we use a power law exponent of $5$. We use the generator of von Looz
\etal\cite{von2015generating}, which is a part of
NetworKit~\cite{staudt2014networkit}, to generate unweighted random hyperbolic
graphs with $2^{20}$ to $2^{25}$ vertices and an average vertex degree of $2^5$
to $2^8$. These graphs generally have very few small cuts and in most instances
there is only one unique minimum cut. Removal of the minimum cut partitions the
set of nodes into two sets of similar size.

\subsection{Configuring the Algorithm}

\begin{table*}[t]
    \begin{tabular}{r||c|c|c|c|c|c|c}
  &\textttA{VCut1} & \textttA{VCut2} &  \textttA{VCut3} &
  \textttA{VCut5} & \textttA{VCut10} & \textttA{VCut25} \\\hline\hline
      \# of non optimal cuts & 29 & 14 & 15 & 19 & 19 & 18 \\ \hline
      average dist. to opt. & 16.2\% & 2.44\% & 2.46\% & 3.80\% & 3.37\% &
      3.14\%\\
     \end{tabular}
     \caption{Error rate for configurations of \textttA{VieCut} in RHG graphs
     (out of 300 instances). The number in configuration name indicates the number
     of iterations in the label propagation step.  \label{fig:rhg_it}}
\end{table*}
    
We performed experiments to tune the number of label propagation iterations and
to find an appropriate amount of randomness for our algorithm. We conducted
these experiments with different configurations on generated hyperbolic
graphs~(see Section~\ref{rhggraphs}) with $2^{15}$ to $2^{19}$ vertices with an
average degree of $2^5$ to $2^8$ and compared error rate and running time. The
instances used here are different to the ones used in later sections.
  
Table~\ref{fig:rhg_it} shows the number of non-optimal cuts returned by
\textttA{VieCut} with different numbers of label propagation iterations
indicated by the integer in the name. Each implementation traverses the graph in
blocks of $256$ randomly shuffled elements as described in
Section~\ref{c:vc:s:shuffle}. The variant \textttA{VieCut25} performs up to 25
iterations or until the label propagation converges so that only up to
$\frac{1}{10000}$ of all nodes change their cluster. On average the variant
performed $20.4$ iterations. The results for all variants with $2$ to $25$
iterations are very similar with $14$ to $19$ non-optimal results and $2.44\%$
and $3.80\%$ average distance to the optimum. As the largest part of the total
running time is in the label propagation step, running the algorithm with a
lower amount of iterations is obviously faster. Therefore we use $2$ iterations
of label propagation in all following experiments.
  
\begin{figure}[t]
  %% FIGURE 3
  \centering
  \includegraphics[width=\textwidth]{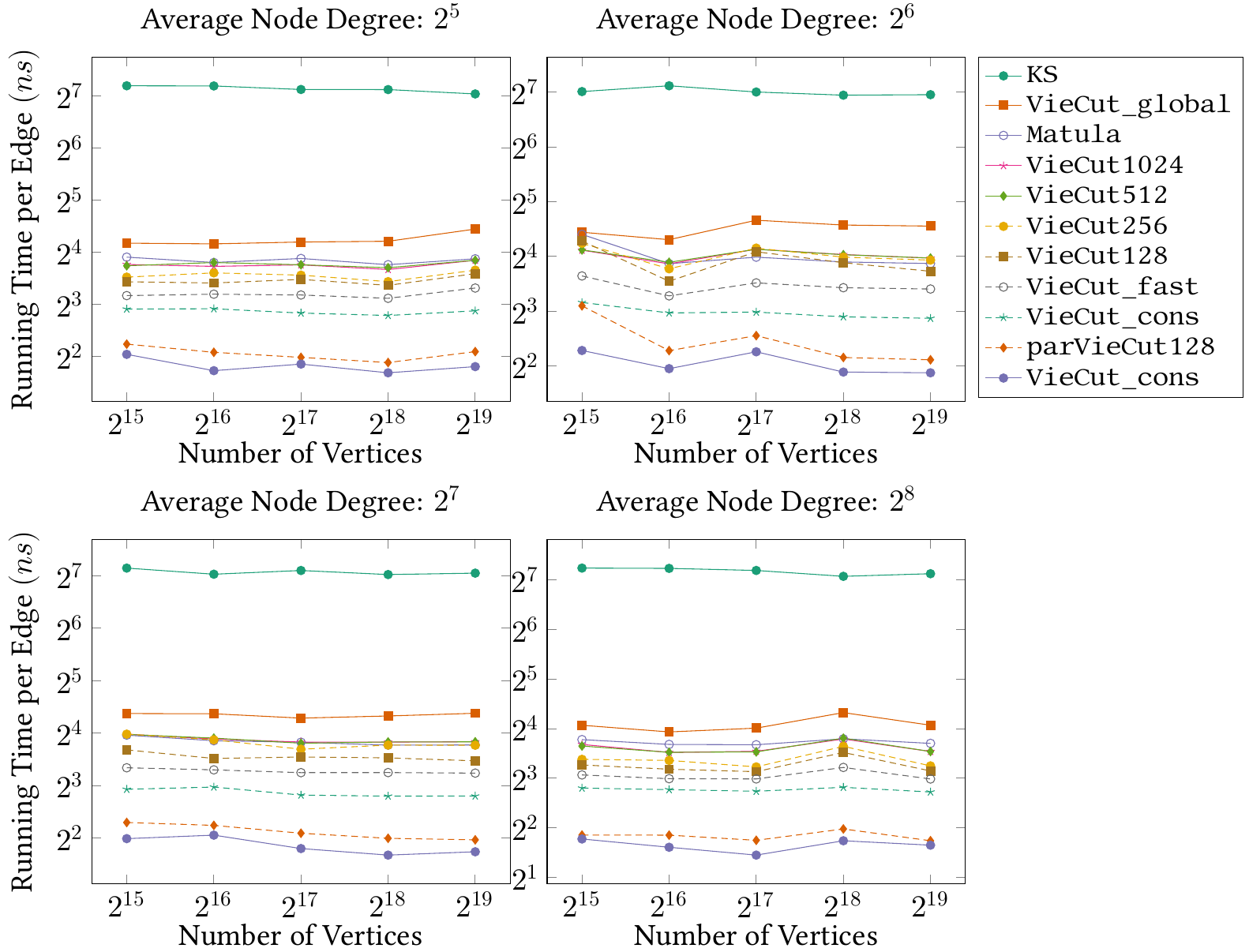}  
  
  \caption{Total running time in nanoseconds per edge in small RHG graphs for different configurations of \textttA{VieCut}}
  \label{fig:rhg}
\end{figure}

To compare the effect of graph traversal strategies, we compared different
configurations of our algorithm. \textttA{VieCut\_cons} does not randomize the
traversal order, \ie it traverses vertices consecutively by ID,
\textttA{VieCut\_global} performs global shuffling, \textttA{VieCut\_fast} swaps
each vertex with a random vertex with a index distance up to 20. The
configurations \textttA{VieCut128}, \textttA{VieCut256}, \textttA{VieCut512},
\textttA{VieCut1024} randomly shuffle blocks of $128$, $256$, $512$, or $1024$
vertices and introduce randomness without losing too much data locality. We also
include the configurations \textttA{parVieCut\_cons} and \textttA{parVieCut128},
which are shared-memory parallel implementation with 12 threads. As a
comparison, we also include the approximation algorithm of Matula and a single
run of the randomized algorithm of Karger and Stein.
    
\label{exp:shuffle}

Figure~\ref{fig:rhg} shows the total running time for different configurations
of \textttA{VieCut}. From the sequential algorithms, \textttA{VieCut\_cons}
has the lowest running time for all algorithms. The algorithm, however, returns
non-optimal cuts in more than $\frac{1}{3}$ of all instances, with an average
distance to the minimum cut of $~44\%$ over all graphs. The best results were
obtained by \textttA{VieCut128}, which has an average distance of $0.83\%$ and
only $10$ non-optimal results out of $300$ instances. The results are very good
compared to \textttA{Matula}, which has $57$ non-optimal results in these $300$
instances and an average distance of $5.57\%$. \textttA{VieCut128} is $20\%$
faster on most graphs than \textttA{Matula}, regardless of graph size or
density.  In the following we use the configuration \textttA{VieCut128} with
$2$ iterations, there named \textttA{VieCut}. On these small graphs, the
parallel versions have a speedup factor of $2$ to $3.5$ compared to their
sequential version. \textttA{parVieCut128} has $17$ non-optimal results and an
average distance of $4.91\%$ while \textttA{parVieCut\_cons} has 29 non-optimal
results and $20\%$ average distance to the minimum cut. Therefore we use
\textttA{parVieCut128} for all parallel experiments (named \textttA{parVieCut}).
We set the bound $n_0$ to \numprint{10000} and did not encounter a single
instance with more than a single bulk contraction step.

\begin{figure}[t!]
  %%% FIGURE 1
  \centering
  \includegraphics[width=\textwidth]{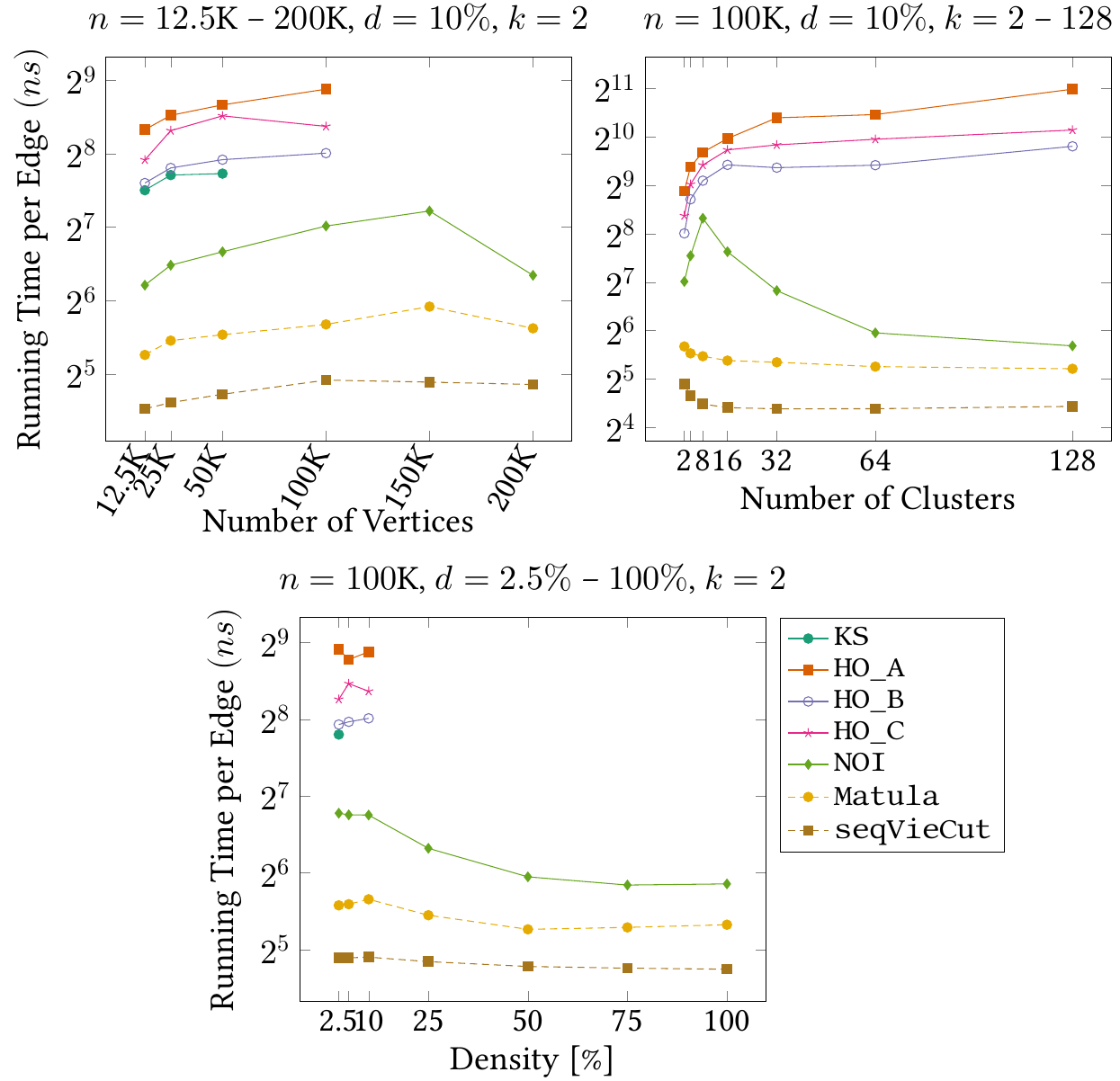}

  \caption{Total running time in nanoseconds per edge in clustered
    Erd\H{o}s-Rényi graphs\label{fig:test3}}
\end{figure}

\subsection{Experimental Results}

\subsubsection{Clustered Erd\H{o}s-Rényi Graphs}

Clustered Erd\H{o}s-R\'enyi graphs have distinct small cuts between the clusters
and do not have any other small cuts. We perform three experiments varying one
parameter of the graph class and use default parameters for the other two
parameters. Our default parameters are $n=\numprint{100000}$, $d=10\%$ and
$k=2$. The code of Chekuri~\etal\cite{Chekuri:1997:ESM:314161.314315} uses 32
bit integers to store vertices and edges. We could therefore not perform the
experiments with $m \geq 2^{31}$ with~\textttA{HO}. Figure~\ref{fig:test3} shows
the results for these experiments. First of all, on $20\%$ of the instances
\textttA{KS} returns non-optimal results. No other algorithm returned any
non-optimal minimum cuts on any graph of this dataset. Moreover,
\textttA{seqVieCut} is the fastest algorithm on all of these instances, followed
by \textttA{Matula}, which is $40\%$ to $100\%$ slower on these instances.

Our algorithm \textttA{seqVieCut} is faster on graphs with a lower number of
vertices, as the array containing cluster affiliations -- which has one entry
per vertex and is accessed for each edge -- fits into cache. In graphs with
$k=2,4,8$, the final number of clusters in the label propagation algorithm is
equal to $k$, as label propagation correctly identifies the clusters. In the
graph contraction step, we iterate over all edges and check whether the incident
vertices are in different clusters. For this branch, the compiler assumes that
they are indeed in different cluster. However, in these graphs, the chance for
any two adjacent nodes being in the same cluster is $\frac{1}{k}$, which is far
from zero. This results in a large amount of branch misses (for $n=$
\numprint{100000}, $d=10\%$, $k=2$: average $14\%$ branch misses, in total $1.5$
billion missed branches). Thus the performance is better with higher values of
$k$. The fastest exact algorithm is \textttA{NOI}. This matches the experimental
results obtained by Chekuri \etal\cite{Chekuri:1997:ESM:314161.314315} on graphs
generated with the same instance generator.

\begin{figure}[t]
    \centering
 %%% FIGURE 2
 \includegraphics[width=\textwidth]{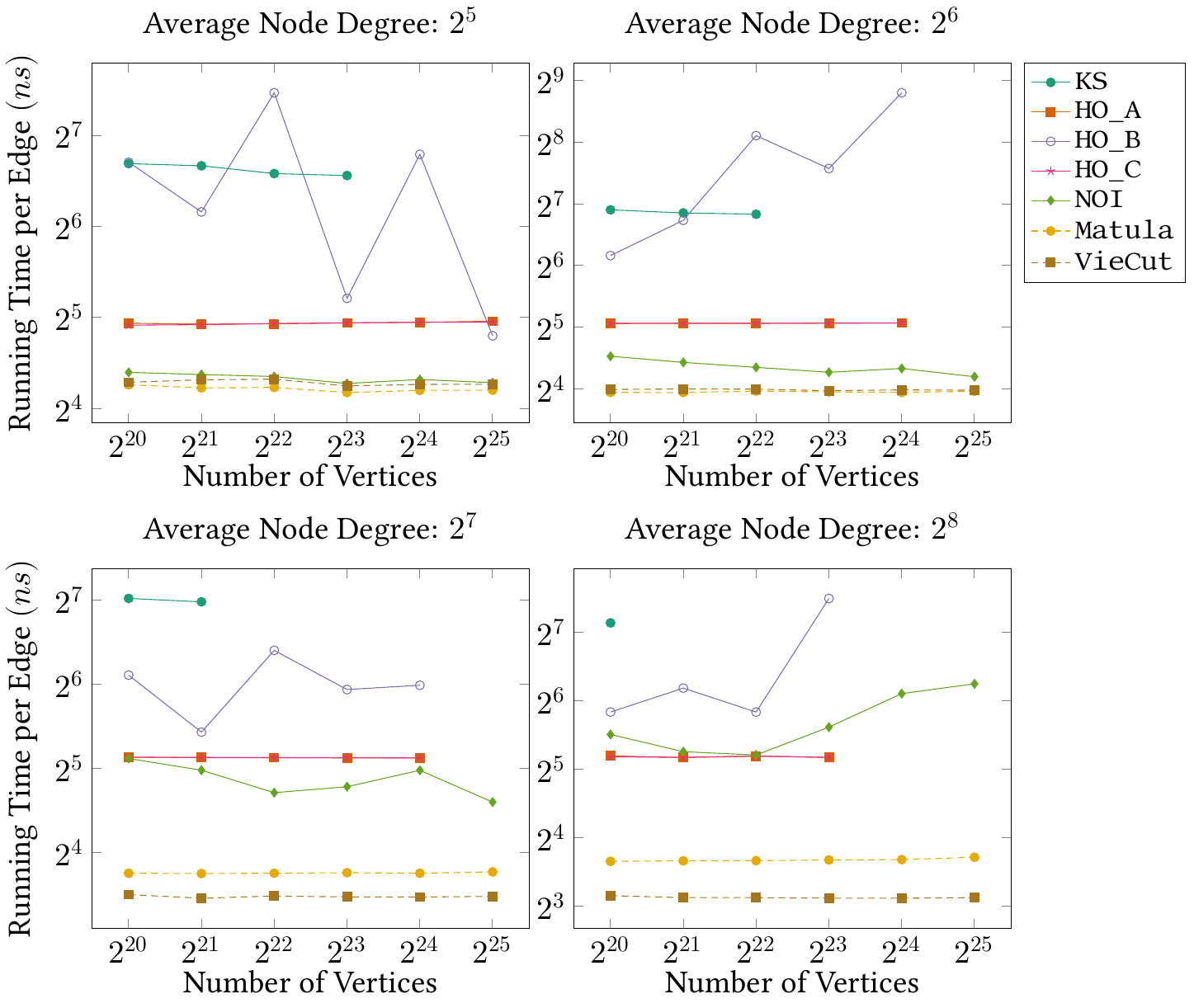}

\caption{Total running time in nanoseconds per edge in RHG graphs
\label{fig:rhgtests}}
\end{figure}

% DONE UNTIL HERE

\begin{figure}[t]
  \centering
  %% FIGURE 4
  \includegraphics[width=.6\textwidth]{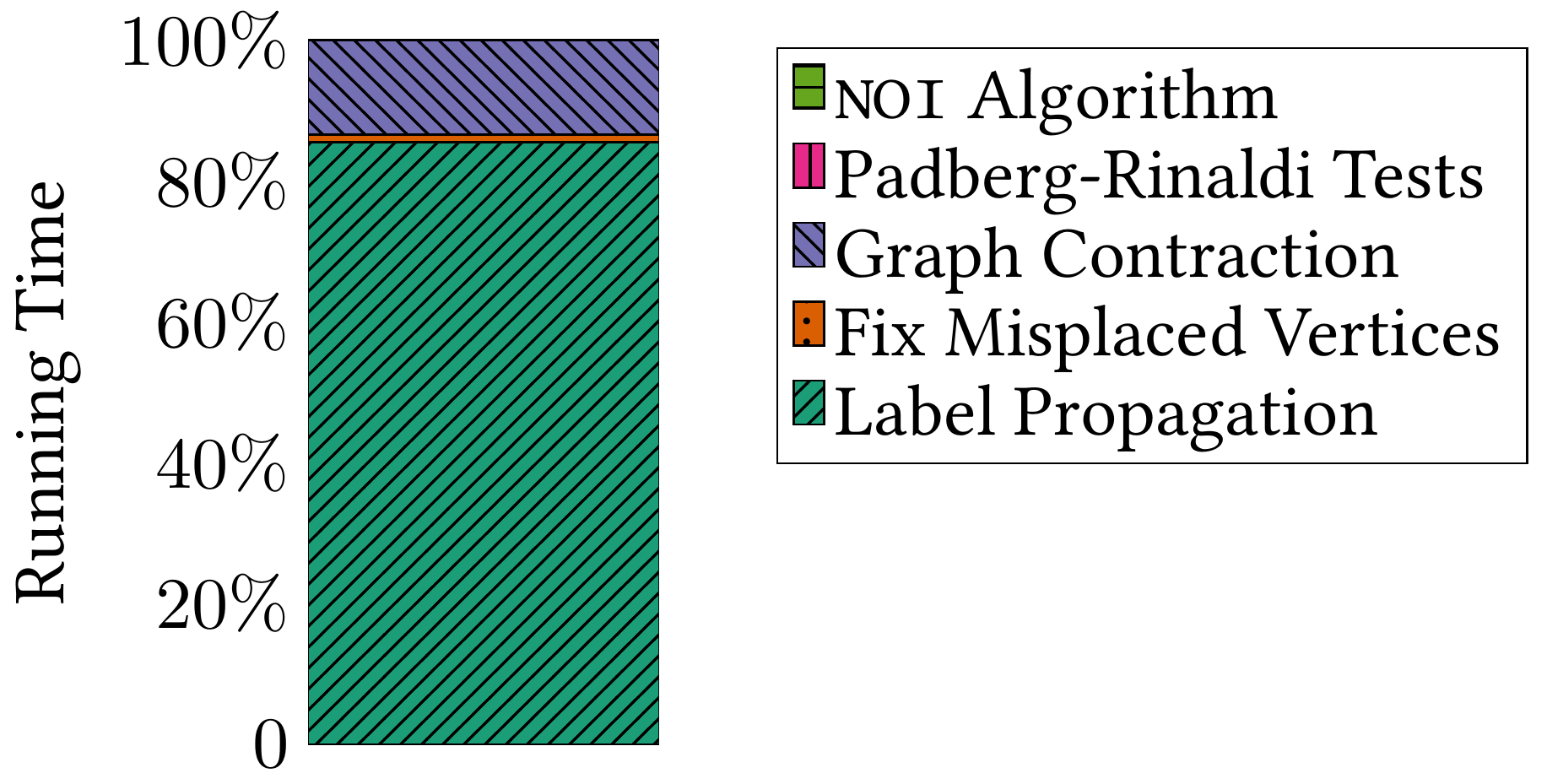}
  \caption{Running Time Breakdown for RHG Graphs with $n=2^{25}$ and
  $m=2^{32}$\label{fig:breakdown}}
  \end{figure}

\subsubsection{Random Hyperbolic Graphs}
\begin{figure}[t]
\centering
\includegraphics[width=.73\textwidth]{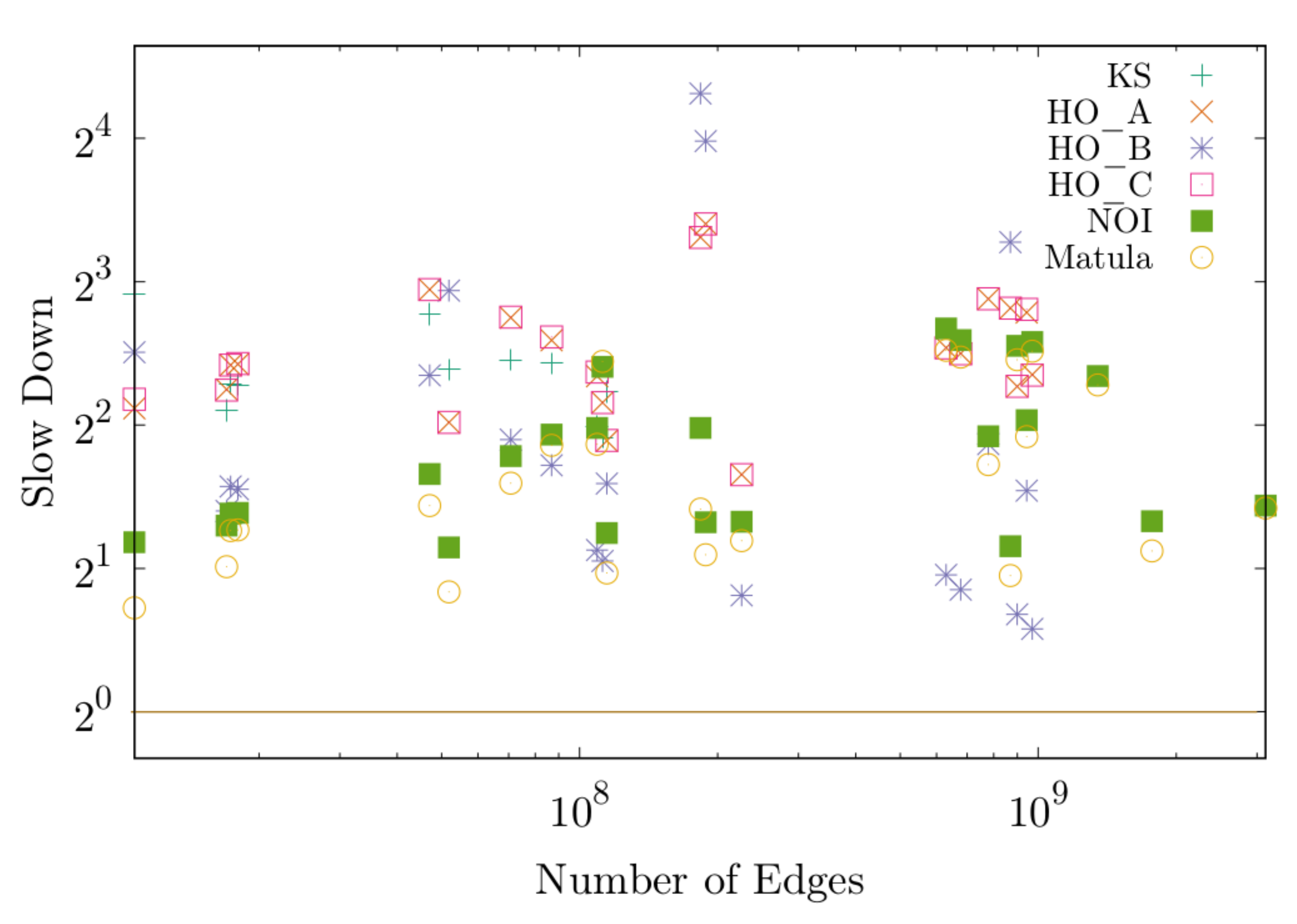}
\includegraphics[width=.73\textwidth]{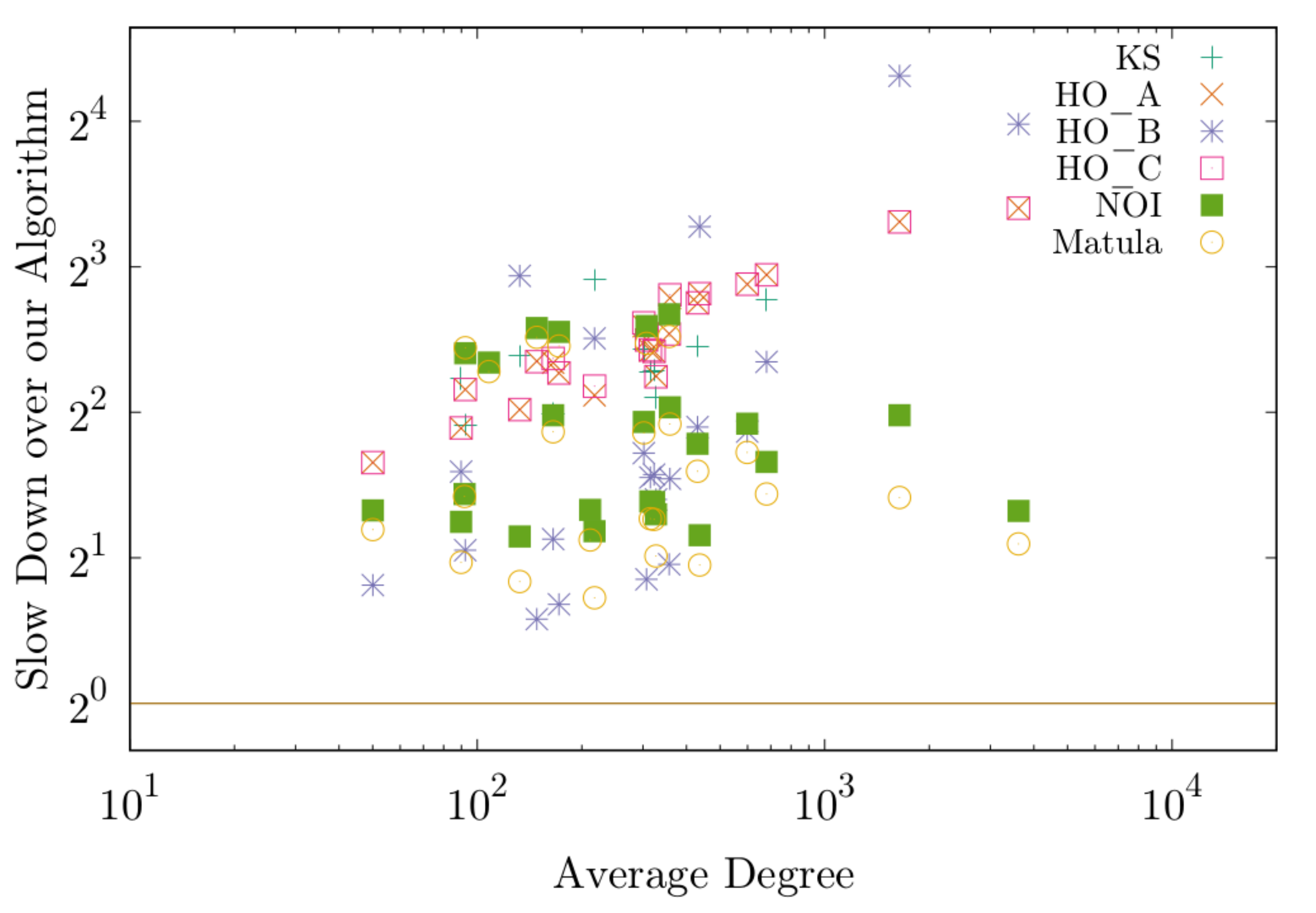}
\caption{Slowdowns of competitors to \textttA{VieCut} in large real-world
graphs. We display slowdowns based on the absolute number of edges (top), and
by the average vertex degree (bottom) in the graph\label{fig:real}}
\end{figure}

We also performed experiments on random hyperbolic graphs with $n=2^{20} -
2^{25}$ and an average degree of $2^5 - 2^8$. We generated $3$ graphs for each
of the $24$ possible combinations of $n$ and average degree yielding a total of
$72$ RHG graphs. Note that these graphs are hard instances for the inexact
algorithms, as they contain few -- usually only one -- small cuts and both sides
of the cut are large. From a total of $360$ runs, \textttA{seqVieCut} does not
return the correct minimum cut in $1\%$ of runs and \textttA{Matula} does not
return the correct minimum cut in $31\%$ of runs. \textttA{KS}, which crashes on
large instances, returns non-optimal cuts in $52\%$ of the runs where
it~ran~to~completion.

Figure~\ref{fig:rhgtests} shows the results for these experiments. On nearly all
of these graphs, \textttA{NOI} is faster than \textttA{HO}. On sparse graphs with
an average degree of $2^5$, \textttA{seqVieCut}, \textttA{Matula} and
\textttA{NOI} nearly have equal running time. On denser graphs with an average
degree of $2^8$, \textttA{seqVieCut} is $40\%$ faster than \textttA{Matula} and
$4$ to $10$ times faster than \textttA{NOI}. \textttA{HO\_A} and \textttA{HO\_C}
use preprocessing with the Padberg-Rinaldi heuristics. Multiple iterations of
this preprocessing contract the RHG graph into two nodes. The running time of
those algorithms is $50\%$ higher on sparse graphs and $4$ times higher on dense
graphs compared to \textttA{seqVieCut}. Figure~\ref{fig:breakdown} shows a time
breakdown for \textttA{seqVieCut} on large RHG graphs with $n=2^{25}$. Around
$85\%$ of the running time is in the label propagation step and the rest is
mostly spent in graph contraction. The correcting step has low running time on
most graphs, as it is not performed on large clusters.

\subsubsection{Real-World Graphs}

The third set of graphs we use in our experiments are $k$-cores of large
real-world social and web graphs. On these graphs, no non-optimal minimum cuts
were returned by any algorithm except for \textttA{KS}, which gave $36\%$
non-optimal results. However, as most of these graph instances have multiple
minimum cuts, even exact algorithms usually output different cuts on multiple
runs. Figure~\ref{fig:real} gives slowdown plots to the fastest algorithm
(\textttA{seqVieCut} in each case) for the real-world graphs. On these graphs,
\textttA{seqVieCut} is the fastest algorithm, far faster than the other
algorithms. \textttA{Matula} is not much faster than \textttA{NOI}, as most of
the running time is in the first iteration of their CAPFOREST algorithm, which
is similar for both algorithms. On the largest real-world graphs,
\textttA{seqVieCut} is approximately $3$ times faster than the next fastest
algorithm \textttA{Matula}. We also see that \textttA{seqVieCut},
\textttA{Matula} and \textttA{NOI} all perform better on denser graphs. For
\textttA{Matula} and \textttA{NOI}, this can most likely be explained by the
smaller vertex priority queue. For \textttA{seqVieCut}, this is mainly due to
better cache locality. As \textttA{HO} does not benefit from denser graphs, it
has high slow down on dense graphs.

The highest speedup in our experiments is in the $10$-core of
\textttA{gsh-2015-host}, where \textttA{seqVieCut} is faster than the next
fastest algorithm (\textttA{Matula}) by a factor of $4.85$. The lowest speedup is
in the $25$-core of \textttA{twitter-2010}, where \textttA{seqVieCut} is $50\%$
faster than the next fastest algorithm (\textttA{HO\_B}). The average speedup
factor of \textttA{seqVieCut} to the next fastest algorithm is $2.37$.
\textttA{NOI} and \textttA{Matula} perform badly on the cores of the graph
\textttA{twitter-2010}. This graph has a very low diameter (average distance on
the original graph is $4.46$), and as a consequence the priority queue used in these algorithms is
filled far quicker than in graphs with higher diameter. Therefore the priority
queue operations become slow and the total running time is very high.

\begin{figure}[t]
  \centering
  %%FIGURE 5
  \includegraphics[width=\textwidth]{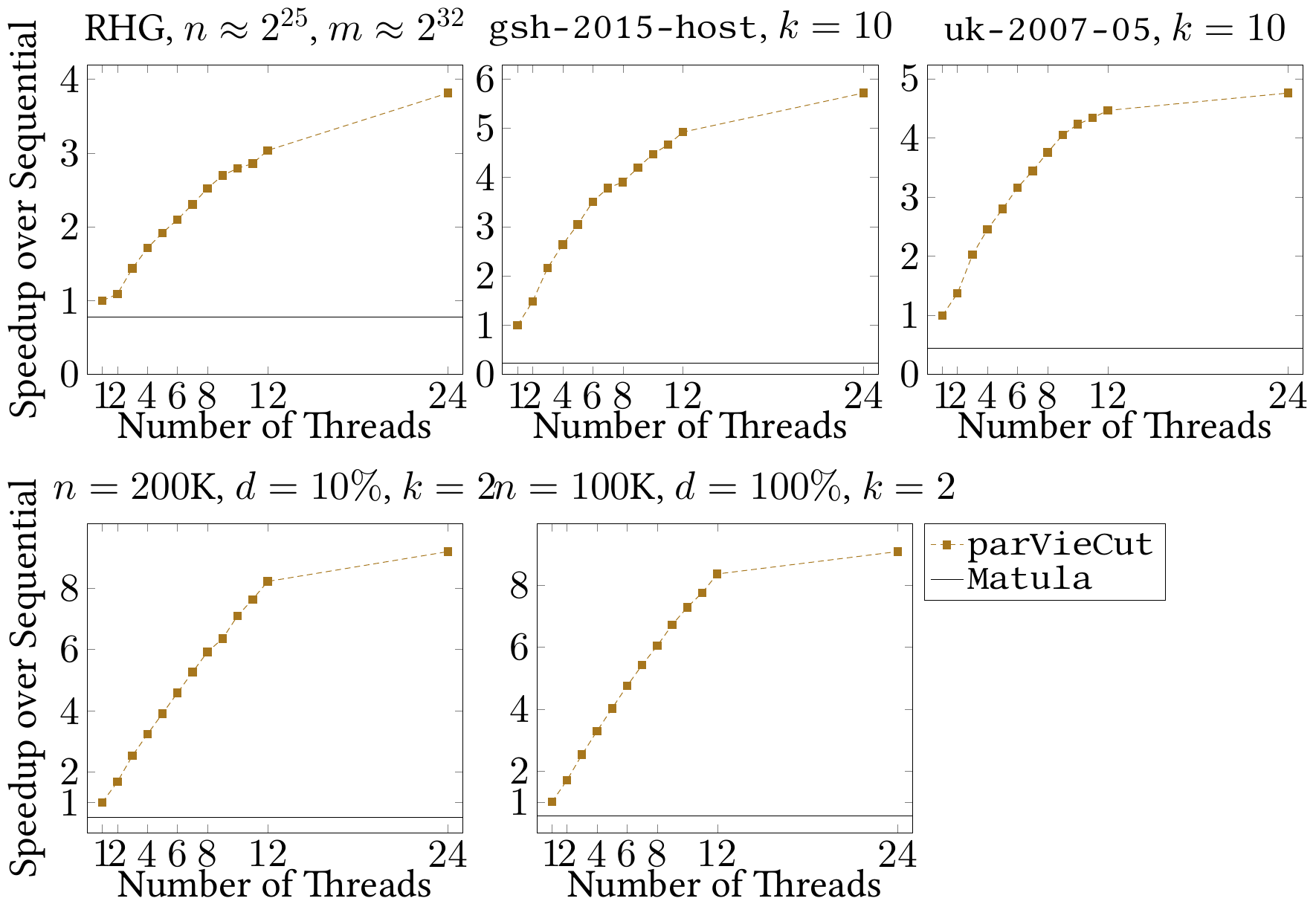}  
  \caption{Speedup on large graphs over \textttA{VieCut} using 1
thread.\label{scaling}}
\end{figure}

\begin{figure}[t]
  \centering
  %%FIGURE 6
  \includegraphics[width=.8\textwidth]{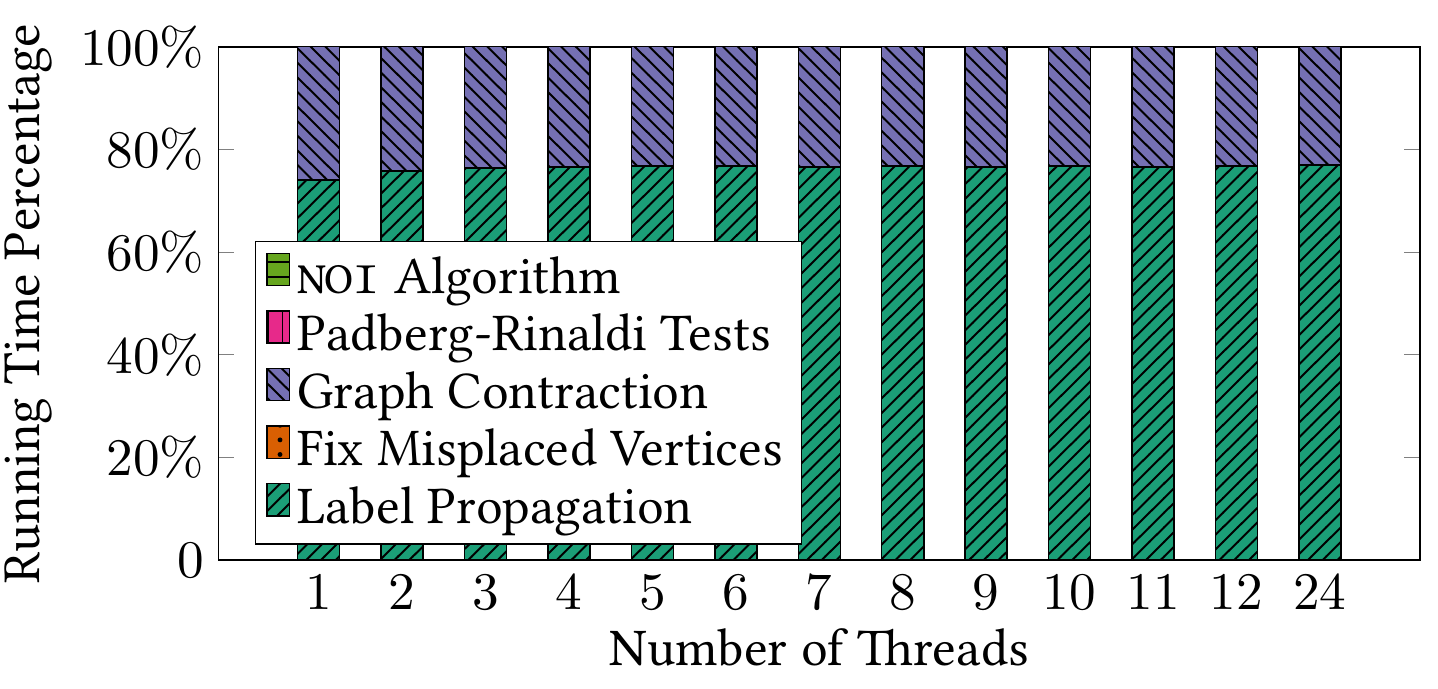}
  \caption{Parallel Running Time Breakdown\label{fig:scaledown}}
  \end{figure}

To summarize, both in generated and real-world graphs, even in sequential runs
\textttA{seqVieCut} is up to a factor of $6$ faster than the state of the art,
while achieving a high solution quality even for hard instances such as the
hyperbolic graphs. The performance of \textttA{seqVieCut} is especially good on
the real-world graphs, presumably as these graphs have high locality.

\subsubsection{Shared-Memory Parallelism}

Figure~\ref{scaling} shows the speedup of \textttA{parVieCut} compared to the
sequential variant and to the next fastest algorithm, which is \textttA{Matula}
in all of the large graph examined. We examine the largest graphs from each of
the three graph classes and perform parallel runs using $1,2,3, \ldots, 12$
threads. We also perform experiments with $24$ threads, as the machine has $12$
cores and supports multi-threading. The harmonic mean of the speedup of
\textttA{parVieCut} on large graphs with $12$ threads is $5.01$. ($24$ threads:
$5.5$) and all runs computed the exact minimum cut. Compared to the next fastest
sequential algorithm \textttA{Matula}, this is an average harmonic speedup factor
of $9.5$ ($24$ threads: $11.1$). \textttA{parVieCut} scales especially well on
the clustered Erd\H{o}s-Rényi graphs, presumably as these dense graphs contain
many high-degree vertices and have a rather low number of vertices.
Figure~\ref{fig:scaledown} shows average running time breakdowns averaged over
all graphs. For this figure, the correcting algorithm is turned off for the two
Erd\H{o}s-Rényi graphs. With one thread, label propagation uses $74\%$ of the
total running time and with $24$ threads, $77\%$ of the total running~time. Thus
the different parts of the algorithm parallelize equally well.

\begin{figure}[t]
  %%FIGURE 7
  \includegraphics[width=\textwidth]{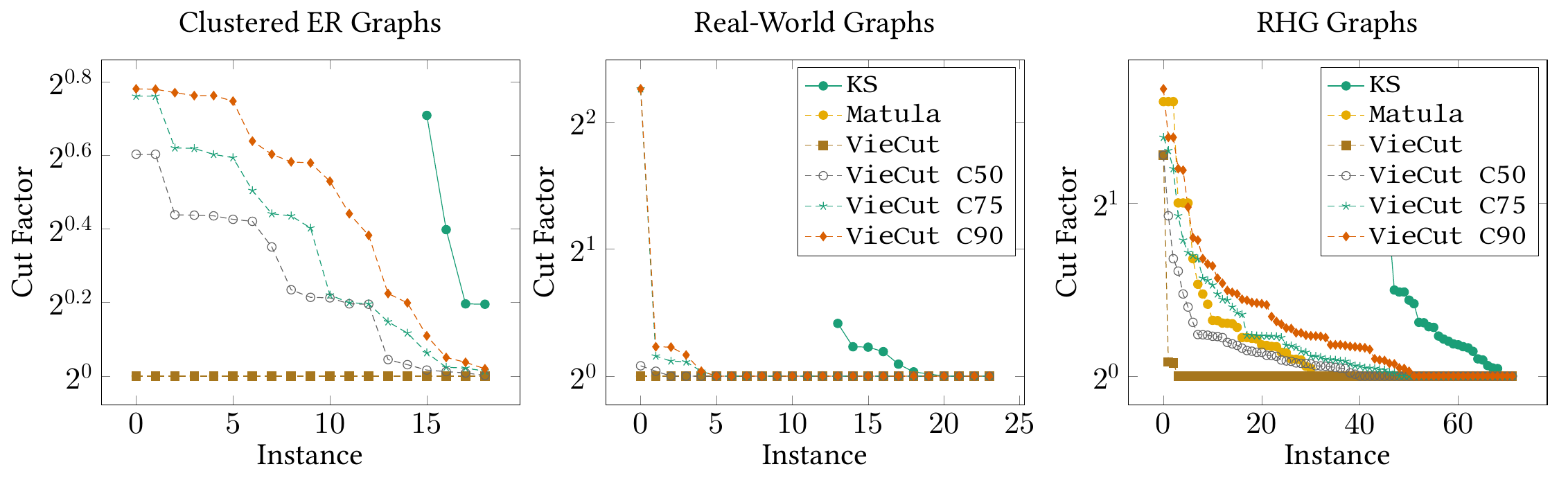}
  \caption{Average factor of result to minimum cut for inexact algorithms\label{fig:error_er}}
\end{figure}

\subsection{Random Edge Contraction}
\label{exp:rec}

We now evaluate the variant of our algorithm that uses random edge contractions
similar to the algorithm of Karger and Stein before running \textttA{VieCut} on
the contracted graph. The edge contractions promise faster results but increase
the error rate of the algorithm. This section shows experiments which detail
error rate, error severity and running times for the heuristic and approximation
algorithms. We repeat the experiments of the previous section, but now with only
inexact algorithms. In addition to (the one iteration-only version of)
\textttA{KS}, \textttA{Matula} and \textttA{VieCut}, we add \textttA{VieCut
C50}, \textttA{VieCut C75}, \textttA{VieCut C90}, which contract $50$, $75$
and $90\%$ of all vertices~before~running~\textttA{VieCut}.

\begin{figure}[t]
  %%FIGURE 8
  \centering
  \includegraphics[width=\textwidth]{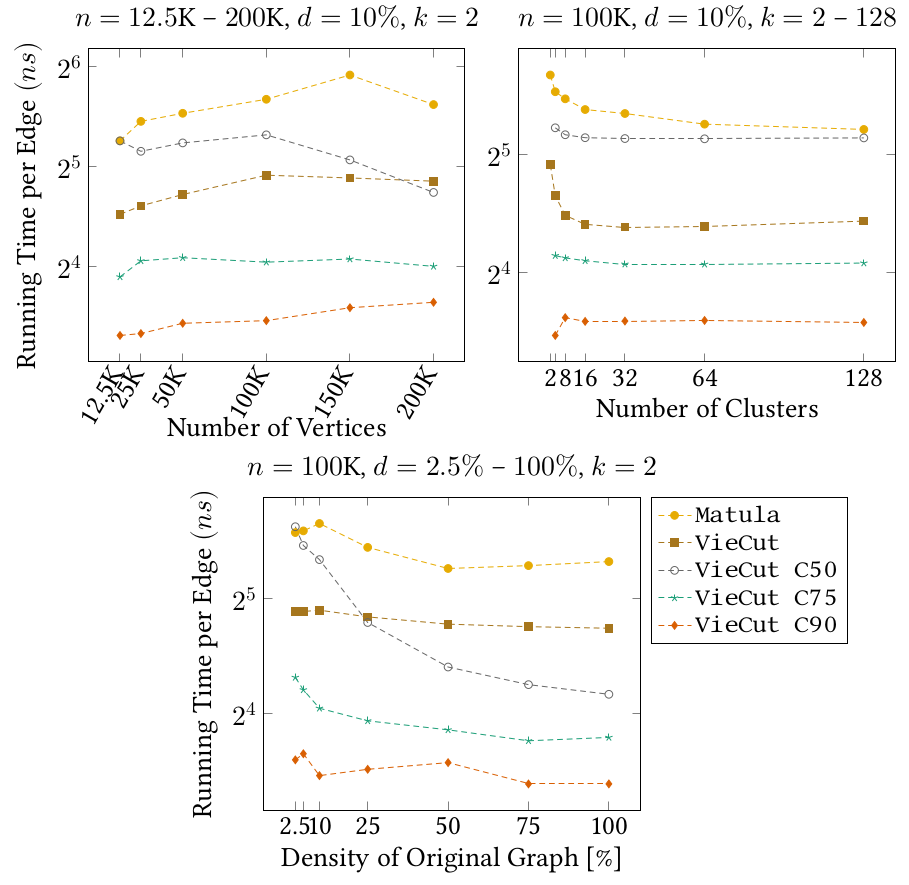}
  \caption{Total running time in nanoseconds per edge in clustered
  Erd\H{o}s-Rényi graphs\label{fig:test3_speed_ks}}
\end{figure}  

\subsubsection{Clustered Erd\H{o}s-Rényi Graphs}

Figure~\ref{fig:test3_speed_ks} shows the running time for the random edge
contraction algorithm variants compared to \textttA{VieCut} on dense clustered
Erd\H{o}s-Rényi graphs. We can see that \textttA{VieCut C75} and
\textttA{VieCut C90} are always faster than \textttA{VieCut}, with
\textttA{VieCut C90} being faster by a factor of $2.3$ to $2.7$ than
\textttA{VieCut} and \textttA{VieCut C75} being faster by a factor of $1.25$
to $1.8$.

Figure~\ref{fig:error_er}~(left) shows the average distance to the optimal cut
for all algorithms which do not guarantee optimality, both as the difference and
the factor of the returned cut to the optimal. In these highly regular graphs,
we find the optimal cut if none of the low-weight edges between the clusters is
contracted. Otherwise we find a cut where one side is only a single vertex. On
average, this cut is around twice the value of the minimum cut. The algorithms
all have an average cut factor of up to $1.7$ on all graphs, depending on how
many vertices we contract. \textttA{KS} has a similar error rate on the graphs
where it finishes. Both \textttA{VieCut} and \textttA{Matula} have no errors on
these graphs.

\begin{figure}[t]
  \centering
\includegraphics[width=.68\textwidth]{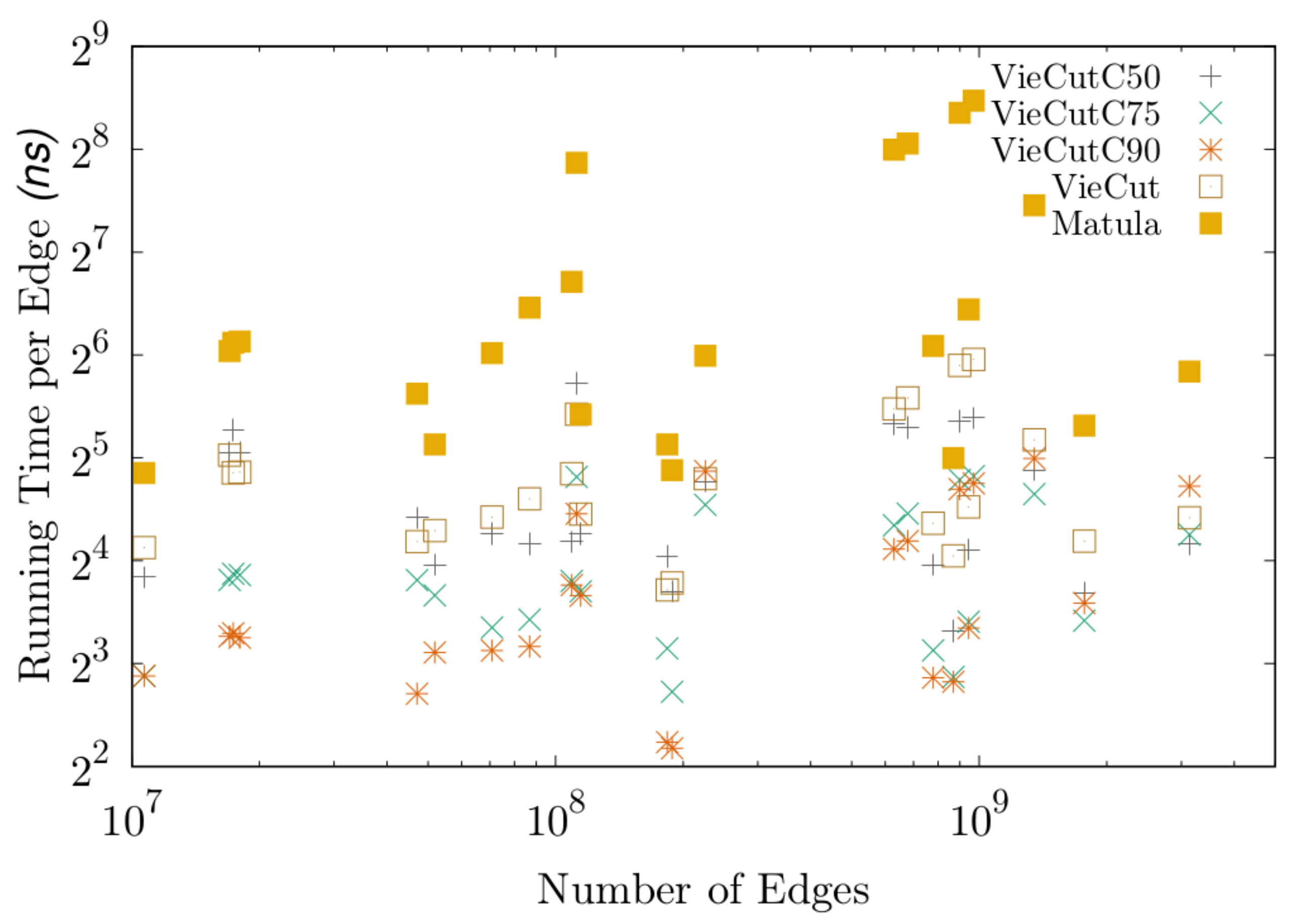}
\includegraphics[width=.68\textwidth]{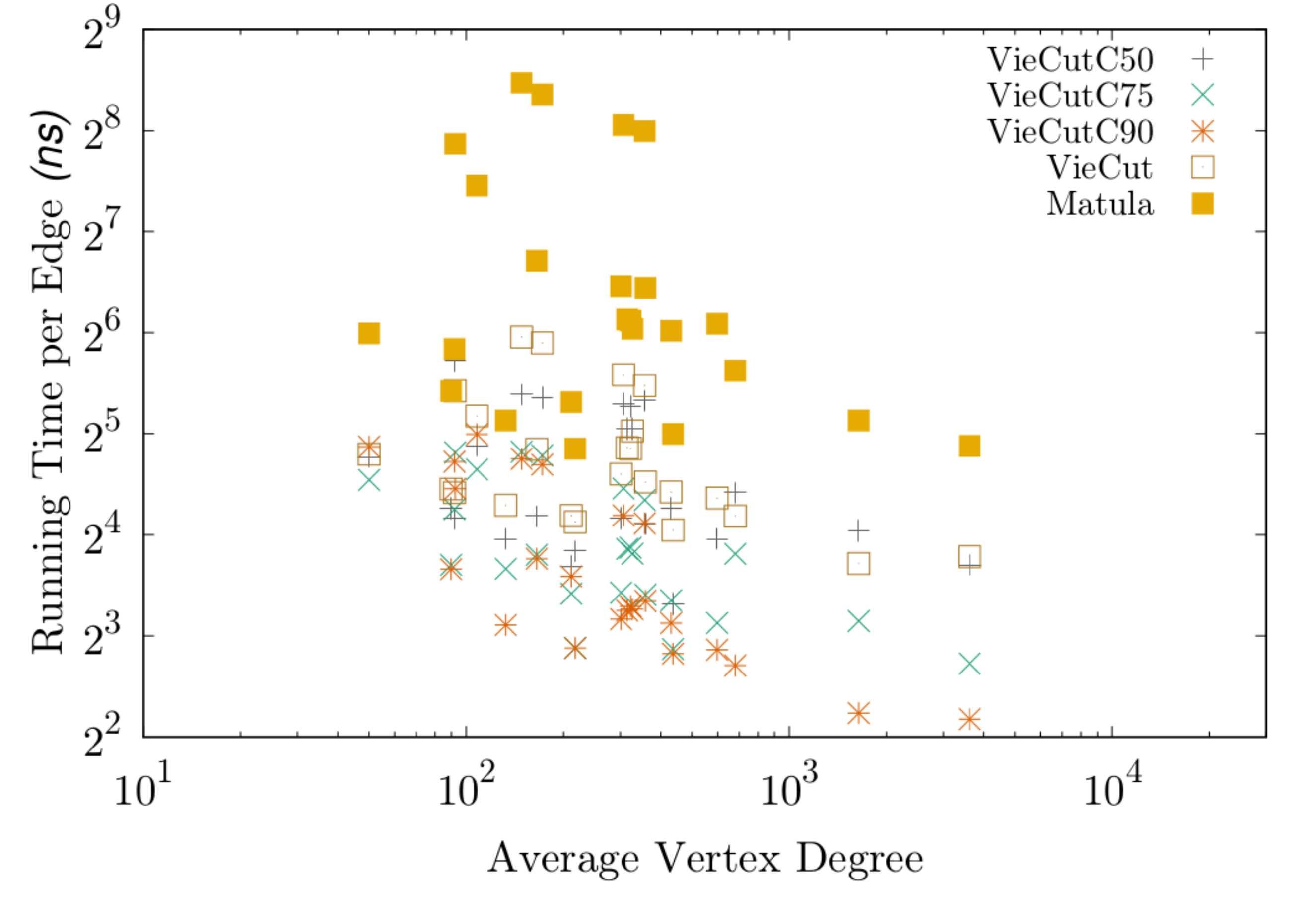}
\caption{Running time per edge in nanoseconds on real-world graphs. We display
running times based on the absolute number of edges (top), and by the average
vertex degree (bottom) in the graph\label{fig:real_speed_ks}}
\end{figure}

\subsubsection{Real-World Graphs}

Figure~\ref{fig:real_speed_ks} shows the average running time for the random
edge contraction variants of \textttA{VieCut} on real-world graphs. \textttA{VieCut C90} is faster than \textttA{VieCut} by a factor of up to
$3.40$. The lowest speedup factor is $0.80$ (a slowdown of $25\%$).
\textttA{VieCut C75} has speedup factors between $1.12$ to $2.35$ compared to
\textttA{VieCut}.

Figure~\ref{fig:error_er}~(middle) shows the average error rates. On these
graphs, the average ratio of cut size to the optimal cut size is very low for
the random contraction algorithms. The outlier for \textttA{VieCut C75} and
\textttA{VieCut C90} is a single run of the graph \textttA{twitter-2010} with
$k=60$, where one of the edges in the unique small cut (of value $3$) is
contracted. The next smallest cut in the graph is a trivial cut with a cut value
of $60$, which is found. The algorithm finds optimal results in the other four
iterations of this graph. On all other graphs, the cut factor is below $1.18$.
\textttA{KS} has an average cut rate of up to $1.33$. \textttA{VieCut} and
\textttA{Matula} have no errors on these graphs.

\begin{figure}[t]
  %%FIGURE 9
  \centering
  \includegraphics[width=\textwidth]{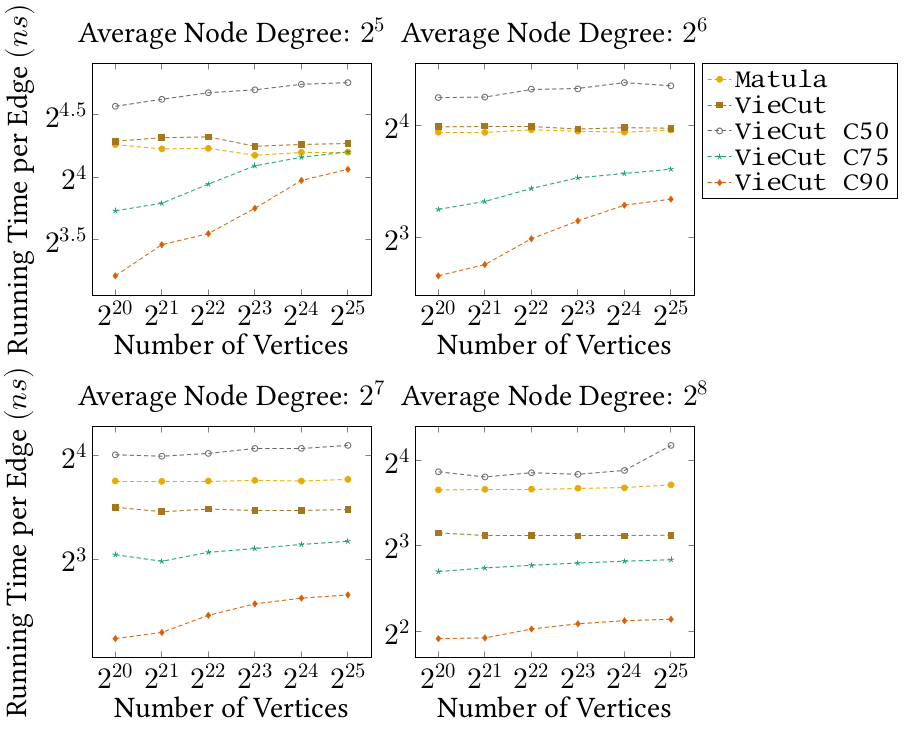}
  \caption{Total running time in nanoseconds per edge in RHG graphs
\label{fig:rhgspeed_ks}} 
\end{figure}

\subsubsection{Random Hyperbolic Graphs}

Figure~\ref{fig:rhgspeed_ks} shows the running time of the algorithms on random
hyperbolic graphs. \textttA{VieCut C75} has a speedup of $1.05$ to $1.67$
compared to \textttA{VieCut} and \textttA{VieCut C90} has a speedup of $1.15$
to $2.54$.

Figure~\ref{fig:error_er}~(right) shows the error rate for RHG graphs. The error
rate of \textttA{VieCut} is far lower than all other non-exact algorithms.
\textttA{Matula} has a non-optimality factor that is between \textttA{VieCut
C50} and \textttA{VieCut C75} but a smaller distance to the optimal solution
than both of them. However, \textttA{VieCut C75} is much faster than
\textttA{Matula}, which is consistently slower than \textttA{VieCut}. Note that
the 3 graphs in which \textttA{Matula} has a cut $3$ times as large as optimal
are graphs with a minimum cut of $1$, where \textttA{Matula} consistently returns
cuts of value $3$. Our implementation of \textttA{Matula} contracts all edges in
the spanning forest with index $\lfloor \frac{\hat\lambda}{2} \rfloor$.

\subsubsection{Shared-Memory Parallelism}

\begin{figure}[t]
  \centering
  %% FIGURE 10
  \includegraphics[width=\textwidth]{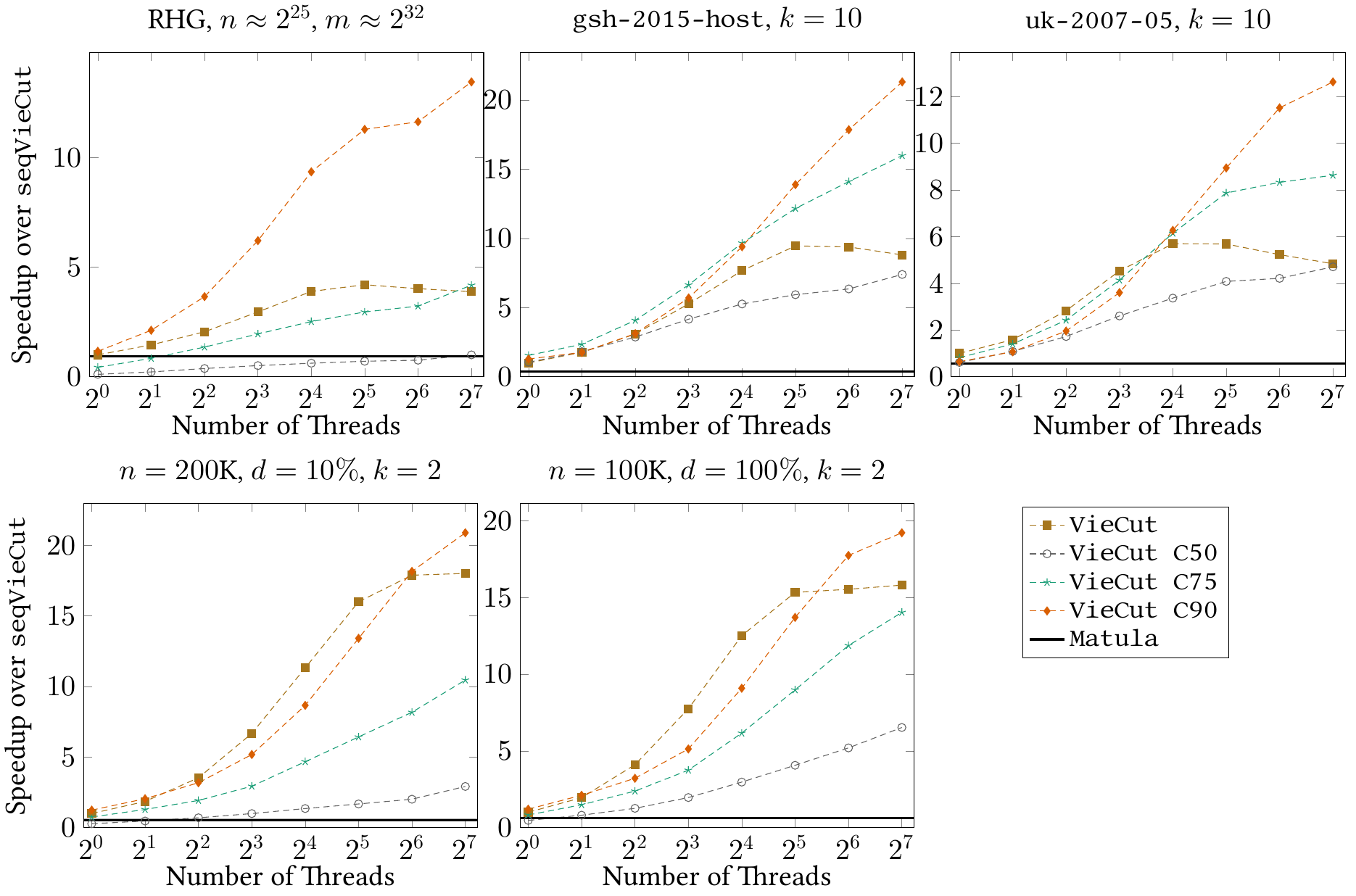}
  \caption{Speedup on large graphs over \textttA{VieCut} using 1
thread.\label{par_cont}}
\end{figure}

Figure~\ref{par_cont} shows the speedup of the random contraction variants in
comparison to \textttA{VieCut}. The RHG graph has a single smallest cut with
value $73$, followed by trivial cuts with a degree of $139$ each.
\textttA{VieCut} and \textttA{Matula} return the correct minimum cut each run,
\textttA{VieCut C50} $34$ out of $40$ times, \textttA{VieCut C75} $28$ times,
\textttA{VieCut C90} $24$ times and $139$ otherwise. Only \textttA{VieCut C90}
scales better than \textttA{VieCut} and has a speedup of up to $13.4$ compared
to sequential \textttA{VieCut}. Due to the large number of clusters, where each
cluster has many incident edges, the contraction step in \textttA{VieCut C75}
and \textttA{VieCut C50} takes a long time.

The random contraction variants scale better in the large real-world graphs,
where all algorithms return the minimum cut in all runs. On graph
\textttA{gsh-2015-host}, \textttA{VieCut C90} has a speedup of over $21$, while
\textttA{VieCut C75} has a speedup of $16$ and \textttA{VieCut} has a speedup
of $9.5$ with $32$ threads. On \textttA{uk-2007-05}, \textttA{VieCut C90} has a
speedup of over $12$, \textttA{VieCut C75} of over $8.5$ and \textttA{VieCut}
a speedup of up to $5.7$.

On the clustered Erd\H{o}s-Rényi graphs, the random edge contraction creates one
or few very large blocks of vertices. Again, \textttA{VieCut} and
\textttA{Matula} always return the correct minimum cut. Out of $80$ runs on
clustered Erd\H{o}s-Rényi graphs, \textttA{VieCut C50} returns the minimum cut
in $46$ cases, \textttA{VieCut C75} in $8$ cases and \textttA{VieCut C90} in
$12$ cases. On these very dense and unstructured graphs, random contraction has
a high error rate and does not significantly speed up \textttA{VieCut}.
\textttA{VieCut} however has a speedup of up to $18$ on these graphs.
Figure~\ref{fig:harmonic} shows the average speedup of the algorithms compared
to their performance with one thread on machine B.

\begin{figure}[t]
  \centering
  %% FIGURE 11
  \includegraphics[width=.5\textwidth]{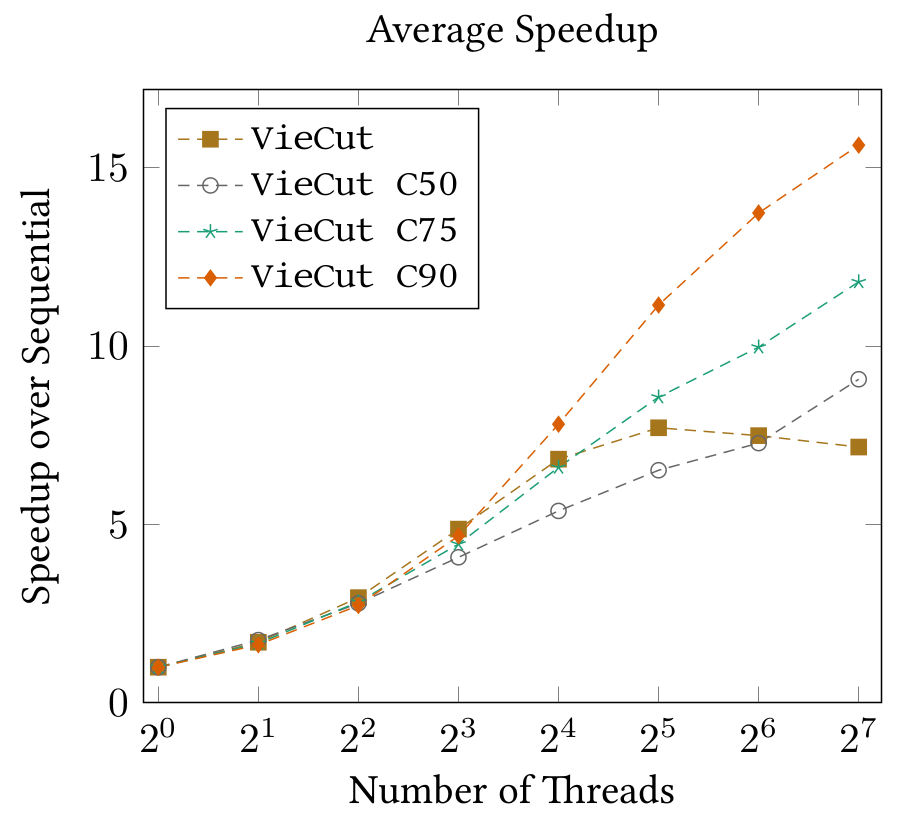}
  \caption{\label{fig:harmonic} Average (harmonic) speedup relative to algorithm
performance using $1$ thread.}
\end{figure}

In summary, the random contraction variants can improve the running time of
\textttA{VieCut} even further, especially when we have many threads. However,
the random contraction increases the error rate of the algorithm by a large
margin and should therefore only be used if running time is more important than
absolute solution quality.

\section{Conclusion}\label{c:vc:s:conclusions}

We presented the linear-time heuristic algorithm \textttA{VieCut} for the
minimum cut problem. \textttA{VieCut} is based on the label propagation
algorithm~\cite{raghavan2007near} and the Padberg-Rinaldi
heuristics~\cite{padberg1990efficient}. Both for real-world graphs and a varied
family of generated graphs, \textttA{VieCut} is significantly faster than the
state of the art. The algorithm has far higher solution quality than other heuristic
algorithms while also being faster. Additionally, we propose a variant of our
algorithm to further speed up computations at the expense of higher error rates.
Important future work includes checking whether using different clustering
techniques affect the observed error probability. However, these clustering
algorithms generally have higher running time.

\chapter{Exact Global Minimum Cut}
\label{c:exmc}

%% IPDPS'19 Paper

In the previous chapter, we introduced a heuristic shared-memory parallel
algorithm for the global minimum cut problem called \textttA{VieCut}. In this
chapter, we combine techniques from that algorithm, the algorithm of
Padberg and Rinaldi~\cite{padberg1990efficient} and the algorithm of Nagamochi, Ono
and Ibaraki~\cite{nagamochi1992computing,nagamochi1994implementing}, as
introduced in Section~\ref{p:mincut:s:algs}, to engineer an exact shared-memory
parallel algorithm for the minimum cut problem. Our algorithm achieves
improvements in running time over existing exact algorithms by a multitude of techniques.
First, we use our fast and parallel \emph{inexact} minimum cut algorithm
\textttA{VieCut} to obtain a better bound for
the problem. Afterwards, we use reductions that depend on this bound to reduce
the size of the graph much faster than previously possible. We use improved data
structures to further lower the running time of our algorithm. Additionally, we
parallelize the contraction routines of
Nagamochi~\etal\cite{nagamochi1992computing,nagamochi1994implementing}. Overall,
we arrive at a system that outperforms the state-of-the-art by a factor of up to
$2.5$ sequentially, and when run in shared-memory parallel, by a factor of up to
$12.9$ using $12$ cores.

The content of this chapter is based on~\cite{DBLP:conf/ipps/HenzingerN019}.

In the following sections we detail our \emph{exact shared-memory parallel algorithm} for
the minimum cut problem that is based on the algorithms of
Nagamochi et al., as
described in Section~\ref{p:mincut:ss:noi} and the \textttA{VieCut} algorithm
described in Chapter~\ref{c:viecut} of this thesis. We aim to modify the algorithm of
Nagamochi~\etal in order to find exact minimum
cuts faster and in parallel. 

We start this chapter with optimizations to the sequential algorithm of
Nagamochi et al. First we show how to save work by first performing the inexact \textttA{VieCut} algorithm to lower the minimum cut upper bound
$\hat \lambda$. As shown in Chapter~\ref{c:viecut}, \textttA{VieCut} often already
finds a cut of value $\lambda$.
We then give different implementations of the priority queue
$\mathcal{Q}$ and detail the effects of the choice of queue on the algorithm. We
show that the algorithm remains correct, even if we limit the priorities in the
queue to $\hat\lambda$, meaning that elements in the queue having a key larger
than that will not be updated. This significantly lowers the number of priority
queue operations necessary. Then we adapt the algorithm so that we are able to
detect contractible edges in parallel efficiently. In
Section~\ref{c:exmc:s:sys}, we put everything together and present a full system
description. We then give experimental setup and results of our work in
Section~\ref{c:exmc:s:exp} before we briefly conclude this chapter in
Section~\ref{c:exmc:s:con}.

\section{Sequential Optimizations}
\label{c:exmc:seqopt}

\subsection{Lowering the Upper Bound $\hat\lambda$}

The upper bound $\hat\lambda$ for the minimum cut is an important parameter for
contraction based minimum cut algorithms. For example, the algorithm of
Nagamochi~\etal\cite{nagamochi1994implementing} computes a lower bound for the
connectivity of the two incident vertices of each edge and contracts all edges
whose incident vertices have a connectivity of at least $\hat\lambda$. Thus, it
is possible to contract more edges if we manage to lower $\hat\lambda$
beforehand.

A trivial upper bound $\hat\lambda$ for the minimum cut is the minimum vertex
degree, as it represents the trivial cut which separates the minimum degree
vertex from all other vertices. We run \textttA{VieCut} to lower $\hat\lambda$
in order to allow us to find more edges to contract. Although \textttA{VieCut}
is an \emph{inexact algorithm}, in most cases it already finds the minimum
cut~\cite{henzinger2018practical} of the graph. As there are by definition no
cuts smaller than the minimum cut, the result of \textttA{VieCut} is guaranteed
to be at least as large as the minimum cut $\lambda$. We set $\hat\lambda$ to
the result of \textttA{VieCut} when running the CAPFOREST routine and can
therefore guarantee a correct result.

A similar idea is employed by the linear time $(2+\epsilon)$-approximation
algorithm of Matula~\cite{matula1993linear}, which initializes the algorithm of
Nagamochi~\etal\cite{nagamochi1994implementing} with $\hat\lambda = (\frac{1}{2}
- \epsilon) \cdot $min degree. The algorithm of Matula does not guarantee
optimality, as this value can be smaller than the minimum cut.

\subsection{Bounded Priority Queues}
\label{c:exmc:sec:lem}

Whenever we visit a vertex in the CAPFOREST algorithm, we update the priority of
all of its neighbors in $\mathcal{Q}$ by adding the respective edge weight. Thus
we perform a total of $|E|$ priority queue increase-weight operations in one
call of the CAPFOREST algorithm. In practice, many vertices reach priority
values much higher than $\hat\lambda$ and perform many priority increases until
they reach their final value. We limit the values in the priority queue
by $\hat\lambda$, \ie we do not update priorities that are already
$\hat\lambda$. Lemma~\ref{c:exmc:lem:pq} shows that this does not affect
correctness of the algorithm.

Let $\tilde q_G(e)$ be the value $q(e)$ assigned to $e$ in the modified
algorithm on graph $G$ and let $\tilde r_G(x)$ be the $r$-value of a node $x$ in
the modified algorithm on $G$.

\begin{lemma} \label{c:exmc:lem:pq} Limiting the values in the priority queue
  $\mathcal{Q}$ used in the CAPFOREST routine to a maximum of $\hat\lambda$ does
  not interfere with the correctness of the algorithm. For every edge $e=(v,w)$
  with $\tilde q_G(e) \geq \hat\lambda$, it holds that $\lambda(G,e) \geq
  \hat\lambda$. Therefore the edge can be contracted.
\end{lemma}
\begin{proof}
  As we limit the priority queue $\mathcal{Q}$ to a maximum value of
  $\hat\lambda$, we cannot guarantee that we always pop the element with
  highest value $r(v)$ if there are multiple elements that have values $r(v)
  \geq \hat\lambda$ in $\mathcal{Q}$. However, we know that the vertex $x$ that
  is popped from $\mathcal{Q}$ is either maximal or has $r(x) \geq \hat\lambda$.

  We prove Lemma~\ref{c:exmc:lem:pq} by creating a graph $G'=(V,E,c')$ by
  lowering edge weights (possibly to $0$, effectively removing the edge) while
  running the algorithm, so that CAPFOREST on $G'$ visits vertices in the same
  order (assuming equal tie breaking) and assigns the same $q$ values as the
  modified algorithm~on~$G$.

  We first describe the construction of $G'$. We initialize the weight of all
  edges in graph $G'$ with the weight of the respective edge in $G$ and run
  CAPFOREST on $G'$. Whenever we check an edge $e = (x,y)$ and update a value
  $r_{G'}(y)$ , we check whether we would set $r_{G'}(y) > \hat\lambda$. If this
  is the case, \ie when $r_{G'}(y) + c(e) > \hat\lambda$, we set $c'(e)$ in $G'$
  to $c(e) - (r_{G'}(y) - \hat\lambda)$, which is lower by exactly the value by
  which $r_{G}(y)$ is larger than $\hat\lambda$, and non-negative. Thus,
  $r_{G'}(y) = \hat\lambda$. As we scan every edge exactly once in a run of
  CAPFOREST, the weights of edges already scanned remain constant afterwards.
  This completes the construction of $G'$

  Note that during the construction of $G'$ edge weights were only decreased and
  never increased. Thus it holds that $\lambda(G',x,y) \leq \lambda(G,x,y)$ for
  any pair of nodes $(x,y)$. If we ran the unmodified CAPFOREST algorithm on
  $G'$ each edge would be assigned a value $q_{G'}(e)$ with $q_{G'}(e) \leq
  \lambda(G',e)$. Thus for every edge $e$ it holds that $q_{G'}(e) \leq
  \lambda(G',e) \leq \lambda(G,e)$.

  Below we will show that $\tilde q_G(e) = q_{G'}(e)$ for all edges $e$. It then
  follows that for all edges $e$ it holds that $\tilde q_G(e) \leq
  \lambda(G,e)$. This implies that if $\tilde q_G(e) \geq \hat\lambda$ then
  $\lambda(G,e) \geq \hat\lambda$, which is what we needed to show.

  It remains to show that for all edges $e$ $\tilde q_G(e) = q_{G'}(e)$. To show
  this claim we will show the following stronger claim. For any $i$ with $1 \leq
  i \leq m$ after the $(i-1)$th and before the $i$th scan of an edge the
  modified algorithm on $g$ and the original algorithm on $G'$ with the same tie
  breaking have visited all nodes and scanned all edges up to now in the same
  order and for all edges $e$ it holds that $\tilde q_G(e) = q_{G'}(e)$ (we
  assume that before scanning an edge $e$, $q(e) = 0$) and for all nodes $x$ it
  holds that $\tilde r_G(x) = r_{G'}(x)$. We show this claim by the induction on
  $i$.

  For $i=1$ observe that before the first edge scan $\tilde q_G(e) = q_{G'}(e) =
  0$ for all edges $e$ and the same node is picked as first node due to
  identical tie breaking and the fact that $G=G'$ at that point. Now for $i > 1$
  assume that the claim holds for $i-1$ and consider the scan of the $(i-1)^\text{th}$
  edge. If for the $(i-1)^\text{th}$ edge scan a new node needs to be chosen from the
  priority queue by one of the algorithms then note that both algorithms will
  have to choose a node and they pick the same node $y$ as $\tilde r_G(x) =
  r_{G'}(x)$ for all nodes $x$. Then both algorithms scan the same incident edge
  of $y$ as in both algorithms the set of unscanned neighbors of $y$ is
  identical. If neither algorithm has to pick a new node then both have scanned
  the same edges of the same current node $y$ and due to identical tie breaking
  will pick the same next edge to scan. Let this edge be $(y,w)$. By induction
  $\tilde r_G(w) = r_{G'}(w)$ at this time. As $(y,w)$ is unscanned $c'(y,w) =
  c(y,w)$ which implies that $\tilde r_G(w) + c(y,w) = r_{G'}(w)+c'(y,w)$. If
  $\tilde r_G(w) + c(y,w) \leq \hat\lambda$ then the modified algorithm on $G$
  and the original algorithm on $G'$ will set the $r$ value of $w$ to the same
  value, namely $\tilde r_G(w) + c(y,w)$. If $\tilde r_G(w) + c(y,w) >
  \hat\lambda$, then $\tilde r_G(w)$ is set to $\hat\lambda$ and $c'(y,w)$ is
  set to $c(y,w)-(r_{G'}(w)-\hat\lambda)$, which leads to $r_{G'}(w)$ being set
  to $\hat\lambda$. Thus $\tilde r_G(w)=r_{G'}(w)$ and by induction $\tilde
  r_G(x) = r_{G'}(x)$ for all $x$. Additionally the modified algorithm on $G$
  sets $\tilde q_G(y,w) = \tilde r_G(w)$ and the original algorithm on $G'$ sets
  $q_{G'}(y,w)=r_{G'}(w)$. It follows that $\tilde q_G(y,w)=q_{G'}(y,w)$ and,
  thus, by induction $\tilde q_G(e) = q_{G'}(e)$ for all $e$. This completes the
  proof of the claim.
\end{proof}

Lemma~\ref{c:exmc:lem:pq} allows us to considerably lower the number of priority queue
operations, as we do not need to update priorities that are bigger than
$\hat\lambda$. This optimization has even more benefit in combination with
running \textttA{VieCut} to lower the upper bound $\hat\lambda$, as we further
lower the number of priority queue operations.

\subsection{Priority Queue Implementations}
\label{c:exmc:ssec:pq}

Nagamochi~\etal\cite{nagamochi1994implementing} use an addressable priority
queue $\mathcal{Q}$ in their algorithm to find contractible edges. In this
section we now address variants for the implementation of the priority queue. As
the algorithm often has many elements with maximum priority in practice, the
implementation of this priority queue can have major impact on the order of
vertex visits and thus also on the edges that will be marked contractible.

\subsubsection{Bucket Priority Queue}
As our algorithm limits the values in the priority queue to a maximum of $\hat
\lambda$, we observe integer priorities in the range of $[0,\hat\lambda]$.
Hence, we can use a bucket queue that is implemented as an array with
$\hat\lambda$ buckets. In addition, the data structure keeps the id of the
highest non-empty bucket, also known as the \emph{top bucket}, and stores the
position of each vertex in the priority queue. Priority updates can be
implemented by deleting an element from its bucket and pushing it to the bucket
with the updated priority. This allows constant time access for all operations
except for deletions of the maximum priority element, which have to check all
buckets between the prior top bucket and the new top bucket, possibly up to
$\hat\lambda$ checks. We give two possible implementations to implement the
buckets so that they can store all elements with a given priority.

The first implementation, \textttA{BStack} uses a dynamic array
(\textttA{std::vector}) as the container for all elements in a bucket. When we
add a new element to the array, we push it to the back of the array.
\texttt{$\mathcal{Q}$.pop\_max()} returns the last element of the top bucket.
Thus, our algorithm will always visit the element next whose priority was just
increased. It thus does not fully explore all vertices in a region and instead
behaves more similar to a depth-first search.

The other implementation, \textttA{BQueue} uses a double ended queue
(\textttA{std::deque}) as the container instead. A new element is pushed to the
back of the queue and \texttt{$\mathcal{Q}$.pop\_max()} returns the \emph{first}
element of the top bucket. This results in a variant of our algorithm, which
behaves more similar to a breadth-first search in that it first explores the
vertices that have been discovered earlier, \ie are closer to the source vertex
in the graph.

\subsubsection{Bottom-Up Binary Heap}
A binary heap~\cite{williams1964heapsort} is a binary tree (implemented as an
array, where element $i$ has its children in index $2i$ and $2i+1$) which
fulfills the heap property, \ie each element has priority that is not lower than
either of its children. Thus the element with highest priority is the root of
the tree. The tree can be made addressable by using an array of indices, in
which we save the position of each vertex. We use a binary heap using the
bottom-up heuristics~\cite{wegener1993bottom}, in which we sift down holes that
were created by the deletion of the top priority vertex. Priority changes are
implemented by sifting the addressed element up or down in the tree. Operations
have a running time of up to $\mathcal{O}(\log n)$ to sift an element up or down
to fix the heap property.

In \texttt{$\mathcal{Q}$.pop\_max()}, the \textttA{Heap} priority queue does not
favor either old or new elements in the priority queue and therefore this
implementation can be seen as a middle ground between the two bucket priority
queues.

\begin{algorithm}[t!]
  \begin{algorithmic}[1]
    \INPUT $G = (V,E,c) \leftarrow$ undirected graph $\hat\lambda \leftarrow$
    upper bound for minimum cut, $\mathcal{T} \leftarrow$ shared array of vertex
    visits 
    \OUTPUT $\mathcal{U} \leftarrow$ union-find data structure to mark
    contractible edges 
    \State Label all vertices $v \in V$ ``unvisited'',
    blacklist $\mathcal{B}$ empty
    \State $\forall v \in V: r(v) \leftarrow 0$
    \State $\forall e \in E: q(e) \leftarrow 0$
    \State $\mathcal{Q} \leftarrow$ empty priority queue
    \State Insert random vertex into $\mathcal{Q}$
    \While{$\mathcal{Q}$ not empty}
    \State $x \leftarrow \mathcal{Q}$.pop\_max()
    \Comment{Choose unvisited vertex with highest priority}
    \State Mark $x$ ``visited''
    \If{$\mathcal{T}(x) = \text{True}$} \Comment{Every vertex is visited only once}
    \State $\mathcal{B}(x) \leftarrow \text{True}$
    \Else
    \State $\mathcal{T}(x) \leftarrow \text{True}$
    \State $\alpha \leftarrow \alpha + c(x) - 2 r(x)$
    \State $\hat\lambda \leftarrow min(\hat\lambda, \alpha)$
    \For{$e = (x,y) \leftarrow$ unscanned edge, where $y \not\in \mathcal{B}$}
    \If{$r(y) < \hat\lambda \leq r(y) + c(e)$}
    \State $\mathcal{U}$.union(x,y) \Comment{Mark edge $e$ to contract}
    \EndIf
    \State $r(y) \leftarrow r(y) + c(e)$
    \State $q(e) \leftarrow r(y)$
    \State $\mathcal{Q}(y) \leftarrow min(r(y), \hat\lambda)$
    \EndFor
    \EndIf
    \EndWhile

  \end{algorithmic}
  \caption{\label{c:exmc:algo:parnoi} Parallel CAPFOREST}
\end{algorithm}

\section{Parallel CAPFOREST}
\label{c:exmc:ssec:pmcap}

We modify the algorithm in order to quickly find contractible edges using
shared-memory parallelism. The pseudocode can be found in
Algorithm~\ref{c:exmc:algo:parnoi}. The proofs in this section show that the
modifications do not violate the correctness of the algorithm. Detailed proofs
for the original CAPFOREST algorithm and the modifications of Nagamochi~\etal
for weighted graphs can be~found~in~\cite{nagamochi1994implementing}.
%The algorithm starts at a random vertex and traverses the graph in the order such that the next visited vertex is the not previously visited vertex that is most strongly connected to the vertices already visited.

The idea of the our algorithm is as follows: We aim to find contractible edges
using shared-memory parallelism. Every processor selects a random vertex and
runs Algorithm~\ref{c:exmc:algo:parnoi}, which is a modified version of
CAPFOREST~\cite{nagamochi1992computing,nagamochi1994implementing} where the
priority values are limited to $\hat\lambda$, the current upper bound of the
size of the minimum cut. We want to find contractible edges without requiring
that every process looks at the whole graph. To achieve this, every vertex will
only be visited by one process. Compared to limiting the number of vertices each
process visits this has the advantage that we also scan the vertices in sparse
regions of the graph which might otherwise not be scanned by any process.

Figure~\ref{c:exmc:fig:parmc} shows an example run of Algorithm~\ref{c:exmc:algo:parnoi}
with~$p=5$. Every process randomly chooses a start vertex and performs
Algorithm~\ref{c:exmc:algo:parnoi} on it to ``grow a region'' of scanned
vertices. As we want to employ shared-memory parallelism to speed up the
algorithm, we share an array $\mathcal{T}$ between all processes to denote
whether a vertex has already been visited. If a vertex $v$ has already been
visited by a process, it will not be visited by any other processes.
Additionally, every process keeps a local blacklist $\mathcal{B}$ for vertices
that the process attempted to visit but that were already visited by another
process before and were thus ignored by this process. Note that $\mathcal{B}$ is
not shared between processes. We need this blacklist to ensure correctness, as a
process may only contract edges that are not adjacent to a vertex previously
blacklisted by that process (proof in Lemma~\ref{c:exmc:lem:parts}). For every
vertex~$v$ we keep a value $r(v)$, which denotes the total weight of edges
connecting $v$ to already scanned vertices. Over the course of a run of the
algorithm, every edge~$e~=~(v,w)$ is given a value $q(e)$ (equal to $r(w)$ right
after scanning $e$) which is a lower bound for the smallest cut
$\lambda(G,v,w)$.  We mark an edge $e$ as contractible (more accurately, we
union the incident vertices in the shared concurrent union-find data
structure~\cite{anderson1991wait}), if $q(e) \geq \hat\lambda$. Note that this
does not modify the graph, it just remembers which nodes to collapse. The actual
node collapsing happens in a postprocessing step. Nagamochi and Ibaraki showed
\cite{nagamochi1994implementing} that contracting only the edges that fulfill
the condition in line $16$ is equivalent.

\begin{figure}[t!]
  \centering
  \includegraphics[width=.6\textwidth]{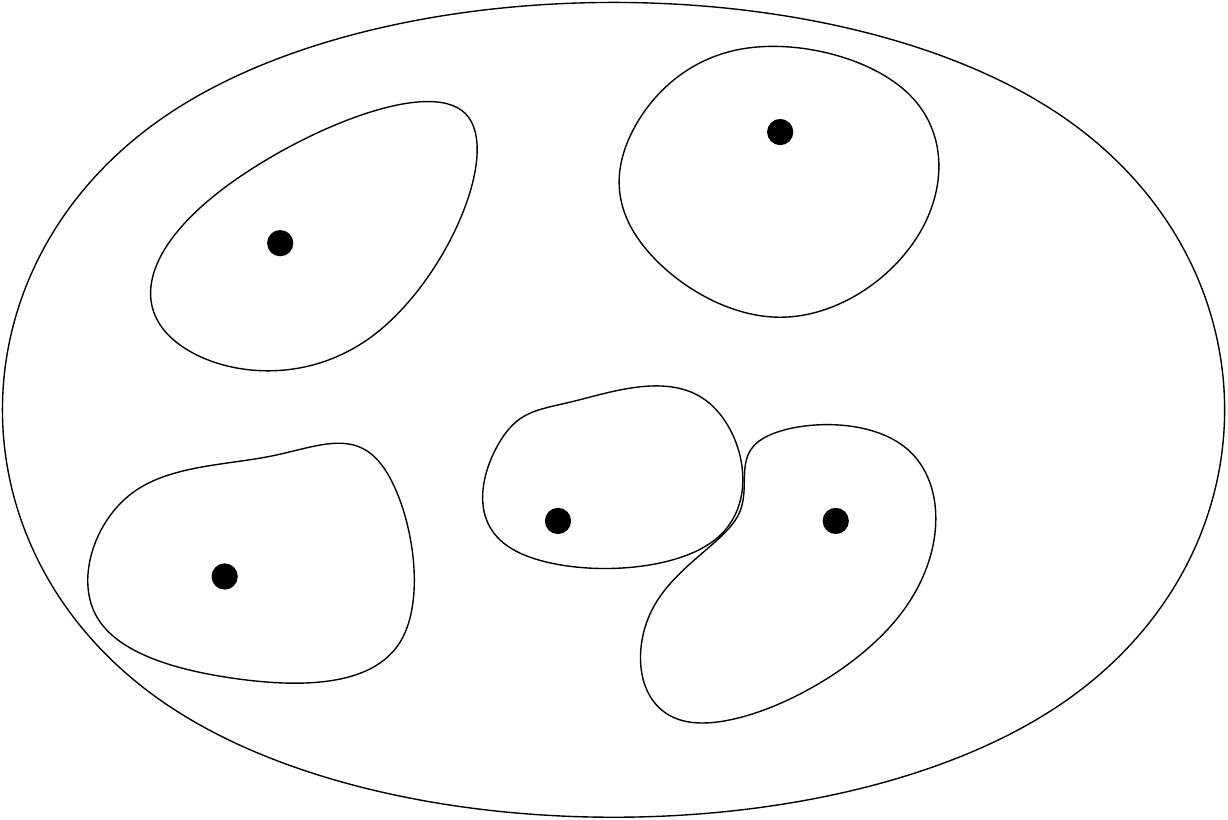}
  \caption{\label{c:exmc:fig:parmc} Example run of Algorithm~\ref{c:exmc:algo:parnoi}. Every
  process starts at a random vertex and scans the region around the start vertex.
  These regions do not overlap.}
\end{figure}

As the set of disconnected edges is different depending on the start vertex, we
looked into visiting every vertex by a number of processes up to a given
parameter to find more contractible edges. However, this did generally result in
higher total running times and thus we only visit every vertex once.

After all processes are finished, every vertex was visited exactly once (or
possibly zero times, if the graph is disconnected). On average, every process
has visited roughly $\frac{n}{p}$ vertices and all processes finish at the same
time. We do not perform any form of locking of the elements of $\mathcal{T}$, as
this would come with a running time penalty for every write and the only
possible race condition with concurrent writes is that a vertex is visited more
often, which does not affect correctness of the algorithm.

However, as we terminate early and no process visits every vertex, we cannot
guarantee that the algorithm actually finds a contractible edge.
However, in practice, this only happens if the graph is already very small
($<50$ vertices in all of our experiments). We can then run the (sequential)
CAPFOREST routine to find at least one edge which can be contracted. In line
$13$ and $14$ of Algorithm~\ref{c:exmc:algo:parnoi} we compute the value of the
cut between the scanned and unscanned vertices and update $\hat\lambda$ if this
cut is smaller than it. This optimization to the CAPFOREST algorithm was first
given by Nagamochi~\etal\cite{nagamochi1994implementing}.

In practice, many vertices reach values of $r(y)$ that are much higher than
$\hat\lambda$ and therefore need to update their priority in $\mathcal{Q}$
often. As previously detailed, we limit the values in the priority queue by
$\hat\lambda$ and do not update priorities that are already greater or equal
to $\hat\lambda$. This allows us to considerably lower the number of priority
queue operations per vertex.

\begin{theorem}
  Algorithm~\ref{c:exmc:algo:parnoi} is correct.
\end{theorem}

As Algorithm~\ref{c:exmc:algo:parnoi} is a modified variant of
CAPFOREST~\cite{nagamochi1992computing,nagamochi1994implementing}, we use the
correctness of their algorithm and show that our modifications cannot result in
incorrect results. In order to show this we need the following lemmas:

  \begin{lemma}
    The following modifications to the CAPFOREST algorithms do not result in
    incorrect results.
  \begin{enumerate}
    \item Multiple instances of Algorithm~\ref{c:exmc:algo:parnoi} can be run in
    parallel with all instances sharing a parallel union-find data structure.
    \item Early termination does not affect correctness
    \item For every edge $e = (v,w)$, where neither $v$ nor $w$ are blacklisted,
    $q(e)$ is a lower bound for the connectivity $\lambda(G, v, w)$, even if the
    set of blacklisted vertices $\mathcal{B}$ is not empty.
    \item When limiting the priority of a vertex in $\mathcal{Q}$ to
    $\hat\lambda$, it still holds that the vertices incident to an edge $e =
    (x,y)$ with $q(e) \geq \hat\lambda$ have connectivity $\lambda(G,x,y) \geq
    \hat\lambda$.
  \end{enumerate}
      \label{c:exmc:lem:parts}
\end{lemma}

\begin{proof}
  A run of the CAPFOREST algorithm finds a non-empty set of edges that can be
  contracted without contracting a cut with value less than
  $\hat\lambda$~\cite{nagamochi1992computing}. We show that none of our
  modifications can result in incorrect results:

  \begin{enumerate}
  \item The CAPFOREST routine can be started from an arbitrary vertex and finds
  a set of edges that can be contracted without affecting the minimum cut
  $\lambda$. This is true for any vertex $v \in V$. As we do not change the
  underlying graph but just mark contractible edges, the correctness is
  obviously upheld when running the algorithm multiple times starting at
  different vertices. This is also true when running the different iterations in
  parallel, as long as the underlying graph is not changed.

    Marking the edge $e=(u,v)$ as contractible is equivalent to performing a
    \textttA{Union} of vertices $u$ and $v$. The \textttA{Union} operation in a
    union-find data structure is commutative and therefore the order of unions
    is irrelevant for the final result. Thus performing the iterations
    successively has the same result as performing them in parallel.

  \item Over the course of the algorithm we set a value $q(e)$ for each edge $e$
  and we maintain a value $\hat\lambda$ that never increases. We contract edges
  that have value~$q(e) \geq \hat\lambda$ at the time when $q(e)$ is set. For
  every edge, this value is set exactly once. If we terminate the algorithm
  prior to setting $q(e)$ for all edges, the set of contracted edges is a subset
  of the set of edges that would be contracted in a full run and all contracted
  edges $e$ fulfill $q(e) \geq \hat\lambda$ at termination. Thus, no edge
  contraction contracts a cut that is smaller than $\hat\lambda$.

    \item Let $e=(v,w)$ be an edge and let $\mathcal{B}_e$ be the set of nodes
    blacklisted at the time when $e$ is scanned. We show that for an edge $e =
    (v, w)$, $q(e) \leq \lambda(\bar G, v, w)$, where $\bar G = (\bar V, \bar
    E)$ with vertices $\bar V = V \backslash \mathcal{B}_e$ and edges $\bar E=
    \{e=(u,v) \in E: u \not\in \mathcal{B}_e $ and $ v \not\in \mathcal{B}_e \}$
    is the graph $G$ with all blacklisted vertices and their incident edges
    removed. As the removal of vertices and edges can not increase edge
    connectivities $q_{\bar G}(e)\leq\lambda (\bar G, v, w) \leq \lambda(G,v,w)$
    and $e$ is contractible.

    Whenever we visit a vertex $b$, we decide whether we blacklist the vertex.
    If we blacklist the vertex $b$, we immediately leave the vertex and do not
    change any values $r(v)$ or $q(e)$ for any other vertex or edge. As vertex
    $b$ is marked as blacklisted, we will not visit the vertex again and the
    edges incident to $b$ only~affect~$r(b)$.
    
    As edges incident to any of the vertices in $\mathcal{B}_e$ do not
    affect~$q(e)$, the value of $q(e)$ in the algorithm with the blacklisted in
    $G$ is equal to the value of $q(e)$ in $\bar G$, which does not contain the
    blacklisted vertices in $\mathcal{B}_e$ and their incident edges. On $\bar
    G$ this is equivalent to a run of CAPFOREST without blacklisted vertices and
    due to the correctness of CAPFOREST~\cite{nagamochi1994implementing} we know
    that for every edge $e \in \bar E: q_{\bar G}(e) \leq \lambda(\bar G, v, w)
    \leq \lambda(G,v,w)$.

    Note that in $\bar G$ we only exclude the vertices that are in
    $\mathcal{B}_e$. It is possible that a node $y$ that was unvisited when $e$
    was scanned might get blacklisted later, however, this does not affect the
    value of $q(e)$ as the value $q(e)$ is set when an edge is scanned and never
    modified afterwards.

\item Proof in Lemma~\ref{c:exmc:lem:pq}.

\end{enumerate}

  We can combine the sub-proofs (3) and (4) by creating the graph $\bar G'$, in
  which we remove all edges incident to blacklisted vertices and decrease edge
  weights to make sure no $q(e)$ is strictly larger than $\hat\lambda$. As we
  only lowered edge weights and removed edges, for every edge between two not
  blacklisted vertices $e = (u,v)$, $q_G(e) \leq \lambda(\bar G',x,y) \leq
  \lambda(G,x,y)$ or $q_G(e) > \hat\lambda$ and thus we only contract
  contractible edges. As none of our modifications can result in the contraction
  of edges that should not be contracted, Algorithm~\ref{c:exmc:algo:parnoi} is
  correct.
\end{proof}

\section{Putting Things Together}
\label{c:exmc:s:sys}

\begin{algorithm}[h!]
  \begin{algorithmic}[1]
    \INPUT $G = (V,E,c)$
    \State $\hat\lambda \leftarrow $ \textttA{VieCut}($G$)
    \State $G_C \leftarrow G$
    \While{$G_C$ has more than $2$ vertices}
    \State $\hat\lambda \leftarrow$ Parallel CAPFOREST($G_C,\hat\lambda$)
    \If{no edges marked contractible}
    \State $\hat\lambda \leftarrow$ CAPFOREST($G_C,\hat\lambda$)
    \EndIf
    \State $G_C, \hat\lambda \leftarrow$ Parallel Graph Contract($G_C$)
    \EndWhile
    \State \Return $\hat\lambda$
  \end{algorithmic}
  \caption{\label{c:exmc:algo:parmc} Parallel Minimum Cut}
\end{algorithm}

Algorithm~\ref{c:exmc:algo:parmc} shows the overall structure of the algorithm.
We first run \textttA{VieCut} to find a good upper bound $\hat\lambda$ for the
minimum cut. Afterwards, we run Algorithm~\ref{c:exmc:algo:parnoi} to find
contractible edges. In the unlikely case that none were found, we run
CAPFOREST~\cite{nagamochi1994implementing} sequentially to find at least one
contractible edge. We create a new contracted graph using parallel graph
contraction with the hash-based shared-memory parallel contraction technique
outlined in the previous chapter. This process is repeated until the graph has
only two vertices left. Whenever we encounter a collapsed vertex with a degree
of lower than $\hat\lambda$, we update the upper bound. We return the smallest
cut we encounter in this process.

If we also want to output the minimum cut, for each collapsed vertex $v_C$ in
$G_C$ we store which vertices of $G$ are included in $v_C$. When we update
$\hat\lambda$, we store which vertices are contained in the minimum cut. This
allows us to see which vertices are on one side of the cut.

\section{Experiments and Results}
\label{c:exmc:s:exp}

\subsection{Experimental Setup and Methodology}

We implemented the algorithms using \CC-17 and compiled all codes using
g++-7.1.0 with full optimization (\textttA{-O3}). Our experiments are conducted
on a machine with two Intel Xeon E5-2643 v4 with 3.4GHz with 6 CPU cores each
and hyper-threading enabled, and 1.5 TB RAM in total. We perform five
repetitions per instance and report average running~time.

\label{c:mc:ss:perf_plots}
\emph{Performance plots}
relate the fastest running time to the running time of each other algorithm on a
per-instance basis. For each algorithm, these ratios are sorted in increasing
order. The plots show the ratio $t_\text{best}/t_\text{algorithm}$ on the
y-axis. % to highlight the instances in which each algorithm performs badly.
A point close to zero indicates that the running time of the algorithm was
considerably worse than the fastest algorithm on the same instance. A value of
one therefore indicates that the corresponding algorithm was one of the fastest
algorithms to compute the solution. Thus an algorithm is considered to
outperform another algorithm if its corresponding ratio values are above those
of the other algorithm. In order to include instances that were too big for an
algorithm, \ie some implementations are limited to 32bit integers, we set the
corresponding ratio below \emph{zero}.

\subsection{Algorithms} There have been multiple experimental studies that
compare exact algorithms for the minimum cut
problem~\cite{Chekuri:1997:ESM:314161.314315,henzinger2018practical,junger2000practical}.
All of these studies report that the algorithm of Nagamochi~\etal and the
algorithm of Hao and Orlin outperform other algorithms, such as the algorithms
of Karger and Stein~\cite{karger1996new} or the algorithm of Stoer and
Wagner~\cite{stoer1997simple}, often by multiple orders of magnitude. Among
others, we compare ourselfs against two available implementations of the
sequential algorithm of
Nagamochi~\etal\cite{nagamochi1992computing,nagamochi1994implementing}.
We use our own implementation of the algorithm of 
Nagamochi~\etal\cite{nagamochi1992computing,nagamochi1994implementing}, written
in \CC~(\textttA{NOI-HNSS}) which was implemented as part of \textttA{VieCut}
(Chapter~\ref{c:viecut}) and uses a binary heap. We use this algorithm with
small optimizations in the priority queue as a base of our implementation.
Chekuri~\etal\cite{Chekuri:1997:ESM:314161.314315} give an implementation of the
flow-based algorithm of Hao and Orlin using all optimizations given in the paper
(variant \textttA{ho} in~\cite{Chekuri:1997:ESM:314161.314315}), implemented in
C, in our experiments denoted as \textttA{HO-CGKLS}. They also give an
implementation of the algorithm of
Nagamochi~\etal\cite{nagamochi1992computing,nagamochi1994implementing}, denoted
as \textttA{NOI-CGKLS}, which uses a heap as its priority queue data structure
(variant \textttA{ni-nopr} in~\cite{Chekuri:1997:ESM:314161.314315}). As their
implementations use signed integers as edge ids, we include their algorithms
only for graphs that have fewer than~$2^{31}$~edges. Most of our discussions
focus on comparisons to the \textttA{NOI-HNSS} implementation as this outperforms
the implementations by~Chekuri~\etal.

Gianinazzi~\etal\cite{gianinazzi2018communication} give a MPI implementation of
the algorithm of Karger and Stein~\cite{karger1996new}. We performed preliminary
experiments on small graphs which can be solved by \textttA{NOI-HNSS},
\textttA{NOI-CGKLS} and \textttA{HO-CGKLS} in less than $3$ seconds. On these
graphs, their implementation using $24$ processes took more than $5$ minutes,
which matches other
studies~\cite{Chekuri:1997:ESM:314161.314315,junger2000practical,henzinger2018practical}
that report bad real-world performance of (other implementations of) the
algorithm of Karger and Stein. Gianinazzi~\etal report a running time of $5$
seconds for RMAT graphs with $n=16000$ and an average degree of $4000$, using
\emph{$1536$ cores}. As \textttA{NOI-HNSS} can find the minimum cut on RMAT
graphs~\cite{wsvr} of equal size in less than $2$ seconds using \emph{a single
core}, we do not include the implementation
in~\cite{gianinazzi2018communication} in our experiments.

As our algorithm solves the minimum cut problem exactly, we do not include the
$(2+\epsilon)$-approximation algorithm of Matula~\cite{matula1993linear} and our
inexact algorithm \textttA{VieCut} in the experiments.

\subsection{Instances}

We use a set of graph instances that was also used for experiments in
Chapter~\ref{c:viecut}. The set of instances contains
$k$-cores~\cite{batagelj2003m} of large undirected real-world graphs taken from
the 10th DIMACS Implementation Challenge~\cite{bader2013graph} as
well as the Laboratory for Web Algorithmics~\cite{BRSLLP,BoVWFI}. Additionally
it contains large random hyperbolic graphs~\cite{krioukov2010hyperbolic,
von2015generating} with $n=2^{20} - 2^{25}$ and $m=2^{24} - 2^{32}$. A detailed
description of the graph instances is given in Section~\ref{rwgraphs} (Graph
family A). These graphs are unweighted, however contracted graphs that are
created in the course of the algorithm have edge weights.

\begin{figure}[t]
  %%FIGURE 3
  \includegraphics[width=\textwidth]{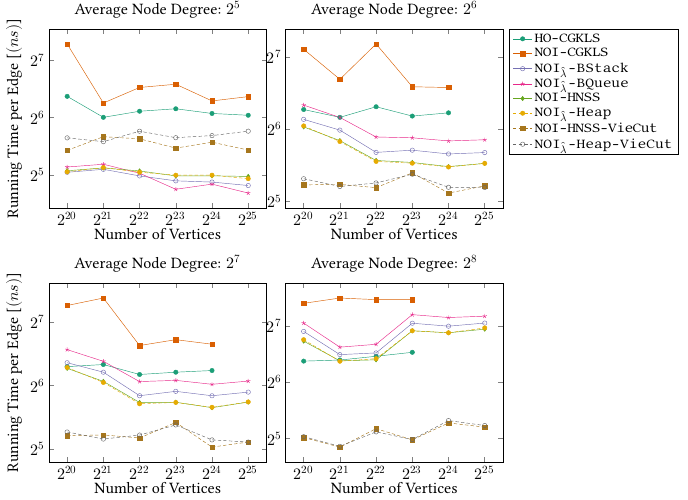}
\caption{Total running time in nanoseconds per edge in random hyperbolic graphs.
\label{c:exmc:fig:rhg}}
\end{figure}

\subsection{Sequential Experiments}
\label{c:exmc:exp:seq}

We limit the values in the priority queue $\mathcal{Q}$ to $\hat\lambda$, in
order to significantly lower the number of priority queue operations needed to
run the contraction routine. In this experiment, we want to examine the effects
of different priority queue implementations and limiting priority queue values
have on sequential minimum cut computations. We also include variants which run
\textttA{VieCut} first to lower $\hat\lambda$.

We start with sequential experiments using the implementation of
  \textttA{NOI-HNSS}. We use two variants: \textttA{NOI}$_{\hat\lambda}$ limits
  values in the priority queue to $\hat\lambda$ while \textttA{NOI-HNSS} allows
  arbitrarily large values in $\mathcal{Q}$. For \texttt{NOI}$_{\hat\lambda}$,
  we test the three priority queue implementations, \textttA{BQueue},
  \textttA{Heap} and \textttA{BStack}. As the priority queue for
  \textttA{NOI-HNSS} has priorities of up to the maximum degree of the graph and
  the contracted graphs can have very large degrees, the bucket priority queues
  are not suitable for \textttA{NOI-HNSS}. Therefore we only use the
  implementation of \textttA{NOI-HNSS}~\cite{henzinger2018practical}.
  
  The variants \textttA{NOI-HNSS-VieCut} and
  \texttt{NOI$_{\hat\lambda}$-Heap-VieCut} first run the shared-memory parallel
  algorithm \textttA{VieCut} using all $24$ threads to lower $\hat\lambda$
  before running the respective sequential algorithm. We report the total
  running time, \eg the sum of \textttA{VieCut} and \textttA{NOI}.
  
  \begin{figure}[t]
    \includegraphics[width=.49\textwidth]{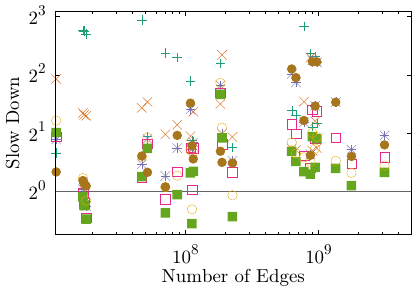}
    \raisebox{1.8mm}{
    \includegraphics[width=.49\textwidth]{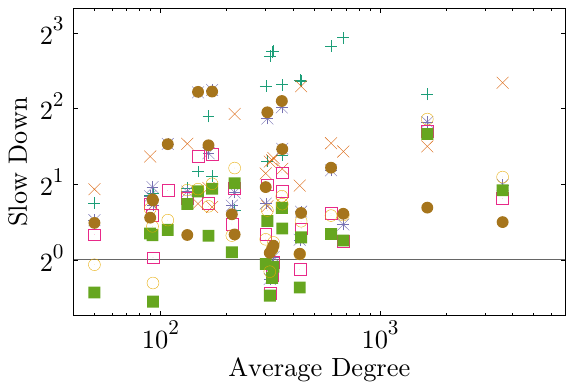}
    }
    \caption{\label{c:exmc:fig:edges} Total running time in real-world graphs, normalized
    by the running time of \texttt{NOI$_{\hat\lambda}$-Heap-VieCut}. (Legend
    shared with Figure~\ref{c:exmc:fig:perf} below)}
  \end{figure}

  \begin{figure}[t]
    \centering
    \includegraphics[width=.7\textwidth]{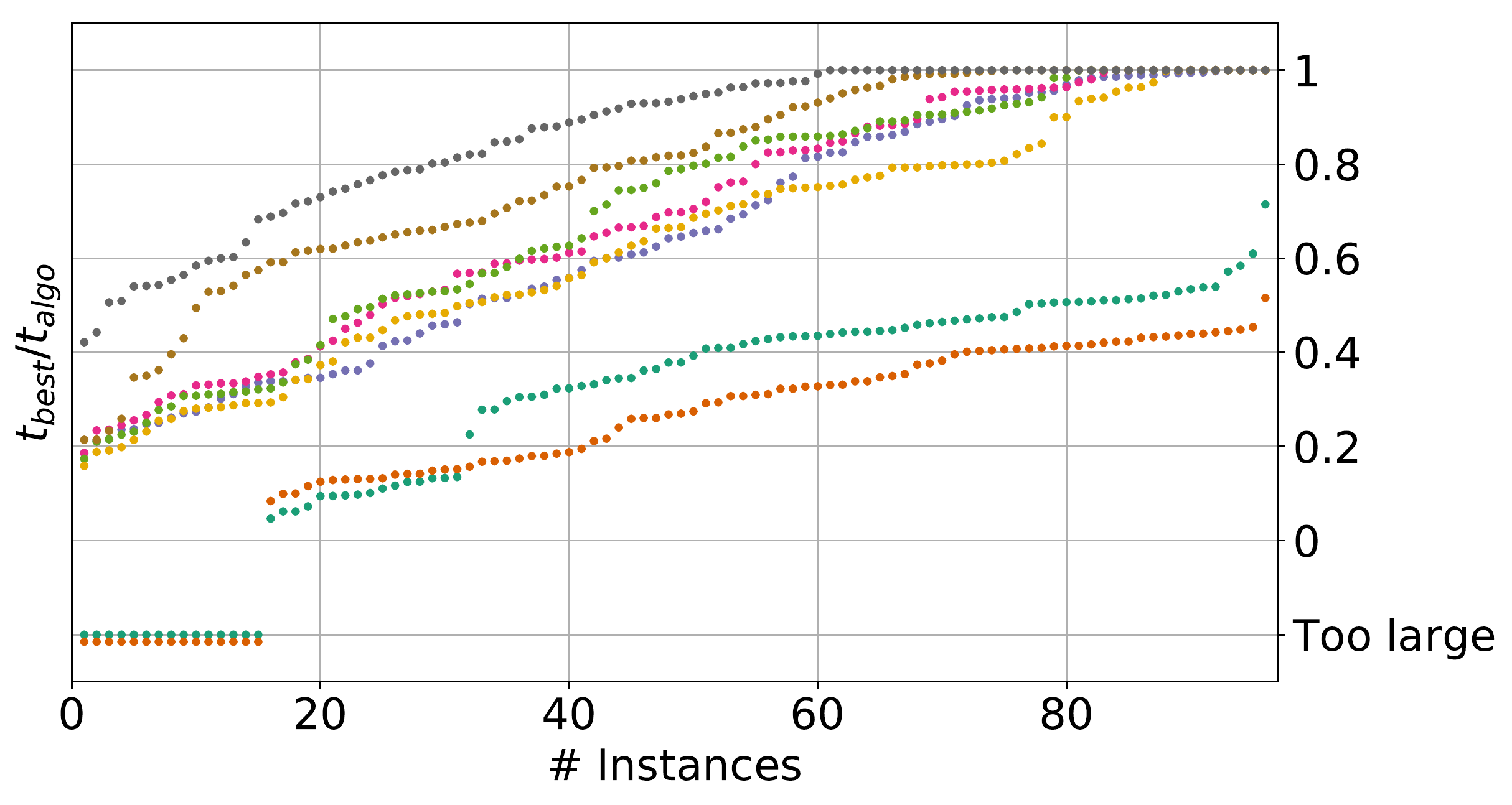}
    \raisebox{2.2cm} {
      \includegraphics[width=.25\textwidth]{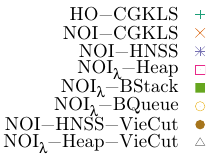}
    }
    \caption{\label{c:exmc:fig:perf} Performance plot for all graphs.}
  \end{figure}

\subsubsection{Priority Queue Implementations}

Figure~\ref{c:exmc:fig:rhg} shows the results for hyperbolic graphs and
Figure~\ref{c:exmc:fig:edges} shows the results for real-world graphs, normalized by
the running time of \texttt{NOI$_{\hat\lambda}$-Heap-VieCut}.
Figure~\ref{c:exmc:fig:perf} gives performance plots for all graphs from both
graph families. We can see that in
nearly all sequential runs, \texttt{NOI$_{\hat\lambda}$-BStack} is $5-10\%$
faster than  \texttt{NOI$_{\hat\lambda}$-BQueue}. This can be explained as this
priority queue uses \textttA{std::vector} instead of \textttA{std::deque} as its
underlying data structure and thus has lower access times to add and remove
elements. As all vertices are visited by the only thread, the scan order does
not greatly influence how many edges~are~contracted.

In the random hyperbolic graphs, nearly no vertices in \textttA{NOI-HNSS} reach
priorities in $\mathcal{Q}$ that are much larger than $\hat\lambda$. Usually,
fewer than $5\%$ of edges do not incur an update in $\mathcal{Q}$. Thus,
\textttA{NOI-HNSS} and \texttt{NOI$_{\hat\lambda}$-Heap} have practically the
same running time.  \texttt{NOI$_{\hat\lambda}$-BStack} is usually $5\%$ slower.

As the real-world graphs are social network and web graphs, they contain
vertices with very high degrees. In these vertices, \textttA{NOI-HNSS} often
reaches priority values of much higher than $\hat\lambda$ and
\texttt{NOI$_{\hat\lambda}$} can actually save priority queue operations. Thus,
\texttt{NOI$_{\hat\lambda}$-Heap} is up to $1.83$ times faster than
\textttA{NOI-HNSS} with an average (geometric) speedup factor of $1.35$. Also, in
contrast to the random hyperbolic graphs, \texttt{NOI$_{\hat\lambda}$-BStack} is
faster than \textttA{NOI-HNSS} on real-world graphs. Due to the low diameter of
web and social graphs, the number of vertices in $\mathcal{Q}$ is very large.
This favors the \textttA{BStack} priority queue, as it has constant~access~times.
The average geometric speedup of \texttt{NOI$_{\hat\lambda}$-BStack} compared to
\texttt{NOI$_{\hat\lambda}$-Heap} is $1.22$.

\subsubsection{Reduction of $\hat\lambda$ by VieCut}

\begin{figure}[t]
  \centering
  \includegraphics[width=.8\textwidth]{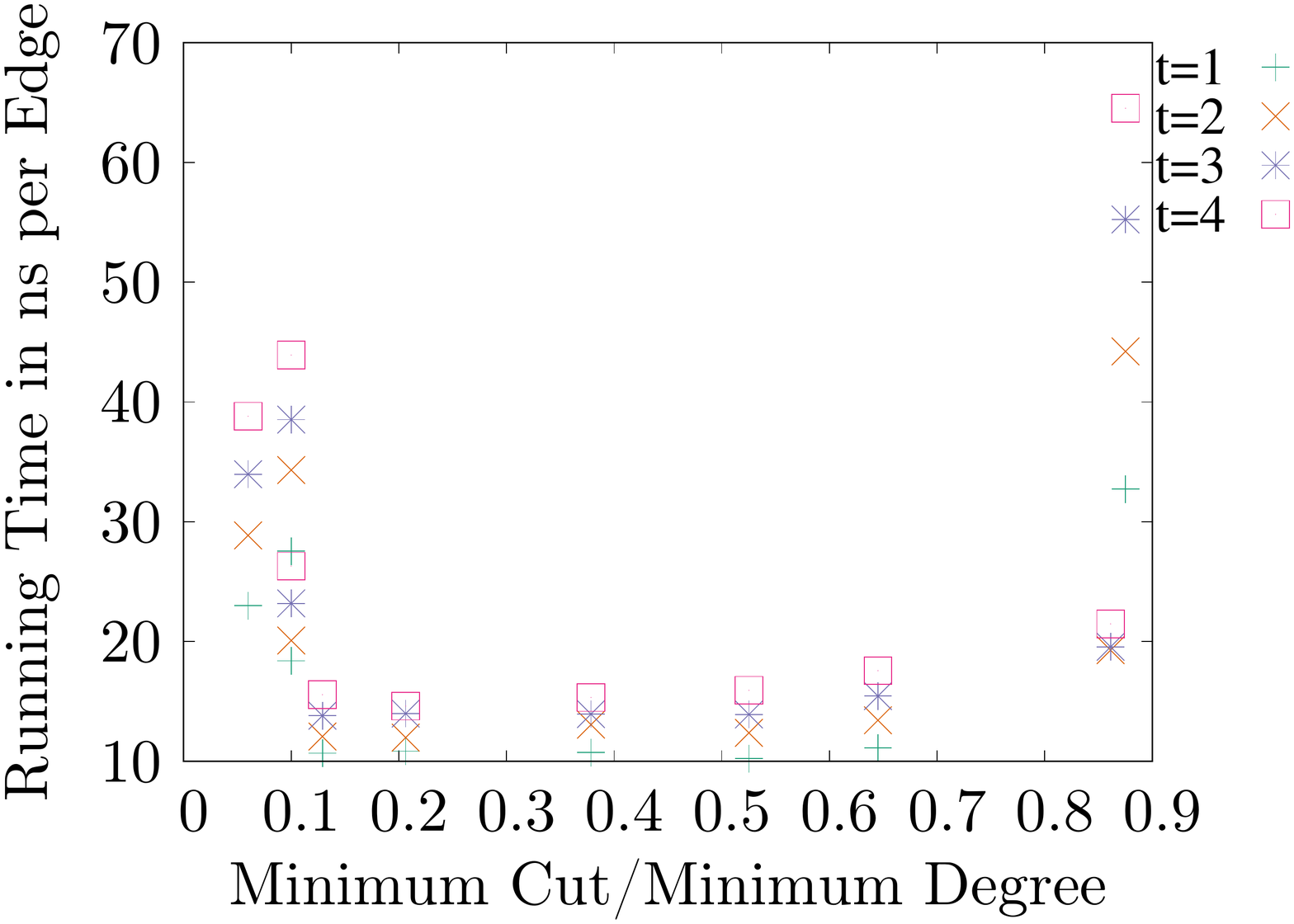}
  \caption{\label{c:exmc:fig:paramt} Results for different values for $t$}
\end{figure}

In this experiment we aim to reduce $\hat\lambda$ by running \textttA{VieCut}
before \textttA{NOI}. While the other algorithms are slower for denser random
hyperbolic graphs, both algorithms \textttA{NOI-HNSS-VieCut} and
\texttt{NOI$_{\hat\lambda}$-Heap-VieCut} are faster in these graphs with higher
density. This happens as the variants without \textttA{VieCut} find fewer
contractible edges and therefore need more rounds of CAPFOREST. The highest
speedup compared to \texttt{NOI$_{\hat\lambda}$-Heap} is reached in random
hyperbolic graphs with $n=2^{23}$ and an average density of $2^8$, where
\texttt{NOI$_{\hat\lambda}$-Heap-VieCut} has a speedup of factor $4$.

\texttt{NOI$_{\hat\lambda}$-Heap-VieCut} is fastest on most real-world graphs,
however when the minimum degree is very close to the minimum cut $\lambda$,
running \textttA{VieCut} can not significantly lower $\hat\lambda$. Thus, the
extra work to run \textttA{VieCut} takes longer than the time saved by lowering
the upper bound $\hat\lambda$. The average geometric speedup factor of
\texttt{NOI$_{\hat\lambda}$-Heap-VieCut} on all graphs compared to the variant
without \textttA{VieCut} is $1.34$.

In the performance plots in Figure~\ref{c:exmc:fig:perf} we can see that
\texttt{NOI$_{\hat\lambda}$-Heap-VieCut} is fastest or close to the fastest
algorithm in all but the very sparse graphs, in which the algorithm of
Nagamochi~\etal\cite{nagamochi1994implementing} is already very
fast~\cite{henzinger2018practical} and therefore using \textttA{VieCut} cannot
sufficiently lower $\hat\lambda$ and thus the running time of the algorithm.
\textttA{NOI-CGKLS} and \textttA{HO-CGKLS} are outperformed on all graphs.

\begin{figure}[t!]
  
  \includegraphics[width=.9\textwidth]{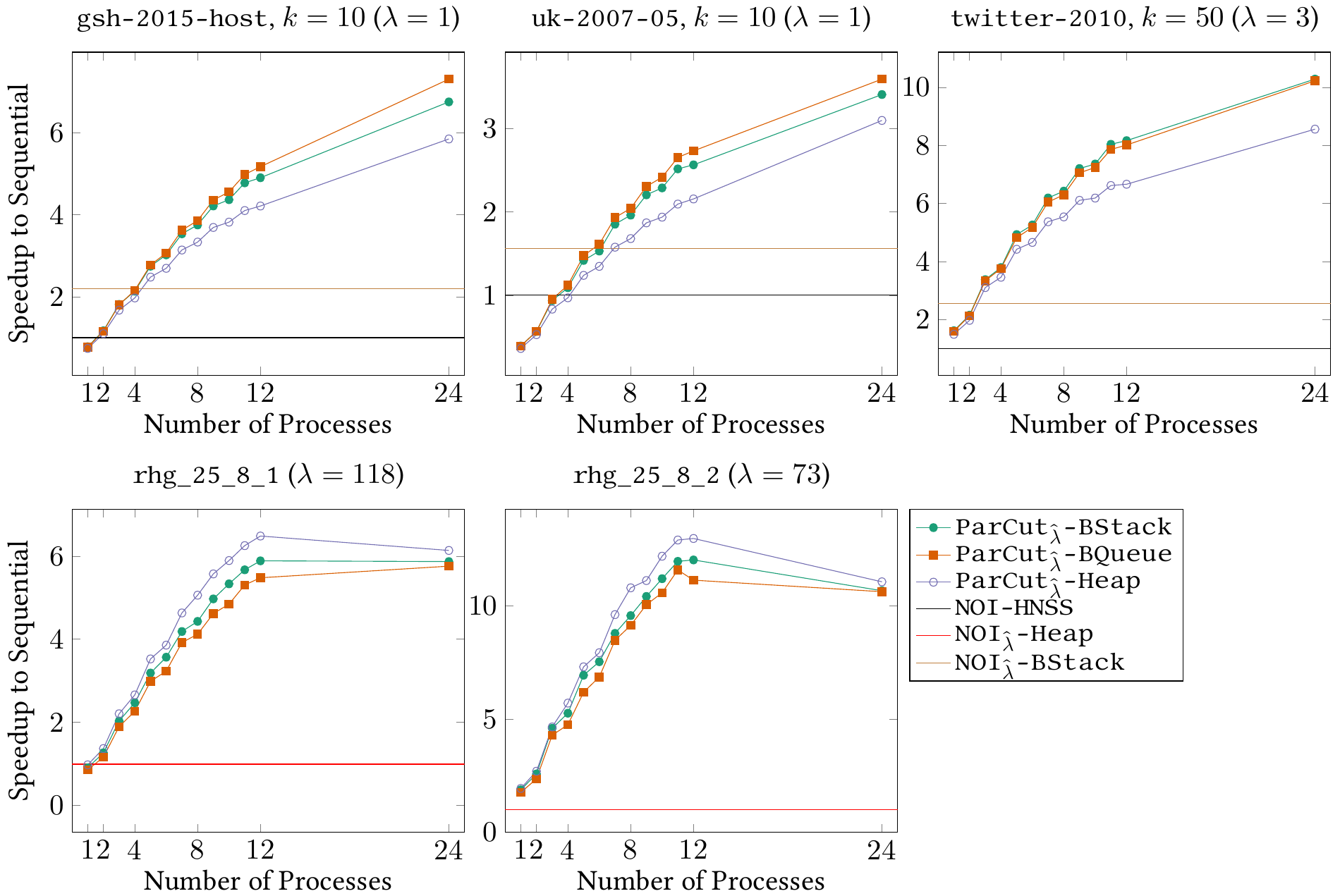}    
  \caption{Scaling plots for large graphs - Scalability}
  \label{c:exmc:fig:scale}
  \end{figure}

\subsection{Shared-memory parallelism}

\subsubsection{Overlapping Scan Regions}

We examine whether it is useful to overlap the regions scanned by each
thread. For this purpose we introduce a parameter $t$ which indicates how
many threads can scan each processor. A value of $t = 1$ executes
algorithm~\ref{c:exmc:algo:parnoi}, any larger value replaces $\mathcal{T}(x)$
with a counter indicating how many threads already scanned vertex $x$. If
$\mathcal{T}(x) \geq t$, no further threads may scan it.

Figure~\ref{c:exmc:fig:paramt} shows results for
\texttt{ParCut$_{\hat\lambda}$-BQueue} with $24$ processes for values of $t$
from $1$ to $4$ on a set of $10$ large graphs ($4$ real-world graphs, $6$ RHG
graphs). In general, lower values of $t$ have lower running times for
Algorithm~\ref{c:exmc:algo:parnoi}, however the amounts of contracted edges can be
lower, especially when many vertices have degree not too much higher than
$\hat\lambda$, as those can only be contracted depending on the order of vertex
scans. On $9$ out of the $10$ graphs, $t=1$ has the best performance, just in
the graph \textttA{rhg\_25\_8\_1} with minimum degree $137$ and $\lambda=118$,
$t=2$ and $t=3$ are slightly faster. Thus we set parameter $t$ to $1$ and do not
use overlapping scan regions.

\begin{figure}[t!]
 \includegraphics[width=.9\textwidth]{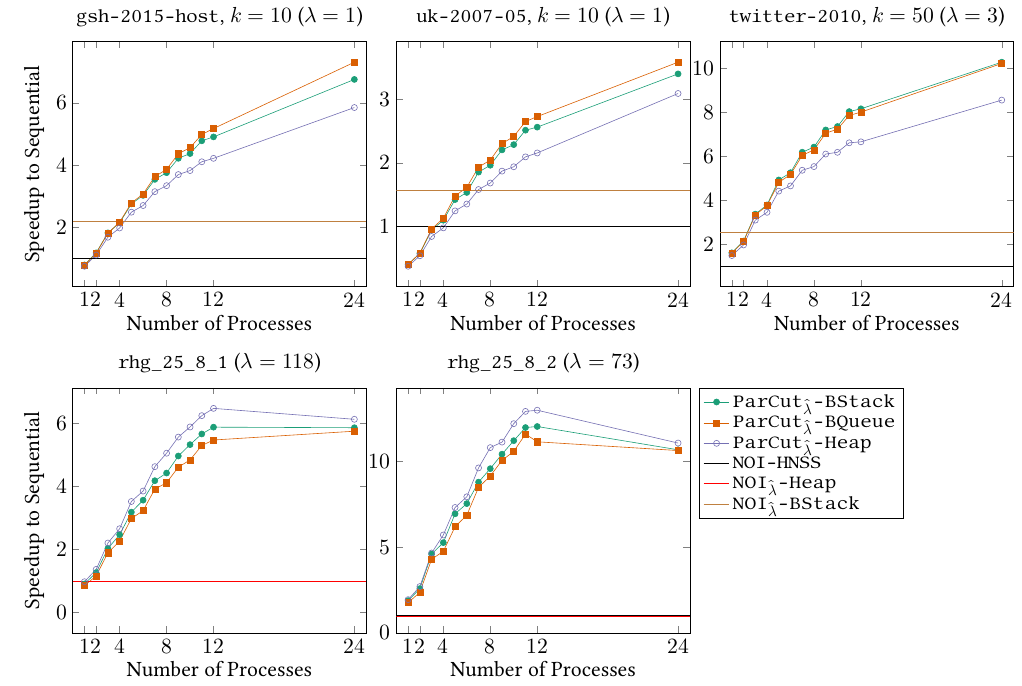}
  \caption{\label{c:exmc:fig:scale2}Scaling plots for large graphs - Speedup
  compared to \textttA{NOI-HNSS} and fastest sequential algorithm (first $3$
  graphs: \texttt{NOI$_{\hat\lambda}$-BStack}¸ last $2$ graphs:
  \texttt{NOI$_{\hat\lambda}$-Heap}).}
\end{figure}

We run experiments on $5$ of the largest graphs in the data sets using up to
$24$ threads on $12$ cores. First, we compare the performance of
Algorithm~\ref{c:exmc:algo:parmc} using different priority queues:
\texttt{ParCut$_{\hat\lambda}$-Heap}, \texttt{ParCut$_{\hat\lambda}$-BStack} and
\texttt{ParCut$_{\hat\lambda}$-BQueue} all limit the priorities to
$\hat\lambda$, the result of \textttA{VieCut}. In these experiments,
\textttA{VieCut} takes up between $19-83\%$ of the total running time with an
average of $51\%$. Figure~\ref{c:exmc:fig:scale} shows how well the algorithms
scale with increased number of processors. Figure~\ref{c:exmc:fig:scale2} shows
the speedup compared to the fastest sequential algorithm of
Section~\ref{c:exmc:exp:seq}. On all graphs,
\texttt{ParCut$_{\hat\lambda}$-BQueue} has the highest speedup when using $24$
threads. On real-world graphs, \texttt{ParCut$_{\hat\lambda}$-BQueue} also has
the lowest total running time. In the large random hyperbolic graphs, in which
the priority queue is usually only filled with up to $1000$ elements, the worse
constants of the double-ended queue cause the variant to be slightly slower than
\texttt{ParCut$_{\hat\lambda}$-Heap} also even when running with $24$ threads.
In the two large real-world graphs that have a minimum degree of $10$, the
sequential algorithm \texttt{NOI$_{\hat\lambda}$-BStack} contracts most edges in
a single run of CAPFOREST - due to the low minimum degree, the priority queue
operations per vertex are also very low. Thus, \texttt{ParCut$_{\hat\lambda}$}
using only a single thread has a significantly higher running time, as it runs
\textttA{VieCut} first and performs graph contraction using a concurrent hash
table, as described in Section~\ref{contract}, which is slower than
sequential graph contraction when using just one thread. In graphs with higher
minimum degree, \textttA{NOI} needs to perform multiple runs of CAPFOREST. By
lowering $\hat\lambda$ using \textttA{VieCut} we can contract significantly more
edges and achieve a speedup factor of up to $12.9$ compared to the fastest
sequential algorithm \texttt{NOI$_{\hat\lambda}$-Heap}. On
\textttA{twitter-2010}, $k=50$, \texttt{ParCut$_{\hat\lambda}$-BQueue} has a
speedup of $10.3$ to \textttA{NOI-HNSS}, $16.8$ to \textttA{NOI-CGKLS} and a
speedup of $25.5$ to \textttA{HO-CGKLS}. The other graphs have more than
$2^{31}$ edges and are thus too large for \textttA{NOI-CGKLS} and
\textttA{HO-CGKLS}.

\section{Conclusion}
\label{c:exmc:s:con}

We presented a shared-memory parallel exact algorithm for the minimum cut
problem. Our algorithm is based on the algorithms of
Nagamochi~\etal\cite{nagamochi1992computing,nagamochi1994implementing} and our
work described in Chapter~\ref{c:viecut}. We use different data structures
and optimizations to decrease the running time of the algorithm of
Nagamochi~\etal by a factor of up to $2.5$. Using additional shared-memory
parallelism we further increase the speedup factor to up to $12.9$. Future work
includes checking whether our sequential optimizations and parallel
implementation can be applied to the $(2+\epsilon)$-approximation algorithm of
Matula~\cite{matula1993linear}.

\chapter{Finding All Minimum Cuts}
\label{c:allmc}
% ESA'20 Paper
We present a practically efficient algorithm that finds all global minimum cuts
in huge undirected graphs. Our algorithm uses a multitude of kernelization rules
to reduce the graph to a small equivalent instance and then finds all minimum
cuts using an optimized version of the algorithm of Nagamochi, Nakao and
Ibaraki~\cite{nagamochi2000fast}. Some of these techniques are adapted from
techniques for the global minimum cut
problem~\cite{padberg1990efficient,nagamochi1994implementing} which we
discussed in the previous chapters of this dissertation. Using these and
newly developed reductions we are able to decrease the running time by up to
multiple orders of magnitude compared to the algorithm of
Nagamochi~\etal\cite{nagamochi2000fast} and are thus able to find all minimum
cuts on graphs with up to billions of edges in a few minutes. Based on the
cactus representation of all minimum cuts, we are able to find the most balanced
minimum cut in time linear to the size of the cactus. As our techniques are able
to find the most balanced minimum cut of graphs with
billions of edges in minutes, this allows the use of minimum cuts as a
subroutine in sophisticated data mining and graph analysis.

The content of this chapter is based on~\cite{henzinger2020finding}.

\section{Algorithm Description}

Our algorithm combines a variety of techniques and algorithms in order to find
all minimum cuts in a graph. The algorithm is based on the contractions of edges
which cannot be part of \emph{any} minimum cut. Thus, we first show that an edge $e$ that
is not part of any minimum cut in graph $G$ can be contracted. In contrast to
the previous chapters, we now aim to maintain \emph{all} minimum cuts.

\begin{lemma}~\cite{karger1996new} \label{lem:contractible} If an edge $e=(u,v)$
is not part of any minimum cut in graph $G$, all minimum cuts of $G$ remain in
the resulting graph $G/e$. 
\end{lemma}

\begin{proof}
  Let $(A,B)$ be an arbitrary minimum cut of $G$. For an edge $e = (u,v)$, which
  is not part of any minimum cut, we know that $e \not\in E[A]$, so either $u$
  and $v$ are both in vertex set $A$ or both in vertex set $B$. This is still
  the case in $G/e$. Thus, the edge $e$ can be contracted even if we aim to find
  every minimum cut of $G$.
\end{proof}

Lemma~\ref{lem:contractible} is very useful to reduce the size of the graph with
the usage of techniques to identify such edges. We now give a short overview
of our algorithm and then explain the techniques in more detail. First, we use
our shared-memory parallel heuristic minimum cut algorithm
\textttA{VieCut} (as described in Chapter~\ref{c:viecut}) in order to find an upper bound
$\hat\lambda$ for the minimum cut which is likely to be the correct value. Having a tight bound for the minimum cut
allows the contraction of many edges, as multiple reduction techniques depend on
the value of the minimum cut. We adapt contraction techniques originally
developed by
Nagamochi~\etal\cite{nagamochi1992computing,nagamochi1994implementing} and
Padberg~\etal\cite{padberg1991branch} (see Section~\ref{p:mincut:s:algs}) to the problem of \emph{finding all minimum
cuts}. Section~\ref{ss:contraction} details these contraction routines. On the
resulting graph we find all minimum cuts using an optimized variant of the
algorithm of Nagamochi, Nakao and Ibaraki~\cite{nagamochi2000fast} and return
the cactus graph which represents them all. A short description of the algorithm
and an explanation of our engineering effort are given in
Section~\ref{ss:cactus}. Afterwards, in Section~\ref{ss:combine} we show how we
combine the parts into a fast algorithm to find all minimum cuts of large
networks. 

\subsection{Edge Contraction}
\label{ss:contraction}

As shown in Lemma~\ref{lem:contractible}, edges that are not part of any minimum
cut can be safely contracted. We build a set of techniques that aim to find
contractible edges and run these in alternating order until neither of them
finds any more contractible edges. We now give a short introduction to these.

For efficiency, we perform contractions in bulk. If our algorithm finds an edge
that can be contracted, we merge the incident vertices in a thread-safe
union-find data structure~\cite{anderson1991wait}. After each run of a
contraction technique that finds contractible edges, we create the contracted
graph using a shared-memory parallel hash table~\cite{maier2016concurrent}. In
this contracted graph, each set of vertices of the original graph is merged into
a single node. The contraction of this vertex set is equivalent to contracting a
spanning tree of the set. After contraction we check whether a vertex in the
contracted graph has degree $< \hat\lambda$. If it does, we found a cut of
smaller value and update $\hat\lambda$ to this value.

\subsubsection{Connectivity-based Contraction}
\label{sss:connectivity}

The connectivity of an edge $e=(s,t)$ is the weight of the minimum cut that
separates $s$ and $t$, \ie the \emph{minimum s-t-cut}. For an edge that has
connectivity $> \hat\lambda$, we thus know that there is no cut separating $s$
and $t$ (\ie no cut that contains $e$) that has value $\leq \hat\lambda$. Thus,
we know that there cannot be a minimum cut that contains $e$, as $\hat\lambda$
is by definition at least as large as $\lambda$. However, solving the minimum
s-t-cut problem takes significant time, so computing the connectivity of each
edge does not scale to large networks. Hence, as part of their algorithm for the
global minimum cut problem,
Nagamochi~\etal\cite{nagamochi1992computing,nagamochi1994implementing} give a
subroutine that computes a lower bound $q(e)$ for the connectivity of every edge
$e$ of $G$ in a total running time of $\Oh{m + n\log{n}}$. Both the algorithm of
Nagamochi~\etal and the CAPFOREST subroutine that computes the connectivity
lower bounds $q(e)$ are outlined in Section~\ref{p:mincut:ss:noi}. Each of the
edges whose connectivity lower bound is already larger than $\hat\lambda$ can be
contracted as it cannot be part of any minimum cut. 

In Chapter~\ref{c:exmc} we give a fast shared-memory parallel variant of their
algorithm. As that algorithm aims to find a single minimum cut,
it also contracts edges that have connectivity equal to $\hat\lambda$, as the
only relevant cuts are ones better than the best cut known
previously. As we want to find all minimum cuts, we can only contract edges
whose connectivity is strictly larger than $\hat\lambda$. Nagamochi~\etal could
prove that at least one edge has value $\hat\lambda$ in their routine and can
thus be contracted. We do not have such a guarantee when trying to find edges
that have connectivity $> \hat\lambda$. Consider for example an unweighted tree,
whose minimum cut has a value of $1$ and each edge has connectivity $1$ as well.

\begin{figure*}[t!]
  \centering
  \includegraphics[width=\textwidth]{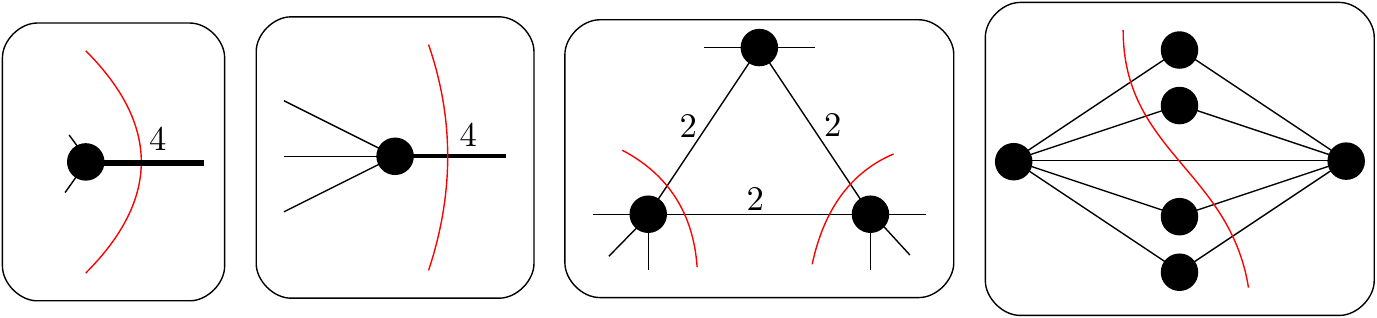}
  \caption{Local reduction rules: (1) \prOne, (2) \prTwo, (3) \prThree, (4) \prFour.}
  \label{fig:all_local}
\end{figure*}

 \subsubsection{Local Contraction Criteria}
 \label{sss:local}
 
 Padberg and Rinaldi~\cite{padberg1991branch} give a set of \emph{local
 reduction} routines which determine whether an edge can be contracted without
 affecting the minimum cut. We describe these reductions in Section~\ref{p:mincut:ss:pr}. Their reduction routines were shown to be very
 useful in order to find a minimum cut fast in
 practice~\cite{Chekuri:1997:ESM:314161.314315,junger2000practical,henzinger2018practical}
 and are also used in our \textttA{VieCut} algorithm in Chapter~\ref{c:viecut}.
 We adapt the routines originally developed for the minimum cut problem so that
 they hold for the problem of for finding all minimum cuts. Thus, we have to
 make sure that we do not contract cuts of value $\hat\lambda$, as they might be
 minimal and additionally make sure that we do not contract edges incident to
 vertices that could have a \emph{trivial minimum cut}, \ie a minimum cut, where
 one side contains only a single vertex. Figure~\ref{fig:all_local} shows
 examples and Lemma~\ref{lem:local_crit} gives a formal definition of these reduction rules.

 \begin{lemma} \label{lem:local_crit} For an edge $e = (u,v) \in E$, $e$ is not
    part of any minimum cut, if $e$ fulfills at least one of the following
    criteria. Thus,  all minimum cuts of $G$ are still present in $G/e$ and $e$
    can be contracted.
    \begin{enumerate}
      \item \prOne: $c(e) > \hat\lambda$
      \item \prTwo:
      \begin{itemize}
        \item $c(v) < 2 c(e)$ and $c(v) > \hat\lambda$, or
        \item $c(u) < 2 c(e)$ and $c(u) > \hat\lambda$
      \end{itemize}      
      \item \prThree:\\
        $\exists w \in V$ with 
      \begin{itemize}
        \item $c(v) < 2 \{c(v,w) + c(e)\}$ and $c(v) > \hat\lambda$, and
        \item $c(u) < 2 \{c(u,w) + c(e)\}$ and $c(u) > \hat\lambda$
      \end{itemize}
      \item \prFour: \\$c(e) + \sum_{w \in V} min\{c(v,w),c(u,w)\} >
      \hat\lambda$
    \end{enumerate}
 \end{lemma}

\begin{proof}
  \begin{enumerate}
    \item If $c(e) > \hat\lambda$, every cut that contains $e$ has capacity $>
    \hat\lambda$. Thus it can not be a minimal cut.
    \item Without loss of generality let $v$ be the vertex in question. The
    condition $c(v) < 2 c(e)$ means that $e$ is heavier than all other edges
    incident to $v$ combined. Thus, for any non-trivial cut that contains $e$,
    we can find a lighter cut by replacing $e$ with all other incident edges to
    $v$, \ie moving $v$ to the other side of the cut. As this is not true for
    the trivial minimum cut $(v,V \backslash v)$, we cannot contract an edge
    incident to a vertex that has weight $\leq \hat\lambda$.
    \item  This condition is similar to (2). Let there be a triangle ${u,v,w}$
    in the graph in which it holds for both $u$ and $v$ that the two incident
    triangle edges are heavier than the sum of all other incident edges. Then,
    every cut that separates $u$ and $v$ can be improved by moving $u$ and $v$
    into the same side. As the cut could have vertex $w$ on either side, both
    vertices need to fulfill this condition. To make sure that we do not
    contract any trivial minimum cut, we check that both $v$ and $u$ have
    weighted vertex degree $> \hat\lambda$ and thus can not represent a trivial
    minimum cut. 
    \item In this condition we check the whole shared neighborhood of vertices
    $u$ and $v$. Every cut that separates $u$ and $v$ must contain $e$ and for
    each shared neighbor $w$ at least one of the edges connecting them to $w$.
    Thus, we sum over the lighter edge connecting them to the shared neighbors
    and have a lower bound of the minimum cut that separates $u$ and $v$. If
    this is heavier than $\hat\lambda$, we know that no minimum cut separates
    $u$ and $v$.
  \end{enumerate}
  \end{proof}

The conditions \prOne{} and \prTwo{} can both be checked for the whole graph in
a single run in linear time. While we can check condition \prThree{} when
summing up the lighter incident edges for condition \prFour{}, exhaustively
checking all triangles incurs a strictly worse than linear runtime, as a graph
can have up to $\Theta(m^{3/2})$ triangles~\cite{schank2005finding}. Thus, we
only perform linear-time runs as developed by
Chekuri~\etal\cite{Chekuri:1997:ESM:314161.314315} by marking the neighborhood
of $u$ and $v$ while we check the conditions and do not perform the test on
marked vertices.

\subsubsection{Vertices with One Neighbor}
\label{sss:degreeone}

Over the run of the algorithm, we occasionally encounter vertices that have only
a single neighbor. Let $v$ be this vertex with one neighbor and $e = (v,w)$ be
the only incident edge. As we update $\hat\lambda$ to the minimum degree
whenever we perform a bulk edge contraction, $c(e) \geq \hat\lambda$: for an
edge whose weight is $> \hat\lambda$, condition \prOne{} will contract it. For
an edge whose weight is $\hat\lambda$, the edge represents a trivial minimum cut
iff $\hat\lambda = \lambda$. This is the only minimum cut that contains $e$, as
every non-trivial cut containing $e$ has higher weight. Thus, we can contract
$e$ for now and remember that it was contracted. If $\hat\lambda$ is decreased,
we can forget about these vertices as the cuts are not minimal. When we are
finished, we can re-insert all contracted vertices that have a trivial minimum
cut. We perform this reinsertion in a bottom-up fashion (\ie in reverse order to
how they were contracted), as the neighbor $w$ could be contracted in a later
contraction. 

\subsection{Finding All Minimum Cuts}
\label{ss:cactus}

We apply the reductions in the previous section exhaustively until they are not
able to find a significant number of edges to contract. On the remaining graph
we aim to find the cactus representation of all minimum cuts. Our algorithm for
this purpose is based on the algorithm of Nagamochi, Nakao and
Ibaraki~\cite{nagamochi2000fast}. While there is a multitude of algorithms for
the problem of finding all minimum cuts, to the best of our knowledge there are
no implementations accessible to the public and there is no practical
experimentation on finding all minimum cuts. We base our algorithm on the
algorithm of Nagamochi, Nakao and Ibaraki~\cite{nagamochi2000fast}, as their
algorithm allows us to run the reduction routines previously detailed in between
recursion steps.

We give a quick sketch of their algorithm, for further details we refer the reader to
\cite{nagamochi2000fast}. To find all minimum cuts in graph $G$, the algorithm
chooses an edge $e=(s,t)$ in $G$ and uses a maximum flow $f$ to find the minimum
s-t-cut $\lambda(s,t)$. If $\lambda(s,t) > \lambda$ there is no minimum cut that
separates $s$ and $t$ and thus $e$ can be contracted. If $\lambda(s,t) =
\lambda$, the edge is part of at least one minimum cut. They show that the
strongly connected components $(V_1,\dots,V_k)$ of the residual graph $G_f$
represent all minimum cuts that contain $e$ (and potentially some more). For
each connected component $V_i$, they build a graph $C_i$, in which all other
connected components are contracted into a single vertex. We recurse on these
component subgraphs and afterwards combine the minimum cut cactus graphs of the
recursive calls to a cactus representation for $G$. The combination of the
cactus graphs begins by building a cactus graph $C$ representing the set of
strongly connected components, in which each $V_i$ is represented by a single
vertex $v_i$. Each cactus $C_i$ is then merged with $C$ by replacing $v_i$ with
$C_i$. Inside this algorithm we re-run the contraction routines of
Section~\ref{ss:contraction}. As they incur some computational cost and the
graph does not change too much over different recursion steps, we only run the
contraction routines every $10$ recursion levels.

As the contraction routines in Section~\ref{ss:contraction} usually mark a large
amount of edges that can be contracted in bulk, we represent the graph in the
compressed sparse row format~\cite{tewarson1973sparse}. This allows for fast and
memory-efficient accesses to vertices and edges, however, we need to completely
rebuild the graph in each bulk contraction and also keep vertex information
about the whole graph hierarchy to be able to see which vertices in the original
graph are encompassed in a vertex in a coarser vertex and to be able to
re-introduce the cactus edges that were removed. While this is efficient for the
bulk contractions performed in the previous section, in this section we often
perform single-edge contractions or contract a small block of vertices. For fast
running times these operations should not incur a complete rebuild of the graph
data structure. We therefore use a mutable adjacency list data structure where
each vertex is represented by a dynamic array of edges to neighboring vertices.
Each edge stores its weight, target and the ID of its reverse edge (as we look
at undirected graphs). This allows us to contract edges and small blocks in time
corresponding to the sum of vertex degrees. For each vertex in the original
graph, we store information which vertex currently encompasses it and every
vertex keeps a list of currently encompassed vertices of the original graph. All
vertex and edge information is updated during each edge contraction. The same
graph data structure is also used for the \emph{multiterminal cut problem} in
Part~\ref{p:mtc} of this work.

\subsubsection{Edge Selection}

The recursive algorithm of Nagamochi, Nakao and Ibaraki~\cite{nagamochi2000fast}
selects an arbitrary edge for the maximum flow problem in each recursion step.
If this edge has connectivity equal to the minimum cut, we create a recursive
subproblem for each connected component of the residual graph. In order to
reduce the graph size - and thus the amount of work necessary - quickly, we aim
to select edges in which the largest connected component of the residual graph
is as small as possible. The edge selection strategy \textttA{Heavy} searches for
the highest degree vertex $v$ and chooses the edge from $v$ to its highest
degree neighbor. The strategy \textttA{WeightedHeavy} does the same, but uses the
vertices whose weighted degree is highest. The idea is that an edge between
high-degree vertices is most likely 'central' to the graph and thus manages to
separate sizable chunks from the graph. The edge selection strategy
\textttA{Central} aims to find a central edge more directly: we aim to find two
vertices $u$ and $v$ with a high distance and take the central edge in their
shortest paths. We find those vertices by performing a breadth-first search from
a random vertex $w$, afterwards performing a breadth-first search from the
vertex encountered last. We then take the central edge in the shortest path (as
defined from the second breadth-first search) from the two vertices encountered
last in the two breadth-first searches. The edge selection strategy
\textttA{Random} picks a random edge.

\subsubsection{Degree-two Reductions}
\label{sss:degreetwo}

Over the course of this recursive contraction-based algorithm, we often
encounter vertices with just two neighbors. Let $v$ be the vertex in question,
which is connected to $u_0$ by edge $e_0$ and to $u_1$ by edge $e_1$. We look at
four cases, each looking at whether the weight of $e_0$ being equal to the
weight of $e_1$ and $c(v)$ being equal to $\lambda$, both conditions that can be
checked in constant time. In three out of four cases, we are able to contract an
incident edge.

$c(e_0) \neq c(e_1)$ and $c(v) > \lambda$: Without loss of generality let $e_0$
be the heavier edge. As $c(v) > \lambda$, the trivial cut $(\{v\}, V \backslash
\{v\})$ is not a minimum cut. As by definition no cut in $G$ is smaller than
$\lambda$, $\lambda(u_0, u_1) \geq \lambda$. Thus, excluding the path through
$v$, they have a connectivity of $\geq \lambda - c(e_1)$ and any cut containing
$e_0$ has weight $\geq \lambda - c(e_1) + c(e_0) > \lambda$ and can thus not be
minimal. We therefore know that $e_0$ is not part of any minimum cuts and can be
contracted according to Lemma~\ref{lem:contractible}.

$c(e_0) \neq c(e_1)$ and $c(v) = \lambda$: Without loss of generality let $e_0$
be the heavier edge. Analogously to the previous case we can show that no
nontrivial cut contains $e_0$. In this case, where $c(v) = \lambda$, the trivial
cut $(\{v\}, V \backslash \{v\})$ is minimal and therefore should be represented
in the cactus graph. For all other minimum cuts that contain $e_1$, we know that
$v$ and $u_0$ will be in the same block (as $c(e_0) > c(e_1)$). Thus, $v$ will
be represented in the cactus as a leaf incident to $u_0$. We contract $e_0$
calling the resulting vertex $u^*$ and store which vertices of the original
graph are represented by $v$. Then we recurse. On return from the recursion we
check which cactus vertex now encompasses $u^*$ and add an edge from this vertex
to a newly added vertex representing all vertices encompassed by $v$.

$c(e_0) = c(e_1)$ and $c(v) > \lambda$: in this case we are not able to contract
any edges without further connectivity information.

$c(e_0) = c(e_1)$ and $c(v) = \lambda$: as $c(v) = \lambda$, the trivial cut
$(\{v\}, V \backslash \{v\})$ is minimal. If there are other minimum cuts that
contain either $e_0$ or $e_1$ (e.g. that separate $u_0$ and $u_1$), we know that
by replacing $e_0$ with $e_1$ (or vice-versa) the cut remains minimal. Such a
minimum cut exists iff $\lambda(u_0, u_1) = \lambda$. We contract $e_0$ and
remember this decision. As $e_1$ is still in the graph (merged with $(u_0,
u_1)$), we are able to find each cut that separates $u_0$ and $u_1$. If none
exists, $\lambda(u_0, u_1) > \lambda$ and $u_0$ and $u_1$ will be contracted
later in the algorithm. When leaving the recursion, we can thus re-introduce
vertex $v$ as a leaf connected to the vertex encompassing $u_0$ and $u_1$. If
$u_0$ and $u_1$ are in different vertices after leaving the recursion, there is
at least one nontrivial cut that contains $e_1$. We thus re-introduce $v$ as a
cycle vertex connected to $u_0$ and $u_1$, each with weight $\frac{\lambda}{2}$,
and subtract $\frac{\lambda}{2}$ from $c(u_0, u_1)$.

In three out of the four cases presented here, we are able to contract an edge
incident to a degree-two vertex. We can check these conditions in total time
$\Oh{n}$ for the whole graph. Over the course of the algorithm, we perform edge
contractions and thus routinely encounter vertices whose neighborhood has been
contracted and thus have a degree of two. Thus, these reductions are able to
reduce the size of the graph significantly even if the initial graph is rather
dense and does not have a lot of low degree vertices.

\subsection{Putting it All Together}
\label{ss:combine}

\begin{algorithm}[t]
  \caption{Algorithm to find all minimum cuts}\label{alg:overview}
  \begin{algorithmic}[1]
  \Procedure{FindAllMincuts}{$G = (V,E)$}
  \State $\hat\lambda \gets \textttA{VieCut}(G)$ \cite{henzinger2018practical}
  \While{not converged}
  \State $(G,D_1, \hat\lambda) \gets \text{contract degree-one vertices}(G,
  \hat\lambda)$
  \State $(G, \hat\lambda) \gets \text{connectivity-based contraction}(G, \hat\lambda)$ 
  \State $(G, \hat\lambda) \gets \text{local  contraction}(G, \hat\lambda)$
  \EndWhile
  \State $\lambda \gets \text{FindMinimumCutValue}(G)$
  \State $C \gets \text{RecursiveAllMincuts}(G,\lambda)$
  (\cite{nagamochi2000fast})
  \State $C \gets \text{reinsert vertices}(C, D_1)$
  \State \Return $(C, \lambda)$
  \EndProcedure
  \end{algorithmic}
  \end{algorithm}

Algorithm~\ref{alg:overview} gives an overview over our algorithm to find all
minimum cuts. Over the course of the algorithm we keep an upper bound
$\hat\lambda$ for the minimum cut, initially set to the result of the inexact
variant of the \textttA{VieCut} minimum cut
algorithm~\cite{henzinger2018practical} (Chapter~\ref{c:viecut}). While the
\textttA{VieCut} algorithm also offers an exact
version~\cite{DBLP:conf/ipps/HenzingerN019} (Chapter~\ref{c:exmc}), we use the
inexact version, as it is considerably faster and gives a low upper bound for
the minimum cut, usually equal to the minimum cut. As described in
Section~\ref{ss:contraction}, we use this bound to contract degree-one vertices,
high-connectivity edges and edges whose local neighborhood guarantees that they
are not part of any minimum cut. We repeat this process until it is converged,
as an edge contraction can cause other edges in the neighborhood to also become
safely contractible. As this process often incurs a long tail of single edge
contractions, we stop if the number of vertices was decreased by less than $1\%$
over a run of all contraction routines. 

We then use the minimum cut algorithm of Nagamochi, Ono and
Ibaraki~\cite{nagamochi1992computing,nagamochi1994implementing} on the remaining
graph, as the following steps need the correct minimum cut. To find all minimum
cuts in the contracted graph, we call our optimized version of the algorithm of
Nagamochi~\etal\cite{nagamochi2000fast}, as sketched in Section~\ref{ss:cactus},
and afterwards re-insert all minimum cut edges that were previously deleted.
Before each recursive call of the algorithm of
Nagamochi~\etal\cite{nagamochi2000fast}, we contract edges incident to
degree-one and eligible degree-two vertices. Every $10$ recursion levels we
additionally check for connectivity-based edge contractions and local
contractions. 

\subsection{Shared-Memory Parallelism}
Algorithm~\ref{alg:overview} employs shared-memory parallelism in every step.
When we run the algorithm in parallel, we use the parallel variant of
\textttA{VieCut}~\cite{henzinger2018practical}. Local contraction and marking of degree
one vertices are parallelized using OpenMP~\cite{dagum1998openmp}. For the first
round of connectivity-based contraction, we use the parallel connectivity
certificate used in the shared-memory parallel minimum cut algorithm detailed in
Chapter~\ref{c:exmc}~\cite{DBLP:conf/ipps/HenzingerN019}. This connectivity certificate is
essentially a parallel version of the connectivity certificate of
Nagamochi~\etal\cite{nagamochi1992computing,nagamochi1994implementing}, in which
the processors divide the work of computing the connectivity bounds for all
edges of the graph. In subsequent iterations every processor runs an independent
run of the connectivity certificate of Nagamochi~\etal on the whole graph
starting from different random vertices in the graph. As the connectivity bounds
given by the algorithm heavily depend on the starting vertex, this allows us to
find significantly more contractible edges per round than running the
connectivity certificate only once. 

We use our exact shared-memory parallel minimum cut algorithm to find the exact
minimum cut of the graph. The algorithm of
Nagamochi~\etal\cite{nagamochi2000fast} is not shared-memory parallel, however
we usually manage to contract the graph to a size proportional to the minimum
cut cactus before calling them. Unfortunately it is not beneficial to perform
the recursive calls embarrassingly parallel, as in almost all cases one of the
connected components of the residual graph contains the vast majority of
vertices and thus also has the overwhelming majority of work. 

\section{Applications} \label{s:balanced}

We can use the minimum cut cactus $C_G$ to find a minimum cut fulfilling certain
balance criteria, such as a most balanced minimum cut, e.g. a minimum cut $(A,V
\backslash A)$ that maximizes min$(|A|,|V \backslash A|)$. Note that this is not
equal to the most balanced $s$-$t$-cut problem, which is NP
hard~\cite{bonsma2007most}. Following that we show how to modify the algorithm
to find the optimal minimum cut for other optimization functions.

One can find a most balanced minimum cut trivially in time $\Oh{(n^*)^3}$, as
one can enumerate all $\Oh{(n^*)^2}$ minimum cuts~\cite{karger2000minimum} and
add up the number of vertices of the original graph $G$ on either side. We now
show how to find a most balanced minimum cut of a graph $G$ in $\Oh{n^* + m^*}$
time, given the minimum cut cactus graph $C_G$. 

For every cut $(A, V\backslash A)$, we define the balance $b(A)$ (or
$b(V\backslash A))$ of the cut as the number of vertices of the original graph
encompassed in the lighter side of the cut. Recall that for any node $v \in
V_G$, $c(v)$ is the number of vertices of $G$ represented by $v$. For a leaf $v
\in V_G$, we set its weight $w(v) = c(v)$ and set the balance $b(v)$ to be the
minimum of $w(v)$ and $n - w(v)$. We root $C_G$ in an arbitrary vertex and
depending on that root define $w(v)$ as the sum of vertex weights in the
subcactus rooted in $v$; and $b(v)$ accordingly. For a cycle $C = \{c_1, \dots,
c_i\}$, we define $b(c_j,\dots,c_{k \mod i})$ with $0 \geq j \geq k$ analogously
as the balance of the minimum cut splitting the cycle so that the sub-cacti
rooted in $c_j, \dots, c_{k \mod i}$ are on one side of the cut and the rest are
on the other side (see blue line in Figure~\ref{fig:mostbalanced} for an
example).

\begin{figure}[t!]\centering
  \includegraphics[width=.7\textwidth]{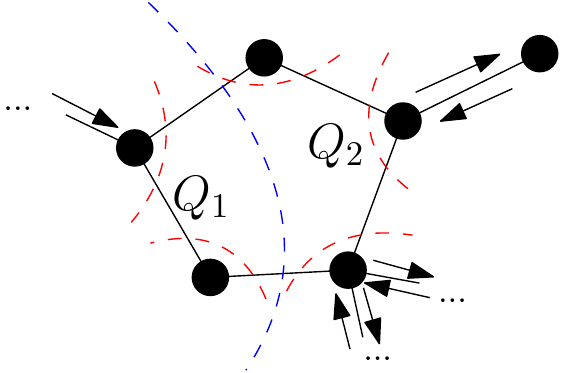}
  \caption{Cycle check in balanced cut algorithm\label{fig:mostbalanced}}
\end{figure}

Let $T_G$ be the tree representation of $C_G$ where each cycle in $C_G$ is
contracted into a single vertex. We perform a depth-first search on $T_G$ rooted
on an arbitrary vertex and check the balance of every cut in $T_G$ when
backtracking. 

As $C_G$ is not necessarily a tree, we might encounter cycles and we explain
next how to extend the depth first search to handle such cycles. Let
$\mathcal{C} = \{c_0,\dots,c_{i-1}\}$ be a cycle and $c_0$ be the vertex
encountered first by the DFS. Due to the cactus graph structure of $C_G$, the
depth-first search backtracks from a vertex $v_{cy}$ in $T_G$ that represents
$\mathcal{C}$ only after all subtrees rooted in $\mathcal{C}$ are explored.
Thus, we know the weight of all subtrees rooted in vertices $c_1,\dots,c_{i-1}$
when backtracking. The weight of $c_0$ is equal to $n$ minus the sum of these
sub-cactus weights. 

Examining all cuts in the cycle would take $i^3$ time, but as we only want to
find the most balanced cut, we can check only a subset of them, as shown in
Algorithm~\ref{alg:cycle}. $Q_1$ and $Q_2$ are \emph{queues}, thus elements are
ordered and the following operations are supported: \emph{queue} adds an element
to the back of the queue, called the \emph{tail} of the queue, \emph{dequeue}
removes the element at the front of the queue, called the \emph{head} of the
queue. We implicitly use the fact that queues can only be appended to, thus an
element $q$ was added to the queue after all elements that are closer to the
head of the queue and before all elements that are closer to its tail.

\begin{algorithm}
  \caption{Algorithm to find most balanced cut in cycle
  $\{c_0,\dots,c_{i-1}\}$}\label{alg:cycle}
  \begin{algorithmic}[1]
  \Procedure{BalanceInCycle}{$G = (V,E), C =\{c_1,\dots,c_i\}$} \State $b_{OPT}
  \gets 0$ \State $Q_1 =$ Queue($\{\}$) \State $Q_2 =$
  Queue($\{c_0,c_1,\dots,c_{i-1}\}$) \While{$c_0$ not $Q_1$.head() for second
  time} \State $b_{OPT} \gets$ checkBalance($Q_1, Q_2$) \If{$w(Q_1) > w(Q_2)$}
  \State $Q_2$.queue($Q_1$.dequeue()) \Else \State $Q_1$.queue($Q_2$.dequeue())
  \EndIf \EndWhile \Return $b_{OPT}$ \EndProcedure
  \end{algorithmic}
  \end{algorithm}

The weight of a queue $w(Q)$ is denoted as the weight of its contents. For queue
$Q = \{c_{j \bmod i},\dots,c_{k \bmod i}\}$ with $0 \leq j \leq k$, we use the
notation $w_{j \bmod i,k \bmod i}$ to denote the weight of $Q$ and
$\overline{w_{j \bmod i, k \bmod i}}$ as the weight of the queue that contains
all cycle vertices not in $Q$. 

In every step of the algorithm, the cut represented by the current state of the
queues consists of the two edges connecting the queue heads to the tails of the
respective other queue. Initially $Q_1$ is empty and $Q_2$ contains all
elements, in order from $c_0$ to $c_{i-1}$. In every step of the algorithm, we
dequeue one element and queue it in the other queue. Thus, at every step each
cycle vertex is in exactly one queue. When we check the balance of a cut, we
compute the weight of each queue at the current point in time; and update
$b_{OPT}$, the best balance found so far, if $(Q_1, Q_2)$ is more balanced. As
we only move one cycle vertex in each step, we can check the balance of an
adjacent cut in constant time by adding and subtracting the weight of the moved
vertex to the weights of each set.

\begin{lemma}
  Algorithm~\ref{alg:cycle} terminates after $O(i)$ steps.
\end{lemma}

\begin{proof}
  In each step of Algorithm~\ref{alg:cycle}, one queue head is moved to the
  other queue. The algorithm terminates when $c_0$ is the head of $Q_1$ for the
  second time. In the first step, $c_0$ is moved to $Q_1$, as the empty queue
  $Q_1$ is the lighter one. The algorithm terminates after $c_0$ then performs a
  full round through both queues and is the head of $Q_1$ again. At termination,
  $c_0$ was thus moved a total of three times, twice from $Q_2$ to $Q_1$ and
  once the other way. As no element can 'overtake' $c_0$ in the queues, every
  vertex will be moved at most three times. Thus, we enter the loop at most $3i$
  times, each time only using a constant~amount~of~time.
\end{proof}

In Algorithm~\ref{alg:cycle}, we only check the balance of a subset of cuts
represented by edges in the cycle $C$. Lemma~\ref{lem:mostbal} shows that none
of the disregarded cuts can have balance better than $b_{OPT}$ and we thus find
the most balanced minimum cut. We call a cut disregarded if its balance was
never checked (Line $6$), and considered otherwise. In order to prove
correctness of Algorithm~\ref{alg:cycle}, we first show the following Lemma:

\begin{lemma}\label{lem:head} Each vertex in the cycle is dequeued from $Q_1$ at
least once in the algorithm.
\end{lemma}

\begin{proof}
  The algorithm terminates when $c_0$ is the head of $Q_1$ for the second time.
  For this, it needs to be moved from $Q_2$ to $Q_1$ twice. As we queue elements
  to the back of a queue, all vertices are dequeued from $Q_2$ before $c_0$ is
  dequeued from it for the second time. In order for $c_0$ to become the head of
  $Q_1$ again, all elements that were added beforehand need to be dequeued from
  $Q_1$.
\end{proof}

\begin{lemma}\label{lem:mostbal} Algorithm~\ref{alg:cycle} finds the most
  balanced minimum cut represented by cycle~$C$.
\end{lemma}

\begin{proof}
  We now prove for each $c_l \in \mathcal{C}$ that all disregarded cuts
  containing the cycle edge separating $c_l$ from $c_{(l-1) \bmod i}$ are not
  more balanced than the most balanced cut found so far. As no disregarded cut
  can be more balanced than the most balanced cut considered in the algorithm,
  the output of the algorithm is the most balanced minimum cut; or one of them
  if multiple cuts of equal balance exist.

  \begin{figure}[H] \centering
    \includegraphics[width=.7\textwidth]{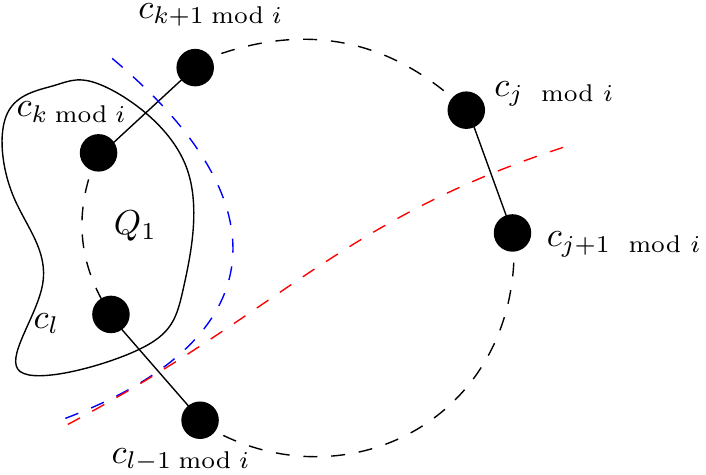}
    \caption{State of $Q_1$ at time $t_l$ (cut in blue). Cut in red denotes cut
    considered at time $t^*$\label{fig:cycle_slice}}.
    \vspace*{-.5cm}
  \end{figure}

  Let $t_l$ be the time that $c_l$ becomes the head of $Q_1$ for the first time.
  Figure~\ref{fig:cycle_slice} shows the state of $Q_1$ at that point in time.
  Let $c_{k \bmod i}$ be the tail of $Q_1$ at time $t_l$ for some integer $k$.
  Right before $t_l$, $c_{l-1 \bmod i}$ was head of the heavier queue $Q_1$ and
  thus dequeued, \ie $Q_1 = \{c_{l-1 \bmod i},\dots,c_{k \bmod i}\}$ has weight
  $w_{l-1 \bmod i, k \bmod i} \geq \overline{w_{l-1 \bmod i, k \bmod i}}$ and
  $c_l$ is now head of $Q_1$.

  From this point $t_l$ the algorithm considers cuts that separate $c_l$ from
  $c_{l-1 \bmod i}$. While $Q_1$ is not heavier than $Q_2$, we add more elements
  to the tail of $Q_1$ (and check the respective cuts) until $Q_1$ is the
  heavier queue. Let $t^*$ be the time when this happens and $c_{j \bmod i}$
  with $j \geq k$ be the tail of $Q_1$ at this point. Note that at time $t^*$,
  $c_l$ is about to be dequeued from $Q_1$. The red cut in
  Figure~\ref{fig:cycle_slice} shows the cut at time $t^*$, where $w_{c_l,c_{j
  \bmod i}} > \overline{w_{c_l,c_{j \bmod i}}}$.
  
  We now prove that all cuts in which $c_l$ is the head of $Q_1$ and its tail is
  not between $c_{k \bmod i}$ and $c_{j \bmod i}$ cannot be more balanced than
  the most balanced cut considered so far.
  
  For all cuts where $c_l$ is head of $Q_1$ and $Q_1$ also contains $c_{j+1
  \bmod i}$, $Q_1$ is heavier than $w_{l,j \bmod i}$, as it contains all
  elements in $c_l,\dots,c_{j \bmod i}$ plus at least one more. As $w_{l,j \bmod
  i} > \overline{w_{l,j \bmod i}}$, \ie $Q_1$ is already heavier when $c_{j
  \bmod i}$ is its tail, all of these cuts are less balanced than
  $(\{c_l,\dots,c_{j \bmod i}\}, \mathcal{C} \backslash \{c_l,\dots,c_{j \bmod
  i}\})$.

  For the cuts in which $c_{k \bmod i}$ is in $Q_2$, \ie $Q_1$ is lighter than
  at time $t_l$, we need to distinguish two cases, depending on whether $w_{l,k
  \bmod i}$ is larger than $\overline{w_{l,k \bmod i}}$ or not.

  If $w_{l,k \bmod i} \leq \overline{w_{l,k \bmod i}}$, all cuts in which $c_l$
  is the head of $Q_1$ and $c_{k \bmod i}$ is in $Q_2$ are less balanced than
  $(\{c_l,\dots,c_{k \bmod i}\}, \mathcal{C} \backslash \{c_l,\dots,c_{k \bmod
  i}\})$, as $Q_1$ is lighter than it is at $t_l$, where it was already not the
  heavier queue.

  If $w_{l,k \bmod i} > \overline{w_{l,k \bmod i}}$, there might be cuts in
  which $c_l$ is the head of $Q_1$ that are more balanced than
  $(\{c_l,\dots,c_{k \bmod i}\}, \mathcal{C} \backslash \{c_l,\dots,c_{k \bmod
  i}\})$ in which $Q_1$ is lighter than at time $t_l$. Thus, consider time $t'$
  when $c_{k \bmod i}$ was added to $Q_1$. Such a time must exist, since $Q_1$
  is initially empty. As $c_{k \bmod i}$ is already the tail of $Q_1$ at time
  $t_l$, $t' < t_l$. At that time $Q_1$ contained $c_{l-1 \bmod i}, \dots,
  c_{k-1 \bmod i}$ and potentially more vertices. 
  
  Still, $w_{l-1 \bmod i, k-1 \bmod i} \leq \overline{w_{l-1 \bmod i, k-1 \bmod
  i}}$, as otherwise $c_{k \bmod i}$ would not have been added to $Q_1$.
  Obviously $w_{l-1 \bmod i, k-1 \bmod i} > w_{l, k-1 \bmod i}$, as $Q_1$ is
  even lighter when $c_{l-1 \bmod i}$ is dequeued. As $w_{l-1 \bmod i, k-1 \bmod
  i}$ is already not heavier than its complement, $(\{c_l,\dots,c_{k-1 \bmod
  i}\}, \mathcal{C} \backslash \{c_l,\dots,c_{k-1 \bmod i}\})$ is more
  imbalanced than the cut examined just before time $t'$. Thus, all cuts where
  $c_l$ is the head of $Q_1$ and $c_{k-1 \bmod i}$ is in $Q_2$ are even more
  imbalanced, as $Q_1$ is even lighter.

  Coming back to the outline shown in Figure~\ref{fig:cycle_slice}, we showed
  that for all cuts in which $c_l$ is head of $Q_1$ and $Q_1$ is lighter than at
  time $t_l$ (left of blue cut) and all cuts where $Q_1$ is heavier than at time
  $t^*$ (below red cut) can be safely disregarded, as a more balanced cut than
  any of them was considered at some point between $t'$ and $t^*$. The algorithm
  considers next all cuts with $c_l$ as head of $Q_1$ and the tail of $Q_1$
  between $c_{k \bmod i}$ and $c_{j \bmod i}$. Thus, the algorithm will return a
  cut that is at least as balanced as the most balanced cut that separates $c_l$
  and $c_{l-1 \bmod i}$. This is true for every cycle vertex $v_l \in
  \mathcal{C}$, which concludes the proof.
\end{proof}

This allows us to perform the depth-first search and find the most balanced
minimum cut in $C_G$ in time $\Oh{n^* + m^*}$. This algorithm can be adapted to
find the minimum cut of any other optimization function of a cut that only
depends on the (weight of the) edges on the cut and the (weight of the) vertices
on either side of the cut. In order to retain the linear running time of the
algorithm, the function needs to be evaluable in constant time on a neighboring
cut. For example, we can find the minimum cut of lowest conductance. The
conductance of a cut $(S, V \backslash S)$ is defined as
$\frac{\lambda(S,(V\backslash S))}{min(a(S), a(V\backslash S))}$, where $a(S)$ is
the sum of degrees for all vertices in set $S$. Note that this is not the
minimum conductance cut problem, which is NP-hard~\cite{andersen2008algorithm},
as we only look at the minimum cuts. To find the minimum cut of lowest
conductance, we set the weight of a vertex $v_{C_G} \in C_G$ to the sum of
vertex degrees encompassed in $v_{C_G}$. Otherwise the algorithm remains the
same.

\section{Experiments and Results} \label{s:experiments}

We now perform an experimental evaluation of the proposed algorithms. This is
done in the following order: first analyze the impact of algorithmic components
on our minimum cut algorithm in a non-parallel setting, i.e.~we compare
different variants for edge selection and see the impact of the various
optimizations detailed in this work. Afterwards, we report parallel speedup on a
variety of large graphs.

\subsubsection{Experimental Setup and Methodology}

We implemented the algorithms using \CC-17 and compiled all code using g++
version 8.3.0 with full optimization (\textttA{-O3}). Our experiments are
conducted on a machine with two Intel Xeon Gold 6130 processors with 2.1GHz with
16 CPU cores each and $256$ GB RAM in total. We perform five repetitions per
instance and report average running time. In this section we first describe our
experimental methodology. Afterwards, we evaluate different algorithmic choices
in our algorithm and then we compare our algorithm to the state of the art. When
we report a mean result we give the geometric mean as problems differ
significantly in cut size and time.

\subsubsection{Instances}

We use a variety of graphs from the 10th DIMACS Implementation
challenge~\cite{bader2013graph} and the SuiteSparse Matrix
Collection~\cite{davis2011university}. These are social graphs, web graphs,
co-purchase matrices, cooperation networks and some generated instances. If a
network has multiple connected components, we run on the largest. The list of
graphs can be found in Section~\ref{rwgraphs}, where graph family (2A) shows a
set of smaller instances and graph family (2B) shows a set of larger and harder
to solve instances.

\subsection{Edge Selection}

Figure~\ref{fig:edgeselect} shows the results for graph family (2A). We compute the cactus graph representing all minimum cuts
using the edge selection variants \textttA{Random}, \textttA{Central},
\textttA{Heavy} and \textttA{HeavyWeighted}, as detailed in
Section~\ref{ss:cactus}. As we want a majority of the running time in the algorithm of Nagamochi~\etal\cite{nagamochi2000fast}, where we
actually select edges, we run a variant of our algorithm that only contracts
edges using connectivity-based contraction and then runs the algorithm of
Nagamochi~\etal\cite{nagamochi2000fast}.

\begin{figure}[t!]\centering
  \includegraphics[width=\textwidth]{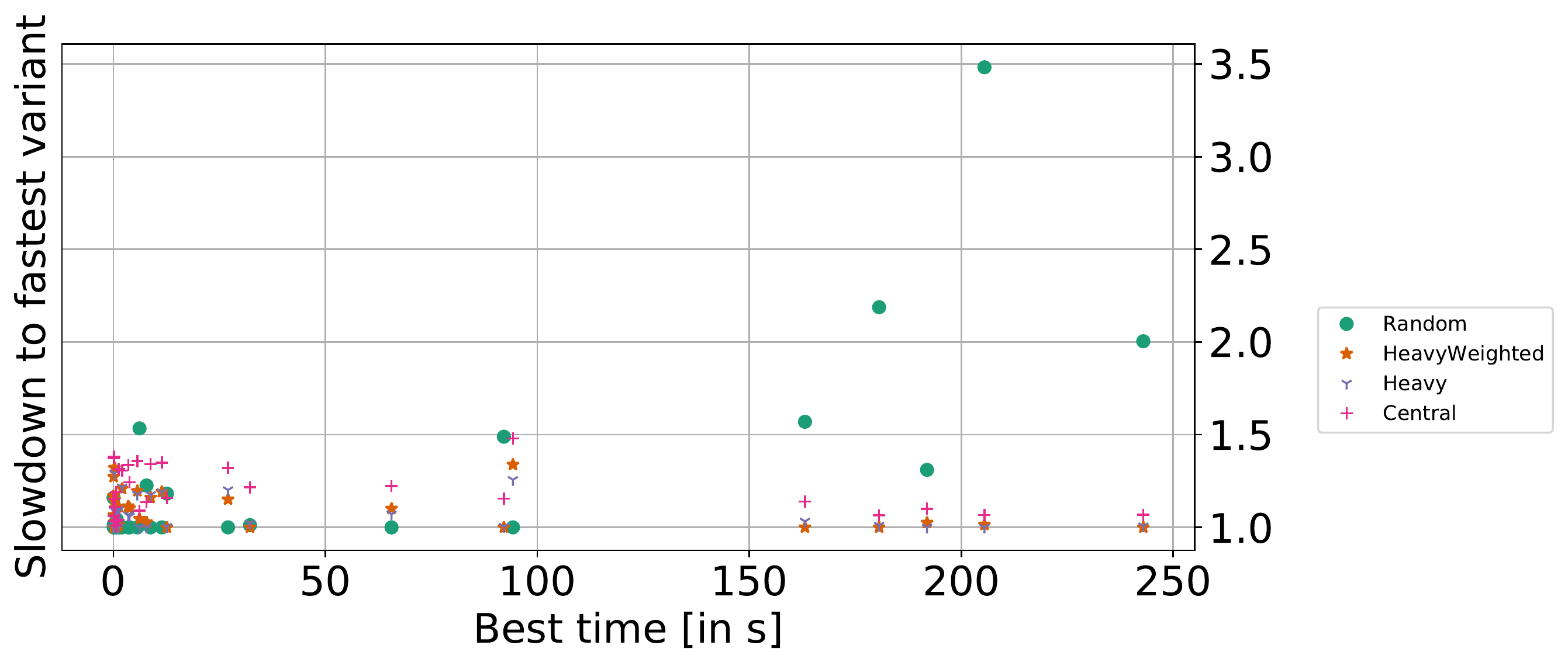}
  \caption{Effect of edge selection strategies.\label{fig:edgeselect}}
\end{figure}%

We can see that in the graphs which cannot be contracted quickly,
\textttA{Random} is significantly slower than all other variants. On
\textttA{cnr-2000}, \textttA{Random} takes over $700$ seconds in average, whereas
all other variants finish in approximately $200$ seconds. This happens
independently of the random seed used, there is no large deviation in the
running time on any of the graphs. On almost all graphs, the variants
\textttA{Heavy} and \textttA{HeavyWeighted} are within $3\%$ of each other, which
is not surprising, as the variants are almost identical. While it optimizes for
'edge centrality' very directly, \textttA{Central} has two iterations of
breadth-first search in each edge selection and thus a sizable overhead. For
this reason it is usually $5-15\%$ slower than \textttA{Heavy} and is not the
fastest algorithm on any graph. On graphs with large $n^*$, all three variants
manage to shrink the graph significantly faster than \textttA{Random}.

On graphs with a low value of $n^*$, we can see that \textttA{Random} is slightly
faster than the other variants. There is no significant difference in the
shrinking of the graph, as almost all selected edges have connectivity larger
than $\lambda$ and thus only trigger a single edge contraction anyway. Thus, not
spending the extra work of finding a `good' edge results in a slightly lower
running time. In the following we will use variant \textttA{Heavy}, which is the
only variant that is never more than $30\%$ slower than the fastest variant on
any graph.

\begin{figure}[t]
    \includegraphics[width=\textwidth]{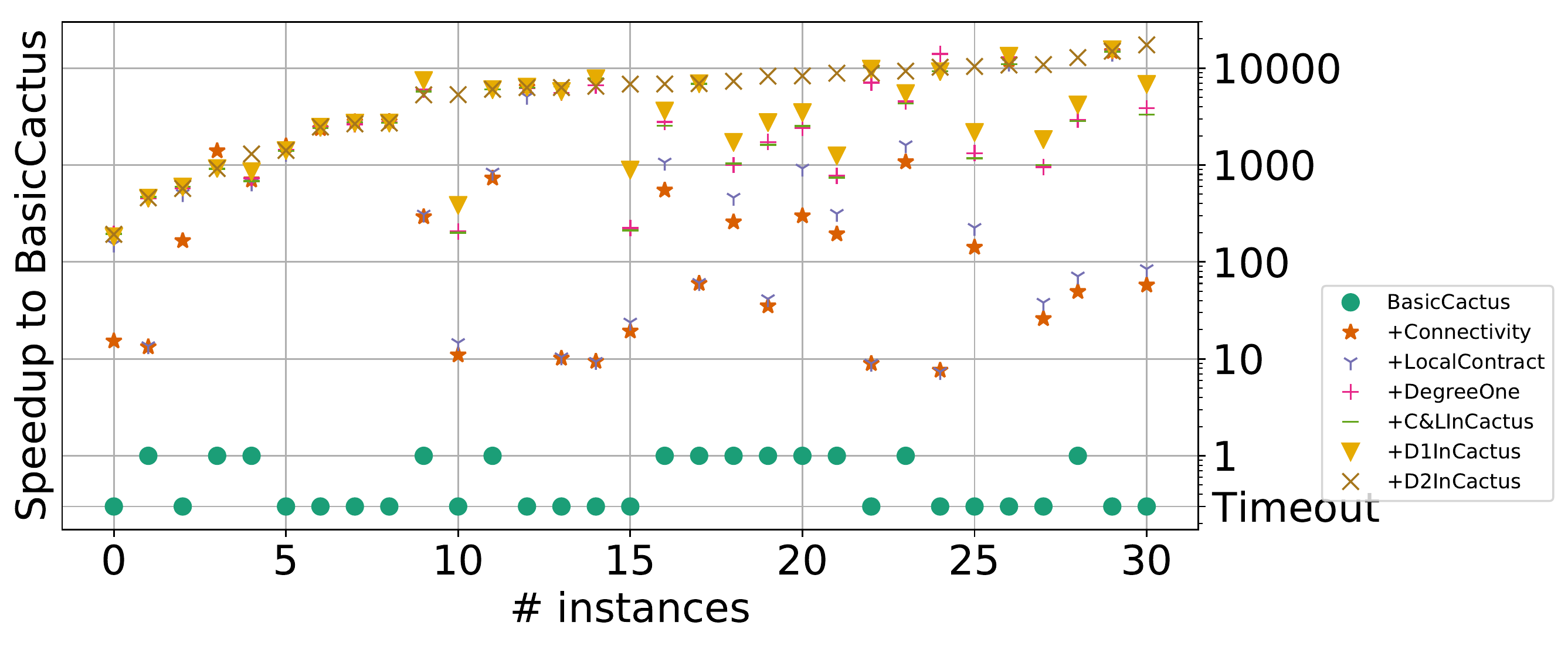}
    \caption{Speedup to \textttA{BasicC} on small graphs
  (Table~\ref{rwgraphs}, graph family 2A)\label{fig:smallopt}}
\end{figure}

\begin{figure}[t]
    \includegraphics[width=\textwidth]{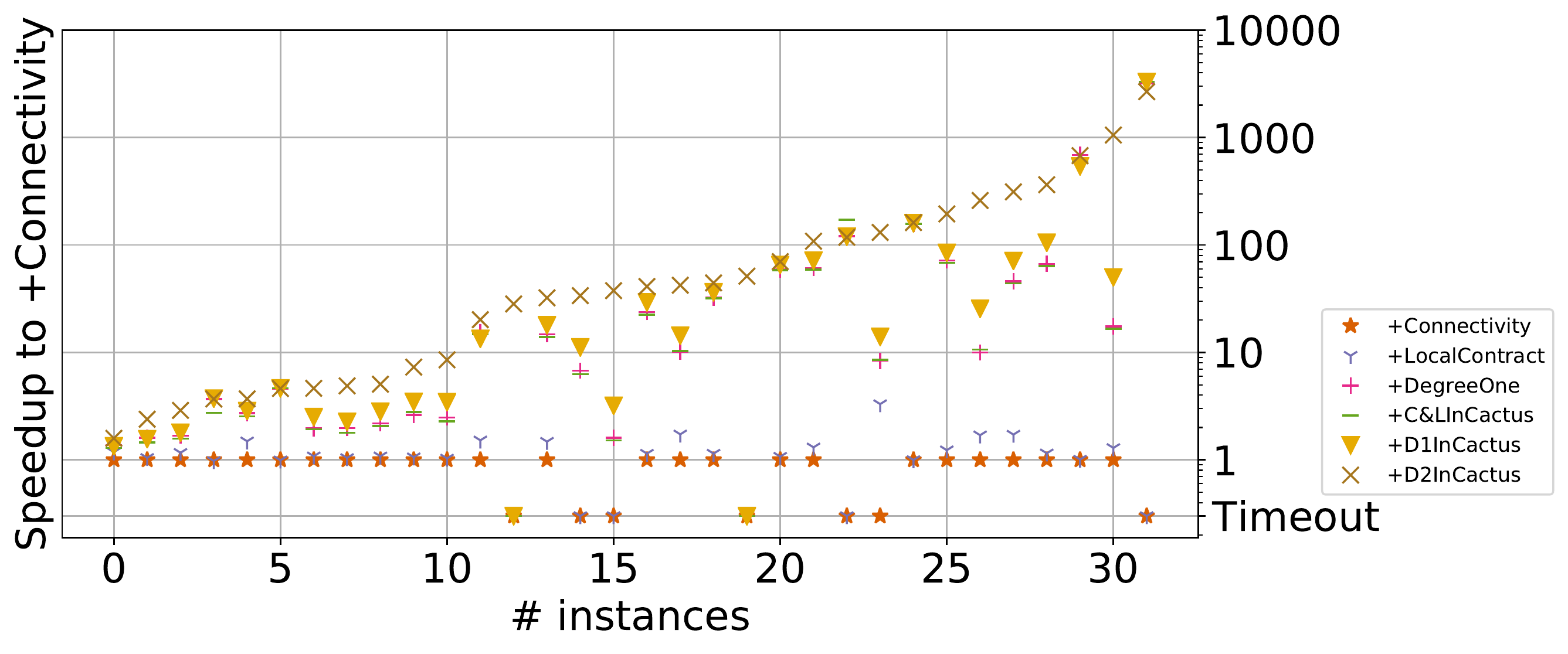}
    \caption{Speedup to \textttA{+Conn} on large graphs
  (Table~\ref{rwgraphs}, graph family 2B)\label{fig:largeopt}}
\end{figure}

\subsection{Optimization}

We now examine the effect of the different optimizations. For this purpose, we
benchmarks different variants on a variety of graphs. We hereby compare the
following variants that build on one another: as a baseline, \optZero{} runs the
algorithm of Nagamochi, Nakao and Ibaraki~\cite{nagamochi2000fast} using
\textttA{Heavy} edge selection on the input graph. \optOne{} additionally runs
\textttA{VieCut}~\cite{henzinger2018practical} to find an upper bound for the minimum cut
and uses this to contract high-connectivity edges as described in
Section~\ref{sss:connectivity}. In addition to this, \optTwo{} also contracts
edges whose neighborhood guarantees that they are not part of any minimum cut,
as described in Section~\ref{sss:local} and Lemma~\ref{lem:local_crit}.
\optThree{} runs also the last remaining contraction routine from
Algorithm~\ref{alg:overview}, contraction and re-insertion of degree-one
vertices as described in Section~\ref{sss:degreeone}. \optFour{} additionally
runs high-connectivity and local contraction in every tenth recursion step.
\optFive{} additionally contracts and re-inserts degree-one vertices in every
recursion step. \optSix{} also runs the degree-two contraction as described in
Section~\ref{sss:degreetwo}. We compare these variants on the graph families
(2A) and (2B) of Table~\ref{rwgraphs}. We use a timeout of $30$ minutes for each
problem. If the baseline algorithm does not finish in the allotted time, we
report speedup to the timeout, so a lower bound for the actual speedup.

Figure~\ref{fig:smallopt} shows the speedup of all variants to the baseline
\optZero{} on all small graphs. We can see that
already just adding \optOne{} gives a speedup of more than an order of magnitude
for each of the graphs in the dataset. Most of the other optimizations manage to
improve the running time of at least some instances by a large margin.
Especially \optThree{}, which is the first contraction for edges that are in a
minimum cut, has speedups of multiple orders of magnitude in some instances.
This is the case as minimum cut edges that are incident to a degree-one vertex
previously incur a flow problem on the whole graph each. However, it is very
easy to see that the edge will be part of exactly one minimum cut, thus we can
contract and re-insert it in constant time. Especially in graphs whose minimum
cut is $1$, all edges can be quickly contracted, as they will either be incident
to a degree-one vertex or be quickly certified to have a connectivity value of
$>1$.

While rerunning \textttA{Connectivity} and \textttA{LocalContract} inside of the
recursive algorithm of Nagamochi~\etal\cite{nagamochi2000fast} does usually not
yield a large speedup, many graphs develop degree-one vertices by having their
whole neighborhood contracted. Thus, \optFive{} has a significant speedup for
most graphs in which $n^*$ is sufficiently large. \optSix{} has an even larger
speedup on these graphs, even when the minimum cut is significantly higher than
$2$, as there are often cascading effects where the contraction of an edge
incident to a degree-two vertex often lowers the degree of neighboring vertices
to two.

Figure~\ref{fig:largeopt} shows the speedup of all variants on large graphs. As
\optZero{} is not able to solve any of these instances in $30$ minutes, we use
\optOne{} as a baseline. The results are similar to Figure~\ref{fig:smallopt},
but we can see even clearer how useful the contraction of degree-two vertices is
in finding all minimum cuts: \optSix{} often has a speedup of more than an order
of magnitude to all other variants and is the only variant that never times out.

\begin{table}[t!]
  \centering
  \caption{Huge social and web graphs. $n^*$ denotes number of vertices in
  cactus graph, max $n$ and max $m$ denote size of smaller block in most
  balanced cut\label{tab:huge}} \resizebox{\textwidth}{!}{
  \begin{tabular}{|r|r|r|r|r|r|r|r|r|r|}
    \hline
    Name & $n$ & $m$ & $n^*$ & $\lambda$ & max. $n$ & max. $m$ & seq. t & par.
    t\\ \hline\hline
    friendster & $65.6$M & $1.81$B & $13.99$M & \numprint{1} & \numprint{897} &
    \numprint{1793} & $1266.35$s & $138.34$s \\\hline
    twitter7 & $41.7$M & $1.20$B & $1.93$M & \numprint{1} & \numprint{47} &
    \numprint{1893} & $524.86$s & $72.51$s \\\hline
    uk-2007-05 & $104.3$M & $3.29$B & $9.66$M & \numprint{1} & \numprint{49984}
    & $13.8$M & $229.18$s & $40.16$s\\\hline
  \end{tabular}
  }
\end{table}

\subsection{Shared-memory Parallelism}

Table~\ref{tab:huge} shows the average running times of our algorithm both
sequential and with $16$ threads on huge social and web graphs. Each of these
graphs has more than a billion of edges and more than a million vertices in the
cactus graph depicting all minimum cuts. On these graphs we have a parallel
speedup factor of $5.7$x to $9.1$x using $16$ threads. On all of these graphs, a
large part of the running time is spent in the first iteration of the
kernelization routines, which already manages to contract most dense blocks in
the graph. Thus, all subsequent operations can be performed on significantly
smaller problems and are therefore much faster.  

\section{Conclusion}

We engineered an algorithm to find all minimum cuts in large undirected graphs.
Our algorithm combines multiple kernelization routines with an engineered
version of the algorithm of Nagamochi, Nakao and
Ibaraki~\cite{nagamochi2000fast} to find all minimum cuts of the reduced graph.
Our experiments show that our algorithm can find all minimum cuts of huge social
networks with up to billions of edges and millions of minimum cuts in a few
minutes on shared memory. We found that especially the contraction of
high-connectivity edges and efficient handling of low-degree vertices can give
huge speedups. Additionally we give a linear time algorithm to find the most
balanced minimum cut given the cactus graph representation of all minimum cuts.
Future work includes finding all near-minimum cuts.

\chapter{Dynamic Minimum Cut}
\label{c:dynmc}

% Unpublished Paper

In this chapter, we give the first implementation of a \emph{fully-dynamic
algorithm} for the \emph{minimum cut problem} in a weighted graph. Our algorithm
maintains an exact global minimum cut under edge insertions and deletions. For
edge insertions, we use the approach of
Henzinger~\cite{henzinger1995approximating} and
Goranci~\etal\cite{goranci2018incremental}, who maintain a compact data
structure of all minimum cuts in a graph and invalidate only the minimum cuts
that are affected by an edge insertion. We use the algorithm presented in
Chapter~\ref{c:allmc} to compute all minimum cuts in a graph. For edge
deletions, we use the push-relabel algorithm of Goldberg and
Tarjan~\cite{goldberg1988new} to certify whether the previous minimum cut is
still a minimum cut. As we only need to certify whether an edge deletion changes
the value of the minimum cut, we can perform optimizations that significantly
improve the speed of the push-relabel algorithm for our application. In
particular, we develop a fast initial labeling scheme and terminate early when
the connecitivity value is certified. 

An important observation for dynamic minimum cut algorithms is that graphs often
have a large set of global minimum cuts. We can see this in the experimental
section of Chapter~\ref{c:allmc}, where we aim to find all minimum cuts in huge graphs. Thus,
dynamic minimum cut algorithms can avoid costly recomputation by storing a
compact data structure representing all minimum
cuts~\cite{henzinger1995approximating,goranci2018incremental} and only
invalidate changed cuts in edge insertion. The data structure we use is a
\emph{cactus graph}, \ie a graph in which every vertex is part of at most one
cycle. A minimum cut in the cactus graph is represented by either a tree edge or
two edges of the same cycle. For a graph with multiple connected components, \ie
a graph whose minimum cut value $\lambda = 0$, the cactus graph $\mathcal{C}$
has an empty edge set and one vertex corresponding to each connected component.

The content of this chapter is based on~\cite{henzinger2021practical}.

The rest of this chapter is organized as follows. We start by explaining the
incremental minimum cut algorithm in Section~\ref{c:dynmc:s:insert}, followed by
a description of the decremental minimum cut algorithm in
Section~\ref{c:dynmc:s:delete}. In Section~\ref{c:dynmc:s:combine}, we show how
to combine the routines into a fully dynamic minimum cut algorithm. In
Section~\ref{c:dynmc:s:experiments}, we perform an experimental evaluation of the
algorithms detailed in this chapter.

\section{Incremental Minimum Cut}
\label{c:dynmc:s:insert}

For incremental minimum cuts, our algorithm is closely related to the exact
incremental dynamic algorithms of Henzinger~\cite{henzinger1995approximating}
and Goranci~\etal\cite{goranci2018incremental}. Upon initialization of the
algorithm with graph $G$, we run the algorithm detailed in Chapter~\ref{c:allmc}
on $G$ to find the weight of the minimum cut $\lambda$ and the cactus graph
$\mathcal{C}$ representing all minimum cuts in $G$. Each minimum cut in
$\mathcal{C}$ corresponds to a minimum cut in $G$ and each minimum cut in $G$
corresponds to one or more minimum cuts in
$\mathcal{C}$~\cite{henzinger1995approximating}.

The insertion of an edge $e=(u,v)$ with positive weight $c(e) > 0$ increases the
weight of all cuts in which $u$ and $v$ are in different partitions, \ie in
different vertices of the cactus graph $\mathcal{C}$. The weight of cuts in
which $u$ and $v$ are in the same partition remains unchanged. As edge weights
are non-negative, no cut weight can be decreased by inserting additional edges.

If $\Pi(u) = \Pi(v)$, \ie both vertices are mapped to the same vertex in
$\mathcal{C}$, there is no minimum cut that separates $u$ and $v$ and all
minimum cuts remain intact. If $\Pi(u) \neq \Pi(v)$, \ie the vertices are mapped
to different vertices in $\mathcal{C}$, we need to invalidate the affected
minimum cuts by contracting the corresponding edges in $\mathcal{C}$.

\subsection{Path Contraction}

Dinitz~\cite{dinitz1993maintaining} shows that for a connected graph with
$\lambda > 0$ the minimum cuts that are affected by the insertion of $(u,v)$
correspond to the minimum cuts on the path between $\Pi(u)$ and $\Pi(v)$. We
find the path using alternating breadth-first searches from $\Pi(u)$ and
$\Pi(v)$. For this path-finding algorithm, imagine the cactus graph
$\mathcal{C}$ as a tree graph in which each cycle is contracted into a single
vertex. On this tree, there is a unique path from $\Pi(u)$ to $\Pi(v)$.

\begin{figure}[t!]
       \includegraphics[width=\textwidth]{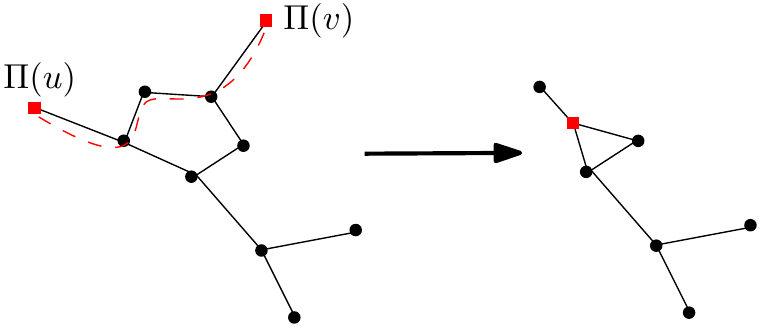}
       \caption{\label{fig:contract} Insertion of edge $e=(u,v)$ - contraction of path in $\mathcal{C}$, squeezing of cycle}
\end{figure}  

For every cycle in $\mathcal{C}$ that contains at least two vertices of the path
between $\Pi(u)$ and $\Pi(v)$, the cycle is ``squeezed'' by contracting the
first and last path vertex in the cycle, thus creating up to two new cycles.
Figure~\ref{fig:contract} shows an example in which a cycle is squeezed. In
Figure~\ref{fig:contract}, the cycle is squeezed by contracting the bottom left
and top right vertices. This creates a new cycle of size $3$ and a ``cycle'' of
size $2$, which is simply a new tree edge in the cactus graph $\mathcal{C}$. For
details and correctness proofs we refer the reader to the work of
Dinitz~\cite{dinitz1993maintaining}. The intuition is that due to the insertion
of the new edge, all cactus vertices in the path from $\Pi(u)$ and $\Pi(v)$ are now
connected with a value $> \lambda$, as their previous connection was $\lambda$
and the newly introduced edge increased it. For any cycle in the path, this also
includes the first and last cycle vertices $x$ and $y$ in the path, as these two vertices now
have a higher connectivity $\lambda(x,y)$. The minimum cuts that are represented
by edges in this cycle that have $x$ and $y$ on the same side are unaffected, as
all vertices in the path from $\Pi(u)$ and $\Pi(v)$ are on the same side of this
cut. As this is not true for cuts that separate $x$ and $y$, we merge $x$ and
$y$ (as well as the rest of the path from $\Pi(u)$ to $\Pi(v)$), which
``squeezes'' the cycle and creates up to two new cycles.

If the graph has multiple connected components, \ie the graph has a minimum cut
value $\lambda = 0$, $\mathcal{C}$ is a graph with no edges where each connected
component is mapped to a vertex. The insertion of an edge between different
connected components $\Pi(u)$ and $\Pi(v)$ merges the two vertices representing
the connected components, as they are now connected.

If $\mathcal{C}$ has at least two non-empty vertices after the edge insertion,
there is at least one minimum cut of value $\lambda$ remaining in the graph, as
all minimum cuts that were affected by the insertion of edge $e$ were just
removed from the cactus graph $\mathcal{C}$. As an edge insertion cannot
decrease any connectivities, $\lambda$ remains the value of the minimum cut. If
$\mathcal{C}$ only has a single non-empty vertex, we need to recompute the
cactus graph $\mathcal{C}$ using the algorithm detailed in Chapter~\ref{c:allmc}.

Checking the set affiliation $\Pi$ of $u$ and $v$ can be done in constant time.
If $\Pi(u) = \Pi(v)$ and the cactus graph does not need to be updated, no
additional work needs to be done. If $\Pi(u) \neq \Pi(v)$, we perform
breadth-first search on $\mathcal{C}$ with $n^* \coloneqq |V(\mathcal{C})|$ and
$m^* \coloneqq |E(\mathcal{C})|$ which has a asymptotic running time of
$\Oh{n^* + m^*} = \Oh{n^*}$, contract the path from $\Pi(u)$ to $\Pi(v)$ in
$\Oh{n^*}$ and then update the set affiliation of all contracted vertices. This
update has a worst-case running time of $\Oh{n}$, however, contracting all
vertices of the path from $\Pi(u)$ to $\Pi(v)$ into the cactus graph vertex that
already corresponds to the most vertices of $G$, we often only need to update
the affiliation of a few vertices. Both the initial computation and a full
recomputation of the minimum cut cactus have a worst-case running time of
$\Oh{nm + n^2 \log{n} + n^*m\log{n}}$.

\section{Decremental Minimum Cut}
\label{c:dynmc:s:delete}

The deletion of an edge $e = (u,v)$ with positive weight $c(e) > 0$ decreases
the weight of all cuts in which $u$ and $v$ are in different partitions. This
might lead to a decrease of the minimum cut value $\lambda$ and thus the
invalidation of the minimum cuts in the existing minimum cut cactus
$\mathcal{C}$. The value of the minimum cut $\lambda(G,u,v)$ that separates
vertices $u$ and $v$ is equal to the maximum flow between them and can be found
by a variety of
algorithms~\cite{dinic1970algorithm,ford1956maximal,goldberg1988new}. In order
to check whether $\lambda$ is decreased by this edge deletion, we need to check
whether $\lambda(G-e,u,v) < \lambda(G)$. For this purpose, we use the
push-relabel algorithm of Goldberg and Tarjan~\cite{goldberg1988new} which aims
to push flow from $u$ to $v$ until there is no possible path remaining. We first
give a short description of the push-relabel algorithm and then show the
adaptions we performed to improve its performance in our application.

\subsection{Push-relabel algorithm}

In this work we use and adapt the push-relabel algorithm of Goldberg and
Tarjan~\cite{goldberg1988new} for the minimum $s$-$t$-cut problem. The algorithm
aims to push as much flow as possible from the \emph{source} vertex $s$ to the
\emph{sink} vertex $t$ and returns the value of the maximum flow between $s$ and
$t$, which is equal to the value of the minimum cut separating
them~\cite{dantzig2003max}. We now give a brief description of the algorithm,
for more details we refer the reader to the original
work~\cite{goldberg1988new}.

Let $G = (V,E,c)$ be a directed edge-weighted graph. An undirected edge
$e=(u,v)$ is hereby interpreted as two symmetric directed edges $(u,v)$ and
$(v,u)$ with $c(e) = c(u,v) = c(v,u)$. In the push-relabel algorithm, each
vertex $v \in V$ has a \emph{distance} or \emph{height label} $d(x)$, initially
$d(x) = 0$ for every vertex except $d(s) = n$. The algorithm handles a
\emph{preflow}, a function $f$ so that for each edge $e$, $0 \geq f(e) \geq
c(e)$ and for each $v \in V \backslash s$, $\sum_{(v,x) \in E} f((v,x)) \leq
\sum_{(y,v)\in E} f((y,v))$ there is at least as much ingoing as outgoing flow.
The difference in ingoing and outgoing flow in a vertex is called the
\emph{excess flow} of this vertex.

First, the algorithm pushes flow from $s$ to all neighboring vertices,
afterwards vertices push their excess flow to neighbors with a lower distance
$d$. If a vertex $v$ has positive excess but no neighbors with a lower distance,
the \emph{relabel} function increases the distance of $v$ until at least one
outgoing preflow $f$ can be increased. At termination, the push-relabel
algorithm reaches a \emph{flow}, where each edge $e$ has $0 \leq f(e) \leq c(e)$
units of flow and the excess of each vertex except $s$ and $t$ is $0$. The value
of the minimum cut $\lambda(s,t)$ separating $s$ and $t$ is equal to the excess
flow on $t$. Inherent to the push-relabel algorithm is the \emph{residual graph}
$G_f = (V,E_f)$ for a given preflow $f$, where $E_f$ contains all edges $e =
(u,v) \in E$ with $f(e) < c(e)$, \ie edges that have capacity to handle
additional flow, and a reverse-edge for every edge where $0 < f(e)$.

\subsection{Early Termination}

We terminate the algorithm as soon as $\lambda(G)$ units of flow reached $v$. If
$\lambda(G)$ units of flow from $u$ reached $v$, we know that $\lambda(G-e,u,v)
\geq \lambda(G)$, \ie the connectivity of $u$ and $v$ on $G-e$ is at least as
large as the minimum cut on $G$, the minimum cut value $\lambda$ remains
unchanged. Note that iff $\lambda(G-e,u,v) = \lambda(G)$, the deletion of $e$
introduces one or more new minimum cuts. We do not introduce these new cuts to
$\mathcal{C}$. The trade-off hereby is that we are able to terminate the
push-relabel algorithm earlier and do not need to perform potentially expensive
operations to update the cactus, but do not necessarily keep all cuts and have
to recompute the cactus earlier. As most real-world graphs have a large number
of minimum cuts, there are far more edge deletions than recomputations of
$\mathcal{C}$.

Each edge deletion calls the push-relabel algorithm using the lowest-label selection
rule with a worst-case running time of $\Oh{n^2m}$~\cite{goldberg1988new}. The
lowest-label selection rule picks the active vertices whose distance label is
lowest, \ie a vertex that is close to the sink $v$. Using
highest-level selection would improve the worst-case running time to
$\Oh{n^2\sqrt{m}}$, but we aim to push as much flow as possible to the sink
early to be able to terminate the algorithm early as soon as $\lambda$ units of
flow reach the sink. Using lowest-level selection prioritizes the vertices close
to the sink and thus increases the amount of flow which reaches the sink at a
given point in time. Preliminary experiments show faster running times using the
lowest-level selection rule.

\subsection{Decremental Rebuild of Cactus Graph}

If the push-relabel algorithm finishes with a value of $< \lambda(G)$, we update
the minimum cut value $\lambda(G-e)$ to $\lambda(G-e,u,v)$. As the minimum cut
value changed by the deletion of $e$ and this deletion only affects cuts which
contain $e$, we know that all minimum cuts of the updated graph $G-e$ separate
$u$ and $v$. We use this information to significantly speed up the cactus
construction. Instead of running the full algorithm from Chapter~\ref{c:allmc},
we run only the subroutine which is used to compute the $(u,v)$-cactus, \ie the
cactus graph which contains all cuts that separate $u$ and $v$, as we know that
all minimum cuts of $G-e$ separate $u$ and $v$. This routine, developed by
Nagamochi and Kameda~\cite{nagamochi1996constructing}, finds a $u$-$v$-cactus a
running time of $\Oh{n+m}$.

Note that the routine of Nagamochi and Kameda~\cite{nagamochi1996constructing}
only guarantees to find all minimum $u$-$v$-cuts if an edge $e = (u,v)$ with
$c(e) > 0$ exists (\cite[Lemma 3.4]{nagamochi1996constructing}). As this edge
was just deleted in $G-e$ and therefore does not exist, it is possible that
\emph{crossing} $u$-$v$-cuts $(X,\overline X)$ and $(Y,\overline Y)$ with $u \in
X$ and $u \in Y$ exist. Two cuts are \emph{crossing}, if both $(\overline X \cap
Y)$ and $(Y \cap \overline X)$ are not empty. As we only find one cut in a pair
of crossing cuts, the $u$-$v$-cactus is not necessarily maximal. However, the
operation is significantly faster than recomputing the complete minimum cut
cactus in which almost all edges are not part of any minimum cut. While it is
not guaranteed that the decremental rebuild algorithm finds all minimum cuts in
$G-e$, every cut of size $\lambda(G-e,u,v)$ that is found is a minimum cut.
As we build the minimum cut cactus out of minimum cuts, it is a valid (but
potentially incomplete) minimum cut cactus and the algorithm is correct.

\subsection{Local Relabeling}

Many efficient implementations of the push-relabel algorithm use the global relabeling heuristic~\cite{cherkassky1997implementing} in order to direct flow towards the sink more efficiently. The push-relabel algorithm maintains a distance label $d$ for each vertex to indicate the distance from that vertex to the sink using only edges that can receive additional flow. The global relabeling heuristic hereby periodically performs backward breadth-first search to compute distance labels on all vertices.

This heuristic can also be used to set the initial distance labels in the flow
network for a flow problem with source $u$ and sink $v$. This has a running time
of $\Oh{n + m}$ but helps lead the flow towards the sink. As our algorithm
terminates the push-relabel algorithm early, we try to avoid the $\Oh{m}$
running time while still giving the flow some guidance. Thus, we perform
\emph{local relabeling} with a \emph{relabeling depth} of $\gamma$ for $\gamma
\in [0, n)$, where we set $d(v) = 0$, $d(u) = n$ and then perform a backward
breadth-first search around the sink $v$, in which we set $d(x)$ to the length
of the shortest path between $x$ and $v$ (at this point, there is no flow in the
network, so every edge in $G$ is admissible). Instead of setting the distance of
every vertex, we only explore the neighborhoods of vertices $x$ with $d(x) <
\gamma$, thus we only set the distance-to-sink for vertices with $d(x) \leq
\gamma$. For every vertex $y$ with a higher distance, we set $d(y) = (\gamma +
1)$. This results in a running time for setting the distance labels of $\Oh{n}$ 
plus the time needed to perform the bounded-depth breadth-first search.

This process creates a ``funnel'' around the sink to lead flow towards it,
without incurring a running time overhead of $\Theta(m)$ (if $\gamma$ is set
sufficiently low). Note that this is useful because the push-relabel algorithm
is terminated early in many cases and thus initializing the distance labels
faster can give a large speedup. We give experimental
results for different relabeling depths $\gamma$ for local relabeling in our
application in Section~\ref{c:dynmc:ss:exp-relabel}.

\subsubsection{Correctness}

Goldberg and Tarjan show that each push and relabel operation in the
push-relabel algorithm preserve a \emph{valid labeling}~\cite{goldberg1988new}.
A valid labeling is a labeling $d$, where in a given preflow $f$ and
corresponding residual graph $G_f$, for each edge $e = (u,v) \in E_f$, $d(u)
\leq d(v) + 1$. We therefore need to show that the labeling $d$ that is given by
the initial local relabeling is a valid labeling.

\begin{lemma} \label{lem:valid-flow}
    Let $G=(V,E,c)$ be a flow-graph with source $s$ and sink $t$ and let $d$ be
    the vertex labeling given by the local relabeling algorithm. The vertex
    labeling $d$ is a valid labeling.
\end{lemma}

\begin{proof}
       The vertex labeling $d$ is generated using breadth-first search. Thus,
       for every edge $e = (u,v)$ where $u \neq s$ and $v \neq s$, $|d(u) -
       d(v)| \leq 1$. We prove this by contradiction. W.l.o.g. assume that $d(u)
       - d(v) > 1$. As $u \neq s$ and $s$ is the only vertex with $d(s) >
       \gamma$, $d(u) \leq \gamma + 1$ and $d(v) < \gamma$. Thus, at some point
       of the breadth-first search, we set the distance labels of all neighbors
       of $v$ that do not yet have a distance label to $d(v) + 1$. As edge $e =
       (u,v)$ exists, $u$ and $v$ are neighbors and the labeling sets $d(u) =
       d(v) + 1$. This contradicts $d(u) - d(v) > 1$.

       This shows that the labeling is valid for every edge not incident to the
       source $s$, as distance labels of incident non-source vertices differ by
       at most $1$. The only edges we need to check are edges incident to $s$.
       In the initialization of the push-relabel algorithm, all outgoing edges
       of the source $s$ are fully saturated with flow and are thus no outgoing
       edge of $s$ is in $E_f$. For ingoing edges $e = (v,s)$, we know that $0
       \leq d(v) \leq \gamma + 1 = n$ and thus know that $d(v) \leq d(s)$. Thus
       $e$ respects the validity of labeling $d$.
\end{proof}

Lemma~\ref{lem:valid-flow} shows that local relabeling gives a valid labeling;
which is upheld by the operations in the push-relabel
algorithm~\cite{goldberg1988new}. Thus, correctness of the modified algorithm
follows from the correctness proof of Goldberg and Tarjan.

Resetting the vertex data structures can be performed in $\Oh{n}$, however there
are $m$ edges whose current flow needs to be reset to $0$. Using early
termination we hope to solve some problems very fast in practice, as we can
sometimes terminate early without exploring large parts of the graph. Thus,
resetting of the edge flows in $\Oh{m}$ is a significant problem and is avoided
using implicit resetting as described in the following paragraph.

Each flow problem that is solved over the course of the dynamic minimum cut
algorithm is given a unique ID, starting at an arbitrary integer and incrementing
from there.  In addition to the current flow on an edge, we also store the ID of
the last problem which accessed the flow on this edge. When the flow of an edge
is read or updated in a flow problem, we check whether the ID of the last access
equals the ID of the current problem. If they are equal, we simply return or
update the flow value, as the edge has already been accessed in this flow
problem and does not need to be reset. Otherwise, we need to reset the edge flow
to $0$ and set the problem ID to the ID of the current problem and then perform
the operation on the updated edge. Thus, we implicitly reset the edge flow on
first access in the current problem. As we increment the flow problem ID after
every flow problem, no two flow problems share the same ID.

Using this implicit reset of the edge flows saves $\Oh{m}$ overhead but
introduces a constant amount of work on each access and
update of the edge flow. It is therefore useful in practice if the problem
terminates with significantly fewer than $m$ flow updates due to early
termination. It does not affect the worst-case running time of the algorithm, as
we only perform a constant amount of work on each edge update. The running time
of the initialization of the implementation is improved from $\Oh{n+m}$ to
$\Oh{n}$, as we do not explicitly reset the flow on each edge.

\section{Fully Dynamic Minimum Cut}
\label{c:dynmc:s:combine}

Based on the incremental and decremental algorithm described in the preceding
sections, we now describe our fully dynamic algorithm. As the operations in the
previous section each output the
minimum cut $\lambda(G)$ and a corresponding cut cactus $\mathcal{C}$ that
stores a set of minimum cuts for $G$, the algorithm gives correct results on all
operations. However, there are update sequences in which every insertion or
deletion changes the minimum cut value and, thus, triggers a recomputation
of the minimum cut cactus $\mathcal{C}$. One such example is the repeated
deletion and reinsertion of an edge that belongs to a minimum cut. In the
following paragraphs we describe a technique that is used to mitigate such
worst-case instances. Nevertheless, it is still possible to construct update
sequences in which the minimum cut cactus $\mathcal{C}$ needs to be recomputed
every $\Oh{1}$ edge updates and thus the worst-case asymptotic running
time per update is equal to the running time of the static algorithm. 

\subsection{Cactus Cache}
\label{c:dynmc:ss:cache}

Computing the minimum cut cactus $\mathcal{C}$ is expensive if there is a large
set of minimum cuts and the cactus is therefore large. Thus, it is beneficial to
reduce the amount of recomputations to speed up the process. On some fully
dynamic workloads, the minimum cut often jumps between values $\lambda_1$ and
$\lambda_2$ with $\lambda_1 > \lambda_2$, where the minimum cut cactus for cut
value $\lambda_1$ is large and thus expensive to recompute whenever the cut
value changes.

A simple example workload is a large unweighted cycle, which has a minimum cut
of $2$. If we delete any edge, the minimum cut value changes to $1$, as the
incident vertices have a degree of $1$. By reinserting the just-deleted edge,
the minimum cut changes to a value of $2$ again and the minimum cut cactus is
equal to the cactus prior to the edge deletion. Thus we can save a
significant amount of work by caching and reusing the previous cactus graph when
the minimum cut is increased to $2$ again.

\subsubsection{Reuse Cactus Graph from Cactus Cache}

Whenever the deletion of an edge $e$ from graph $G$ decreases the minimum cut
value from $\lambda_1$ to $\lambda_2$, we cache the previous cactus
$\mathcal{C}$. After this point, we also remember all edge insertions, as these
can invalidate minimum cuts in $\mathcal{C}$. If at a later point the minimum
cut is again increased from $\lambda_2$ to $\lambda_1$ and the number of edge
insertions divided by the number of vertices in $\mathcal{C}$ is smaller than a
parameter $\delta$, we recreate the cactus graph from the cactus cache instead of
recomputing it. The default value for $\delta$ is $2$. The algorithm does not
store the intermediate edge deletion, as there can only lower connectivities and
by computing the minimum cut value we know that there is no cut of value $<
\lambda_1$ and thus all cuts of value $\lambda_1$ are global minimum cuts.

For each edge insertion since caching the cactus we perform the edge insertion operation
from Section~\ref{c:dynmc:s:insert} to eliminate all cuts that are invalidated
by the edge insertion. All cuts that remain in $\mathcal{C}$ are still minimum
cuts. If there are only a small amount of edge insertions since the cactus was
cached, this is significantly faster than recomputing the cactus from scratch.
As we do not remember edge deletions, the cactus might not contain all minimum
cuts and thus require slightly earlier recomputation.

\section{Experiments and Results} \label{c:dynmc:s:experiments}

We now perform an experimental evaluation of the proposed algorithms.
This is done in the following order. We use the static and dynamic graph
instances detailed in Section~\ref{rwgraphs} and Table~\ref{p:mincut:table:graphs}.
In Section~\ref{c:dynmc:ss:exp-relabel}, we analyze the impact of local relabeling on the
static preflow-push algorithm to determine with value of the relabeling depth to
use in the experiments on dynamic graphs. Then, in
Sections~\ref{c:dynmc:ss:exp-dynamic}~and~\ref{c:dynmc:ss:exp-static}, we evaluate our dynamic
algorithms on a wide variety of instances. In
Section~\ref{c:dynmc:ss:exp-worstcase}, we generate a set of worst-case problems
and use these to evaluate the
performance of our algorithm on instances that were specifically created to be difficult.

\subsubsection{Experimental Setup and Methodology}

We implemented the algorithms using \CC-17 and compiled all code using g++
version 8.3.0 with full optimization (\texttt{-O3}). Our experiments are
conducted on a machine with two Intel Xeon Gold 6130 processors with 2.1GHz with
16 CPU cores each and $256$ GB RAM in total. In this section, we first describe
our experimental methodology. Afterwards, we evaluate different algorithmic
choices in our algorithm and then we compare our algorithm to the state of the
art. When we report a mean result we give the geometric mean as problems differ
significantly in cut size and time.

\begin{figure}[t!]
    \centering
    \includegraphics[width=\textwidth]{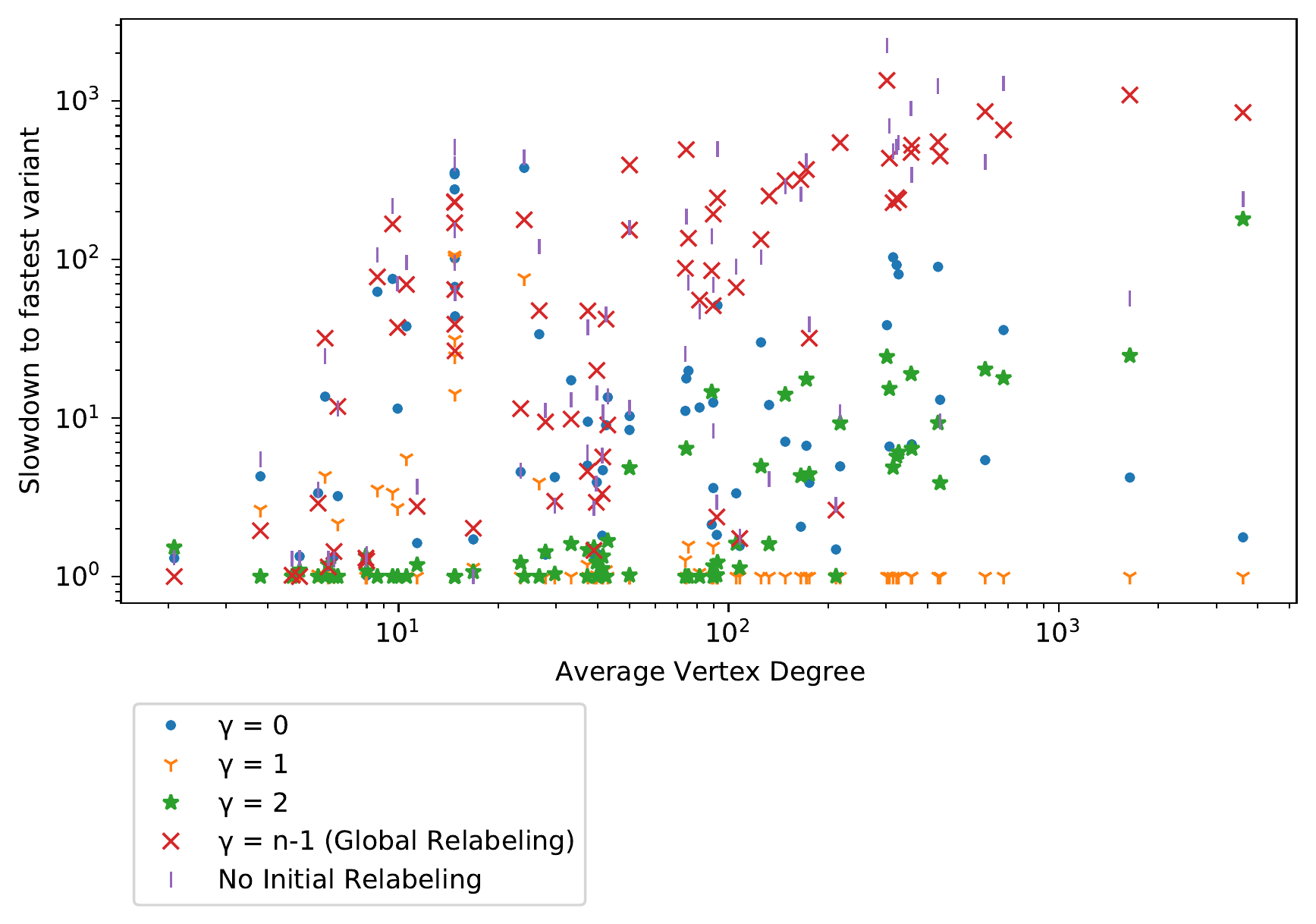}
    \caption{\label{fig:localrelabel}Effect of local relabeling depth on running time of delete operations.}
\end{figure}

\subsection{Local Relabeling}
\label{c:dynmc:ss:exp-relabel}

In order to examine the effects of local relabeling with different values of
relabeling depth $\gamma$, we run experiments using all static graph instances
(Graph Family A and Graph Family B) from Table~\ref{p:mincut:table:graphs}, in which we
delete $1000$ random edges in random order. We report the total time spent
executing delete operations. We compare a total of $5$ variants, one that does
not run initial relabeling, three variants with relabeling depth $\gamma =
0,1,2$ and one variant which performs global relabeling in the initialization
process, \ie local relabeling with depth $\gamma = (n-1)$. Local relabeling with
$\gamma = 0$ is very similar to no relabeling, however the distance value of
non-sink vertices are set to $(\gamma + 1) = 1$ and not to $0$. 

In Figure~\ref{fig:localrelabel}, we report the slowdown to the fastest variant
for all static graph instances from Table~\ref{p:mincut:table:graphs}. The x-axis shows
the average vertex degree for the instances. On most instances, the fastest
variant is local relabeling with $\gamma = 1$. Depending on the graph instance,
this variant spends $25-90\%$ of the deletion time in the initialization
(including initial relabeling). An increase in labeling depth increases the
initialization running time, but decreases the subsequent algorithm running
time. Thus we aim to find a labeling depth value that maintains some balance
between initial labeling and the subsequent algorithm execution. On some
instances, it is outperformed by local relabeling with $\gamma = 2$, which is
slower by a factor of $3-10$x on most instances, with $90-99\%$ of the total
running time spent in the initialization of the algorithm. We can see that in
instances with a higher average degree, local relabeling with $\gamma = 1$
performs better. This is an expected result, as the larger local relabeling is
more expensive in higher-average-degree graphs, as the $2$-neighborhood of a
vertex is much larger. Local relabeling with $\gamma = 2$ spends $90-99\%$ of
the total running time in initialization and initial relabeling. The same effect
is even more pronounced for the variant which performs global relabeling in
initialization. On vertices with a low average degree, we can perform global
relabeling in reasonable time, which makes the variant competitive with the
local relabeling variants. However, in high average degree instances, the
excessive running time of a global relabeling step causes the variant to have
slowdowns of up to $1000$x compared to the fastest variant. On all instances,
the vast majority of running time is spent in initialization including initial
global relabeling.

One graph family where local relabeling with $\gamma = 1$ performs badly are the
graph instances based on \texttt{auto}~\cite{karypis1998fast}, a 3D finite
element mesh graph. These graphs are rather sparse (average degree $~15$) and
planar. On these graphs, the value of the minimum cut divided by the average
degree is very large, as they do not contain any vertices of degree $1,2,3$.
Thus, the variants which perform only minor local relabeling do not guide the
flow enough and therefore the push-relabel algorithm takes a long time. On most other instances
in our test set, local relabeling with $\gamma = 1$ is enough to guide at least
$\lambda$ flow to the sink quickly.

Local relabeling with a relabeling depth $\gamma = 0$ (\ie we set the distance
of the sink to $0$, the source to $n$ and all other vertices to $1$) has a slowdown factor of
$10-100$x with only $1-10\%$ of the running time spent in the initialization.
The slowdown factor is generally increasing for larger values of the minimum cut
$\lambda$ and average degree, which indicates that ``the lack of guidance
towards the sink'' causes the algorithm to send flow to regions of the graph
that are far away from the source. For graphs with large minimum cut value
$\lambda$, the algorithm does not terminate early and needs to perform
a significant amount of push and relabel steps. In variants that perform more
relabeling at initialization, the flow is guided towards the sink by the
distance labels and the termination trigger is reached faster. The variant which
does not include any relabeling in the initialization phase has similar issues
with an even larger slowdown factor of $10-2000$x, as even flow that is already
incident to the sink does not necessarily flow straight to the sink.

On most instances, local relabeling with depth $\gamma=1$ performed best, as it
helps guide the flow towards the sink with additional work (compared to no
relabeling) only equal to the degree of the sink. While performing more
relabeling can increase this guidance even further, it comes with a trade-off in
additional time spent in the initialization. Note that this is not a general
observation for the push-relabel algorithm and can only be applied to our
application, in which the push-relabel algorithm is terminated early as soon as
$\lambda$ units of flow reach the sink vertex. Based on these experiments, we
use local relabeling with $\gamma = 1$ for edge deletions in all following
experiments.

\subsection{Dynamic Graphs}
\label{c:dynmc:ss:exp-dynamic}

\begin{figure}[t!]
       \centering
       \includegraphics[width=\textwidth]{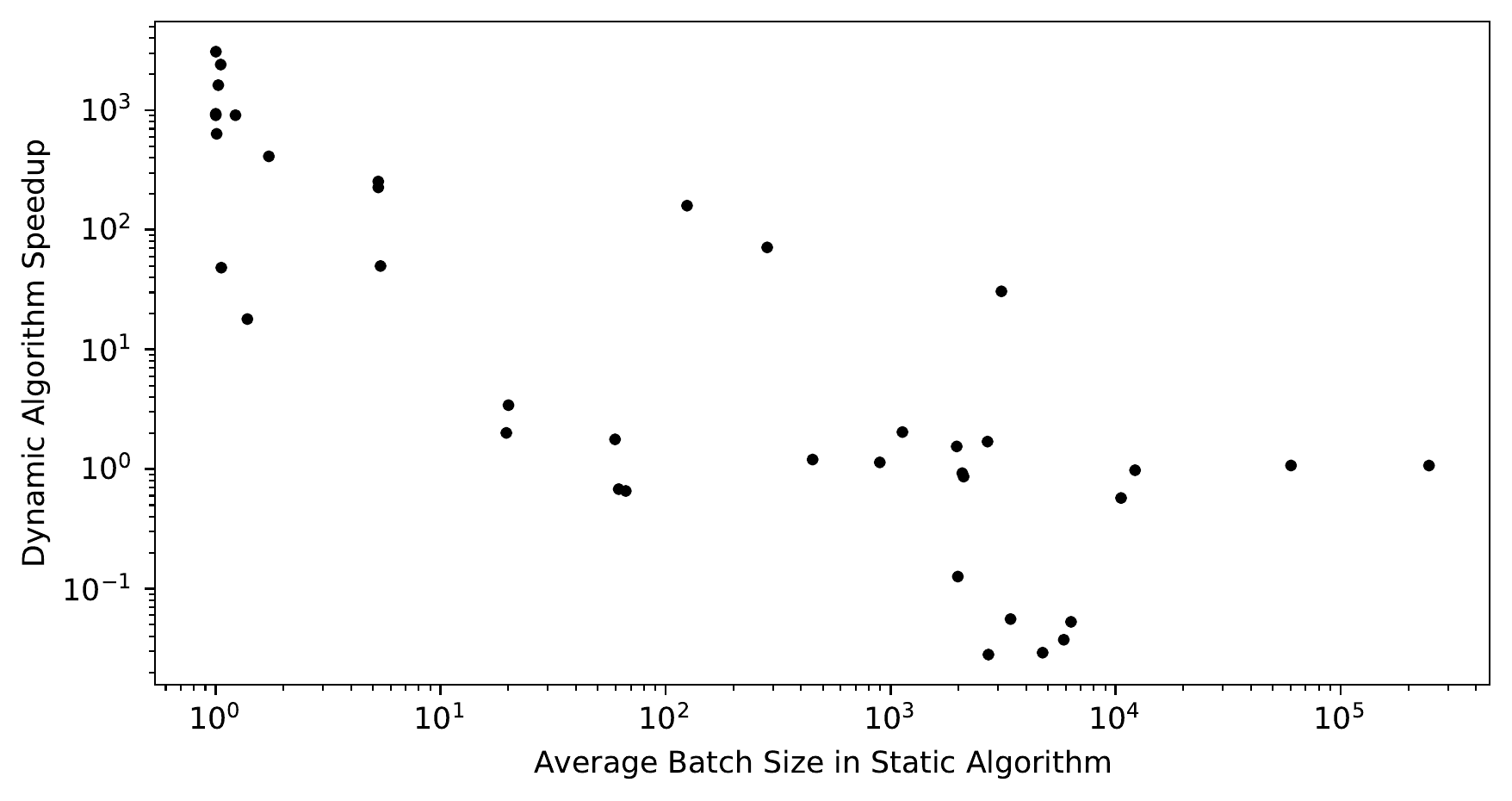}
       \caption{\label{fig:dyng}Speedup of Dynamic Algorithm.}
\end{figure}

Figure~\ref{fig:dyng} shows experimental results on the dynamic graph instances
from Graph Family C in  Table~\ref{p:mincut:table:graphs}. These graph instances are
mostly incremental with some being fully dynamic and most instances have
multiple connected components, \ie a minimum cut value $\lambda = 0$, even after
all insertions. On these incremental graphs with multiple connected components, our algorithm
behaves similar to a simple union-find based connected components algorithm that
for edge insertion checks whether the incident vertices already belong to the same
connected component and merges their connected components if they are different.

In this section we compare our dynamic minimum cut algorithm to
the static algorithm of Nagamochi~\etal\cite{nagamochi1994implementing}, which
has been shown to be one of the fastest sequential algorithms for the minimum
cut problem~\cite{Chekuri:1997:ESM:314161.314315,henzinger2018practical}. The
static algorithm performs the updates batch-wise, \ie the static algorithm is
not called inbetween multiple edge updates with equal timestamp. In
Figure~\ref{fig:dyng}, we show the dynamic speedup in comparison to the average
batch size. As expected, there is a large speedup factor of up to $1000$x for
graphs with small batch sizes; and the speedup decreases for increasing batch
sizes. The family of instances in which the dynamic algorithm is outperformed by
the static algorithm is the \texttt{insecta-ant-colony} graph
family~\cite{mersch2013tracking}. These graphs have a very high minimum cut
value and fewer batches than changes in the minimum cut value. Therefore, the
dynamic algorithm which updates on every edge insertion needs to recompute the
minimum cut cactus more often than the static algorithm is run and, thus, takes a longer time.

As these dynamic instances do not have sufficient diversity, we also perform experiments on static graphs in graph family B in which a subset of edges is inserted or removed dynamically. We report on this experiment in the following section.

\subsection{Random Insertions and Deletions from Static Graphs}
\label{c:dynmc:ss:exp-static}

\begin{figure}[t!]
       \centering
       \includegraphics[width=\textwidth]{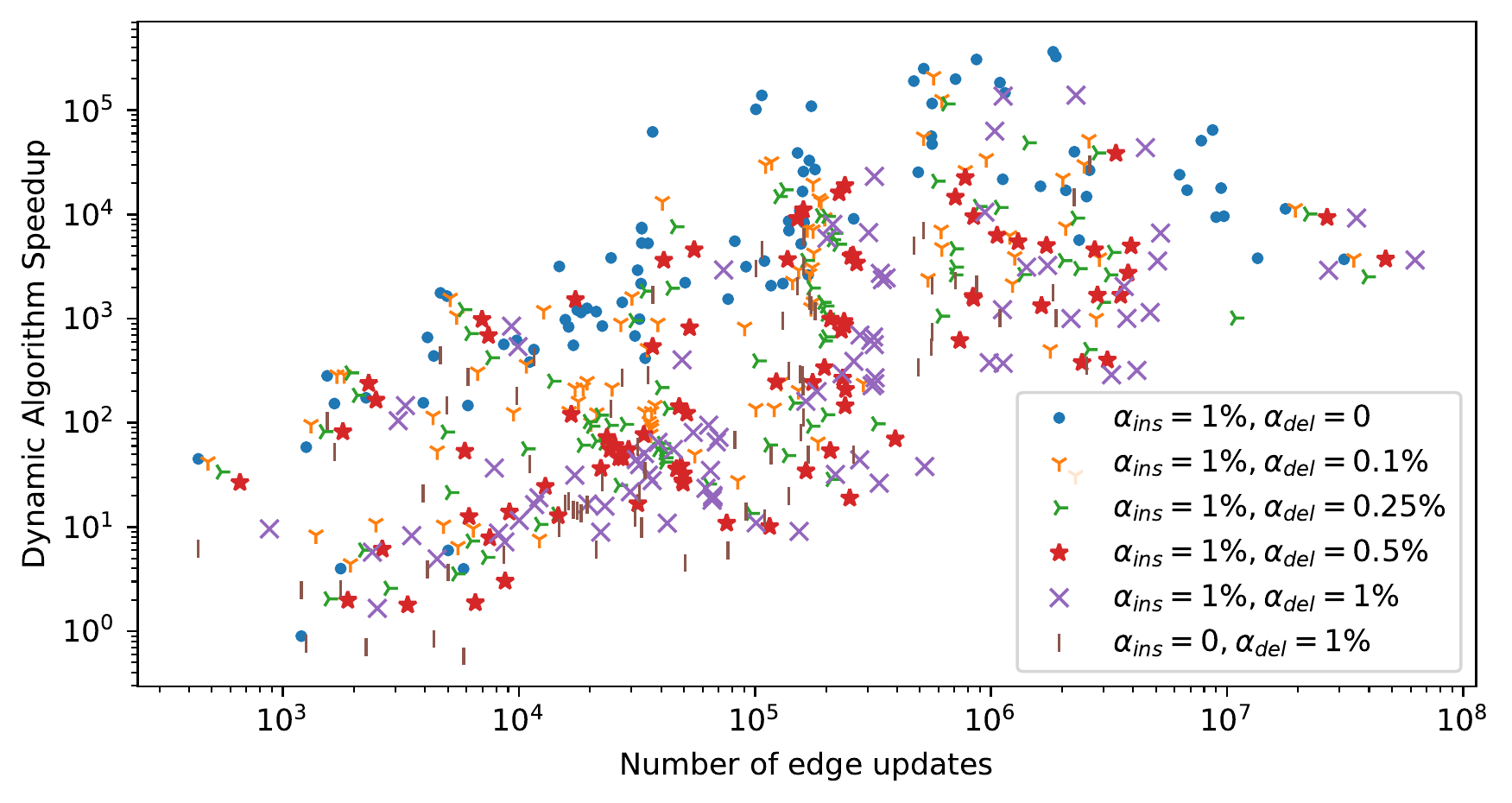}
       \caption{\label{fig:static}Speedup of Dynamic Algorithm on Random Insertions and Deletions from Static Graphs.}
\end{figure}

Figure~\ref{fig:static} shows results for dynamic edge insertions and deletions
from all graphs in Graph Family A and B from Table~\ref{p:mincut:table:graphs}. These
graphs are static, we create a dynamic problem from graph $G=(V,E,c)$ as
follows: let $\alpha_{ins} \in (0,1)$ and $\alpha_{del} \in (0,1)$ with
$\alpha_{ins} + \alpha_{del} < 1$ be the edge insertion and deletion rate. We
randomly select edge lists $E_{ins}$ and $E_{del}$ with $|E_{ins}| =
\alpha_{ins} \cdot |E|$, $|E_{del}| = \alpha_{del} \cdot |E|$ and $E_{ins} \cap
E_{del} = \emptyset$. For every vertex $v \in V$, we make sure that at least one
edge incident to $v$ is neither in $E_{ins}$ nor in $E_{del}$, so that the
minimum degree of $(V,E \backslash (E_{ins} \cap E_{del}), c)$ is strictly
greater than $0$ at any point in the update sequence.

We initialize the graph as $(V,E \backslash E_{ins}, c)$ and create a sequence
of edge updates $E_u$ by concatenating $E_{ins}$ and $E_{del}$ and randomly
shuffling the combined list. Then we perform edge updates one after another and
compute the minimum cut - either statically using our efficient implementation
of the algorithm of
Nagamochi~\etal\cite{nagamochi1994implementing} or by performing an update in
the dynamic algorithm - after every update. Note that all of these algorithms
are sequential. We report the total running time of
either variant and give the speedup of the dynamic algorithm over the static
algorithm as a function of the number of edge updates performed. For each graph
we create problems with $\alpha_{ins} = 1\%$ and $\alpha_{del} \in
\{0,0.1\%,0.25\%,0.5\%,1\%\}$; and additionally 
a decremental problem with $\alpha_{ins} = 0$ and $\alpha_{del} = 1\%$. We set
the timeout for the static algorithm to $1$ hour, if the algorithm does not
finish before timeout, we approximate the total running time of the static
algorithm by performing $100$ or $1000$ updates in batch.

Dynamic edge insertions are generally much faster than edge deletions, as most
real-world graphs have large sets that are not separated by any global minimum
cut. When inserting an edge where both incident vertices are in the same set in
$\mathcal{C}$, the edge insertion only requires two array accesses; if they are
in different sets, it requires a breadth-first search on the relatively small
cactus graph $\mathcal{C}$ and only if there are no minimum cuts remaining, an
edge insertion requires a recomputation. In contrast to that, every edge
deletion requires solving of a flow problem and therefore takes significantly
more time in average. Therefore, the average speedup is larger on problems with
a higher rate of edge insertions.

Generally, the speedup of the dynamic algorithm increases with larger problems
and more edge updates. For larger graphs with $\geq 10^6$ edge updates, the
average speedup is more than four orders of magnitude for instances with
$\alpha_{del} = 0$ and still more than two orders of magnitude for large
instances when $\alpha_{del} = \alpha_{ins} = 1\%$. Note that in this
experiment, the number of edge updates is a function of the number of edges,
thus instances with more updates directly correspond to graphs with more edges.

For decremental instances with $\alpha_{ins} = 0$, the speedup is generally
lower, but still reaches multiple orders of magnitude in larger instances.

\subsubsection{Most Balanced Minimum Cut}

In Section~\ref{s:balanced} we show that given the cactus graph
$\mathcal{C}$ we can compute the most balanced minimum cut, \ie the minimum cut
which has the highest number of vertices in the smaller partition, in $\Oh{n^*}$
time. In our algorithm for the dynamic minimum cut problem we also compute a
cactus graph of minimum cuts, however this cactus graph does not necessarily
contain all minimum cuts in $G$, as we do not introduce new minimum cuts added
by edge deletions. 

We use the algorithm given in Chapter~\ref{c:allmc} to find the
most balanced minimum cut for all instances of Graph Family B every $1000$ edge
updates and compare it to the most balanced minimum cut found by our algorithm.
In instances that are not just decremental, in $97.3\%$ of all cases where there
is a nontrivial minimum cut (\ie smaller side contains multiple vertices), both
algorithms give the same result, \ie our algorithm can almost always output the most balanced
minimum cut. In the instances that are purely decremental, \ie $|E_{ins}| = 0$, we
only find the most balanced minimum cut in $25.4\%$ of cases where there is a
non-trivial minimum cut. This is the case because an increase of the minimum cut
prompts a full recomputation of a cactus graph that represents all (potentially
many) minimum cuts, thus also the most balanced minimum cut. Only if this cut in
particular is affected by an edge update, the dynamic algorithm ``loses'' it. In
the purely decremental case, the minimum cut value only decreases. Thus, the
dynamic algorithm only knows one or a few minimum cuts. All cuts that reach the
same value $\lambda$ in later edge deletions are not in $\mathcal{C}$, as we do
not add cuts of the same value to it. As these decremental instances do not
have any edge insertions that can increase the value of these cuts, there is
eventually a large set of minimum cuts of which the algorithm only knows a few.
If maintaining a balanced minimum cut is a requirement, this can easily be
achieved by occasionally recomputing the entire cactus graph $\mathcal{C}$ from scratch.

\subsection{Worst-case Instances}
\label{c:dynmc:ss:exp-worstcase}

\begin{figure}[t!]
       \centering
       \includegraphics[width=\textwidth]{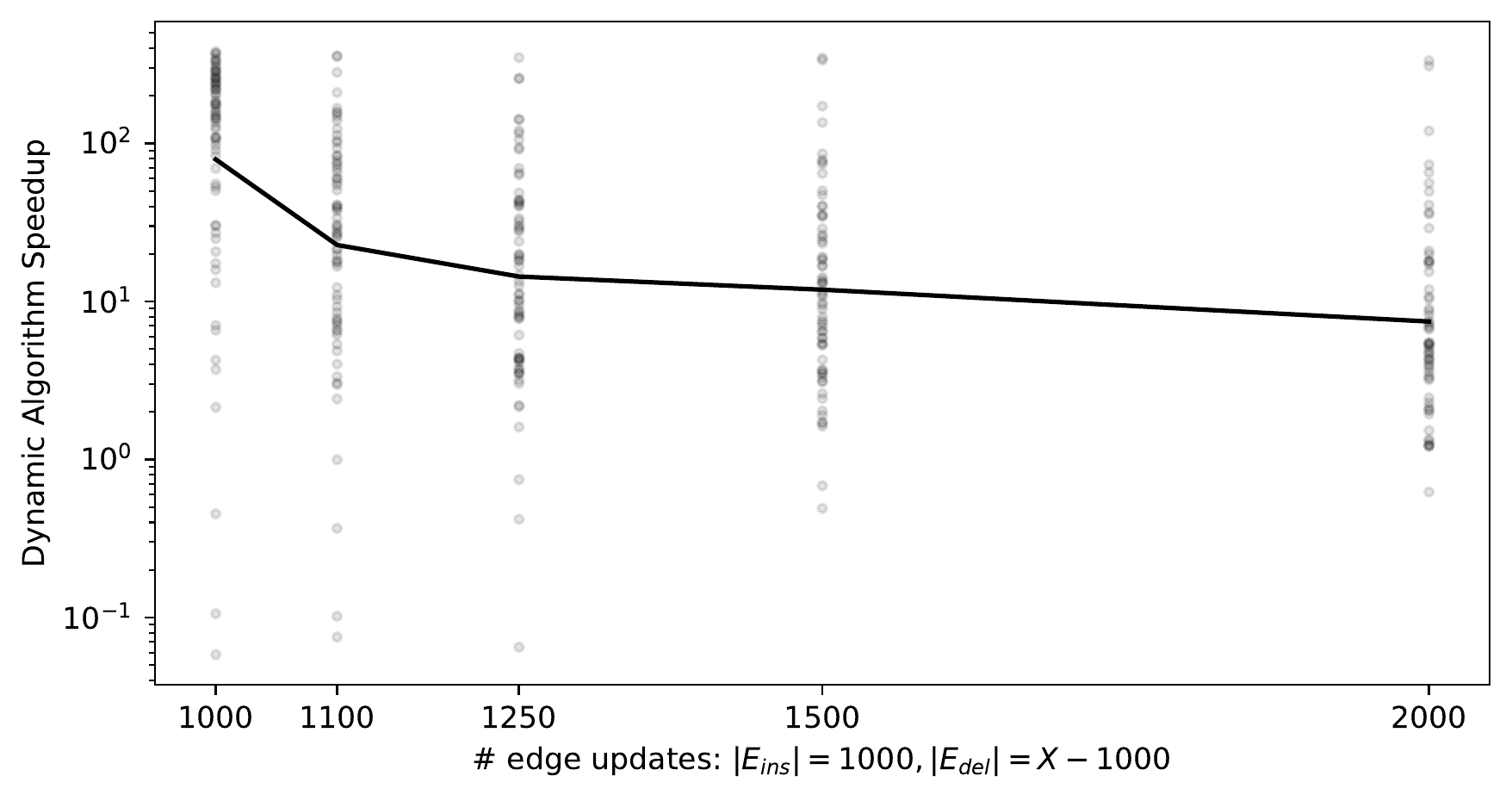}
       \caption{\label{fig:worstcase}Speedup of Dynamic Algorithm on Worst-case Insertions and Deletions from Static Graphs.}
\end{figure}

On random edge insertions, there is a high chance that the vertices incident to
the newly inserted edge were not separated by a minimum cut and therefore
require no update of the cactus graph $\mathcal{C}$. In this experiment we aim
to generate instances that aim to maximize the work performed by the dynamic
algorithm. We initialize the graph as $G=(V,E,c)$ and add random unit-weight
edges $e = (u,v)$ where $\Pi(u) \neq \Pi(v)$ for every newly added edge. Then we
randomly select $|E_{ins}| = 1000$ edges to add so that for each such edge $(u,v)$, $\Pi(u)
\neq \Pi(v)$ before inserting $(u,v)$, and select a subset $E_{del}
\subseteq E_{ins}$ to delete. For each graph we create $5$ problems, with
$|E_{del}| \in \{0,100,250,500,1000\}$. We randomly shuffle the edge updates
while making sure that an edge deletion is only performed after the respective
edge has been added to the graph, but still interspersing edge insertions and
deletions to create true worst-case instances for the dynamic algorithm, as each
edge deletion or insertion affects one or multiple minimum cuts in the graph.

Figure~\ref{fig:worstcase} shows the results of this experiment. Each low-alpha
dot shows the speedup of the dynamic algorithm on a single problem, the black
line gives the geometric mean speedup. As indicated in previous experiments, we
can see that the average speedup decreases when the ratio of deletions is
increased. However, even on these worst-case instances, the mean speedup factor
is still $7.46$x for $|E_{ins}| = |E_{del}| = 1000$ up to $79.2$x for the purely
incremental instances on instances where both algorithms finished before timeout
at one hour. Similar to previous experiments, the speedup factor increases with
the graph size.

On these problem instances we can see interesting effects. Especially in
instances with $|E_{del}| = 500$ we can see many instances where the minimum cut
fluctuates between two different values in more than half of all edge updates.
As the larger of the values usually has a large cactus graph $\mathcal{C}$, this
would result in expensive recomputation on almost every update. However, using
the cactus caching technique detailed in Section~\ref{c:dynmc:ss:cache} we can
save this overhead and simply reuse the almost unchanged previous cactus graph.
In some cases, this reduces the number of calls to the static all-minimum-cut
algorithm by more than a factor of $10$.

We also find some instances where the static graph has few minimum cuts, but
there is a large set of cuts slightly larger than lambda. One such example are planar graphs
derived from Delaunay triangulation~\cite{lee1980two} that have a few vertices
of minimal degree near the edges of the triangulated object, but a large number
of vertices with a slightly larger degree. If we now add edges to increase the
degree of the minimum-degree vertices, the resulting graph has a huge number of
minimum cuts and computing all minimum cuts is significantly more expensive than
computing just a single minimum cut. In these instances the dynamic algorithm is
actually slower than rerunning the static algorithm on every edge update. The
dynamic algorithm is slower than the static algorithm in $3.9\%$ of the worst-case instances.

\section{Conclusion}

In this chapter, we presented the first implementation of a fully-dynamic
algorithm that maintains the minimum cut of a graph under both edge insertions
and deletions. Our algorithm combines ideas from the theoretical foundation with
efficient and fine-tuned implementations to give an algorithm that outperforms
static approaches by up to five orders of magnitude on large graphs. In our
experiments, we show the performance of our algorithm on a wide variety of graph
instances.

Future work includes maintaining all global minimum cuts also under edge
deletions and employing shared-memory or distributed parallelism to further
increase the performance of our algorithm.

\part{The Balanced Graph Partitioning Problem}
\label{p:gp}

\chapter{ILP-based Local Search for Graph Partitioning}

% SEA'18 Paper, JEA'20 Paper
Computing high-quality balanced graph partitions is a challenging problem with
numerous applications.  In this chapter, we present a novel meta-heuristic for
the balanced graph partitioning problem.  Our approach is based on integer
linear programs that solve the partitioning problem to optimality.  However,
since those programs typically do not scale to large inputs, we adapt them to
heuristically improve a given partition.  We do so by defining a much smaller
model that allows us to use symmetry breaking and other techniques that make the
approach scalable. For example, in Walshaw’s well-known benchmark tables, we are
able to improve roughly half of all entries when the number of blocks is high.
Additionally, we include our techniques in a memetic framework and develop a
crossover operation based on the proposed techniques.

The content of this chapter is based
on~\cite{henzinger2020ilp}~and~\cite{henzinger2018ilp}.

\section{Introduction}
\emph{Balanced graph partitioning} is an important problem in computer science
and engineering with an abundant amount of application domains, such as VLSI
circuit design, data mining and distributed systems~\cite{schulz2018graph}. It
is well known that this problem is NP-complete~\cite{BuiJ92} and that no
approximation algorithm with a constant ratio factor exists for general graphs
unless P=NP~\cite{BuiJ92}. Still, there is a large amount of literature on
methods (with worst-case exponential time) that solve the graph partitioning
problem to optimality. This includes methods dedicated to the bipartitioning
case
\cite{armbruster2007branch,Armbruster2008,delling2012exact,delling2012better,feldmann2011n,felner2005,karisch2000solving,HagerPZ13,gp:lp,sellmann2003multicommodity}
and some methods that solve the general graph partitioning
problem~\cite{ferreira1998node,sensen2001lower}. Most of these methods rely on
the branch-and-bound framework \cite{land1960automatic}. However, these methods
can typically solve only very small problems as their running time grows
exponentially, or if they can solve large bipartitioning instances using a
moderate amount of time \cite{delling2012exact,delling2012better}, the running
time highly depends on the bisection width of the graph. Methods that solve the
general graph partitioning problem \cite{ferreira1998node,sensen2001lower} have
huge running times for graphs with up to a few hundred vertices. Thus in
practice mostly heuristic algorithms are used.

Typically the graph partitioning problem asks for a partition of a graph into
$k$ blocks of about equal size such that there are few edges between them. Here,
we focus on the case when the bounds on the size are very strict, including the
case of \emph{perfect balance} when the maximal block size has to equal the
average block size.

Our focus here is on solution quality, \ie minimize the number of edges
that run between blocks. During the past two decades there have been numerous
researchers trying to improve the best graph partitions in Walshaw's well-known
partitioning benchmark~\cite{soper2004combined,wswebsite}. Overall there have
been more than forty different approaches that participated in this benchmark.
Indeed, high solution quality is of major importance in applications such as
VLSI Design \cite{alpert1995rdn,alpert1999spectral} where even minor
improvements in the objective can have a large impact on the production costs
and quality of a chip. High-quality solutions are also favorable in applications
where the graph needs to be partitioned only once and then the partition is used
over and over again, implying that the running time of the graph partitioning
algorithms is of a minor
concern~\cite{DellingGPW11,heuvelinecoop,klsv-dtdch-10,Lau04,wagner2005pgs,ls-csarr-12}.
Thirdly, high-quality solutions are even important in areas in which the running
time overhead is paramount \cite{soper2004combined}, such as finite  element
computations \cite{schloegel2000gph} or the direct solution of sparse linear
systems \cite{george1973nested}. Here, high-quality graph partitions can be
useful for benchmarking purposes, \ie measuring how much more running time can
be saved by higher quality solutions.

In order to compute high-quality solutions, state-of-the-art local search
algorithms exchange vertices between blocks of the partition trying to decrease
the cut size while also maintaining balance. This highly restricts the set of
possible improvements. Sanders and Schulz introduced new techniques that relax the
balance constraint for vertex movements but globally maintain balance by
combining multiple local searches~\cite{kabapeE}.  
This was done by reducing this combination problem to finding negative cycles in
a graph. Here, we extend the neighborhood of the combination problem by
employing integer linear programming. This enables us to find even more complex
combinations and hence to further improve solutions. More precisely, our
approach is based on integer linear programs that solve the partitioning problem
to optimality. However, these programs typically do not scale to large inputs,
in particular because the graph partitioning problem has a very large amount of
symmetry -- given a partition of the graph, each permutation of the block IDs
gives a solution having the same objective and balance. Hence, we adapt the
integer linear program to improve a given input partition. We do so by defining
a much smaller graph, called \emph{model}, and solve the graph partitioning
problem on the model to optimality by the integer linear program. More
specifically, we select vertices close to the cut of the given input partition
for potential movement and contract all remaining vertices of a block into a
single vertex. A feasible partition of this model corresponds to a partition of
the input graph having the same balance and objective. Moreover, this model
enables us to use symmetry breaking, which allows us to scale to much larger
inputs. To make the approach even faster, we combine it with initial bounds on
the objective provided by the input partition, as well as providing the input
partition to the integer linear program solver. Overall, we arrive at a system
that is able to improve more than half of all entries in Walshaw's benchmark
when the number of blocks is high. We include our integer linear program-based
operation into the memetic graph partitioner KaBaPE~\cite{kabapeE}.
Additionally, we develop a crossover operation which is also based on our linear
program. This crossover operation contracts blocks of vertices, which all
partitions place in the same block. The extended memetic algorithm computes
graph partitions from scratch and manages to improve $17\%$ of the entries in
Walshaw's benchmark on the instances with $8,16,32$ or $64$ partitions and and a
maximum allowed imbalance of $3\%$ or $5\%$. In roughly half of all problems
considered, KaBaPE+ILP either reproduces or improves the previous best solution.

In Section~\ref{c:gp:s:preliminaries} we first introduce basic concepts. After
presenting some related work in Section~\ref{c:gp:s:related} we outline the integer
linear program as well as our novel local search algorithm in
Section~\ref{c:gp:s:algorithm}. Here, we start by explaining the technique we
use to find combinations of simple vertex movements. We then explain our
strategies to improve the running time of the solver and vertex selection strategies. In Section~\ref{c:gp:s:evo} we detail how the algorithm can be
used in the context of memetic graph partitioning. A summary of extensive
experiments done to evaluate the performance of our algorithms is presented in
Section~\ref{c:gp:s:experiments}. We conclude in Section~\ref{c:gp:s:conclusion}.

\section{Preliminaries}
\label{c:gp:s:preliminaries}

Let $G=(V=\{0,\ldots, n-1\},E)$ be an undirected graph.  We consider positive,
real-valued edge and vertex weight functions $c$ resp. $\omega$ and extend them
to sets, i.e., $c(E') \Is \sum_{x\in E'}c(x)$ and $\omega(V')\Is \sum_{x\in
V'} \omega(x)$. We use the same terminology to describe graphs as in
Part~\ref{p:mincut} of this dissertation. A vertex is a
\emph{boundary vertex} if it is incident to at least one vertex in a different
block.  We are looking for disjoint \emph{blocks} of vertices $V_1$,\ldots,$V_k$
that partition $V$; i.e., $V_1\cup\cdots\cup V_k=V$. The \emph{balancing
constraint} demands that each block has weight $\omega(V_i)\leq
(1+\epsilon)\lceil\frac{\omega(V)}{k}\rceil=:L_{\max}$ for some imbalance parameter
$\epsilon$. We call a block~$V_i$ \emph{overloaded} if its weight exceeds
$L_{\max}$. The objective of the problem is to minimize the total \emph{cut}
$c(E\cap\bigcup_{i<j}V_i\times V_j)$ subject to the balancing constraint.
%We define the \emph{gain} of a vertex as the maximum decrease in the cut value
%when moving it to a different block.

\section{Related Work}
\label{c:gp:s:related}
There has been a \emph{huge} amount of research on graph partitioning and we
refer the reader to the surveys given
in~\cite{GPOverviewBook,SPPGPOverviewPaper,schloegel2000gph,Walshaw07} for
most of the material. Here, we focus on issues closely related to our main
contributions. All general-purpose methods that are able to obtain good
partitions for large real-world graphs are based on the multi-level
principle. Well-known software packages based on this approach include
Jostle~\cite{Walshaw07},
KaHIP~\cite{kaffpa},~Metis~\cite{karypis1998fast}~and~Scotch~\cite{pellegrini1996scotch}.

Walshaw's well-known benchmark archive for the balanced graph partitioning
problem has been established in 2001~\cite{soper2004combined,wswebsite}. Overall
it contains 816 instances (34 graphs, 4 values of imbalance, and 6 values of
$k$). In this benchmark, the running time of the participating algorithms is not
measured or reported. Submitted partitions will be validated and added to the
archive if they improve on a particular result. This can either be an
improvement in the number of cut edges or, if they match the current best cut
size, an improvement in the weight of the largest block. Most entries in the
benchmark have as of Jan. $2021$ been obtained by
Galinier~\etal\cite{galinier2011efficient} (more precisely an implementation of
that approach by Schneider), Hein and Seitzer~\cite{hein2011beyond},
the Karlsruhe High-Quality Graph Partitioning (KaHIP) framework~\cite{kabapeE}
and the local search techniques described in this work.
More precisely, Galinier \etal\cite{galinier2011efficient} use a memetic
algorithm that is combined with tabu search to compute solutions and Hein and
Seitzer~\cite{hein2011beyond} solve the graph partitioning problem by providing
tight relaxations of a semi-definite program into a continuous problem.

Bisseling~\etal\cite{bisseling2006partitioning} use integer linear programming
to solve the graph partitioning problem in directed graphs. In contrast to our
work, they aim to minimize the number of vertices that have incoming edges from
a different block. Miyauchi~\etal\cite{miyauchi2015redundant} use integer linear
programming to solve the graph partitioning problem on fully connected
edge-weighted graphs.

The Karlsruhe High-Quality Graph Partitioning (\emph{KaHIP}) framework
implements many different algorithms, for example flow-based methods and local
searches, as well as several coarse-grained parallel and sequential
meta-heuristics. KaBaPE~\cite{kabapeE} is a coarse-grained parallel memetic
algorithm, \ie each processor has its own population (set of partitions) and a
copy of the graph. After initially creating the local population, each processor
performs multi-level combine and mutation operations on the local population.
This is combined with a meta-heuristic that combines local searches that
individually violate the balance constraint into a more global feasible
improvement. For more details, we refer the reader to \cite{kabapeE}.

\section{Local Search based on Integer Linear Programming}
\label{c:gp:s:algorithm}
We now explain our algorithm that combines integer linear programming and local
search. We start by explaining the integer linear program that can solve the
graph partitioning problem to optimality. However, out-of-the-box this program
does not scale to large inputs, in particular because the graph partitioning
problem has a very large amount of symmetry. Thus, we reduce the size of the
graph by first computing a partition using an existing heuristic and based on it
collapsing parts of the graph. Roughly speaking, we compute a small graph,
called \emph{model}, in which we only keep a small number of selected vertices
for potential movement and perform graph contractions on the remaining ones. A
partition of the model corresponds to a partition of the input network having
the same objective and balance. The computed model is then solved to optimality
using the integer linear program. As we will see this process enables us to use
symmetry breaking in the linear program, which in turn drastically speeds up
computation times.

\subsection{Integer Linear Program for the Graph Partitioning Problem}
\label{c:gp:ss:ilp}

We now introduce a generalization of an integer linear program formulation for
balanced bipartitioning~\cite{brillout2009multi} to the general graph
partitioning problem. First, we introduce binary decision variables for all
edges and vertices of the graph. More precisely, for each edge $e=\{u,v\}\in E$,
we introduce the variable $e_{uv} \in \{0,1\}$ which is one if $e$ is a cut edge
and zero otherwise. Moreover, for each $v\in V$ and block $k$, we introduce the
variable $x_{v,k} \in \{0,1\}$ which is one if $v$ is in block $k$ and zero
otherwise. Hence, we have a total of $|E| + k|V|$ variables. We use the
following constraints to ensure that the result is a valid $k$-partition:
\begin{ceqn}
  \begin{align}
  \label{eq:ilp-con-1a}
 \forall \{u,v\} \in E, \forall k&: e_{uv} \geq x_{u,k} - x_{v,k}\\
  \label{eq:ilp-con-1b}
 \forall \{u,v\} \in E, \forall k&: e_{uv} \geq x_{v,k} - x_{u,k}\\
  \label{eq:ilp-con-2}
  \forall k&: \sum_{v \in V} x_{v,k} \omega(v) \leq L_{\text{max}}\\
  \label{eq:ilp-con-3}
    \forall v \in V&: \sum_k x_{v,k} = 1
  \end{align}
\end{ceqn}

The first two constraints ensure that $e_{uv}$ is set to one if the vertices $u$
and $v$ are in different blocks. For an edge $\{u,v\}\in E$ and a block $k$, the
right-hand side in this equation is one if one of the vertices $u$ and $v$ is in
block $k$ and the other one is not. If both vertices are in the same block then
the right-hand side is zero for all values of $k$. Hence, the variable can
either be zero or one in this case. However, since the variable participates in
the objective function and the problem is a minimization problem, it will be
zero in an optimum solution.

The third constraint ensures that the balance constraint is satisfied for each
partition. And finally, the last constraint ensures that each vertex is assigned
to exactly one block. To sum up, our program has $2k|E|+k+|V|$ constraints and
$k \cdot (6 |E| + 2 |V|)$ non-zeros. Since we want to minimize the weight of cut
edges, the objective function of our program is~written~as:

\begin{ceqn}
\begin{equation}
  \label{eq:ilp-sum}
  \min \sum_{\{u,v\} \in E} e_{uv} \cdot c(\{u,v\})
\end{equation}
\end{ceqn}

\subsection{Local Search}
The graph partitioning problem has a large amount of symmetry -- each
permutation of the block IDs gives a solution with equal objective and balance.
Hence, the integer linear program described above will scan many branches that
contain essentially the same solutions so that the program does not scale to
large instances. Moreover, it is not immediately clear how to improve the
scalability of the program by using symmetry breaking or other techniques. For
the closely related problem of vertex partitioning,
Bisseling~\etal\cite{bisseling2006partitioning} report that using symmetry
breaking is highly important in order to get optimal solutions in reasonable
time.

Our goal in this section is to develop a local search algorithm using the
integer linear program above. Given a partition as input to be improved, our
\emph{main idea} is to contract vertices ``that are far away'' from the cut of
the partition. In other words, we want to keep vertices close to the cut and
contract all remaining vertices into one vertex for each block of the input
partition. This ensures that a partition of the contracted graph yields a
partition of the input graph with the same objective and balance. Hence, we
apply the integer linear program to the model and solve the partitioning problem
on it to optimality. Note, however, that due to the performed contractions this
does not imply an optimal solution on~the~input~graph.

We now outline the details of the algorithm. Our local algorithm has two inputs,
a graph~$G$ and a partition $V_1, \ldots, V_k$ of its vertices. For now assume
that we have a set of vertices $\mathcal{K} \subset V$ which we want to keep in
the coarse model, \ie a set of vertices which we do not want to contract. We
outline in Section~\ref{c:gp:ss:find} which strategies we have to select the vertices
$\mathcal{K}$. For the purpose of contraction we define $k$ sets $\mathcal{V}_i
:= V_i \backslash \mathcal{K}$. We obtain our coarse model by contracting each of
these vertex sets. The contraction of a vertex set $\mathcal{V}_i$ works by
iteratively contracting all pairs of vertices in that set until only one node is
left.
%the set of vertices is contracted into a single vertex $\mu_i$. The weight of
%$\mu_i$ is set to the sum of the weight of all vertices in the set that is
%contracted. There is an edge between two vertices $\mu_i$ and $v$ in the
%contracted graph if there is an edge between a vertex of the set and $v$ in the
%original graph $G$. The weight of an edge $(\mu_i,v)$ is set to the sum of the
%weight of edges that run between the vertices of~the~set~and~$v$. 
After all contractions have been performed the coarse model contains
$k+|\mathcal{K}|$ vertices, and potentially much fewer edges than the input
graph. Figure~\ref{c:gp:fig:abstractexample} gives an abstract example of our model.

There are two things that are important to see: first, due to the way we perform
contraction, the given partition of the input network yields a partition of our
coarse model that has the same objective and balance simply by putting $\mu_i$
into block $i$ and keeping the block of the input for the vertices in
$\mathcal{K}$. Moreover, if we compute a new partition of our coarse model, we
can build a partition in the original graph with the same properties by putting
the vertices~$\mathcal{V}_i$ into the block of their coarse representative
$\mu_i$ together with the vertices of $\mathcal{K}$ that are in this block.
Hence, we can solve the integer linear program on the coarse model to compute
a~partition~for~the~input~graph. After the solver terminates, \ie found an
optimum solution of our mode or has reached a predefined time limit
$\mathcal{T}$, we transfer the best solution to the original graph. Note that
the latter is possible since an integer linear program solver typically computes
intermediate solutions that may not be optimal.

\begin{figure}[t]
\centering
\includegraphics[width=.49\textwidth]{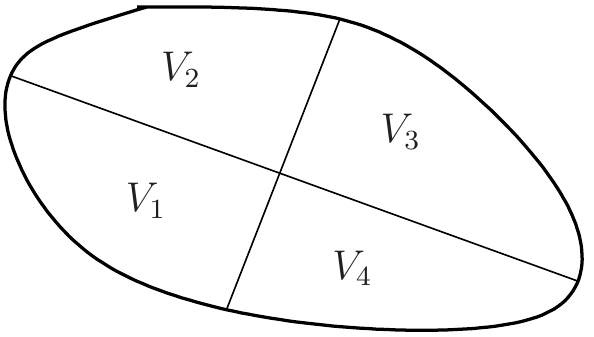}
\includegraphics[width=.49\textwidth]{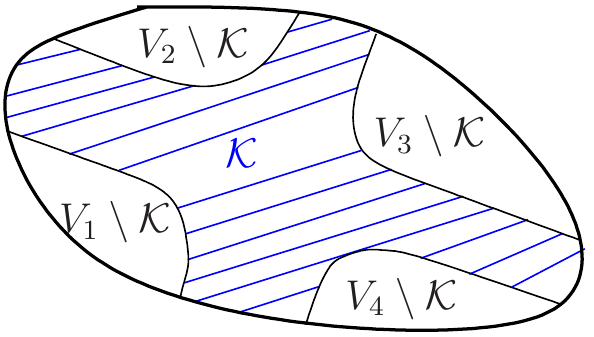}
\includegraphics[width=.49\textwidth]{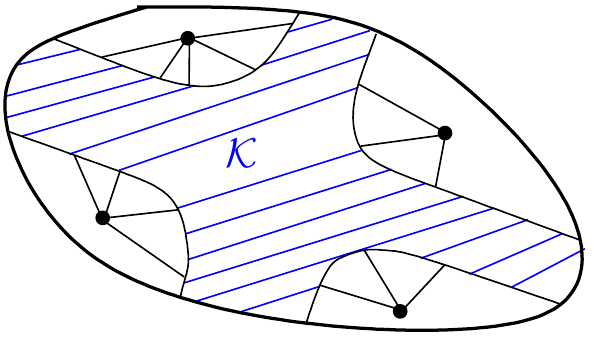}
\caption{Top left: a graph that is partitioned into four blocks, top right: the
set $\mathcal{K}$ close to the boundary that will stay in the model, bottom:
the model in which the sets $V_i \backslash \mathcal{K}$ have been~contracted.}
\label{c:gp:fig:abstractexample}
\end{figure}

\subsection{Optimizations}
\label{c:gp:ss:opti}
Independent of the vertices $\mathcal{K}$ that are selected to be kept in the
coarse model, the approach above allows us to define optimizations to solve our
integer linear program faster. We apply four strategies: (i) symmetry breaking,
(ii) providing a start solution to the solver, (iii) add the objective of the
input as a constraint as well as (iv) using the parallel solving facilities of
the underlying solver. We outline the first three strategies in greater detail:

\subsubsection{Symmetry Breaking}
If the set $\mathcal{K}$ is small, then the solver will find a solution much
faster. Typically, our algorithms selects the vertices $\mathcal{K}$ such that
$\omega(\mu_i)+\omega(\mu_j) > L_{\text{max}}$. In other words, no two contracted vertices
can be clustered in one block. We can use this to break symmetry in our integer
linear programming by adding constraints that fix the block of $\mu_i$ to block
$i$, \ie we set $x_{\mu_i,i} = 1$ and $x_{\mu_i,j} =0$ for  $i\neq j$. Moreover,
for those vertices we can remove the constraint which ensures that the vertex is
assigned to a single unique block---since we assigned those vertices to a block
using the new~additional~constraints. Note that we perform symmetry breaking
even if it is mathematically possible that multiple $\mu_i$ could be in the same
block.

\subsubsection{Providing a Start Solution to the Solver}
The integer linear program performs a significant amount of work in branches
which correspond to solutions that are worse than the input partition. Only very
few - if any - solutions  are better than the given partition. However, we
already know a fairly good partition (the given partition from the input) and
give this partition to the solver by setting according initial values for all
variables. This ensures that the integer linear program solver can omit many
branches and hence speeds up the time needed to solve the integer linear
program.

\subsubsection{Solution Quality as a Constraint}

Since we are only interested in improved partitions, we can add an additional
constraint that disallows solutions which have a worse objective than the input
partition. Indeed, the objective function of the linear program is linear, and
hence the additional constraint is also linear. Depending on the objective
value, this reduces the number of branches that the linear program solver needs
to look at. However, note that this comes at the cost of an additional
constraint that needs to be evaluated. Also note that if we provide a start
solution to the solver, the solver already knows a solution of said quality.
Thus, the solver is then able to prune worse solutions by itself.

\subsubsection{Row Generation}
Equation~\ref{eq:ilp-con-2} ensures that the balancing constraints in the graph
partitioning problem are adhered to. However, checking these constraints comes
with a computational cost. The idea of row generation is to initially omit these
constraints and lazily introduce balance constraints when a given solution
violates them. For each solution found by the ILP solver, we check whether any
block is heavier than $L_{\text{max}}$. If none is, the solution is valid. For
each block $V_k$ heavier than $L_{\text{max}}$ we introduce a new constraint
which makes sure that a subset of $V_k$ with a total weight of $>
L_{\text{max}}$ is not in block $k$ and thus reject the solution, as it violates
the new constraint.

In preliminary experiments this yields mixed results for $k=\{2,4\}$, but slowed
down the ILP for $k \geq 8$, as most solutions without balancing constraints are
too heavy in multiple blocks and thus the row generation introduces a large
amount of balancing constraints over the course of the solving process. We
therefore do not employ row generation in our experiments.

\subsection{Vertex Selection Strategies}
\label{c:gp:ss:find}
The algorithm above works for different vertex sets $\mathcal{K}$ that should be
kept in the coarse model. There is an obvious trade-off: on the one hand, the
set $\mathcal{K}$ should not be too large, otherwise the coarse model would be
large and hence the linear programming solver needs a large amount of time to
find a solution. On the other hand, the set should also not be too small, since
this restricts the amount of possible vertex movements, and hence the approach
is unlikely to find an improved solution. We now explain different strategies to
select the vertex set $\mathcal{K}$. In any case, while we add vertices to the
set $\mathcal{K}$, we compute the number of non-zeros in the corresponding ILP.
We stop to add vertices when the number of non-zeros in the corresponding ILP is
larger than a parameter $\mathcal{N}$.

\subsubsection{Vertices Close to Input Cut}
The intuition of the first strategy, \textttA{Boundary}, is that changes or
improvements of the partition will occur reasonable close to the input
partition. In this simple strategy our algorithm tries to use all \emph{boundary
vertices} as the set $\mathcal{K}$. In order to adhere to the constraint on the
number of non-zeros in the ILP, we add the vertices of the boundary uniformly at
random and stop if the number of non-zeros $\mathcal{N}$ is reached. If the
algorithm managed to add all boundary vertices whilst not exceeding the
specified number of non-zeros, we do the following extension: we perform a
breadth-first search that is initialized with a random permutation of the
boundary vertices. All additional vertices that are reached by the BFS are added
to $\mathcal{K}$. As soon as the number of non-zeros $\mathcal{N}$ is reached,
the algorithm~stops. 

\subsubsection{Start at Promising Vertices}
Especially for high values of $k$ the boundary contains many vertices. The
\textttA{Boundary} strategy quickly adds a lot of random vertices while ignoring
vertices that have high gain. Note that even in good partitions it is
possible that vertices with positive gain exist but cannot be moved due to the
balance constraint.

Hence, our second strategy, \textttA{Gain}$_\rho$, tries to fix this issue by
starting a breadth-first search initialized with only high gain vertices. More
precisely, we initialize the BFS with each vertex having gain $\geq \rho$
where~$\rho$~is~a~tuning~parameter. Our last strategy,
\textttA{TopVertices}$_\delta$, starts by sorting the boundary vertices by their
gain. We break ties uniformly at random. Vertices are then traversed in
decreasing order (highest gain vertices first) and for each start vertex $v$ our
algorithm adds all vertices with distance $\leq \delta$ to the model. The
algorithm stops as soon as the number of non-zeros exceeds $\mathcal{N}$.

Early gain-based local search heuristics for the $\epsilon$-balanced graph
partitioning problem searched for pairwise swaps with positive
gain~\cite{fiduccia1982lth,Kernighan70}. More recent algorithms generalized this
idea to also search for cycles or paths with positive total gain~\cite{kabapeE}.
An important advantage of our new approach is that we solve the combination
problem to optimality, \ie our algorithm finds the best combination of vertex
movements of the vertices in $\mathcal{K}$ with respect to the input partition of the
original graph. Therefore we can also find more complex optimizations that
cannot be reduced to positive gain cycles and paths. 

\section{Integer Linear Programming based Crossover}
\label{c:gp:s:evo}

A memetic algorithm is a population-based metaheuristic algorithm for an
optimization problem. The general outline of a memetic algorithm is such that we
first create a population of solutions and then use crossover and mutation
operations to generate new individuals out of existing ones. Generally, a
mutation operation has a single input partition and a cross operation has
multiple input partitions. If those new individuals are sufficiently fit, they
evict the lowest fitness individual from the population.
\emph{KaBaPE}~\cite{kabapeE} is a distributed parallel memetic algorithm for the
graph partitioning problem that provides multiple cross and mutation operations.
Based on the optimization techniques in this work, we now describe new mutation
and cross operations. These operations are added to the existing portfolio of
operations of KaBaPE.
        %The algorithm uses the \emph{KaFFPa} graph partitioner~\cite{kaffpa} to
        %create a population of graph partitions. When this population reaches a
        %pre-defined size, KaBaPE uses a multitude of cross and mutation
        %operations to find better solutions. A new individual that has a better
        %cut value than the worst cut in the population will evict the most
        %similar cut that has worse value. An individual that is better than the
        %best currently known partition is distributed using a randomized rumor
        %spreading protocol. KaBaPE employs crossover operations in which two
        %individuals are merged into a new individual whose quality is at least
        %as good as the best input individual. This is achieved by not
        %contracting any cut edges in any input partition in the coarsening
        %step. 

More precisely, the ILP-based local search algorithm described in
Section~\ref{c:gp:s:algorithm} can be used as a mutation operation directly. For
this, we take an individual from the population and run ILP-based local search
on the individual. If this results in an improved cut value, the new individual
is added to the population.

\subsection{ILP on Overlap Graph}

\begin{figure}[!t]
  \centering
  \includegraphics[width=.6\textwidth]{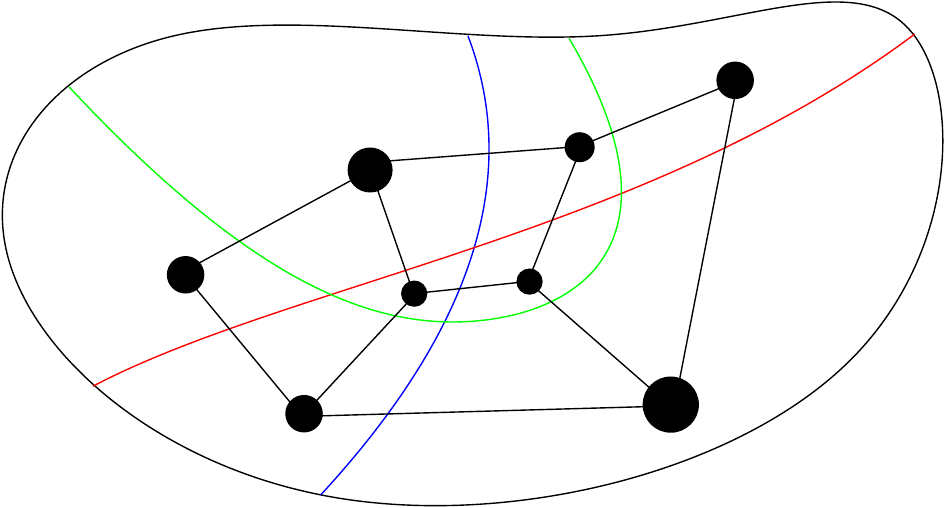}
  \caption{Example overlap graph for a graph with $3$ partitions and $2$ blocks
  each. \label{c:gp:fig:overlap}}
\end{figure}

Our new cross operation builds and solves an integer linear program from
multiple individuals. For this operation, we take $l$ individuals and build an
overlap graph $G_O = (V_O,E_O)$ out of $G$ by contracting regions that are in
the same block in every partition. An example overlap graph can be found in
Figure~\ref{c:gp:fig:overlap}. 
%In other words, if for two vertices $u,v \in V$, $u$ and $v$ are in the same
%block in every partition, they will be merged into the same vertex in $V_O$.
%Otherwise they will be in different vertices in $V_O$. 
In the literature this concept is also called overlap clustering
\cite{DBLP:journals/jmlr/StrehlG02}.

The weight of a vertex $v_O \in V_O$ is equal to the weight sum of all vertices
that are contracted into $v_O$. For vertices $u_O$ and $v_O$ in $V_O$,
$e_O=(u_O,v_O)$ exists if there is an edge from any vertex in $u_O$ to any
vertex in $v_O$. If there are multiple edges, the edge weight $c(e_O)$ is
equal to the sum of their weights. The fundamental idea behind that contraction
is such that if a region of vertices that is in the same block in all
partitions, most good partitions will have them in the same block. It is
therefore more valuable to model regions in which the partitions 'disagree' on
the vertex placement to make the ILP tractable.

In order to break symmetries, we select a subset of vertices $I$ where no two
vertices in $I$ are in the same block for any individual used to create the
overlap graph $G_O$. We choose the first vertex $v \in V_O$ in $I$ at random and
add vertices that share no block with any vertex in $I$ in any individual until
we can't find such a vertex in $V_O$ any more. In this way we break many
symmetries and only disallow solutions that aim to place vertices in the same
block that were placed in different blocks in every individual used to create
$G_O$.

We model the overlap graph as an ILP and initialize the block affiliations in
the ILP according to the partition that has the lowest cut value. When multiple
partitions have the same cut value, we choose any of them at random. For each
vertex $v \in V_O$, the block affiliation of $v$ is set to the block ID of the
vertices merged into $v$. Thus, we already have a solution that has value equal
to the best partition used for the overlap. If the ILP finds a better solution,
we insert the individual into the population.

As there is a very high variability in running times of the ILP operations, we
do not give a fixed ratio of ILP operation calls. Instead, we limit the total
running time fraction used in the ILP operations, so that they never take up
more than a third of the total running time. 

Each process in KaBaPE+ILP keeps two timers, one for each ILP operation, to
count the total running time used for all calls of the operation. If the sum of
them is at least $33\%$ of the total time used for the algorithm, we will not
choose them. If no time has been spent in the ILP operations yet, we choose one
of them with a probability of $75\%$. In between, we perform linear interpolation,
\ie the probability of performing an ILP operation is $(0.33 - \alpha) \cdot
0.75$, where $\alpha$ is the fraction of the total runtime due to ILP
operations. Thus, in graphs where the ILP operations are very fast, we use them
often. However, if the ILP operations are slow compared to other operations they
will not use up the majority of the running time. We also use linear
interpolation over the total running times to determine fairly which ILP
operation is used. As the solution quality improves more rapidly in the start of
the memetic algorithm, we gradually increase the time limit given to the ILP
solver in an ILP operation. The time limit is equal to $10\%$ of the current
total running time. These parameters were obtained from preliminary experiments,
however, in general, the algorithm is not very susceptible to those parameters
within reasonable limits.

We denote the extended memetic algorithm as KaBaPE+ILP.

\subsection{Post-processing}

We also employ a similar strategy to find the overlap graph for all high-quality
partitions. After the memetic algorithm is terminated, we collect all unique
partitions. We then build the overlap graph $G_{\mathcal{O}}$ on the best
$\kappa$ partitions, where $\kappa$ is a tuning parameter. In this graph,
vertices are merged if every high-quality partition in the population places
them in the same block. Thus, if the diversity of the population is large
enough, it is highly likely that the vertices will be placed in the same block
in any good partition. We run the memetic algorithm KaBaPE+ILP again, this time
on $G_{\mathcal{O}}$. As $G_{\mathcal{O}}$ has significantly fewer vertices and
edges than $G$, all operations perform faster and convergence is faster.
However, this also limits the solution space, as partitions that place merged
vertices into different blocks can not be found on $G_{\mathcal{O}}$. Thus, we
might be converging to a local optimum.

%\vfill \pagebreak
\section{Experiments}\label{c:gp:s:experiments}

\subsection{Experimental Setup and Methodology}

We implemented the algorithms described in the previous sections using
\CC\textttA{-17} and compiled all codes with full
optimization enabled (\textttA{-O3}). We use Gurobi as an ILP solver and use its
shared-memory parallel version. The experiments in
Sections~\ref{c:gp:ss:exp_opti},~\ref{c:gp:s:node_selection}~and~\ref{c:gp:ss:walshaw}
were conducted on a machine with two Haswell Xeon E5-2697 v3 processors, using
\textttA{g++-7.2.0} and Gurobi 7.5.2. The
machine has 28 cores at 2.6GHz as well as 64GB of main memory and runs the SUSE
Linux Enterprise Server (SLES) operating system. Unless otherwise mentioned, our
approach uses the shared-memory parallel variant of Gurobi using all 28 cores.
The experiments in Section~\ref{c:gp:ss:exp_evo} use \textttA{g++-8.3.0} and
Gurobi 8.1.1 and were conducted on a machine with two Intel Xeon E5-2643 v4 with
3.4GHz with 6 CPU cores each and 1.5 TB RAM in total. As the memetic algorithm
in this section has multiple parallel threads that perform cross and mutation
operations independent from each other, KaBaPE+ILP uses the sequential variant
of Gurobi. In general, we perform five repetitions per instance and report the
average running time as well as cut. Unless otherwise mentioned, we use a time
limit for the integer linear program. When the time limit is passed, the integer
linear program solver outputs the best solution that has currently been
discovered.
This solution does not have to be optimal. Note that we do not perform
experiments with Metis~\cite{karypis1998fast} and
Scotch~\cite{pellegrini1996scotch}, since previous papers, \eg
\cite{kaffpa,kabapeE}, have already shown that solution quality obtained is much
worse than results achieved in the Walshaw benchmark. When averaging over
multiple instances, we use the geometric mean in order to give every instance
the same influence on the \textit{final score}. We use performance plots to
compare the performance of different algorithm configurations on a per-instance
basis. For an explanation of these performance plots, we refer the reader to
Section~\ref{c:mc:ss:perf_plots}.

\begin{table}[t]
  \centering
  \caption{Basic properties of the benchmark instances. }
  \small
  \begin{tabular}{| l | r | r || l | r | r |}
    \hline
    Graph & $n$& $m$  & Graph & $n$ & $m$ \\
    \hline \hline
    \multicolumn{3}{|c||}{Walshaw Graphs (Set B)} &  \multicolumn{3}{c|}{Walshaw
    Graphs (Set B)}\\
    \hline

    add20       & \numprint{2395} & \numprint{7462} &wing        &
    \numprint{62032} & $\approx121$K \\
    data        & \numprint{2851} & \numprint{15093} &brack2      &
    \numprint{62631} & $\approx366$K\\
    3elt        & \numprint{4720} & \numprint{13722} &finan512    &
    \numprint{74752} & $\approx261$K\\
    uk          & \numprint{4824} & \numprint{6837} &fe\_tooth   &
    \numprint{78136} & $\approx452$K\\
    add32       & \numprint{4960} & \numprint{9462} &fe\_rotor   &
    \numprint{99617} & $\approx662$K \\
    bcsstk33    & \numprint{8738} & $\approx291$K & \numprint{598}a        &
    \numprint{110971} & $\approx741$K \\
    whitaker3   & \numprint{9800} & \numprint{28989} &fe\_ocean   &
    \numprint{143437} & $\approx409$K\\
    crack       & \numprint{10240}& \numprint{30380} & \numprint{144}         &
    \numprint{144649} & $\approx1.1$M\\
    wing\_nodal & \numprint{10937}& \numprint{75488} &wave        &
    \numprint{156317} & $\approx1.1$M\\
    fe\_4elt2   & \numprint{11143}& \numprint{32818} &m14b        &
    \numprint{214765} & $\approx1.7$M\\
    vibrobox    & \numprint{12328}& $\approx165$K &auto        &
    \numprint{448695} & $\approx3.3$M \\
    \cline{4-6}

    bcsstk29    & \numprint{13992}& $\approx302$K &\multicolumn{3}{c|}{}\\
    4elt        & \numprint{15606}& \numprint{45878} &
    \multicolumn{3}{c|}{Parameter Tuning (Set A)}\\
    \cline{4-6}
    fe\_sphere  & \numprint{16386}& \numprint{49152} & delaunay\_n15 &
    \numprint{32768} & \numprint{98274}\\
    cti         & \numprint{16840}& \numprint{48232} & rgg\_15 &
    \numprint{32768} & $\approx160$K\\
    memplus     & \numprint{17758}& \numprint{54196} & \numprint{2}cubes\_sphere
    & \numprint{101492} & $\approx772$K\\
    cs4         & \numprint{22499}& \numprint{43858} & cfd2            &
    \numprint{123440} & $\approx1.5$M \\
    bcsstk30    & \numprint{28924}& $\approx 1.0$M & boneS01         &
    \numprint{127224} & $\approx3.3$M\\
    bcsstk31    & \numprint{35588}& $\approx572$K & Dubcova3        &
    \numprint{146689} & $\approx1.7$M\\
    fe\_pwt     & \numprint{36519}& $\approx144$K & G2\_circuit   &
    \numprint{150102} & $\approx288$K \\
    bcsstk32    & \numprint{44609}& $\approx985$K & thermal2  &
    \numprint{1227087} & $\approx3.7$M\\
    fe\_body    & \numprint{45087}& $\approx163$K & as365     &
    \numprint{3799275} & $\approx11.4$M\\
    t60k        & \numprint{60005}& \numprint{89440} & adaptive  &
    \numprint{6815744} & $\approx13.6$M\\

  \hline
  \end{tabular}
  \label{c:gp:tab:test_instances_walshaw}
\end{table}

\paragraph{Instances.}
We perform experiments on two sets of instances. Set $A$ is used to determine
the performance of the integer linear programming optimizations and to tune the
algorithm. We obtained these instances from the Florida Sparse Matrix collection
\cite{davis2011university} and the~10th~DIMACS Implementation
Challenge~\cite{bader2013graph} to test our algorithm. Set $B$
are all graphs from Chris Walshaw's graph partitioning benchmark
archive~\cite{soper2004combined,wswebsite}. This archive is a collection of
instances from finite-element applications, VLSI design and is one of the
default benchmarking sets~for~graph~partitioning. 

Table~\ref{c:gp:tab:test_instances_walshaw} gives basic properties of the graphs
from both benchmark sets. We ran the unoptimized integer linear program that
solves the graph partitioning problem to optimality from
Section~\ref{c:gp:ss:ilp} on the five smallest instances from the Walshaw
benchmark set. With a time limit of $30$ minutes, the solver has only been able
to compute a solution for the graphs uk and add32 with $k=2$. For higher values
of $k$ the solver was unable to find any solution in the time limit. Even giving
a starting solution does not increase the number of ILPs solved. Hence, we omit
further experiments in which we run an ILP solver on the full graph.

\subsection{Impact of Optimizations}
\label{c:gp:ss:exp_opti}
We now evaluate the impact of the optimization strategies for the ILP that we
presented in Section~\ref{c:gp:ss:opti}.  
In this section, we use the variant of our local search algorithm in which
$\mathcal{K}$ is obtained by starting depth-one breadth-first search at the $25$
highest gain vertices, and set the limit on the non-zeros in the ILP to
$\mathcal{N}=\infty$. However, due to preliminary experiments we expect the
results in terms of speedup to be similar for different vertex selection
strategies. To evaluate the ILP performance, we run KaFFPa using the strong
preconfiguration on each of the graphs from set $A$ using $\epsilon=0$ and
$k\in\{2,4,8,16,32,64\}$ and then use the computed partition as input to each
ILP (with the different optimizations). As the optimizations do not change the
objective value achieved in the ILP and we only look at ILP formulations solved
to optimality in this subsection, we only report running times of our different
approaches. We set the time limit of the ILP solver to 30 minutes.

We use five variants of our algorithm in this experiment: \textttA{Basic} does
not contain any optimizations; \textttA{BasicSym} enables symmetry breaking;
\textttA{BasicSymSSol} additionally gives the input partition to the ILP solver.
The two variants \textttA{BSSSConst=} and \textttA{BSSSConst$<$} are the same as
\textttA{BasicSymSSol} with additional constraints to the solution quality:
\textttA{BSSSConst=} has the additional constraint that the objective has to be
smaller or equal to the start solution,  \textttA{BSSSConst$<$} has the
constraint that the objective value of a solution must be better than the
objective value of the start solution.
Figure~\ref{c:gp:fig:plotresultsilpvariants} summarises the results.

\begin{figure}[t]
  \centering
  \includegraphics[width=.9\textwidth]{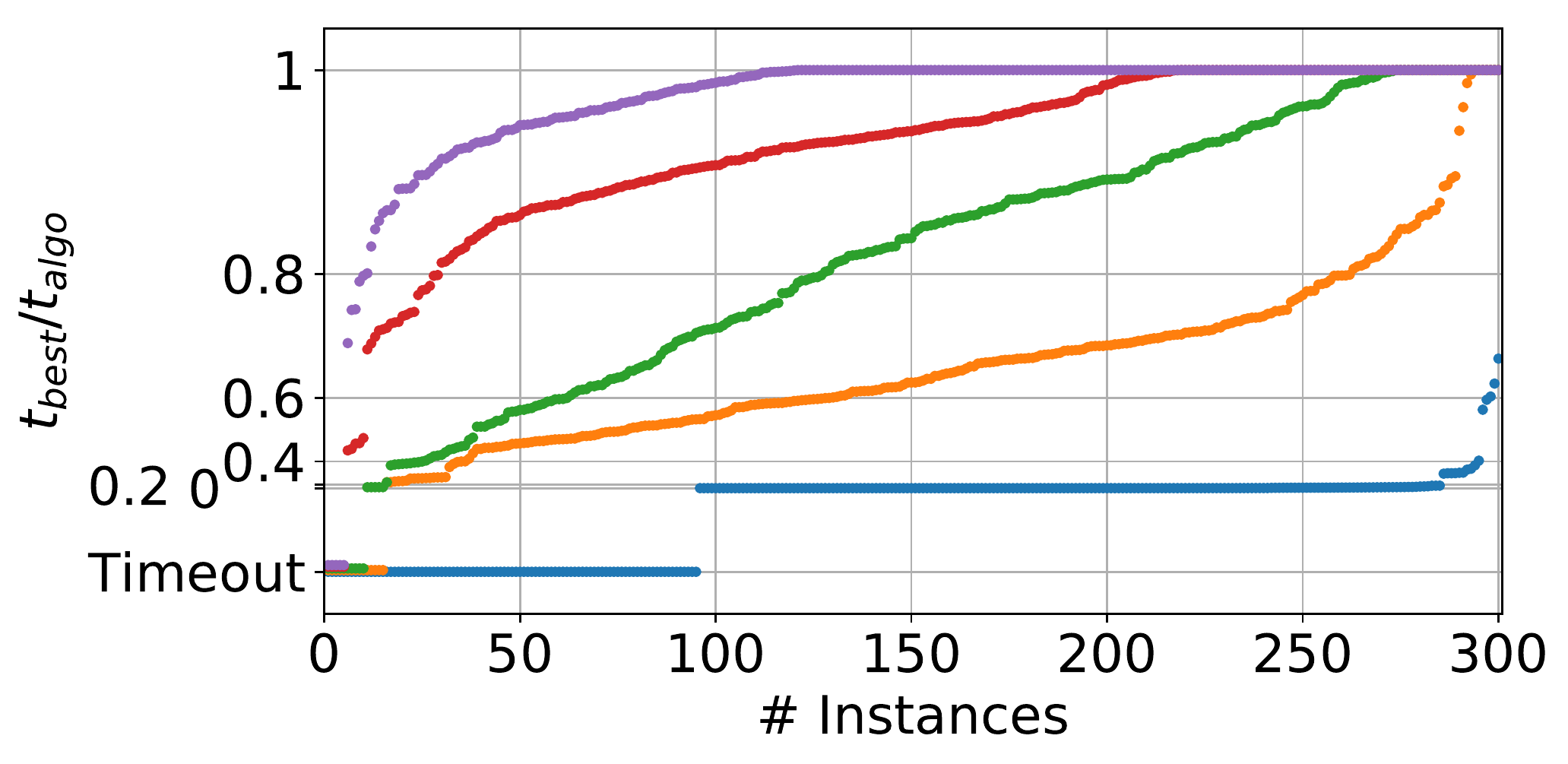}
  \centering
  \includegraphics[width=\textwidth]{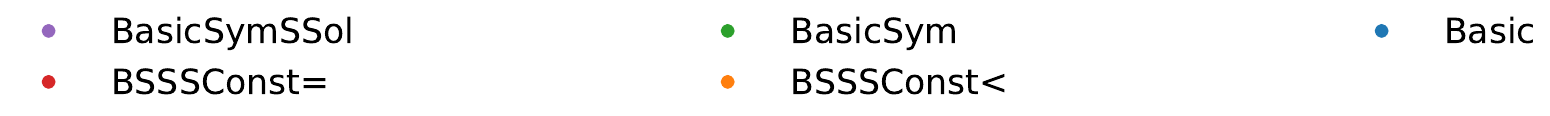}
  \caption{Performance plot for five variants of our algorithm: \textttA{Basic}
  does not contain any optimizations; \textttA{BasicSym} enables symmetry
  breaking; \textttA{BasicSymSSol} additionally gives the input partition to the
  ILP solver. The two variants \textttA{BSSSConst=} and \textttA{BSSSConst$<$} are
  the same as \textttA{BasicSymSSol} with additional constraints:
  \textttA{BSSSConst=} has the additional constraint that the objective has to be
  smaller or equal to the start solution,  \textttA{BSSSConst$<$} has the
  constraint that the solution must be better than~the~start~solution.}
  \label{c:gp:fig:plotresultsilpvariants}
\end{figure}

In our experiments, which are detailed in
Figure~\ref{c:gp:fig:plotresultsilpvariants}, the basic configuration reaches
the time limit in 95 out of the 300 runs. Overall, enabling symmetry breaking
drastically speeds up computations.  On all of the instances which the
\textttA{Basic} configuration could solve within the time limit, each other
configuration is faster than the \textttA{Basic} configuration. Symmetry breaking
speeds up computations by a factor of 41 in the geometric mean on those
instances. The largest obtained speedup on those instances was a factor of 5663
on the graph \textttA{adaptive} for $k=32$. The configuration solves all but the
two instances (\textttA{boneS01}, $k=32$) and (\textttA{Dubcova3}, $k=16$) within
the time limit. Providing the start solution (\textttA{BasicSymSSol}) gives an
additional speedup of 22\% on average. Over the \textttA{Basic} configuration,
the average speedup is 50 with the largest speedup being 6495 and the smallest
speedup being 1.47. This configuration can solve all instances within the time
limit except the instance \textttA{boneS01} for $k=32$. Providing the objective
function as a constraint (or strictly smaller constraint) does not further
reduce the running time of the solver. Instead, the additional constraints even
increase the running time. We attribute this to the fact that the solver has to
do additional work to~evaluate~the~constraint. We conclude that
\textttA{BasicSymSSol} is the fastest configuration of the ILP. Hence, we use
this configuration in all the following experiments. Moreover, from
Figure~\ref{c:gp:fig:plotresultsnonzero} we can see that this configuration can
solve most of the instances within the time limit if the number of non-zeros in
the ILP is below $10^6$. Hence, we set the parameter $\mathcal{N}$ to~$10^6$ in
the following section.

\begin{figure}[t!]
  \centering
\includegraphics[width=.9\textwidth]{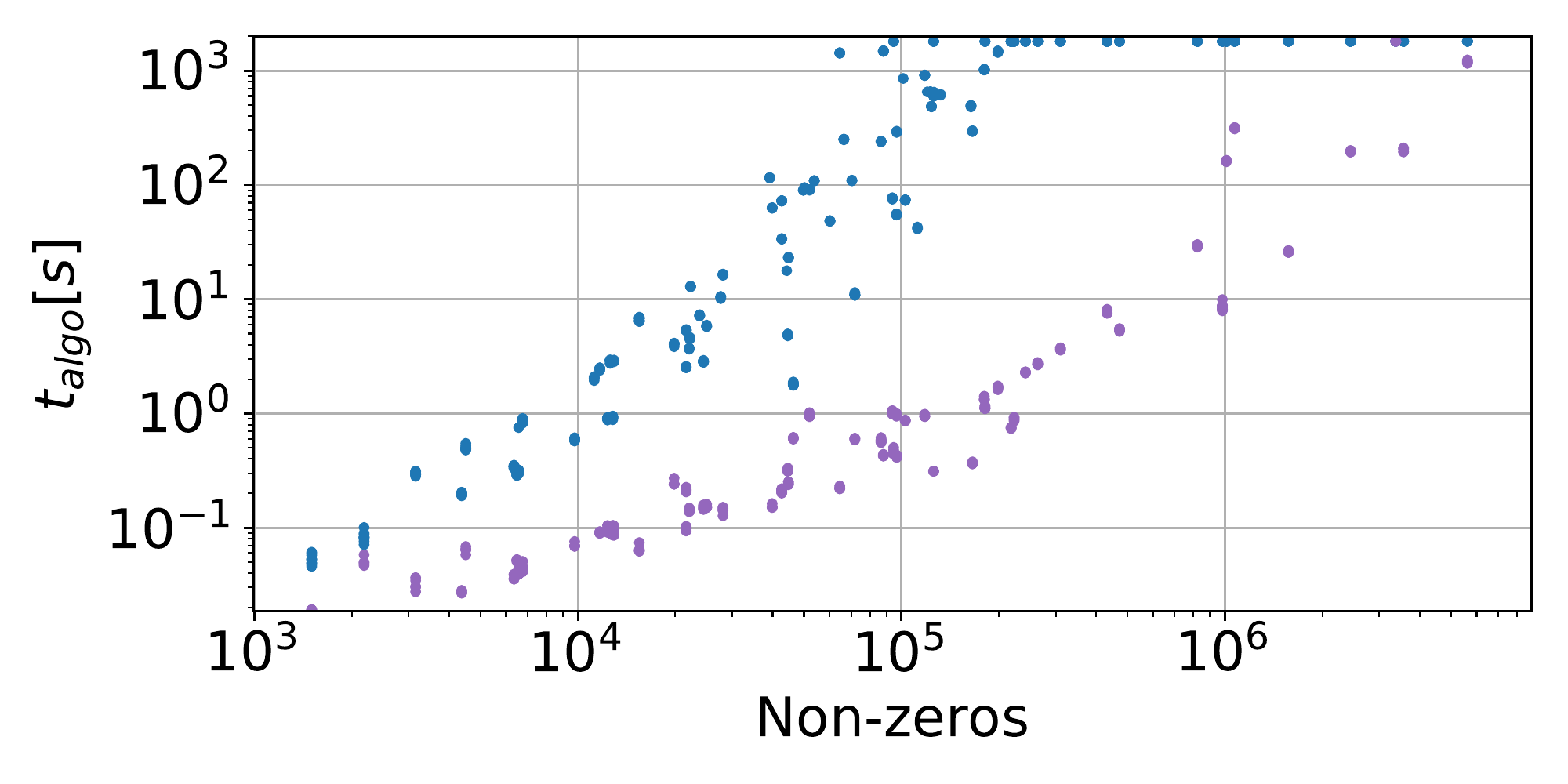}
\centering
\includegraphics[width=\textwidth]{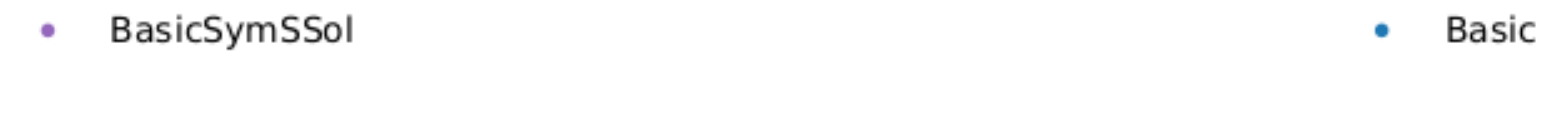}

\caption{Performance of the slowest (\textttA{Basic}) and fastest ILPs
(\textttA{BasicSymSSol}) depending on the number of non-zeros in the ILP.}
\label{c:gp:fig:plotresultsnonzero}
\end{figure}

\subsection{Vertex Selection Rules}
\label{c:gp:s:node_selection}

\begin{figure}[t!]
  \centering
  \includegraphics[width=.9\textwidth]{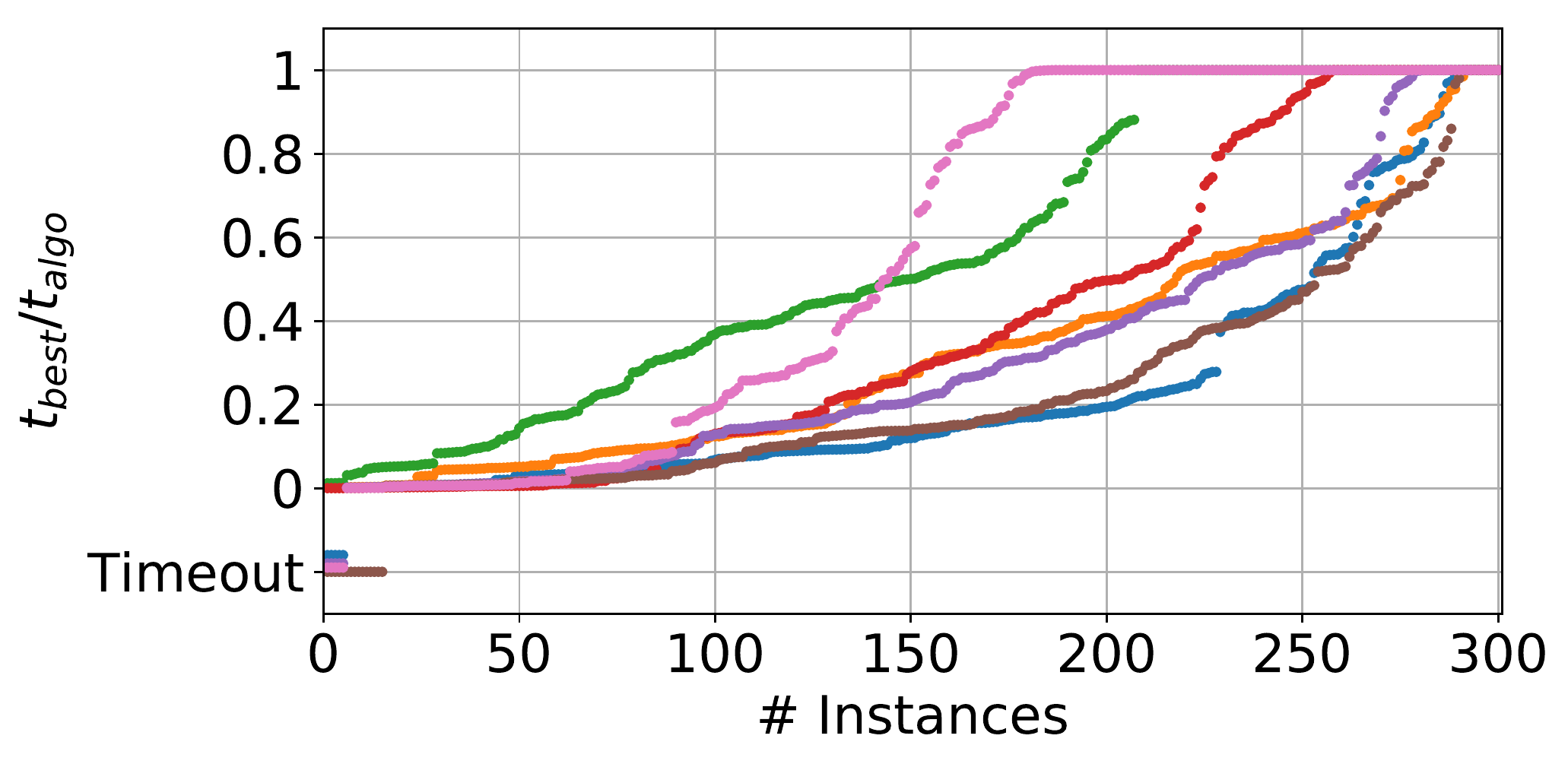}
  \centering
  \includegraphics[width=\textwidth]{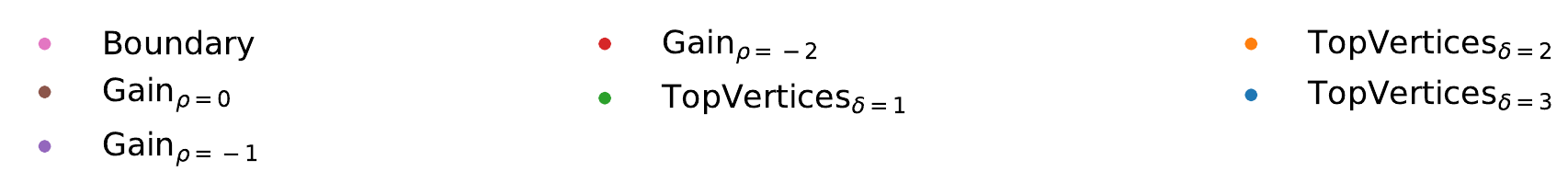} 
  \caption{Performance plot for all vertex selection strategies.}
  \label{c:gp:fig:plotperf}
\end{figure}

\begin{figure}[b!]
  \centering
\includegraphics[width=.9\textwidth]{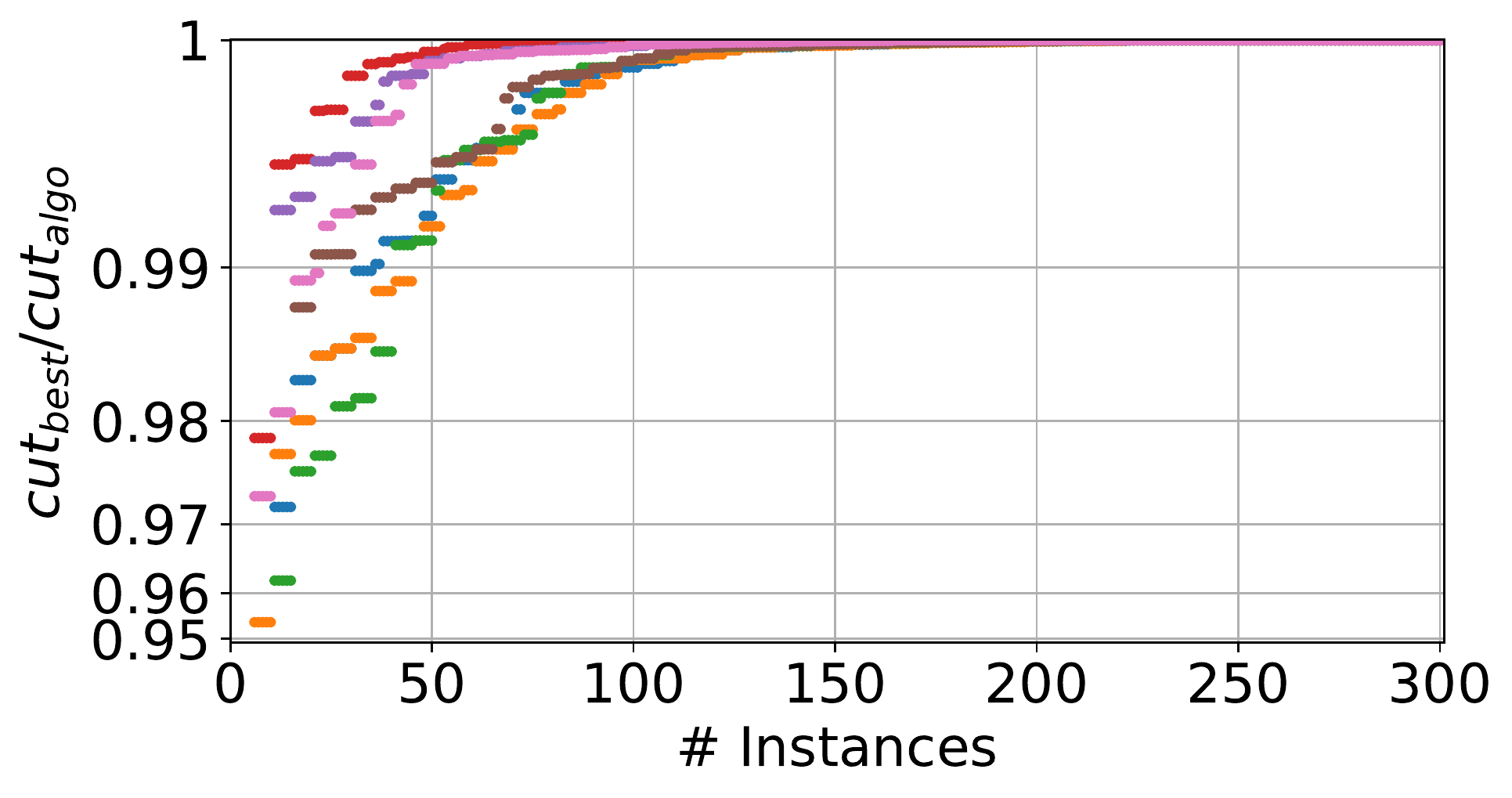}

  \includegraphics[width=\textwidth]{img/legend2.pdf}
\caption{Cut value of vertex selection strategies in comparison to the best
 result given by any strategy.}
\label{c:gp:fig:plotresults}
\end{figure}

We now evaluate the vertex selection strategies to find the set of vertices
$\mathcal{K}$ that model the ILP. We look at all strategies described in
Section~\ref{c:gp:ss:find}, \ie \textttA{Boundary}, \texttt{Gain$_{\rho}$} with
the parameter $\rho \in \{-2,-1,0\}$ as well as \textttA{TopVertices}$_\delta$
for $\delta \in \{1,2,3\}$. To evaluate the different selection strategies, we
use the best of five runs of KaFFPa-strong on each of the graphs from set $A$
using imbalance $\epsilon=0$ and number of partitions $k\in\{2,4,8,16,32,64\}$
and then use the computed partition as input to the ILP
(with~different~sets~$\mathcal{K}$). Table~\ref{table:variants} summarizes the
results of the experiment, \ie the number of cases in which our algorithm was
able to improve the result, the average running time in seconds for these
selection strategies as well as the number of cases in which the strategy
computed the best result (the partition having the lowest cut). We set the time
limit to $2$ days to be able to finish almost all runs without running into
timeout. For the average running time we exclude all graphs in which at least
one algorithm did not finish in $2$ days (\textttA{rgg\_15} $k=16$,
\textttA{delaunay\_n15} $k=4$, \textttA{G2\_circuit} $k=4,8$). If multiple runs
share the best result, they are all counted. However, when no algorithm improves
the input partition on a graph, we do not count them.

 \begin{table}[t!]
    \centering
    \small

    \caption{From top to bottom: Number of improvements found by different
    vertex selection rules relative to the total number of instances, average
    running time of the strategy on the subset of instances (graph, $k$) in
    which all strategies finished within the time limit, and the relative number
    of instances in which the strategy computed the lowest cut. Best values are
    highlighted in bold. \label{table:variants}}
    \resizebox{\textwidth}{!}{
      \begin{tabular}{|l||rrr|rrr|r|}

        \hline
          &\textttA{Gain} &&&\textttA{TopVertices} &&& \textttA{Boundary} \\
      $k$& $\rho = 0 $& $\rho = -1 $ & $\rho=-2$ &$\delta = 1$ & $\delta=2$&
      $\delta=3$& \\\hline \hline
&      \multicolumn{7}{c|}{Relative Number of Improvements}  \\\hline
       2 & \textbf{\numprint{70}\%} & \textbf{\numprint{70}\%} &
       \textbf{\numprint{70}\%} & \numprint{50}\% & \textbf{\numprint{70}\%} &
       \textbf{\numprint{70}\%} & \textbf{\numprint{70}\%} \\
       4 & \numprint{50}\% & \numprint{60}\% & \textbf{\numprint{80}\%} &
       \numprint{70}\% & \numprint{70}\% & \numprint{70}\% &
       \textbf{\numprint{80}\%} \\
       8 & \numprint{50}\% & \numprint{60}\% & \textbf{\numprint{78}\%} &
       \numprint{60}\% & \numprint{60}\% & \numprint{60}\% & \numprint{48}\% \\
      16 & \numprint{30}\% & \numprint{50}\% & \textbf{\numprint{70}\%} &
      \numprint{40}\% & \numprint{30}\% & \numprint{30}\% & \numprint{40}\% \\
      32 & \textbf{\numprint{60}\%} & \textbf{\numprint{60}\%} & \numprint{46}\%
      & \numprint{50}\% & \numprint{50}\% & \numprint{20}\% & \numprint{20}\% \\
      64 & \textbf{\numprint{70}\%} & \textbf{\numprint{70}\%} & \numprint{50}\%
      & \numprint{30}\% & \numprint{20}\% & \numprint{20}\% & \numprint{0}\% \\
  \hline\hline
&      \multicolumn{7}{c|}{Average Running Time} \\
\hline
                2 & \numprint{189}.943s & \numprint{292}.573s &
                \numprint{357}.145s &  \textbf{\numprint{34}.045s} &
                \numprint{61}.152s & \numprint{92}.452s & \numprint{684}.198s \\
       4 & \numprint{996}.934s & \numprint{628}.950s & \numprint{428}.353s &
       \textbf{\numprint{87}.357s} & \numprint{255}.223s & \numprint{558}.578s &
       \numprint{1467}.595s \\
       8 & \numprint{552}.183s & \numprint{244}.470s & \numprint{244}.046s &
       \numprint{105}.737s & \numprint{167}.164s & \numprint{340}.900s &
       \textbf{\numprint{96}.763s} \\
      16 & \numprint{118}.532s & \numprint{52}.547s & \numprint{90}.363s &
      \numprint{53}.385s & \numprint{141}.814s & \numprint{243}.957s &
      \textbf{\numprint{34}.790s} \\
               32 & \numprint{40}.300s & \numprint{24}.607s & \numprint{94}.146s
               & \numprint{27}.156s & \numprint{80}.252s & \numprint{116}.023s &
               \textbf{\numprint{7}.596s} \\
               64 & \numprint{15}.866s & \numprint{21}.908s & \numprint{24}.253s
               & \numprint{14}.627s & \numprint{30}.558s & \numprint{44}.813s &
               \textbf{\numprint{4}.187s} \\
      \hline
      \hline
&      \multicolumn{7}{c|}{Relative Number Best Algorithm} \\
\hline

   2 & \numprint{20}\% & \textbf{\numprint{60}\%} & \numprint{50}\% &
   \numprint{10}\% & \numprint{10}\% & \numprint{0}\% & \textbf{\numprint{60}\%}
   \\
                4 & \numprint{10}\% & \numprint{0}\% & \textbf{\numprint{50}\%}
                & \numprint{10}\% & \numprint{0}\% & \numprint{0}\% &
                \numprint{30}\% \\
                8 & \numprint{0}\% & \numprint{20}\% & \textbf{\numprint{30}\%}
                & \numprint{10}\% & \numprint{10}\% & \numprint{10}\% &
                \numprint{26}\% \\
               16 & \numprint{0}\% & \numprint{10}\% & \textbf{\numprint{54}\%}
               & \numprint{10}\% & \numprint{0}\% & \numprint{10}\% &
               \numprint{20}\% \\
               32 & \numprint{0}\% & \numprint{8}\% & \textbf{\numprint{38}\%} &
               \numprint{0}\% & \numprint{0}\% & \numprint{0}\% & \numprint{4}\%
               \\
      64 & \numprint{0}\% & \numprint{16}\% &          \textbf{\numprint{36}\%}
      & \numprint{0}\% & \numprint{0}\% & \numprint{0}\% & \numprint{0}\% \\

        \hline
    \end{tabular}
    }
  \end{table}

Looking at the number of improvements, the \textttA{Boundary} strategy is able to
improve the input for small values of $k$, but with increasing number of blocks
$k$ improvements decrease to no improvement in all runs with $k=64$. Because of
the limit on the number of non-zeros, the ILP contains only random boundary
vertices for large values of $k$ in this case. Hence, there are not sufficiently
many high gain vertices in the model and fewer improvements for large values of
$k$ are expected. For small values of $k \in \{2,4\}$, the \textttA{Boundary}
strategy can improve as many as the \texttt{Gain$_{\rho=-2}$} strategy but the
average running times are higher.

For $k = \{2,4,8,16\}$, the strategy \texttt{Gain$_{\rho=-2}$} has the highest
number of improvements, for $k=\{32,64\}$ it is surpassed by the strategy
\texttt{Gain$_{\rho=-1}$}. However, the strategy \texttt{Gain}$_{\rho=-2}$ finds
the best cuts in most cases among all tested strategies. Due to the way these
strategies are designed, they are able to put a lot of high gain vertices into
the model as well as vertices that can be used to balance vertex movements. The
\textttA{TopVertices} strategies are overall also able to find a large number of
improvements. However, the improvements are typically smaller than for the
\textttA{Gain} strategies. This is due to the fact that the \textttA{TopVertices}
strategies grow BFS balls with a predefined depth around high gain vertices
first, and later on are not able to include vertices that could be used to
balance their movement. Hence, there are less potential vertex movements that
could yield an improvement.

For almost all strategies, we can see that the average running time decreases as
the number of blocks $k$ increases. This happens because we limit the number of
non-zeros $\mathcal{N}$ in our ILP. As the number of non-zeros grows linearly
with the underlying model size, the models are far smaller for higher values of
$k$. Using symmetry breaking, we already fixed the block of the $k$ vertices
$\mu_i$ which represent the vertices not part of $\mathcal{K}$. Thus the ILP
solver can quickly prune branches which would place vertices connected heavily
to one of these vertices in~a~different~block. Additionally, our data indicates
that a large number of small areas in our model results faster in solve times
than when the model contains few large areas. The performance plot in
Figure~\ref{c:gp:fig:plotperf} shows that the strategies \textttA{Boundary},
\textttA{TopVertices}$_{\delta=1}$ and \texttt{Gain$_{\rho=-2}$} have lower
running times than other strategies. These strategies all select a large number
of vertices to initialize the breadth-first search. Therefore they output a
vertex set $\mathcal{K}$ that is the union of many small areas around these
vertices. Variants that initialize the breadth-first search with fewer vertices
have fewer areas, however each area is larger. Figure~\ref{c:gp:fig:plotresults}
shows that for almost all instances the variants \texttt{Gain$_{\rho=-1}$} and
\texttt{Gain$_{\rho=-2}$} give very good solutions, even if they are not the
best variant on that particular instance.

\subsection{Walshaw Benchmark}
\label{c:gp:ss:walshaw}
  \begin{table}[b!]
    \centering
            \caption{Relative number of improved instances by performing an
            ILP-based local search in the Walshaw Benchmark starting from
            current entries reported in the Walshaw benchmark.
            \label{c:gp:tab:existingimprovment}}
     \begin{tabular}{|l|r|r|r|r|r|r||r|}
      \hline
      $\epsilon \backslash k$ & \numprint{2} & \numprint{4} & \numprint{8} &
      \numprint{16} & \numprint{32} & \numprint{64} & overall \\ \hline \hline
      $0\%$ & $6\%$ & $18\%$ & $26\%$ & $50\%$ & $62\%$ & $68\%$ & $38\%$ \\
      $1\%$ & $12\%$ & $9\%$ & $24\%$ & $26\%$ & $47\%$ & $59\%$ & $29\%$ \\
      $3\%$ & $6\%$ & $6\%$ & $12\%$ & $29\%$ & $47\%$ & $71\%$ & $28\%$ \\
      $5\%$ & $6\%$ & $18\%$ & $15\%$ & $29\%$ & $53\%$ & $76\%$ & $33\%$ \\
      \hline
     \end{tabular}
    \end{table}

In this section, we present the results when running our best configuration on
all graphs from Walshaw's benchmark archive. Note that the rules of the
benchmark imply that running time is not an issue, but algorithms should achieve
the smallest possible cut value while satisfying the balance constraint. We run
our algorithm in the following setting: We take existing partitions from the
archive and use those as input to our algorithm. As indicated by the experiments
in Section~\ref{c:gp:s:node_selection}, the vertex selection strategies
\texttt{Gain$_{\rho \in \{-1,-2\}}$} perform best for different values of $k$.
Thus we use the variant \texttt{Gain$_{\rho=-2}$} for $k \leq 16$ and both
\texttt{Gain$_{\rho=-2}$} and \texttt{Gain$_{\rho=-1}$} otherwise in this
section. We repeat the experiment once for each instance (graph, $k$) and run
our algorithm for $k = \{2,4,8,16,32,64\}$ and $\epsilon \in \{0, 1\%, 3\%,
5\%\}$. For larger values of $k \in \{32,64\}$, we strengthen our strategy and
use  $\mathcal{N} = 5 \cdot 10^6$ as a bound for the number of non-zeros. We set
the time limit to two hours. Table~\ref{c:gp:tab:existingimprovment} summarizes
the results. Detailed per-instance results are given in
Section~\ref{c:gp:s:additional}.

\begin{figure}
  \centering
  \begin{tabular}{c c}
    \includegraphics[width=.45\textwidth]{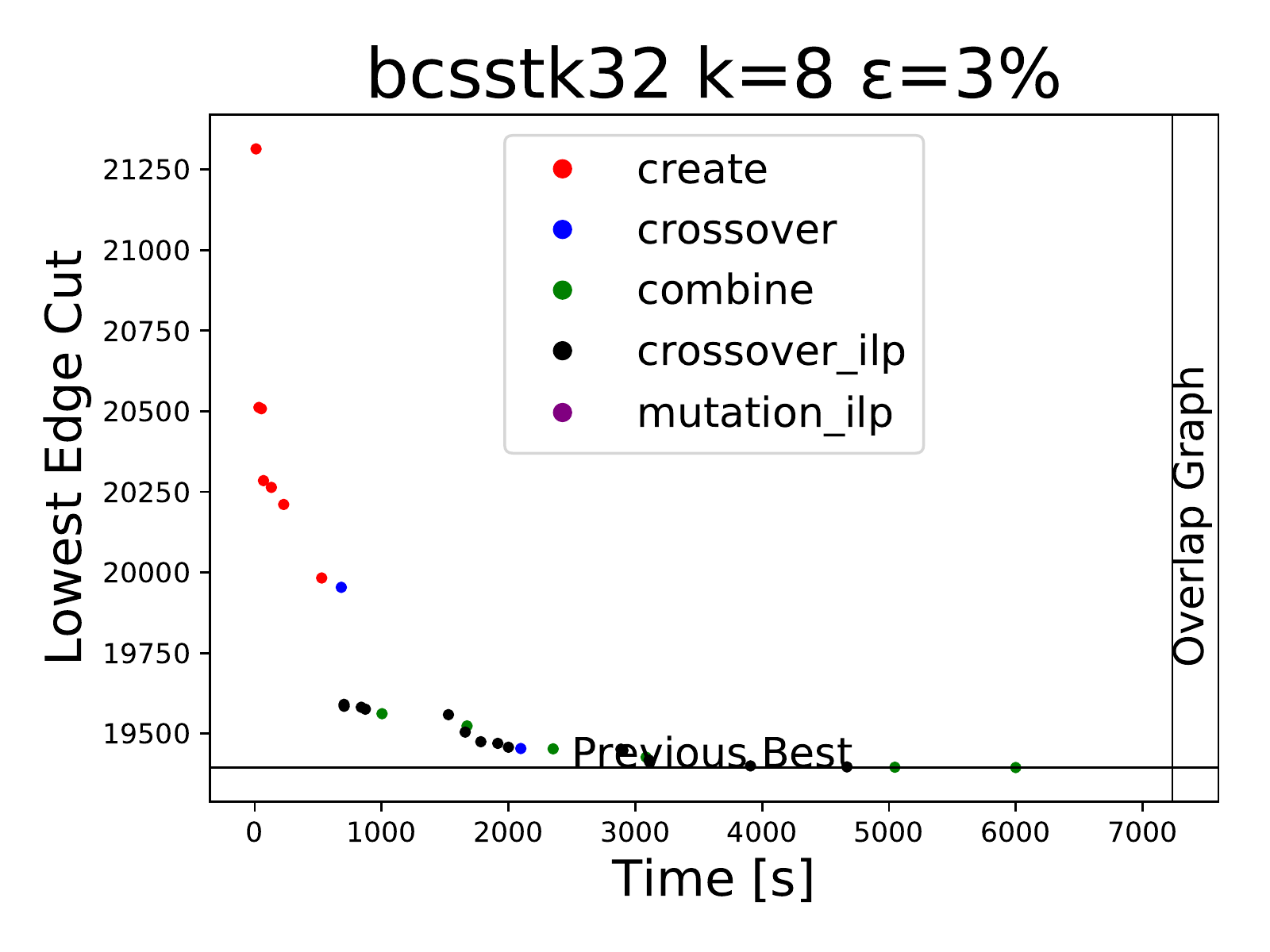} &
    \includegraphics[width=.45\textwidth]{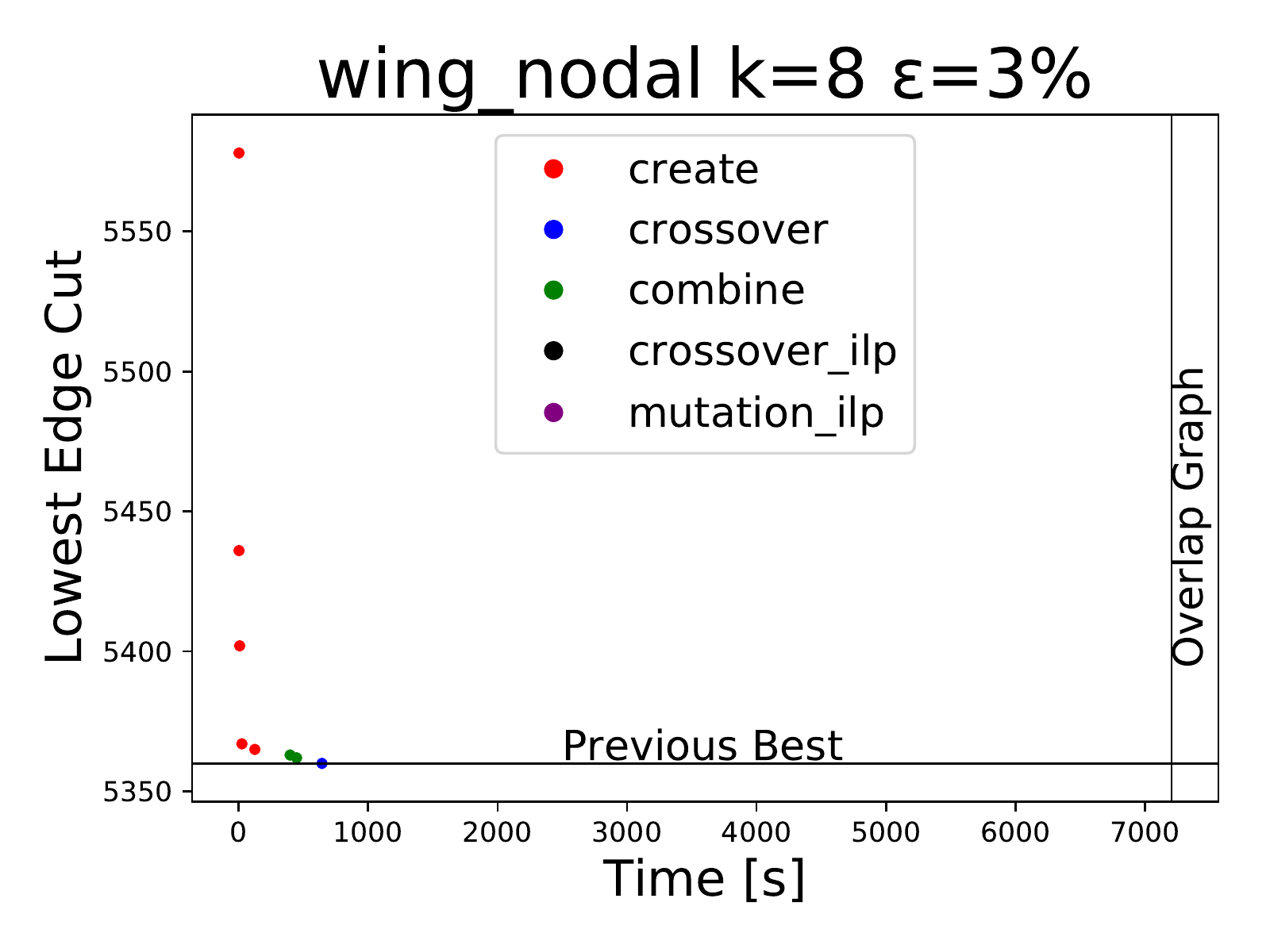} \\

    \includegraphics[width=.45\textwidth]{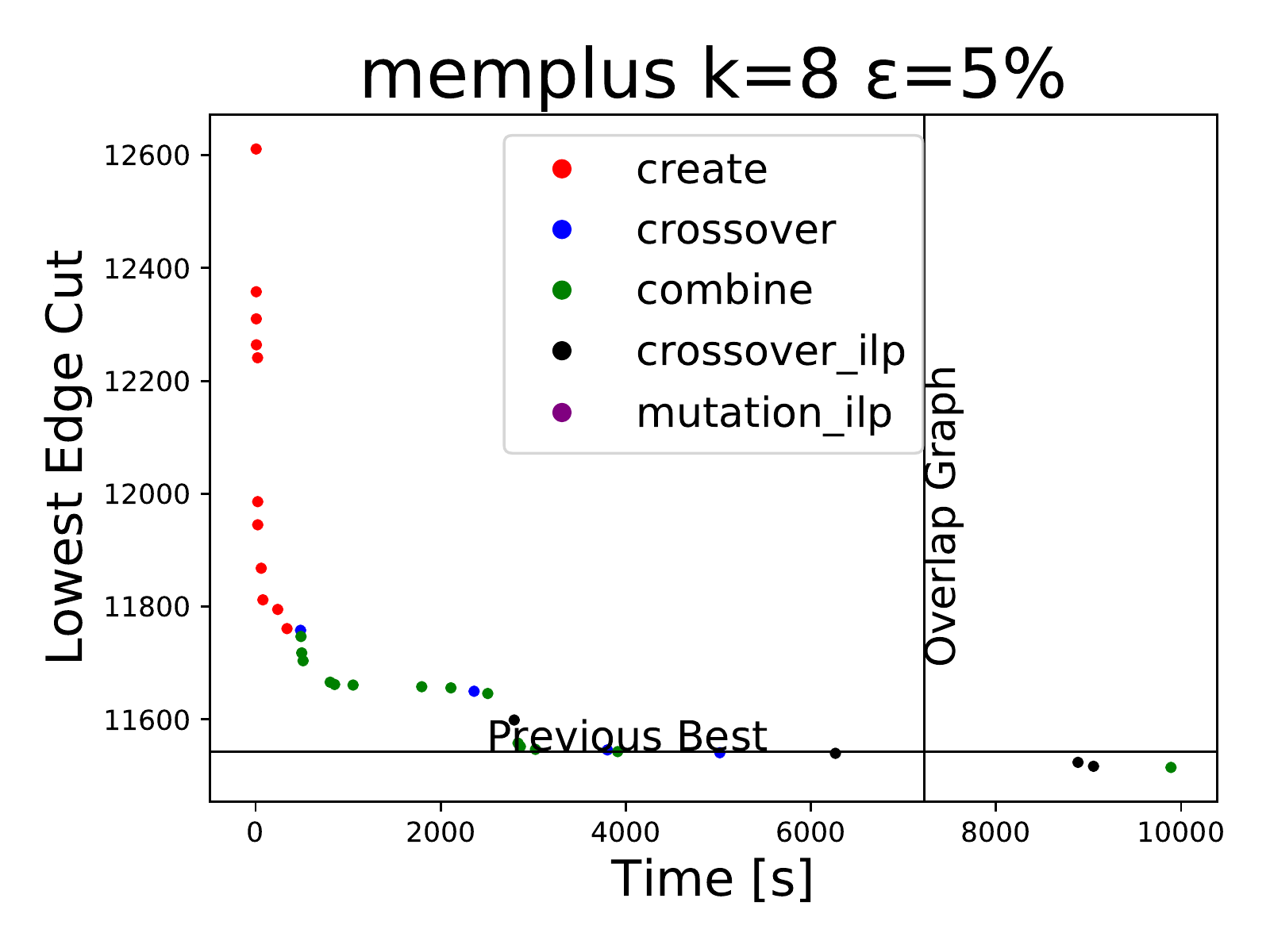} &
    \includegraphics[width=.45\textwidth]{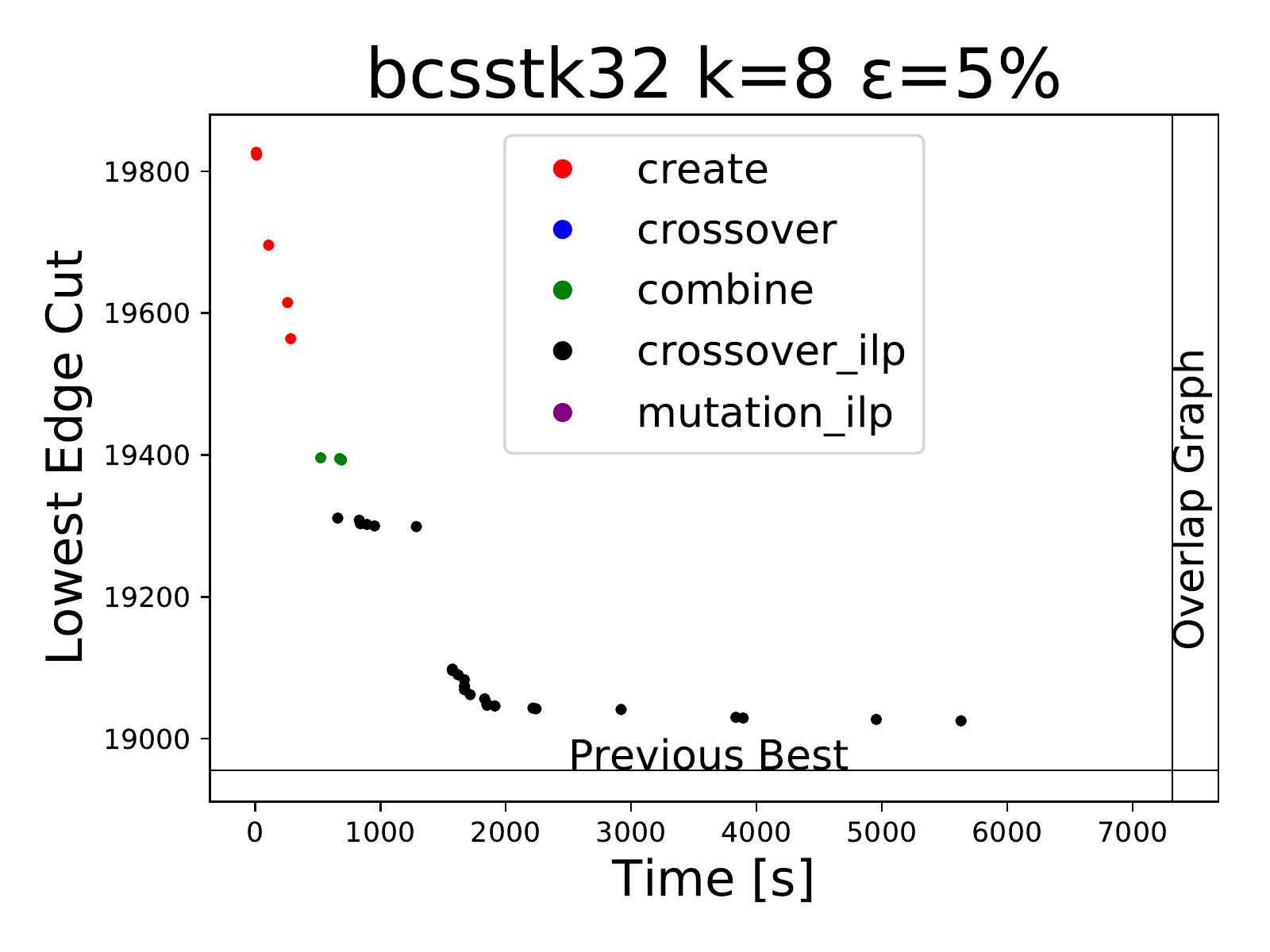} \\

    \includegraphics[width=.45\textwidth]{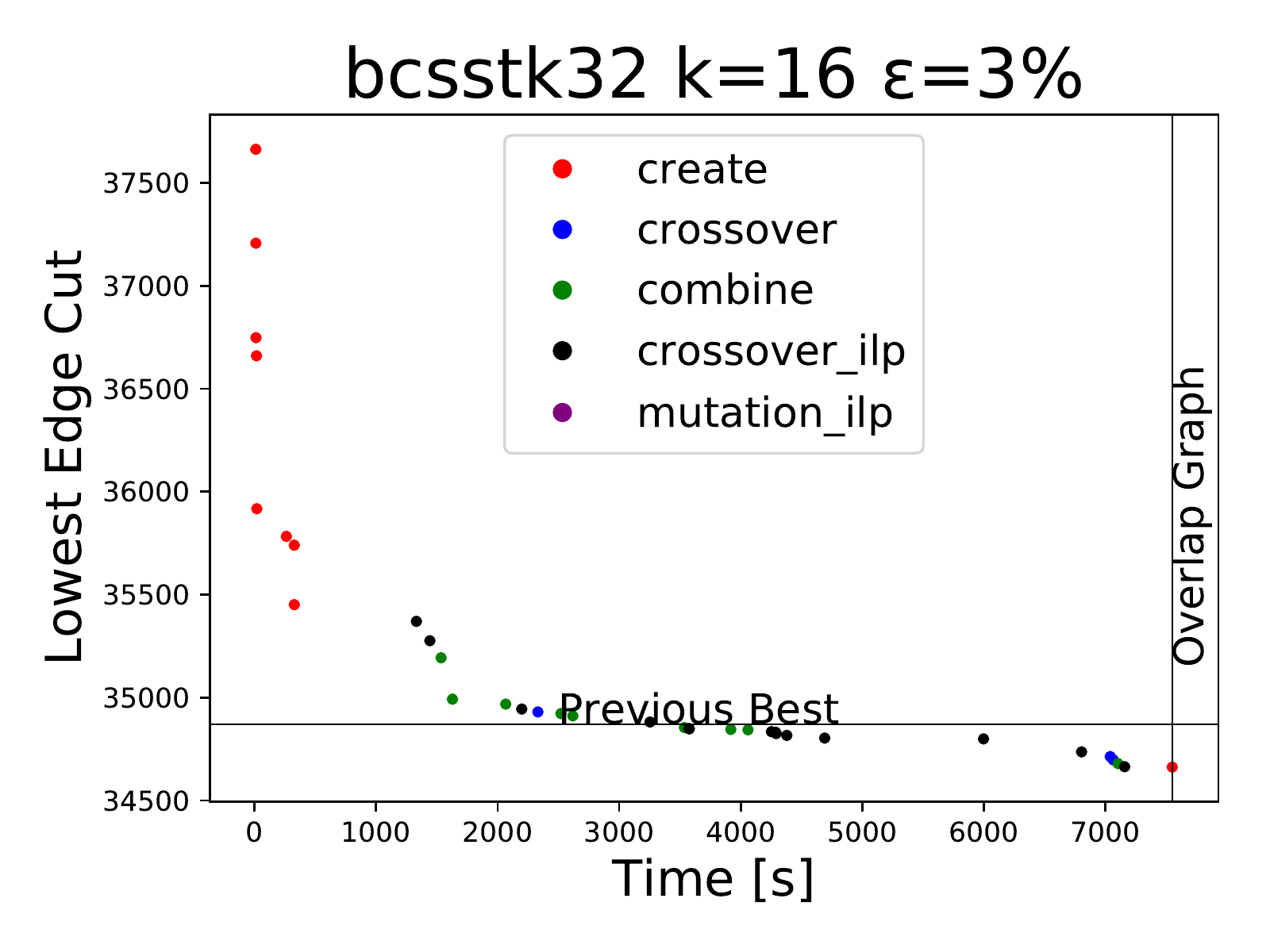} &
    \includegraphics[width=.45\textwidth]{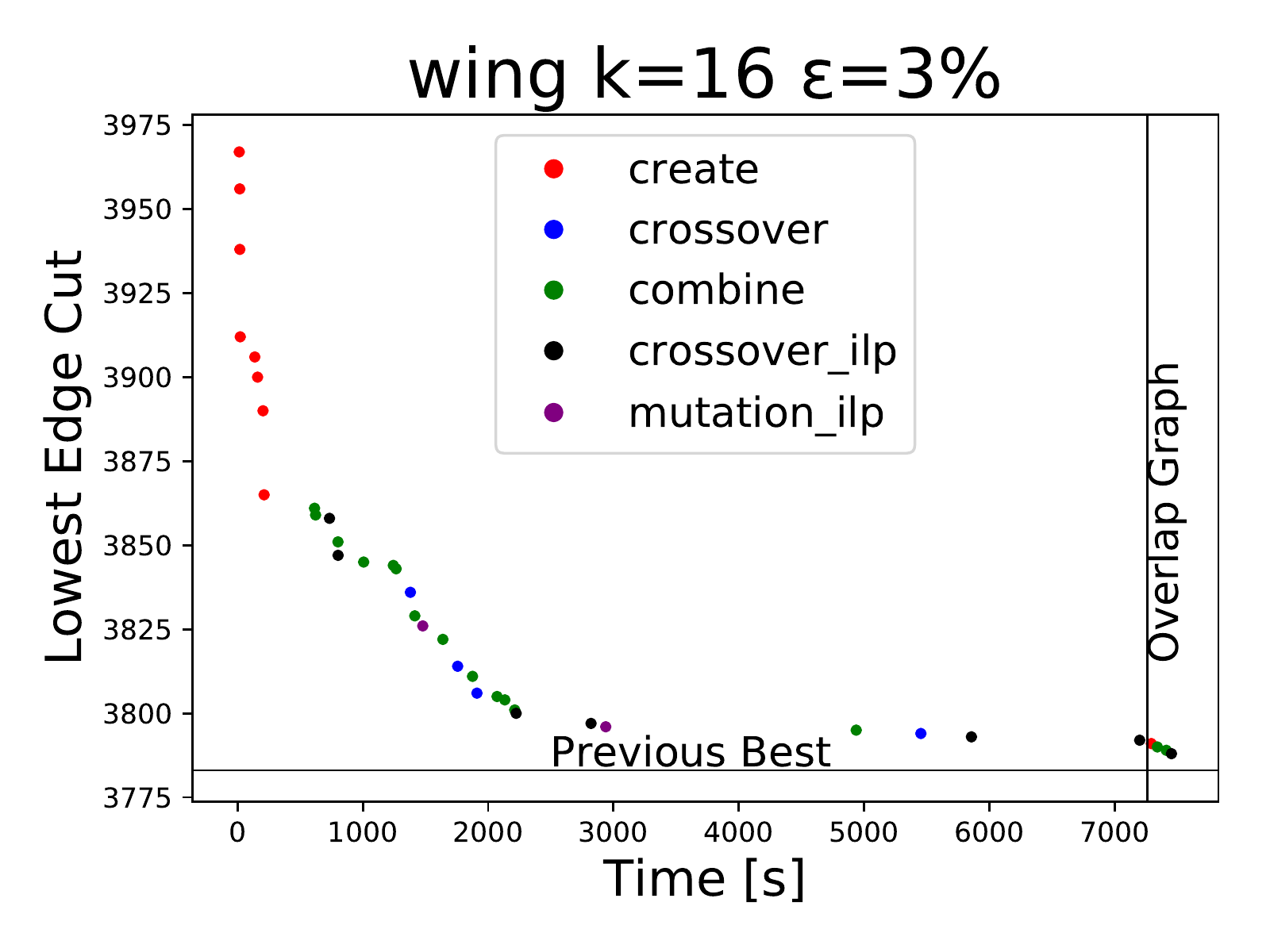} \\

    \includegraphics[width=.45\textwidth]{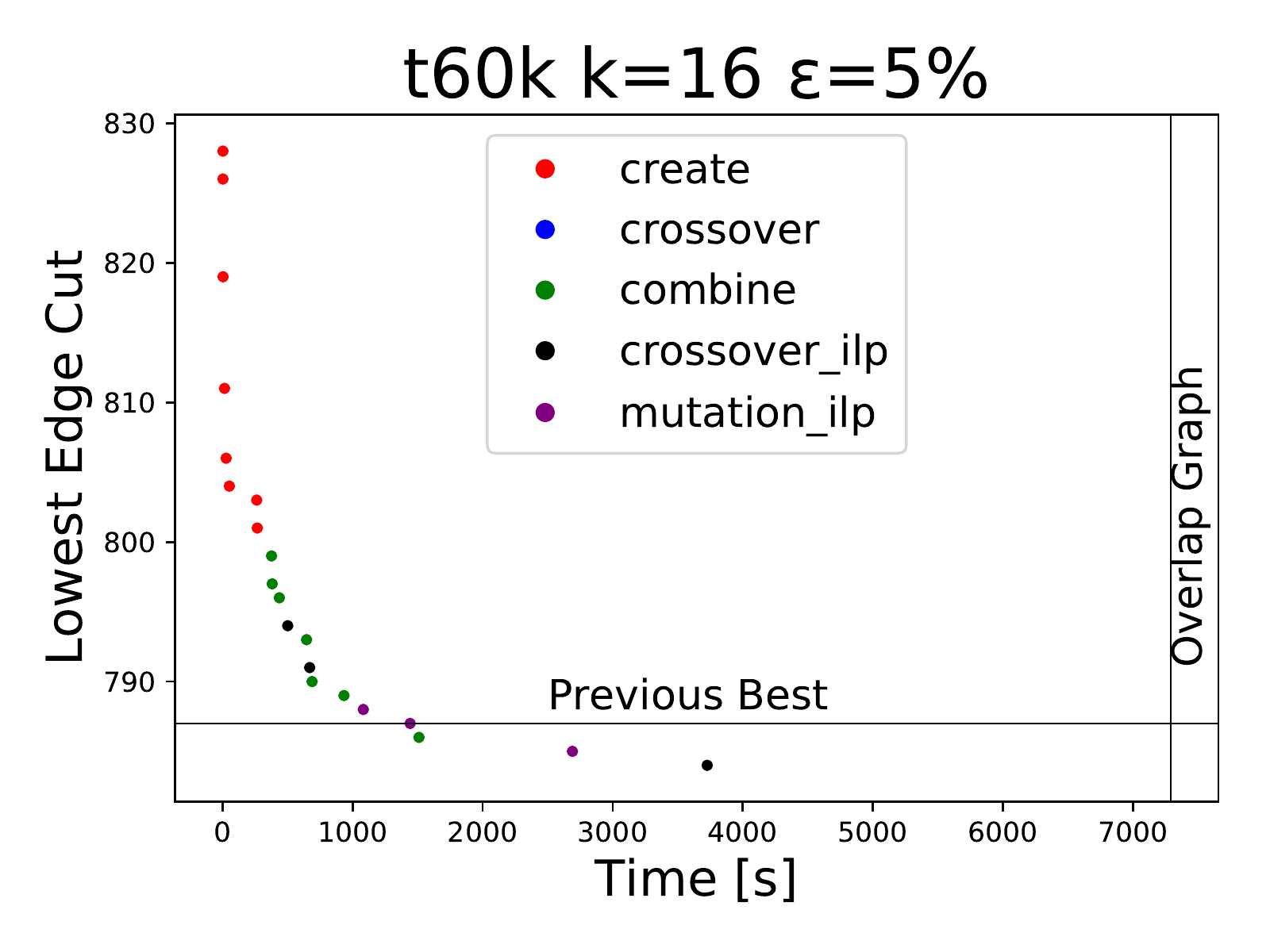} &
    \includegraphics[width=.45\textwidth]{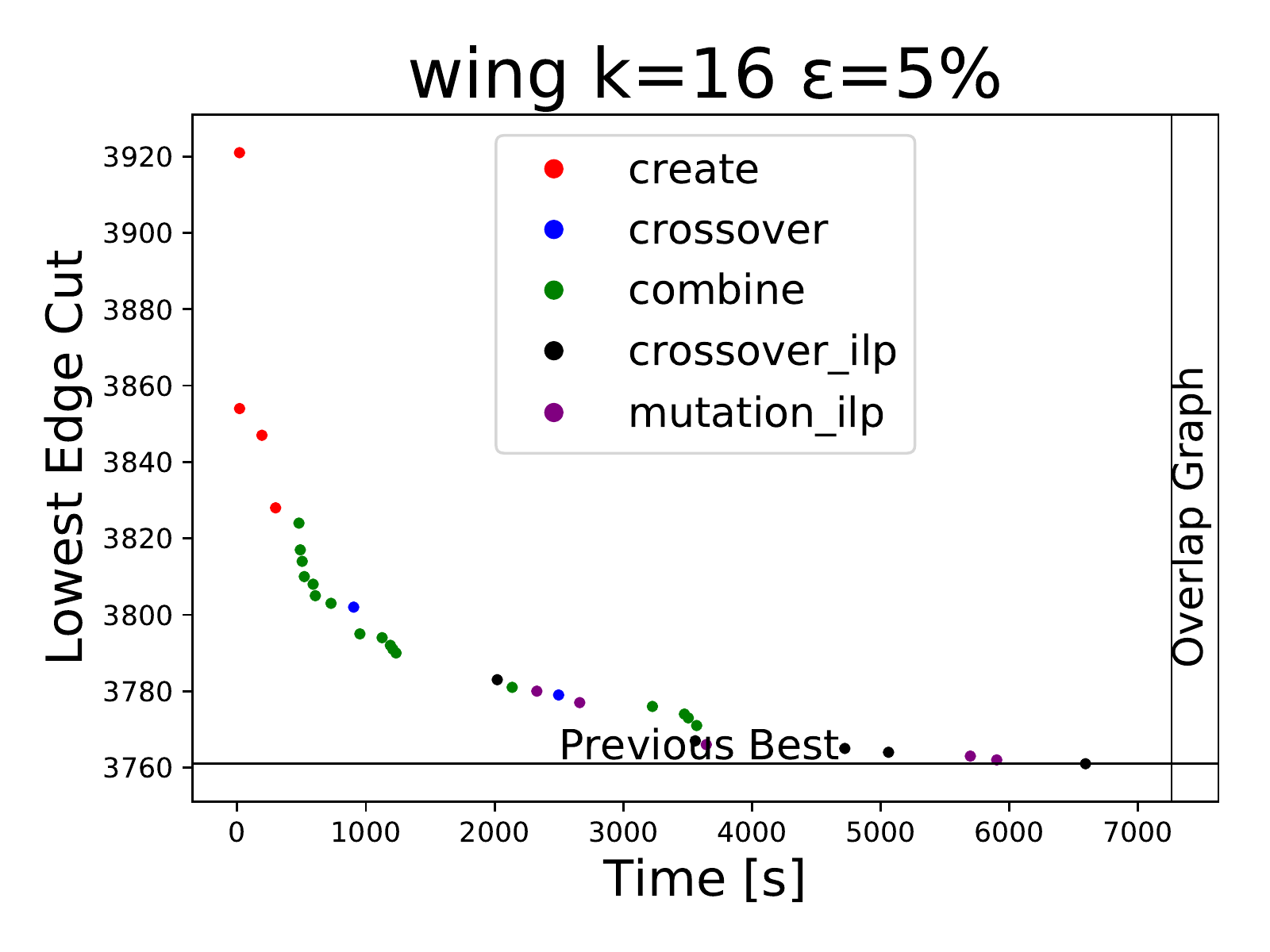} \\
  \end{tabular}
  \caption{Improvement of best partition over time, compared to previously best
  solution.\label{c:gp:fig:evo1}}
\end{figure}  

\begin{figure}
  \centering
  \begin{tabular}{c c}

    \includegraphics[width=.45\textwidth]{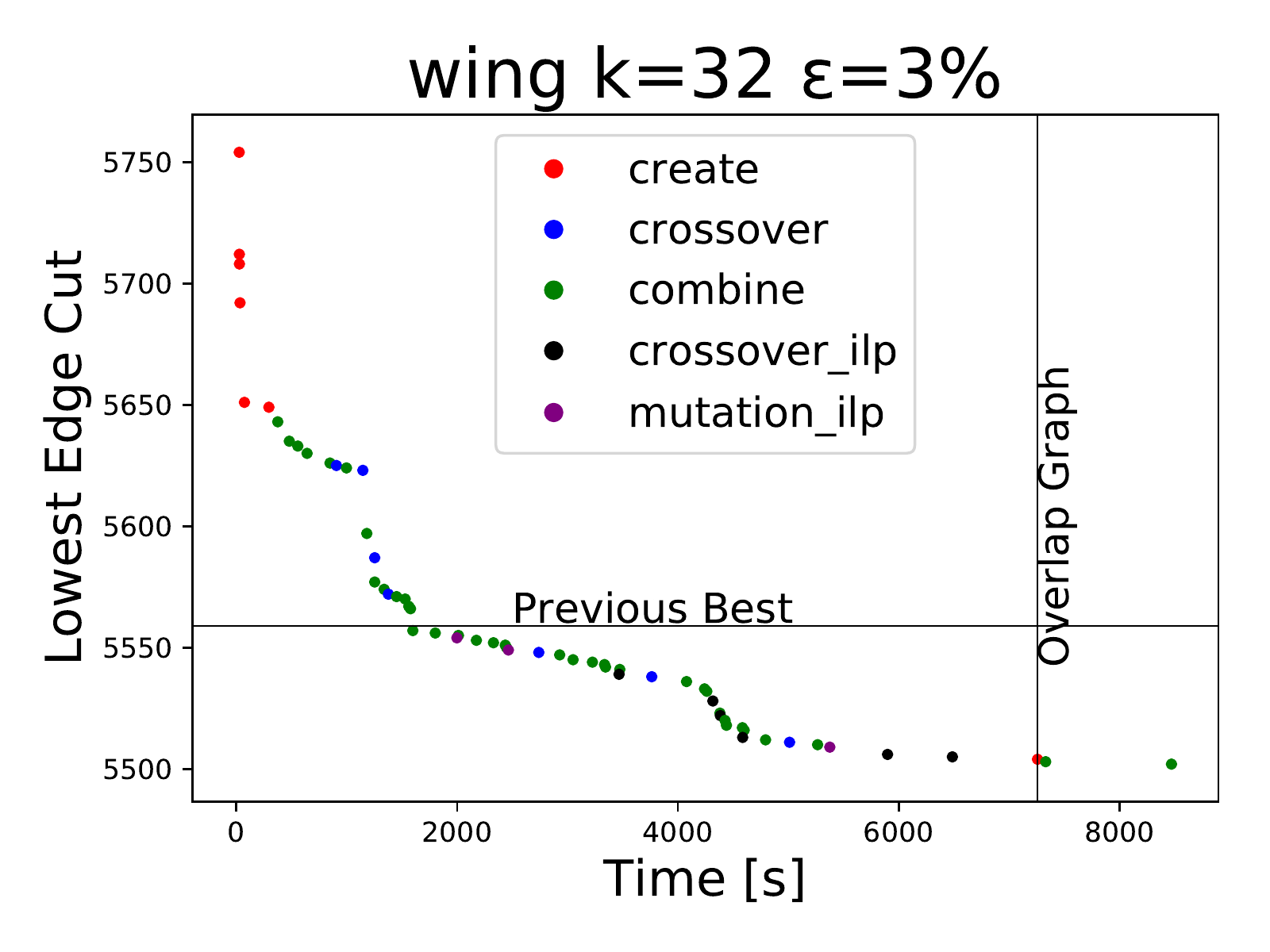} &
    \includegraphics[width=.45\textwidth]{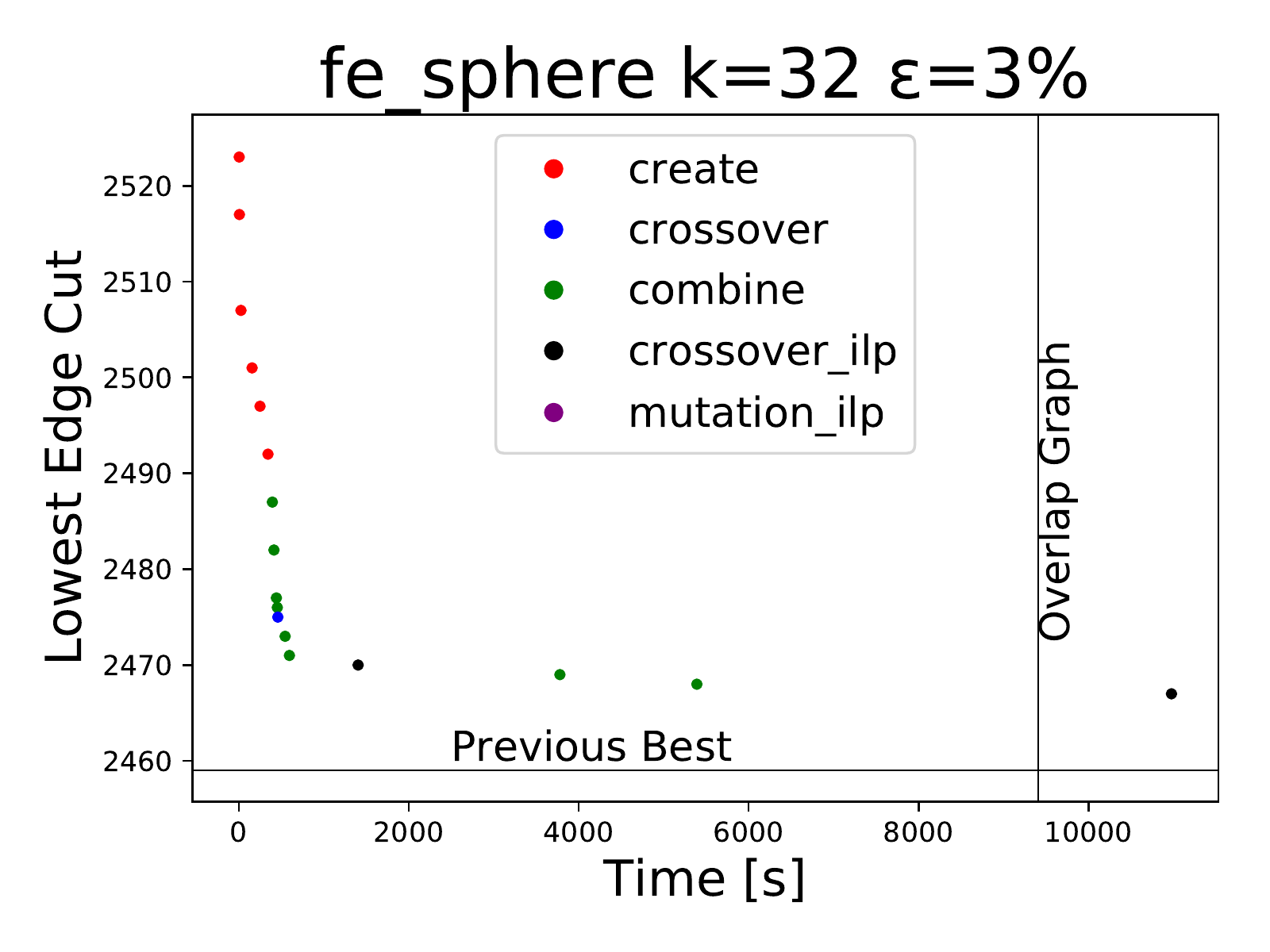} \\
    
    \includegraphics[width=.45\textwidth]{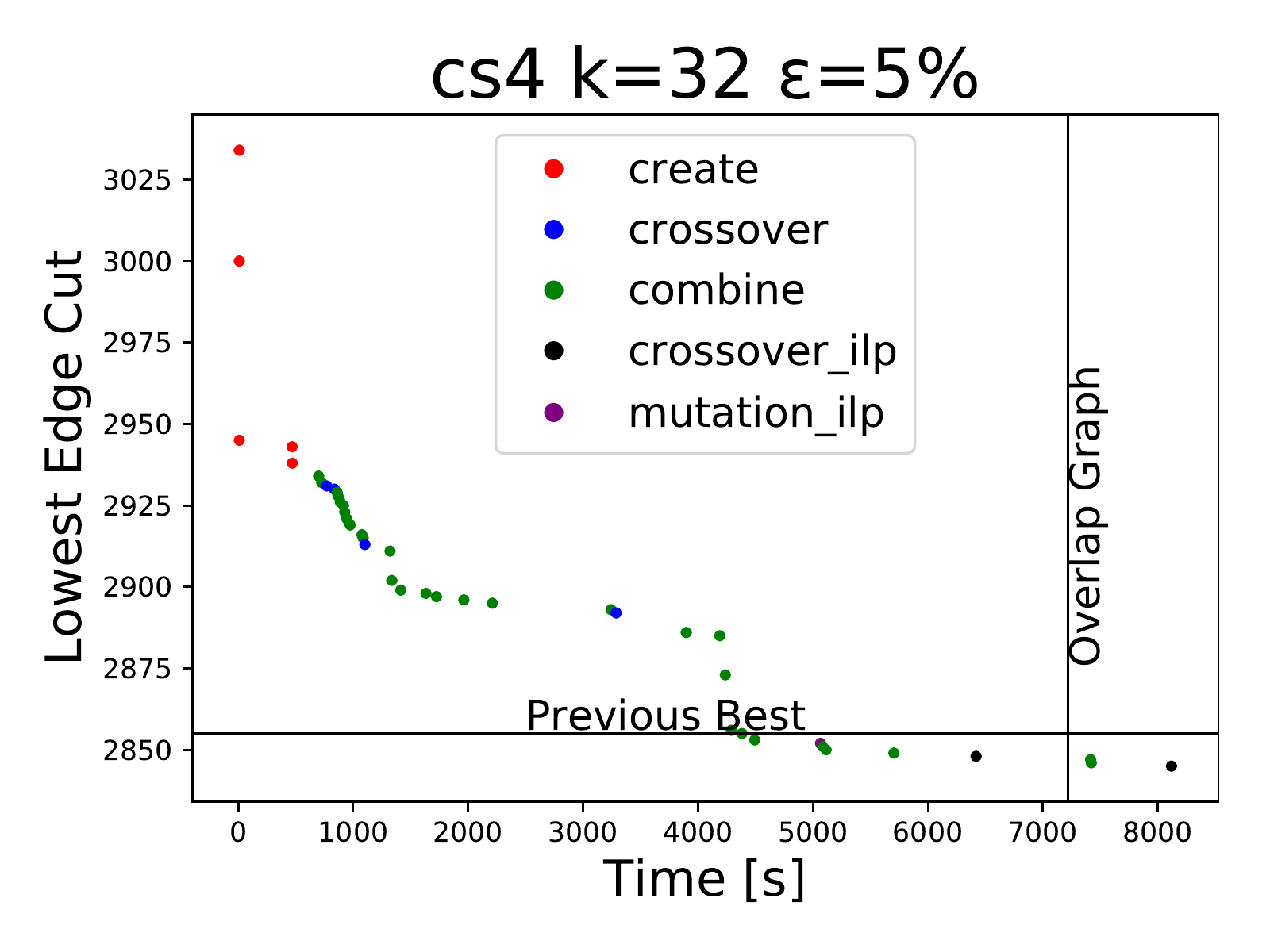} &
    \includegraphics[width=.45\textwidth]{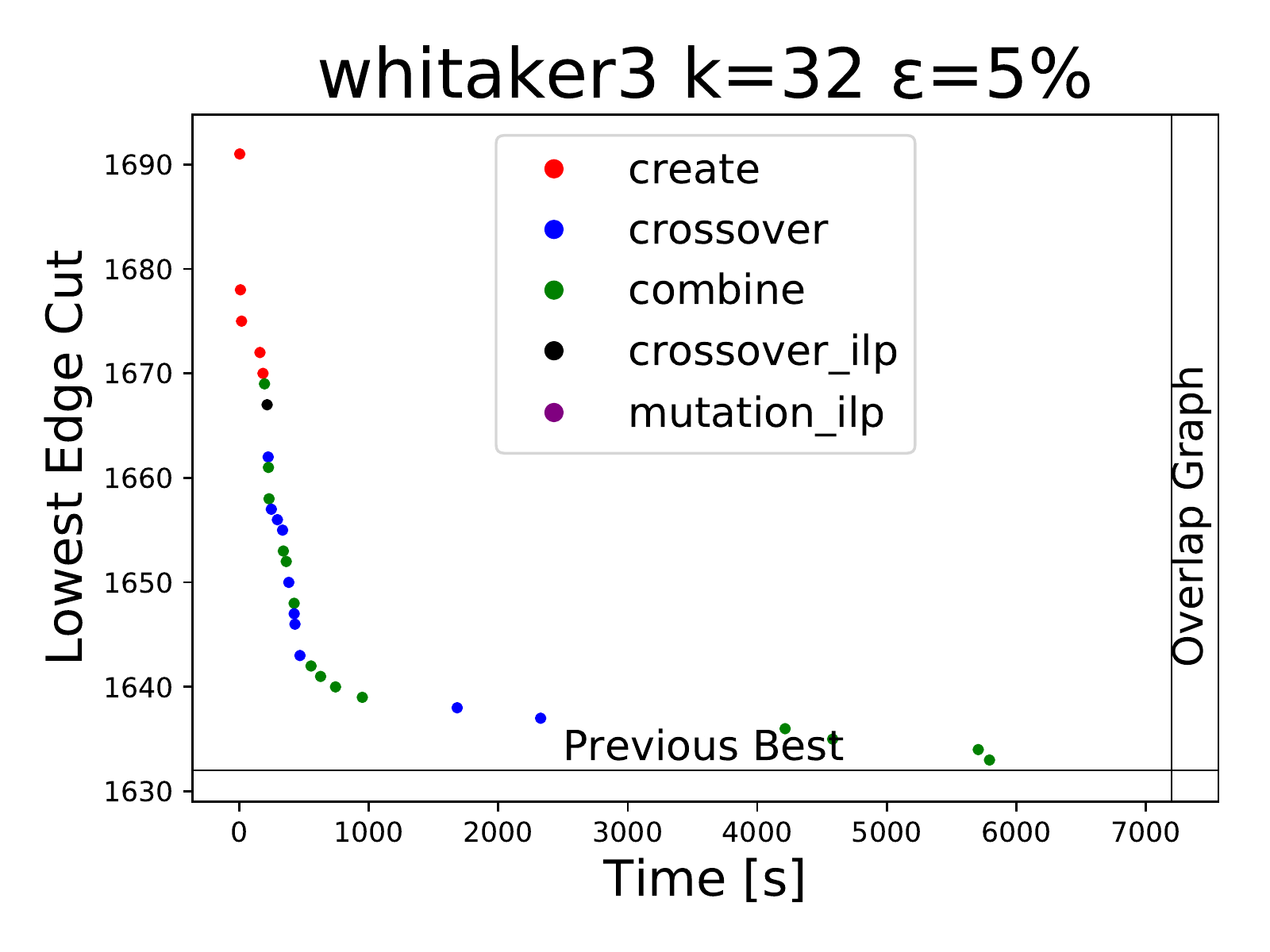} \\

    \includegraphics[width=.45\textwidth]{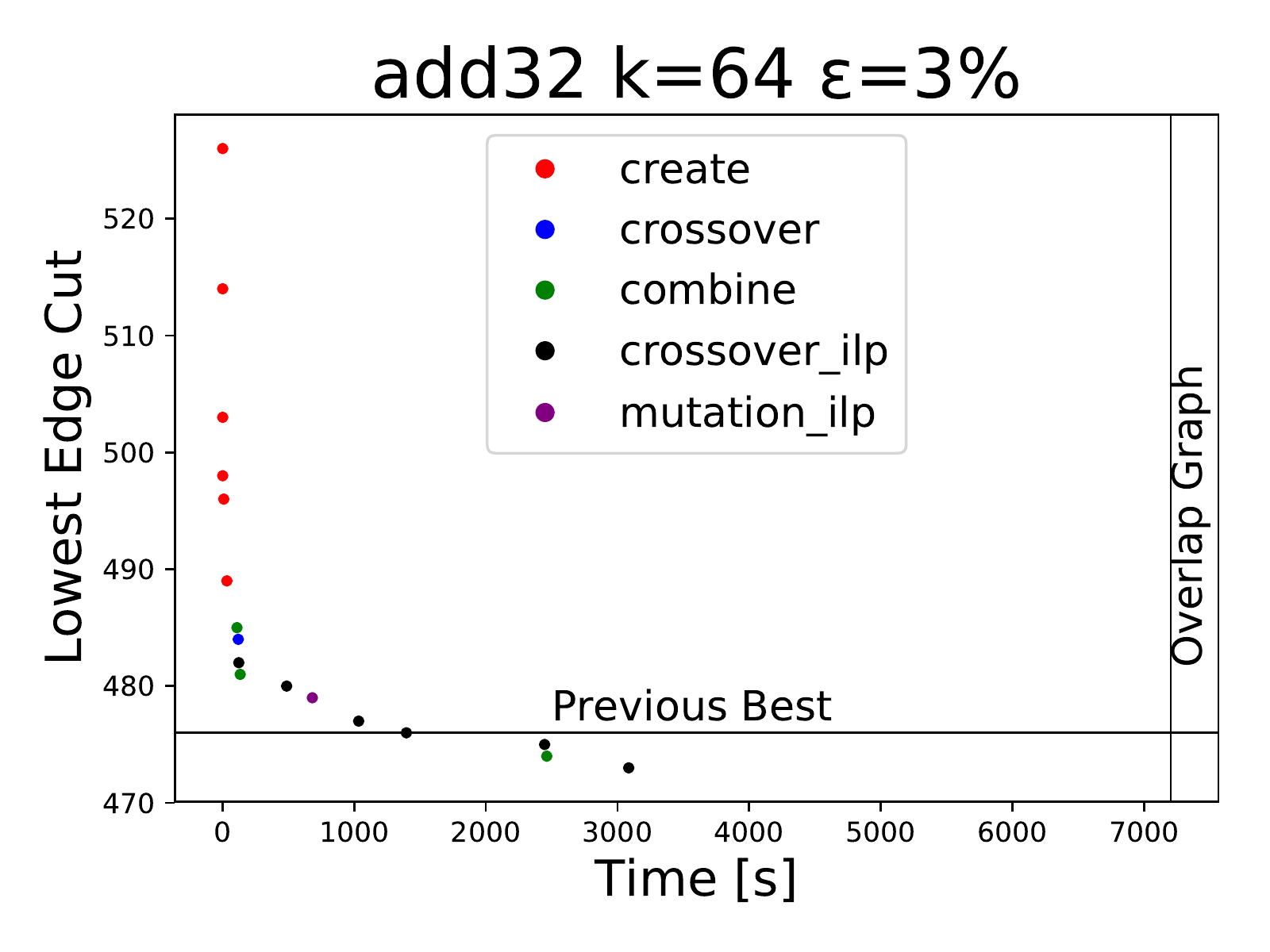} &
    \includegraphics[width=.45\textwidth]{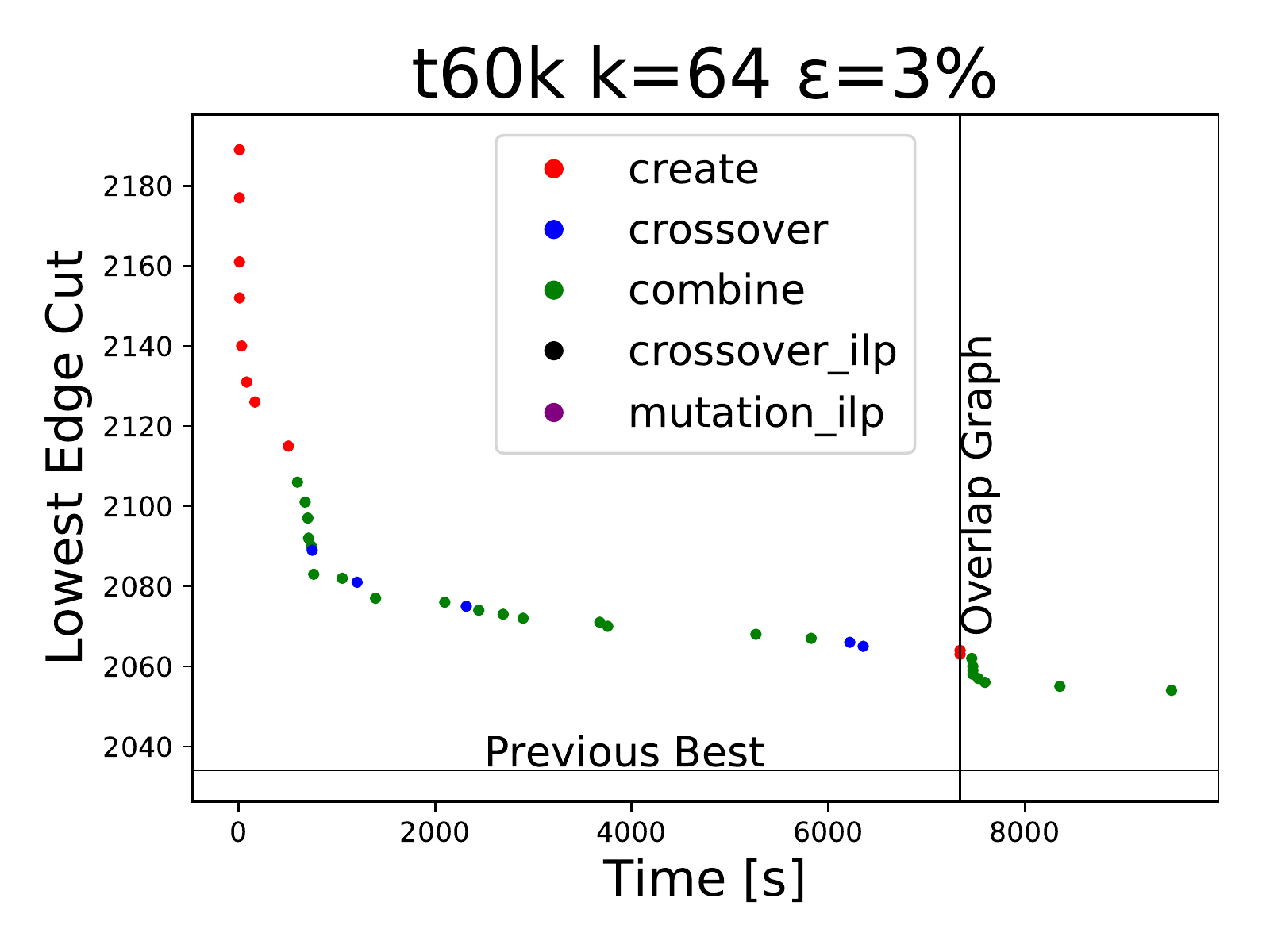} \\

    \includegraphics[width=.45\textwidth]{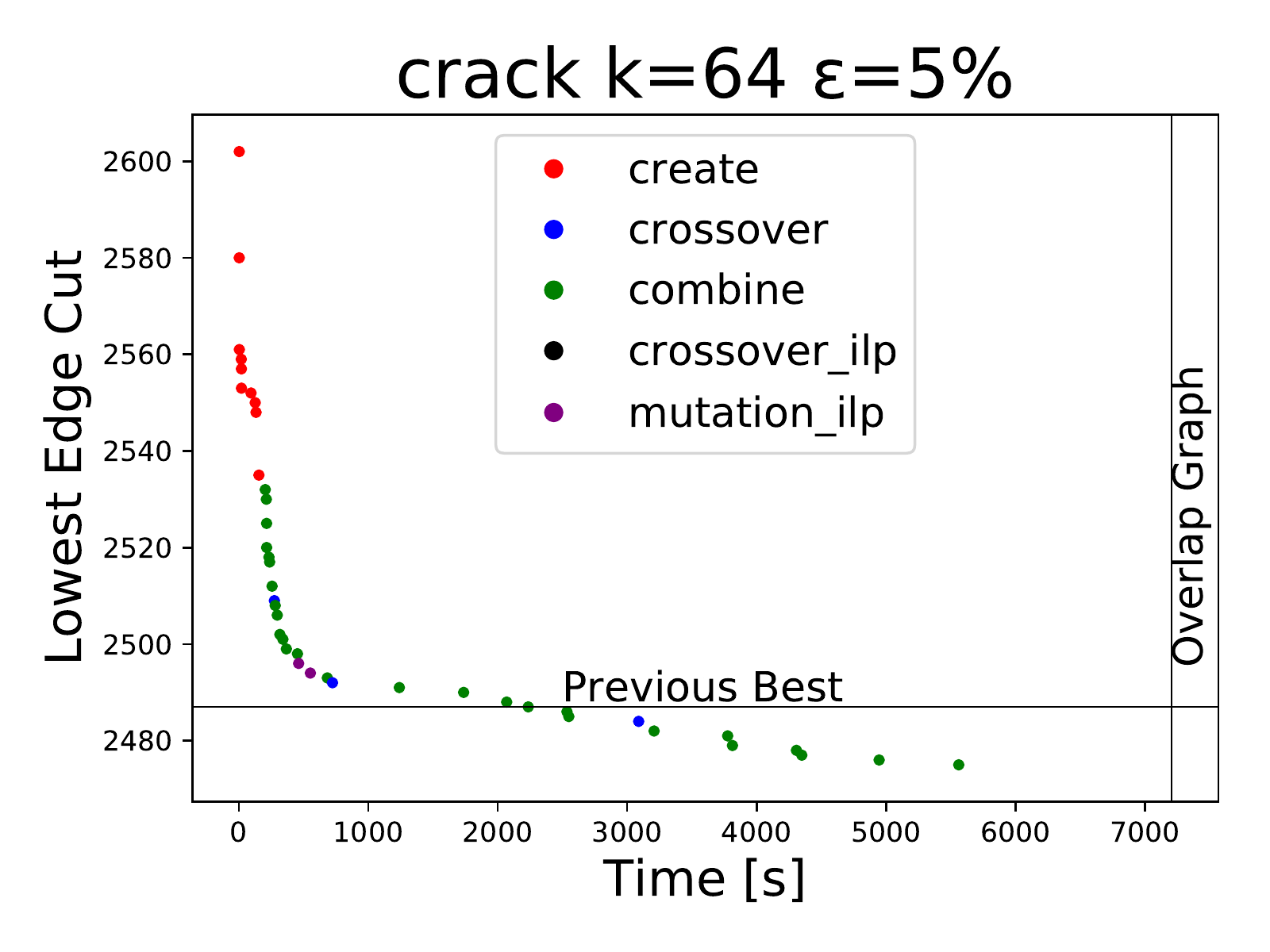} &
    \includegraphics[width=.45\textwidth]{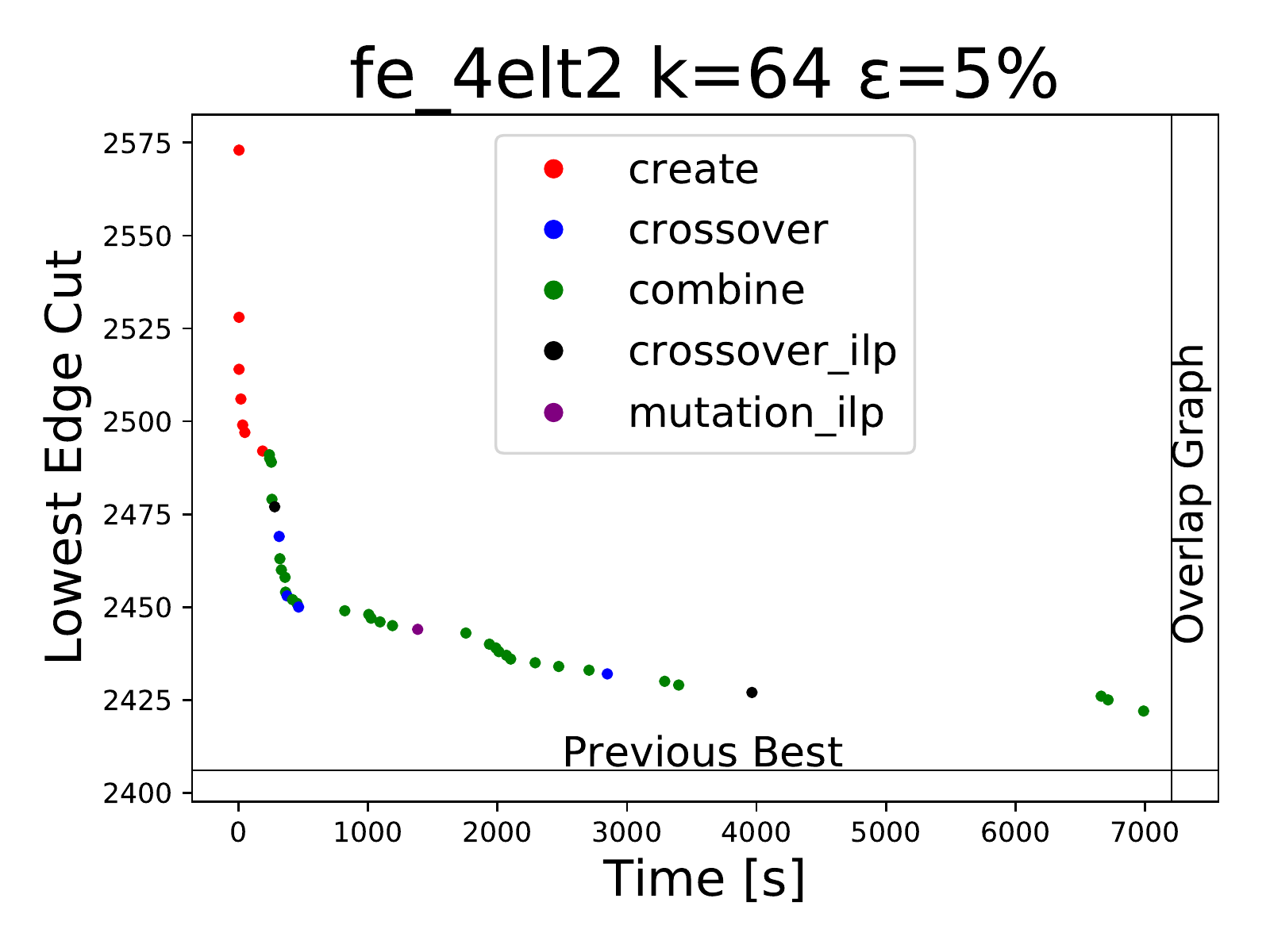}\\
  \end{tabular}
  \caption{Improvement of best partition over time, compared to previously best
  solution.\label{c:gp:fig:evo}}
\end{figure}

When running our algorithm using the currently best partitions provided in the
benchmark, we are able to improve $38\%$ of the currently reported perfectly
balanced results. We are able to improve a larger number of results for larger
values of $k$, more specifically, out of the partitions with $k \geq 16$, we can
improve $60\%$ of all perfectly balanced partitions. There is a wide range of
improvements with the smallest improvement being $0.0008\%$ for graph auto with
$k=32$ and $\epsilon = 3\%$ and with the largest improvement that we found being
$1.72\%$ for fe\_body for $k=32$ and $\epsilon = 0\%$. The largest absolute
improvement we found is~$117$ for bcsstk32 with $k=64$ and $\epsilon = 0\%$. In
general, the total number of improvements is lower if some imbalance is allowed.
This is also expected since traditional local search methods have a larger
amount of freedom to move vertices. However, the number of improvements still
shows that the method is also able to improve many partitions even if some
imbalance is allowed. We submitted the improved partitionings of our ILP-based
local search algorithm to the Walshaw graph partitioning
archive~\cite{wswebsite}, where it is denoted by \texttt{$^*$-ILP}.

\subsection{Integration into KaBaPE}
\label{c:gp:ss:exp_evo}

Section~\ref{c:gp:s:evo} shows how we integrate our approach into the memetic
graph partitioning algorithm KaBaPE. We detail the two new operations that we
introduce to KaBaPE. In KaBaPE, we use the standard parameters given by the
original authors~\cite{kabapeE}.

We run experiments on the small and medium sized graphs of the Walshaw graph
partitioning benchmark archive~\cite{wswebsite} as shown in
Tables~\ref{c:gp:tab:evo3}~and~\ref{c:gp:tab:evo5}, which are the graphs also
used in the original KaBaPE paper~\cite{kabapeE}. Similar to their experiments,
we also give $2$ hours for each problem. Afterwards, we perform post-processing
by running the algorithm for $1$ hour on the overlap graph given by the best
$100$ unique partitions. Note, that even though we have a total running time of
$3$ hours instead of $2$ hours in the results of KaBaPE~\cite{kabapeE}, all
problems in which KaBaPE+ILP outperforms the current best solution in the
Walshaw archive, the solution was already better before post processing.

We run experiments on the problems that have $k \in \{8,16,32,64\}$ and
$\epsilon \in \{3\%, 5\%\}$. These are the hard instances of the benchmark, in
which algorithms do not just reproduce the same solution as previous approaches.
Figures~\ref{c:gp:fig:evo1}~and~\ref{c:gp:fig:evo} show the development of the
fittest individual over the course of the algorithm for a variety of graphs. A
summary of the results is shown in Table~\ref{c:gp:tab:evo}, complete results
for KaBaPE+ILP on all problems are given in
Tables~\ref{c:gp:tab:evo3}~and~\ref{c:gp:tab:evo5} in the appendix. 

\begin{table}[t!]
  \centering
          \caption{Relative number of improved instances by KaBaPE in the
          Walshaw Benchmark starting from scratch. \label{c:gp:tab:evo}}
   \begin{tabular}{|l|r|r|r|r||r|}
    \hline
    $\epsilon \backslash k$ & \numprint{8} & \numprint{16} & \numprint{32} &
    \numprint{64} & overall \\ \hline \hline
    $3\%$ & $16\%$ & $4\%$ & $20\%$ & $4\%$ & $11\%$ \\
    $5\%$ & $8\%$ & $32\%$ & $28\%$ & $24\%$ & $23\%$ \\ \hline
   \end{tabular}
  \end{table}

On those $200$ problems we manage to improve the best known solution in $35$
cases. The previously best results hereby include the improvements given in the
previous experiments. In $62$ of the problems, KaBaPE+ILP reproduces the best
known cut. The highest improvement can be found on graph \textttA{bcsstk32},
$k=32$, $\epsilon = 5\%$, where we improve the best known edge cut by a value of
more than $800$.

Note that feeding the best known solution from the Walshaw archive into the
population as a seed partition does not increase the quality of the solution.
For all $35$ instances in which KaBaPE+ILP outperforms the best solution from
the Walshaw archive, a larger improvement is only seen in $3$ instances when
additionally using a seed partition.

\section{Conclusion}\label{c:gp:s:conclusion} We presented a
novel meta-heuristic for the balanced graph partitioning problem. Our approach
is based on an integer linear program that solves a model to combine
unconstrained vertex movements into a global feasible improvement. Through a
given input partition, we were able to use symmetry breaking and other
techniques that make the approach scale to large inputs. In Walshaw’s benchmark,
we were able to improve a large number of partitions. 

We also integrated the algorithm into the KaHIP framework by adding new
crossover operations based on integer linear programs into the evolutionary
algorithm KaBaPE~\cite{kabapeE}. This extended evolutionary algorithm produces
high quality partitions from scratch. On half of the hard problems from
Walshaw's benchmark, our new algorithm produces a result that is at least as
good as the previously best result. On $17\%$, the solution given is better than
the previous best solution.

We would like to look at other objective functions as long as they can be
modelled linearly. Moreover, we want to investigate whether this kind of
contractions can be useful for other ILPs. Besides using other exact techniques
like branch-and-bound to solve the model, it may also be worthwhile to use a
heuristic algorithm instead. In the Walshaw graph partitioning
benchmark~\cite{wswebsite}, the results given by this algorithm are denoted by
\textttA{KaBaPE+ILP}.

\newpage

\section{Additional Tables}
\label{c:gp:s:additional}

  \begin{table}[ht!]
   \centering
    \small
      \caption{Improvement of existing partitions from the Walshaw benchmark
with $\epsilon = 0\%$ using our ILP approach. In each $k$-column the results
computed by our approach are on the left and the current Walshaw cuts are on the
right. Results achieved by \texttt{Gain$_{\rho=-1}$} are marked with \^\ and
results achieved by \texttt{Gain$_{\rho=-2}$} are marked with *.
\label{table:detailedimprovement}} \resizebox{\linewidth}{!}{

\begin{tabular}{|l|rr|rr|rr|rr|rr|rr|}
\hline
      Graph / k & \multicolumn{2}{|c|}{2} & \multicolumn{2}{|c|}{4} &
      \multicolumn{2}{|c|}{8} & \multicolumn{2}{|c|}{16} &
      \multicolumn{2}{|c|}{32} & \multicolumn{2}{|c|}{64} \\ \hline \hline
            add20 & \numprint{596} & \numprint{596} &  \numprint{1151} &
            \numprint{1151} &  \numprint{1681} & \numprint{1681} &
            \numprint{2040} & \numprint{2040} &   \textbf{*\numprint{2360}} &
            \numprint{2361} &   \textbf{\^{}\numprint{2947}} & \numprint{2949}
            \\
             data & \numprint{189} & \numprint{189} & \numprint{382} &
             \numprint{382} & \numprint{668} & \numprint{668} &  \numprint{1127}
             & \numprint{1127} & \numprint{1799} & \numprint{1799} &
             \numprint{2839} & \numprint{2839} \\
             3elt & \numprint{90} & \numprint{90} & \numprint{201} &
             \numprint{201} & \numprint{345} & \numprint{345} & \numprint{573} &
             \numprint{573} & \numprint{960} & \numprint{960} &  \numprint{1532}
             & \numprint{1532} \\
               uk & \numprint{19} & \numprint{19} & \numprint{41} &
               \numprint{41} & \numprint{83} & \numprint{83} & \numprint{145} &
               \numprint{145} &    \textbf{*\^{}\numprint{246}} & \numprint{247}
               & \numprint{408} & \numprint{408} \\
            add32 & \numprint{11} & \numprint{11} & \numprint{34} &
            \numprint{34} & \numprint{67} & \numprint{67} & \numprint{118} &
            \numprint{118} & \numprint{213} & \numprint{213} & \numprint{485} &
            \numprint{485} \\
         bcsstk33 & \numprint{10171} & \numprint{10171} & \numprint{21717} &
         \numprint{21717} & \numprint{34437} & \numprint{34437} &
         \numprint{54680} & \numprint{54680} & \numprint{77414} &
         \numprint{77414} & \numprint{107185} & \numprint{107185} \\
\hline
    \end{tabular}
   }
   \vspace{-3cm}
    \end{table}

    \begin{table}[b!]
   \centering
    \small
      \caption{Improvement of existing partitions from the Walshaw benchmark
with $\epsilon = 1\%$ using our ILP approach. In each $k$-column the results
computed by our approach are on the left and the current Walshaw cuts are on the
right. Results achieved by \texttt{Gain$_{\rho=-1}$} are marked with \^\ and
results achieved by \texttt{Gain$_{\rho=-2}$} are marked with *.
\label{table:detailedimprovement1}} \resizebox{\linewidth}{!}{

  \begin{tabular}{|l|rr|rr|rr|rr|rr|rr|}

\hline
      Graph / k & \multicolumn{2}{|c|}{2} & \multicolumn{2}{|c|}{4} &
      \multicolumn{2}{|c|}{8} & \multicolumn{2}{|c|}{16} &
      \multicolumn{2}{|c|}{32} & \multicolumn{2}{|c|}{64} \\ \hline \hline

            add20 & \numprint{585} & \numprint{585} & \numprint{1147} &
            \numprint{1147} &  \textbf{*\^{}\numprint{1680}} & \numprint{1681} &
            \numprint{2040} & \numprint{2040} & \numprint{2361} &
            \numprint{2361} & \numprint{2949} & \numprint{2949} \\
             data & \numprint{188} & \numprint{188} & \numprint{376} &
             \numprint{376} & \numprint{656} & \numprint{656} & \numprint{1121}
             & \numprint{1121} & \numprint{1799} & \numprint{1799} &
             \numprint{2839} & \numprint{2839} \\
             3elt & \numprint{89} & \numprint{89} & \numprint{199} &
             \numprint{199} & \numprint{340} & \numprint{340} & \numprint{568} &
             \numprint{568} & \numprint{953} & \numprint{953} & \numprint{1532}
             & \numprint{1532} \\
               uk & \numprint{19} & \numprint{19} & \numprint{40} &
               \numprint{40} & \numprint{80} & \numprint{80} & \numprint{142} &
               \numprint{142} & \numprint{246} & \numprint{246} & \numprint{408}
               & \numprint{408} \\
            add32 & \numprint{10} & \numprint{10} & \numprint{33} &
            \numprint{33} & \numprint{66} & \numprint{66} & \numprint{117} &
            \numprint{117} & \numprint{212} & \numprint{212} & \numprint{485} &
            \numprint{485} \\
         bcsstk33 & \numprint{10097} & \numprint{10097} & \numprint{21338} &
         \numprint{21338} & \numprint{34175} & \numprint{34175} &
         \numprint{54505} & \numprint{54505} & \numprint{77195} &
         \numprint{77195} & \numprint{106902} & \numprint{106902} \\
\hline
    \end{tabular}

   }

    \end{table}

    \begin{table}[ht!]
   \centering
    \small
      \caption{Improvement of existing partitions from the Walshaw benchmark
with $\epsilon = 3\%$ using our ILP approach. In each $k$-column the results
computed by our approach are on the left and the current Walshaw cuts are on the
right. Results achieved by \texttt{Gain$_{\rho=-1}$} are marked with \^\ and
results achieved by \texttt{Gain$_{\rho=-2}$} are marked with *.
\label{table:detailedimprovement3}} \resizebox{\linewidth}{!}{

  \begin{tabular}{|l|rr|rr|rr|rr|rr|rr|}

\hline
      Graph / k & \multicolumn{2}{|c|}{2} & \multicolumn{2}{|c|}{4} &
      \multicolumn{2}{|c|}{8} & \multicolumn{2}{|c|}{16} &
      \multicolumn{2}{|c|}{32} & \multicolumn{2}{|c|}{64} \\ \hline \hline

            add20 & \numprint{560} & \numprint{560} & \numprint{1134} &
            \numprint{1134} & \numprint{1673} & \numprint{1673} &
            \numprint{2030} & \numprint{2030} & \numprint{2346} &
            \numprint{2346} & \numprint{2920} & \numprint{2920} \\
             data & \numprint{185} & \numprint{185} & \numprint{369} &
             \numprint{369} & \numprint{638} & \numprint{638} & \numprint{1088}
             & \numprint{1088} & \numprint{1768} & \numprint{1768} &
             \textbf{*\numprint{2781}} & \numprint{2783} \\
             3elt & \numprint{87} & \numprint{87} & \numprint{198} &
             \numprint{198} & \numprint{334} & \numprint{334} & \numprint{561} &
             \numprint{561} & \numprint{944} & \numprint{944} & \numprint{1512}
             & \numprint{1512} \\
               uk & \numprint{18} & \numprint{18} & \numprint{39} &
               \numprint{39} & \numprint{78} & \numprint{78} & \numprint{139} &
               \numprint{139} & \numprint{240} & \numprint{240} & \numprint{397}
               & \numprint{397} \\
            add32 & \numprint{10} & \numprint{10} & \numprint{33} &
            \numprint{33} & \numprint{66} & \numprint{66} & \numprint{117} &
            \numprint{117} & \numprint{212} & \numprint{212} & \numprint{476} &
            \numprint{476} \\
         bcsstk33 & \numprint{10064} & \numprint{10064} & \numprint{20762} &
         \numprint{20762} & \numprint{34065} & \numprint{34065} &
         \numprint{54354} & \numprint{54354} & \numprint{76749} &
         \numprint{76749} & \textbf{*\numprint{105737}} & \numprint{105742} \\
\hline
    \end{tabular}
   }
    \end{table}

    \begin{table}[ht!]
   \centering
    \small
      \caption{Improvement of existing partitions from the Walshaw benchmark
with $\epsilon = 5\%$ using our ILP approach. In each $k$-column the results
computed by our approach are on the left and the current Walshaw cuts are on the
right. Results achieved by \texttt{Gain$_{\rho=-1}$} are marked with \^\ and
results achieved by \texttt{Gain$_{\rho=-2}$} are marked with *.
\label{table:detailedimprovement5}} \resizebox{\linewidth}{!}{

  \begin{tabular}{|l|rr|rr|rr|rr|rr|rr|}

\hline
      Graph / k & \multicolumn{2}{|c|}{2} & \multicolumn{2}{|c|}{4} &
      \multicolumn{2}{|c|}{8} & \multicolumn{2}{|c|}{16} &
      \multicolumn{2}{|c|}{32} & \multicolumn{2}{|c|}{64} \\ \hline \hline

            add20 & \numprint{536} & \numprint{536} & \numprint{1120} &
            \numprint{1120} & \numprint{1657} & \numprint{1657} &
            \numprint{2027} & \numprint{2027} & \numprint{2341} &
            \numprint{2341} & \numprint{2920} & \numprint{2920} \\
             data & \numprint{181} & \numprint{181} & \numprint{363} &
             \numprint{363} & \numprint{628} & \numprint{628} & \numprint{1076}
             & \numprint{1076} & \numprint{1743} & \numprint{1743} &
             \numprint{2747} & \numprint{2747} \\
             3elt & \numprint{87} & \numprint{87} & \numprint{197} &
             \numprint{197} & \numprint{329} & \numprint{329} & \numprint{557} &
             \numprint{557} & \numprint{930} & \numprint{930} & \numprint{1498}
             & \numprint{1498} \\
               uk & \numprint{18} & \numprint{18} & \numprint{39} &
               \numprint{39} & \numprint{75} & \numprint{75} & \numprint{137} &
               \numprint{137} & \numprint{236} & \numprint{236} & \numprint{394}
               & \numprint{394} \\
            add32 & \numprint{10} & \numprint{10} & \numprint{33} &
            \numprint{33} & \numprint{63} & \numprint{63} & \numprint{117} &
            \numprint{117} & \numprint{212} & \numprint{212} & \numprint{476} &
            \numprint{476} \\
         bcsstk33 & \numprint{9914} & \numprint{9914} & \numprint{20158} &
         \numprint{20158} & \numprint{33908} & \numprint{33908} &
         \numprint{54119} & \numprint{54119} &  \textbf{\^{}\numprint{76070}} &
         \numprint{76079} & \textbf{*\numprint{105297}} & \numprint{105309} \\
\hline
    \end{tabular}
   }

    \end{table}

    \begin{table}
      \caption{Complete results for KaBaPE+ILP compared to best of previous
      approaches on Walshaw's benchmark with $\epsilon=3\%$.
      \label{c:gp:tab:evo3}} \resizebox{\linewidth}{!}{
      
        \begin{tabular}{|l|rr|rr|rr|rr|}
          
            \hline      
            Graph / k,$\epsilon$ & \multicolumn{2}{|c|}{$k=8$, $\epsilon=3\%$} &
            \multicolumn{2}{|c|}{$k=16$, $\epsilon=3\%$} &
            \multicolumn{2}{|c|}{$k=32$, $\epsilon=3\%$} &
            \multicolumn{2}{|c|}{$k=64$,$\epsilon=3\%$} \\ \hline \hline
      
                add20 & \textbf{\numprint{1664}} & \numprint{1673} &
                \numprint{2030} & \numprint{2030} & \numprint{2350} &
                \numprint{2346} & \numprint{2932} & \numprint{2920} \\
                data & \numprint{638} & \numprint{638} & \numprint{1088} &
                \numprint{1088} & \numprint{1768} & \numprint{1768} &
                \numprint{2791} & \numprint{2781} \\
                3elt & \numprint{334} & \numprint{334} & \numprint{561} &
                \numprint{561} & \numprint{944} & \numprint{944} &
                \numprint{1521} & \numprint{1512} \\
                uk & \textbf{\numprint{77}} & \numprint{78} & \numprint{139} &
                \numprint{139} & \textbf{\numprint{239}} & \numprint{240} &
                \numprint{401} & \numprint{397} \\
                add32 & \numprint{66} & \numprint{66} & \numprint{117} &
                \numprint{117} & \numprint{212} & \numprint{212} &
                \textbf{\numprint{471}} & \numprint{476} \\
                bcsstk33 & \numprint{34065} & \numprint{34065} &
                \numprint{54354} & \numprint{54354} & \numprint{76879} &
                \numprint{76749} & \numprint{106263} & \numprint{105737}\\
                whitaker3 & \numprint{649} & \numprint{649} & \numprint{1077} &
                \numprint{1073} & \numprint{1653} & \numprint{1647} &
                \numprint{2468} & \numprint{2456} \\
\hline
          \end{tabular}
         }
      
          \end{table}

          \begin{table}
            \caption{Complete results for KaBaPE+ILP compared to best of
            previous approaches on Walshaw's benchmark with $\epsilon=5\%$.
            \label{c:gp:tab:evo5}} \resizebox{\linewidth}{!}{
            
              \begin{tabular}{|l|rr|rr|rr|rr|}\hline

                  Graph / k,$\epsilon$ & \multicolumn{2}{|c|}{$k=8$,
                  $\epsilon=5\%$} & \multicolumn{2}{|c|}{$k=16$, $\epsilon=5\%$}
                  & \multicolumn{2}{|c|}{$k=32$, $\epsilon=5\%$}  &
                  \multicolumn{2}{|c|}{$k=64$, $\epsilon=5\%$} \\ \hline \hline
            
                      add20 & \textbf{\numprint{1651}} & \numprint{1657} &
                      \textbf{\numprint{2024}} & \numprint{2027} &
                      \numprint{2341} & \numprint{2341} & \numprint{2925} &
                      \numprint{2920} \\
                      data & \numprint{628} & \numprint{628} & \numprint{1076} &
                      \numprint{1076} & \numprint{1747} & \numprint{1743} &
                      \numprint{2761} & \numprint{2747} \\
                      3elt & \numprint{329} & \numprint{329} & \numprint{557} &
                      \numprint{557} & \numprint{931} & \numprint{930} &
                      \numprint{1499} & \numprint{1498} \\
                      uk & \numprint{75} & \numprint{75} & \numprint{137} &
                      \numprint{137} & \textbf{\numprint{235}} & \numprint{236}
                      & \textbf{\numprint{392}} & \numprint{394} \\
                      add32 & \numprint{63} & \numprint{63} & \numprint{117} &
                      \numprint{117} & \numprint{212} & \numprint{212} &
                      \textbf{\numprint{471}} & \numprint{476} \\
                      bcsstk33 & \numprint{33908} & \numprint{33908} &
                      \numprint{54137} & \numprint{54119} & \numprint{76213} &
                      \numprint{76070} & \numprint{105746} & \numprint{105297}
                      \\
                      whitaker3 & \numprint{644} & \numprint{644} &
                      \numprint{1068} & \numprint{1068} & \numprint{1633} &
                      \numprint{1632} & \numprint{2439} & \numprint{2425} \\
\hline
                \end{tabular}
               }
            \end{table}

\part{The Multiterminal Cut Problem}
\label{p:mtc}

\chapter[Branch-and-Reduce for Multiterminal Cut]{
  Shared-memory Branch-and-Reduce for Multiterminal Cut}
\label{c:mtc}
% 2 Multiterminal Cut Papers (ALENEX'20, unpublished)

We introduce the fastest known exact algorithm~for~the multiterminal cut problem
with $k$ terminals. In particular, we engineer existing as well as new highly
effective data reduction rules to transform the graph into a smaller equivalent
instance. We use these rules within a branch-and-reduce framework as well as to boost the
performance of an ILP formulation. In addition, we present a local search
algorithm that can significantly improve a given solution to the multiterminal
cut problem. Our algorithms achieve improvements in running time of up to
\emph{multiple orders of magnitudes} over the ILP formulation without data
reductions, which has been the de facto standard used by practitioners. This
allows us to solve instances to optimality that are significantly larger than
was previously possible; and give better solutions for problems that are too
large to be solved to optimality. Furthermore, we give an inexact heuristic
algorithm that computes high-quality solutions for very hard instances in
reasonable time.

The content of this chapter is based
on~\cite{henzinger2020faster}~and~\cite{henzinger2020sharedmemory}.

\section{Introduction}
We consider the multiterminal cut problem with $k$ terminals. Its input is an
undirected edge-weighted graph $G=(V,E,w)$ with edge weights $w: E \mapsto
\MdN_{>0}$ and its goal is to divide its set of nodes into~$k$ blocks such that
each blocks contains exactly one terminal and the weight sum of the edges
running between the~blocks~is~minimized. The problem has applications in a wide
range of areas, for example in multiprocessor
scheduling~\cite{DBLP:journals/tse/Stone77},
clustering~\cite{DBLP:journals/njc/PferschyRW94} and
bioinformatics~\cite{karaoz2004whole,nabieva2005whole,vazquez2003global}. It is
a fundamental combinatorial optimization problem which was first formulated by
Dahlhaus~\etal\cite{dahlhaus1994complexity} and
Cunningham~\cite{cunningham1989optimal}. It is NP-hard for $k \geq
3$~\cite{dahlhaus1994complexity}, even on planar graphs, and reduces to the
minimum $s$-$t$-cut problem, which is in P, for $k=2$. The minimum $s$-$t$-cut
problem aims to find the minimum cut in which the vertices $s$ and $t$ are in
different blocks. Most algorithms for the minimum multiterminal cut problem use
minimum s-t-cuts as a subroutine. Dahlhaus~\etal\cite{dahlhaus1994complexity}
give a $2(1-1/k)$ approximation algorithm with polynomial running time. Their
approximation algorithm uses the notion of \emph{isolating cuts}, \ie the
minimum cut separating a terminal from all other terminals. They prove that the
union of the $k-1$ smallest isolating cuts yields a valid multiterminal cut with
the desired approximation ratio. The currently best known approximation
algorithm by Buchbinder~\etal\cite{buchbinder2013simplex} uses linear program
relaxation to achieve an approximation ratio of $1.323$.

While the multiterminal cut problem is NP-hard, it is \emph{fixed-parameter
tractable} (FPT), parameterized by the multiterminal cut weight $\wgt(G)$.
Marx~\cite{marx2006parameterized} proves that the multiterminal cut problem is
FPT and Chen~\etal\cite{chen2009improved} give the first FPT algorithm with a
running time of $4^{\wgt(G)}\cdot n^{\Oh{1}}$, later improved by
Xiao~\cite{xiao2010simple} to $2^{\wgt(G)}\cdot n^{\Oh{1}}$ and by
Cao~\etal\cite{cao20141} to $1.84^{\wgt(G)}\cdot n^{\Oh{1}}$. However, to the
best of our knowledge, there is no actual implementation for
any~of~these~algorithms. 

The minimum $s$-$t$-cut problem and its equivalent counterpart, the maximum
$s$-$t$-flow problem~\cite{ford1956maximal} were first formulated by
Harris~\etal\cite{harris1955fundamentals}. Ford and
Fulkerson~\cite{ford1956maximal} gave the first algorithm for the problem with a
running time of $\Oh{mn\wgt}$. One of the fastest known algorithms in practice
is the push-relabel algorithm of Goldberg and Tarjan~\cite{goldberg1988new} with
a running~time~of~$\Oh{mn\log(n^2/m)}$.

Problems related to the minimum multiterminal cut problem also appear in the
data mining community, namely the very similar and heavily studied \emph{seed
expansion problem}, for which the aim is to find ground-truth clusters when
given a small subset of the cluster vertices. In contrast to the minimum
multiterminal cut problem, these clusters might overlap. There is a multitude of
approaches adding and removing vertices
greedily~\cite{andersen2006communities,clauset2005finding,luo2008exploring,mislove2010you}.
PageRank~\cite{page1999pagerank} is reported to be well suited for the
problem~\cite{kloumann2014community} and there are multiple approaches that aim
to make PageRank perform even
better~\cite{andersen2006local,bian2017many,leskovec2010empirical}. Another
approach is to use machine learning methods such as
geometric~\cite{ye2017learning} or relational~\cite{macskassy2003simple}
neighborhood classifiers.

Closely related to the problem is also the global minimum cut problem, which is
discussed in Part~\ref{p:mincut} of this work. In this chapter, we adapt some of the
reductions discussed there that are applicable to the minimum multiterminal cut
problem and use them to reduce the size of the problem.

Our work on the multiterminal cut problem has the following \emph{main
contributions}: We engineer existing as well as new data reduction rules for the
minimum multiterminal cut problem with $k$ terminals. These reductions are used
within a branch-and-reduce framework as well as to boost the performance of an
ILP formulation for the problem. Through extensive experiments we show that
kernelization has a significant impact on both, the branch-and-reduce framework
as well as the ILP formulation. Our experiments also show a clear trade-off:
combining reduction rules with the ILP is very fast for problems which have a
small kernel but a high cut value and the fixed-parameter tractable
branch-and-reduce algorithm is highly efficient when the cut value is small.
Using this observation we combine the branch-and-reduce framework with an ILP
formulation and solve subproblems using the solver better suited to the
subproblem in question. In addition, we present a local search algorithm that
can significantly improve a given solution to the multiterminal cut problem.
Overall, we obtain algorithms that are multiple orders of magnitude faster than
the ILP formulation which is de facto standard to solve the problem to
optimality. Additionally, we give an inexact algorithm that gives high-quality
solutions to hard problems in reasonable time, but does not give an optimality
guarantee. 

%=====================================================================
\section{Preliminaries}\label{c:mc:s:preliminaries}

\subsection{Basic Concepts}
Let $G = (V, E, w)$ be a weighted undirected graph with vertex set $V$, edge set
$E \subset V \times V$ and non-negative edge weights $c: E \rightarrow
\mathbb{N}$.  We use the same terminology to describe graphs as in
Parts~\ref{p:mincut}~and~\ref{p:gp} of this dissertation. A \emph{$k$-cut}, or
\emph{multicut}, is a partitioning of $V$ into $k$ disjoint non-empty blocks,
\ie $V_1 \cup \dots \cup V_k = V$. The weight of a $k$-cut is defined as the
weight sum of all edges crossing block boundaries, \ie $c(E \cap
\bigcup_{i<j}V_i\times V_j)$.

\subsection{Multiterminal Cuts}

A \emph{multiterminal cut} for a graph $G=(V,E)$ with $k$ terminals $T =
\{t_1,\shorterDots,t_k\}$ is a multicut with~$t_1 \in V_1,\shorterDots,t_k \in
V_k$. Thus, a multiterminal cut pairwisely separates all terminals from each
other. The edge set of the multiterminal cut with minimum weight of $G$ is
called $\cut(G)$ and the associated optimal partitioning of vertices is denoted
as $\vopt = \{\vopt_1, \dots, \vopt_k\}$. $\cut$ can be seen as the set of all
edges that cross block boundaries in $\vopt$, \ie $\cut(G) = \bigcup \{e = (u,v)
\mid \vopt_u \neq \vopt_v\}$. The weight of the minimum multiterminal cut is
denoted as $\wgt(G) = c(\cut(G))$. At any point in time, the best currently
known upper bound for $\wgt(G)$ is denoted as $\bestwgt(G)$ and the best
currently known multiterminal cut is denoted as $\bestcut(G)$. If graph $G$ is
clear from the context, we omit it in the notation. There may be multiple
minimum multiterminal cuts, however, we aim to find one multiterminal cut with
minimum weight.
%For two vertices $s$ and $t$, we denote $\wgt(G,s,t)$ as the smallest
%multiterminal cut of $G$, where $s$ and $t$ are on different blocks.

In this paper we use \emph{minimum s-T-cuts}. For a vertex $s$ (\emph{source})
and a non-empty vertex set $T$ (\emph{sinks}), the minimum s-T-cut is the
smallest cut in which $s$ is one side of the cut and all vertices in $T$ are on
the other side. This is a generalization of minimum s-t-cuts that allows
multiple vertices in $t$ and can be easily replaced by a minimum s-t-cut by
connecting every vertex in $T$ with a new super-sink by infinite-capacity edges.
We denote the capacity of a minimum-s-T-cut, \ie the sum of weights in the
smallest cut separating $s$ from $T$, by $\lambda(G,s,T)$.

The examples in Figures~\ref{c:mc:fig:example_mtcut}~and~\ref{c:mc:fig:st_cuts}
show graphs with $4$ terminals each.
The minimum s-T-cut for each terminal with $T$ being the set of all terminals is
shown in red and the minimum multiterminal cut is shown in blue. We can see that
any $k-1$ minimum s-T-cuts (in red) separate all terminals and are thus a valid
multiterminal cut. 
%When \emph{clustering} a graph, we are looking for \emph{blocks} of nodes
%$V_1$,\ldots,$V_k$ that partition $V$, i.e., $V_1\cup\cdots\cup V_k=V$ and
%$V_i\cap V_j=\emptyset$ for $i\neq j$. The parameter $k$ is usually not given
%in advance.
In our algorithm we use \emph{graph contraction} and \emph{edge deletions}.
Given an edge $e = (u, v) \in E$, we define $G/e$ to be the graph  
after \emph{contracting} $e$. In the contracted graph, we delete vertex $v$ and
all incident edges. For each edge $(v, x) \in E$, we add an edge $(u, x)$ with
$c(u, x) = c(v, x)$ to $G$ or, if the edge already exists, we give it the edge
weight $c(u,x) + c(v,x)$. For the \emph{edge deletion} of an edge $e$, we define
$G-e$ as the graph $G$ in which $e$ has been removed. Other vertices and edges
remain the same.

For a given multiterminal cut $S$, the graph $G\backslash S$ splits $G$ into $k$
blocks as defined by the cut edges in $S$, each containing exactly one terminal.
Let the residual $R(t_i)$ be the connected component of $G \backslash S$
containing $t_i$ and $\delta(t_i) = |E(R(t_i), V \backslash R(t_i))|$ be the
edges in $S$ incident to $t_i$.

\begin{figure}[t]
   \centering
   \includegraphics[width=.7\textwidth]{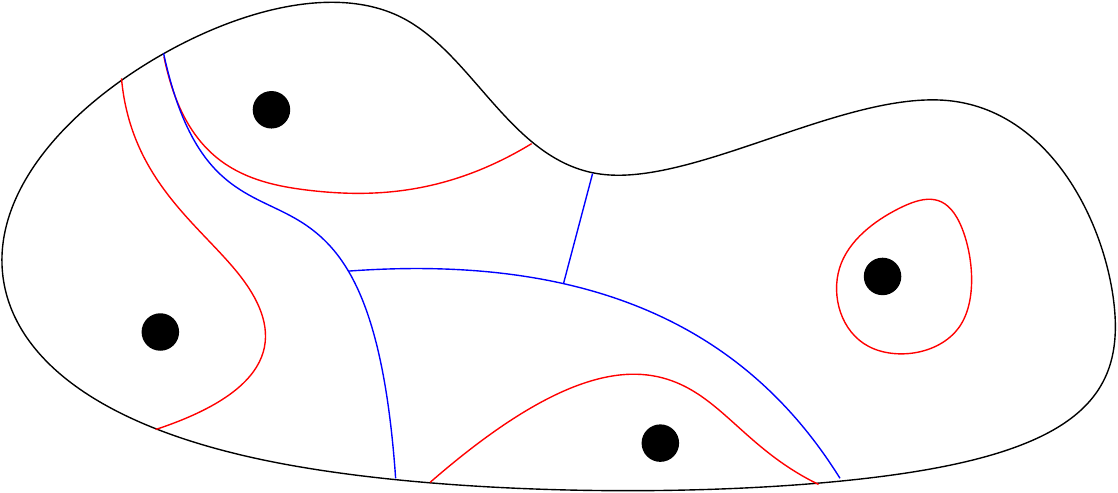}
   \caption{\label{c:mc:fig:example_mtcut}Graph with $4$ terminals. Minimum
   $s$-$T$-cut for each terminal shown in red, minimum multiterminal cut $\cut$
   shown in blue.}
 \end{figure}

\section{Branch and Reduce for Multiterminal Cut}
\label{c:mc:s:algorithm}

In this section we give an overview of our approach to find the optimal
multiterminal cut in large graphs. Our algorithm combines kernelization
techniques with an engineered bounded search. 

We begin by finding all connected components of~$G$. We can then look at all
connected components independently from each other, as there is a trivial cut of
weight $0$ between different connected components. If a connected component
contains only one terminal $t$, it can be separated from all other terminals by
using the whole connected component as the block $\vopt_t$ belonging to terminal
$t$. Due to it being not connected to any other terminals, the cut value is $0$.
If a connected component contains no terminals, the result $\wgt$ is identical
no matter which block $\vopt$ the connected component belongs to. For a
connected component $C$ with two terminals $s$ and $t$, we can run a minimum
s-t-cut algorithm on $C$ to find the minimum cut. The optimal blocks $\vopt_s$
and $\vopt_t$ then consist of the two sides of the s-t-cut. On a connected
component with more than two terminals, the problem is
NP-hard~\cite{dahlhaus1994complexity}. We run our branch and reduce algorithm on
this component. As those runs are completely independent, we only look at one
connected component in the following and disregard the rest of the graph for
now. 

\begin{figure}[t]
  \centering
  \includegraphics[width=.4\textwidth]{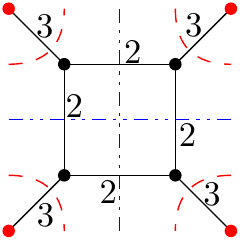}
  \caption{\label{c:mc:fig:st_cuts}Simple graph with $4$ terminals, in which minimum
  s-T-cuts (red) are different from minimum multiterminal cut (blue).}
\end{figure}

For a graph $G$, Dahlhaus~\etal~\cite{dahlhaus1994complexity} show that the sum
of minimum s-T-cut weights minus the heaviest of them is an upper bound
$\bestwgt$ of the weight of the minimum multiterminal cut, as denoted in Equation~\ref{c:mc:eq:upper}.

\begin{equation} \label{c:mc:eq:upper}
\wgt(G) \leq \bestwgt(G) = \sum_{s\in T} \lambda(G,s,T\backslash\{s\}) - \arg\,max_{s \in T}
\lambda(G,s,T\backslash\{s\})
\end{equation}

The intuition behind Equation~\ref{c:mc:eq:upper} is that any set of $t-1$
s-T-cuts pairwisely separates all terminals and is thus a valid multiterminal
cut of weight $\bestwgt(G)$. However, $\bestwgt(G)$ is not necessarily the value
of the minimum multiterminal cut $\cut(G)$, as the minimum s-T-cuts might share
edges -- which then do not need to be counted twice -- and the minimum
multiterminal cut might be smaller. For a simple example where the minimum
multiterminal cut is smaller than any set of $t-1$ minimum multiterminal cuts,
see Figure~\ref{c:mc:fig:st_cuts}, where any set of $t-1$ minimum s-T-cuts
result in a multiterminal cut of weight $9$ whereas the minimum multiterminal
cut has a weight of $8$.

Dahlhaus~\etal~\cite{dahlhaus1994complexity} also give a lower bound for the
minimum multiterminal cut: as $\lambda(G,s,T\backslash\{s\})$ is by definition
minimal, $\cut$ has at least as
many edges incident to terminal $s$ as $\lambda(G,s,T\backslash\{s\})$. As this
is true for every terminal (and every edge is only incident to two vertices),
$\cut(G) \cdot 2 \geq \sum_{s \in T} \lambda(G,s,T\backslash\{s\})$, so that
$\cut(G) \geq \sum_{s \in T} \lambda(G,s,T\backslash\{s\}) / 2$.  

In our algorithm, we keep a queue $\queue$ of problems. A problem in $\queue$
consists of a graph $G_{\queue}$, a set of terminals, the upper and lower bound
for $\wgt(G_{\queue})$ and the weight sum of all deleted edges in $G_{\queue}$.
When our algorithm is initialized, $\queue$ is initialized with a single
problem, whose graph is $G$ and whose set of terminals is $T$. The problem has
$0$ deleted edges and its lower and upper bound for $\wgt(G)$ can be set as
previously described. As the problem is currently the only one, the global upper
bound $\hat{\wgt(G)}$ is equal to the upper bound of $G$. Over the course of the
algorithm, we repeatedly take a problem from $\queue$ and check whether we can
reduce the graph size using our kernelization techniques outlined in
Section~\ref{c:mc:ss:kernel}. When possible, we perform the kernelization and
push the kernelized problem to $\queue$. Otherwise, we branch on an edge $e$
adjacent to one of the terminals. 

The kernelization techniques detailed in Section~\ref{c:mc:ss:kernel} reduce the size
of the graph by finding edges that are (1) either guaranteed to be in a minimum
multiterminal cut or (2) guaranteed not to be part of at least one minimum
multiterminal cut. As we only want to find a single multiterminal cut with
minimum sum of edge weights, we can delete edges in (1) and contract edges in
(2). 

In Section~\ref{c:mc:ss:branch} we detail the branching procedure which is used
if these reduction techniques are unable to find any further reduction
possibilities. For any edge~$e$, either it is in the multiterminal cut or it is
not. We create two subproblems for $G$: $G/e$ and $G-e$. We aim to find the
minimum multiterminal cut on either. We also give an enhanced branching scheme
that aims to increase performance by creating more than two subproblems. Further
details on the branching and edge selection are given in
Section~\ref{c:mc:ss:branch}. 

We compute upper and lower
bounds for each of the problems and follow the branches whose lower bounds are
lower than $\bestwgt$, the best cut weight previously found. In
Section~\ref{c:mc:ss:queue_impl} we discuss queue implementation and whether using a
priority queue to first process 'promising' problems is useful in practice. We
employ shared-memory parallelism by having multiple threads pull problems from
$\queue$.

In Section~\ref{c:mc:s:local} we describe our local search algorithm which can
improve a given solution by iteratively moving vertices on the original graph
until the solution reaches a local optimum. This allows us to significantly
lower $\bestwgt$ and therefore improve performance by pruning subproblems whose
lower bound is $\geq \bestwgt$.

We then give a variant of our algorithm in Section~\ref{c:mc:s:inexact} that
does not guarantee optimality but is able to solve significantly larger
instances. This variant aggressively prunes problems that are unlikely to
improve the solution quality and performs additional data reductions that do not
have an optimality guarantee but can significantly shrink the graph while
maintaining the most promising regions therein.

\subsection{Kernelization}
\label{c:mc:ss:kernel}

We now show how to reduce the size of our graph to make the problem more
manageable. This is achieved by contracting edges that are guaranteed not to be
in the minimum multiterminal cut and deleting edges that are guaranteed to be in
it. Before we detail the kernelization rules we show that edges not in $\cut$
can be safely contracted and edges in $\cut$ can be safely deleted if we store
the weight sum of all deleted edges so far. The kernelization rules given in the
following and outlined in Figure~\ref{c:mc:fig:various_reductions} are used to
identify such edges.

\begin{figure}[t!]
  \centering
  \includegraphics[width=\textwidth]{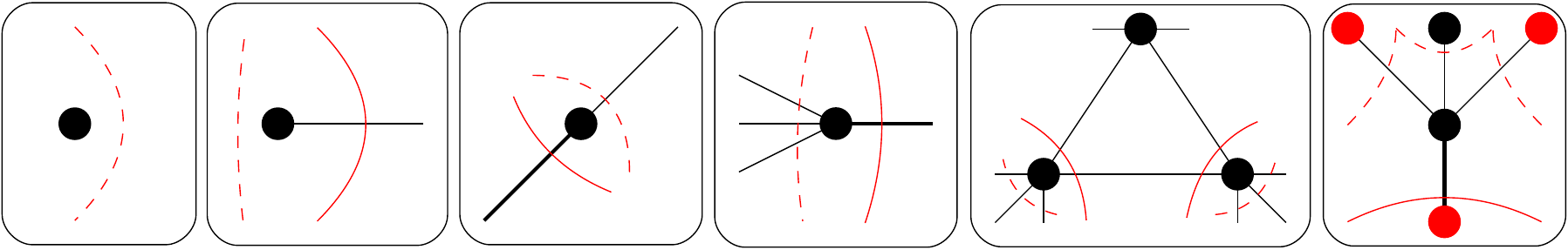}
  \caption{\label{c:mc:fig:various_reductions} Reductions. Solid line cannot be
  minimal as dashed line has smaller weight: (1) \textttA{IsolatedVertex}, (2)
  \textttA{DegreeOne}, (3) \textttA{DegreeTwo} , (4) \textttA{HeavyEdge},
  (5)~\textttA{HeavyTriangle} and (6) \textttA{SemiEnclosedVertex}.}
\end{figure}

\begin{lemma}\label{c:mc:lem:cont} \cite{cao20141} If an edge $e = (u,v) \in G$ is
  guaranteed not to be in at least one multiterminal cut $\cut(G)$ (\ie $P_u =
  P_v$), we can contract $e$ and $\wgt(G/e) = \wgt(G)$.
\end{lemma}

\begin{proof}
    As $e \not\in \cut(G)$, $\cut(G/e)$ is equal to $\cut(G)$ and thus still has
    weight equal to $c(\cut(G)) = \wgt(G)$. As an edge contraction only removes
    cuts and does not create any new cuts, an edge contraction can not lower the
    weight of the minimum multiterminal cut, \ie $\wgt(G/e) \geq \wgt(G)$. As
    $\cut(G/e)$ has weight $\wgt(G)$, it is a multiterminal cut in $G/e$ with
    weight equal to $\wgt(G)$. Thus it is definitely a minimum multiterminal cut
    with weight $\wgt(G)$.
\end{proof}

Lemma~\ref{c:mc:lem:cont} allows us to reduce the graph size by contracting an edge
if we can prove that both incident vertices are in the same partition in
$\vopt$. The lemma can be generalized trivially to contract a connected vertex
set by applying the lemma to each edge connecting two vertices of the set.

\begin{lemma}\label{c:mc:lem:del} \cite{cao20141} If an edge $e = (u,v) \in E$ is
  guaranteed to be in a minimum multiterminal cut, \ie there is a minimum
  multiterminal cut $\cut(G)$ in which $P_u \neq P_v$, we can delete $e$ from
  $G$ and $\cut(G - e)$ is still a valid minimum multiterminal cut.
\end{lemma}

\begin{proof}
  Let $\wgt(G)$ be the weight of the minimum multiterminal cut $\cut(G)$. We
  show that for an edge $e \in \cut(G)$, $\wgt(G - e) = \wgt(G) - c(e)$. Thus,
  we can delete $e$ (and thus replace $G$ with $G - e$) and store the weight of
  the deleted edge. Obviously, $\cut(G - e)$ has weight equal to $\wgt(G) -
  c(e)$, as we just deleted $e$ and all other edges in $\cut(G)$ are still in
  $G$. By deleting $e$, the weight of any multiterminal cut can be decreased by
  at most $c(e)$ (as a multiterminal cut is a set of edges and $e$ can at most
  be once in that set). As $\wgt(G)$ is minimal by definition and no cut weight
  can be decreased by more than $c(e)$, $G - e$ cannot have a minimum
  multiterminal cut with weight $< \wgt(G) - c(e)$. Thus, $\cut(G-e)$ is a
  minimum multiterminal cut of $G-e$ with weight $\wgt(G-e)$.
\end{proof}

\subsubsection{Minimum Isolating Cuts}
When we look at a problem, we first solve the minimum s-T-cut problem for each
terminal $s \in T$. This results in one or multiple minimum cuts that separate
$s$ from all other terminals. We call the side of the cut containing $s$ the
\emph{isolating cut} of $s$. Dahlhaus~\etal\cite{dahlhaus1994complexity} prove
that there is a minimum multiterminal cut $\cut$ in which the complete isolating
cut is in $\vopt_s$. Thus, according to Lemma~\ref{c:mc:lem:cont}, we can contract all
vertices of the largest isolating cut into a single vertex. In
Figure~\ref{c:mc:fig:example_mtcut} this would result in contracting the red areas
into their respective terminals. This contraction might result in edges
connecting terminals. Such an edge $e=(u,v)$, where both $u$ and $v$ are
terminal vertices is guaranteed to be a part of $\cut(G)$. This comes from the
fact that we know $\vopt_u \neq \vopt_v$, \ie $u$ and $v$ are not in the same
block in the minimum multiterminal cut, as both $u$ and $v$ are terminals.
According to Lemma~\ref{c:mc:lem:del} they can therefore be deleted.

\subsubsection{Local Contraction}
We aim to find edges that cannot be part of the minimum multiterminal cut. If we
find an edge that can be contracted, we mark it in a union find data
structure~\cite{gabow1985linear}. This union-find structure is initialized with
each vertex as its own block, an edge contraction then merges the two blocks of
incident vertices. After all kernelization criteria are tested, we contract all
edges that are marked as contractible. As a contraction might open up new
contractions in its neighborhood, we run the contraction routines until they do
not find any more contractible edges. To ensure low overhead, we run only the
first iteration completely and subsequently check only the neighborhoods of
vertices that were changed in the previous iteration.

\subsubsection{Low-Degree Vertices~\cite{cao20141}}
Figures~\ref{c:mc:fig:various_reductions}.(1),~\ref{c:mc:fig:various_reductions}.(2)
and \ref{c:mc:fig:various_reductions}.(3) show examples of non-terminal vertices
with degree $\leq 2$ that can be contracted while maintaining a minimum
multiterminal cut. A non-terminal vertex with no neighbors
(\textttA{IsolatedVertex}) can be deleted as there is no incident edge that could
affect a cut. For a non-terminal vertex $v$ with only one adjacent edge $e =
(v,x)$ (\textttA{DegreeOne}), $e$ can not be part of the minimum multiterminal
cut $\cut(G)$. Any multiterminal cut that contains $e$ can be improved by
removing $e$ and moving $v$ to the block of its neighbour $x$. Thus, we can
contract $e$. On a non-terminal vertex with two adjacent edges $e_1$ and $e_2$
(\textttA{DegreeTwo}), the heavier edge $e_1$ can not be part of $\cut$, as
replacing it with $e_2$ improves the cut value. If $e_1$ and $e_2$ have equal
weight, we can contract either (but not both!). These reductions are performed
in a single run, which we denote as \textttA{Low}.

\subsubsection{Heavy Edges}
We now look to contract heavy edges. The reductions
\textttA{HeavyEdge}~(\ref{c:mc:fig:various_reductions}.(4)) and
\textttA{HeavyTriangle}~(\ref{c:mc:fig:various_reductions}.(5)) were originally
used for the minimum cut
problem~\cite{Chekuri:1997:ESM:314161.314315,henzinger2018practical,padberg1990efficient}
and are described in Part~\ref{p:mincut}~(Section~\ref{p:mincut:ss:pr}) of this work.
We adapt them and transfer them to the minimum multiterminal cut problem. 

\textttA{HeavyEdge} says that an edge $e=(u,v)$ which has a weight of at least
half of the total edge degree of a non-terminal vertex $u$ can be contracted, as
any cut containing $e$ can instead also contain all other edges incident to $u$.
If $e$ has at least $\frac{deg(u)}{2}$, all other incident edges together are
not heavier.

For a \textttA{HeavyTriangle} with vertices $v_1$, $v_2$ and $v_3$, we can relax
the condition. If for two of the vertices the incident triangle edges together
are at least as heavy as all other incident edges, we can contract those, as
shown in Figure~\ref{c:mc:fig:various_reductions}.(5). Each of the continuous lines
between $v_1$ and $v_2$ can be replaced with the dashed line without increasing
the value of the cut. Thus, in every case ($v_3$ can be on either side of the
cut), there is an optimal solution in which $v_1$ and $v_2$ are in the same
block. Thus, we can contract the edge according to Lemma~\ref{c:mc:lem:cont}.

The condition \textttA{SemiEnclosed}, shown in
Figure~\ref{c:mc:fig:various_reductions}.(6), considers a vertex $v$ which is mostly
incident to terminal vertices. Let $t_1$ be the terminal that is most strongly
connected to $v$ and $t_2$ the terminal with second highest connection strength.
Now say that $v$ is contracted into any terminal vertex. All edges connecting
$v$ with other terminals are then edges connecting terminals and are guaranteed
to be in $\cut$. If $c(v, t_1) > c(v, t_2) + \sum_{u \in V \backslash T} c(v,
u)$, \ie $(v, t_1)$ is heavier than the sum of $(v, t_2)$ and all edges
connecting $v$ with non-terminals, we can contract $v$ into $t_1$. This follows
from the fact that the weight of cut edges incident to $v$ is at most $deg(v) -
c(v, t_1)$ if $v$ is in the same block as $t_1$. If we instead add $v$ to the
block of $t_2$ (or any other block), at most $c(v, t_2) + \sum_{u \in V
\backslash T} c(v, u)$ of the edges incident to $v$ would not be part of the
cut. Thus, the locally best choice is contracting $v$ into $t_1$. As this does
not affect any other graph areas, this choice is guaranteed to be optimal. We
check both \textttA{HeavyEdge} and \textttA{SemiEnclosed} in a single run labelled
\textttA{High}. \textttA{HeavyTriangle} is checked in a run named
\textttA{Triangle}.

\subsubsection{High-connectivity edges}

The \emph{connectivity} of an edge $e=(u,v)$ is the value of the minimum cut
separating $u$ and $v$. If an edge has connectivity $\geq \bestwgt(G)$, it is
guaranteed that $u$ and $v$ are in the same block in $\vopt$, as there can not
be a multiterminal cut that separates them and has value $<\bestwgt(G)$. We can
therefore contract $u$ and $v$. We now show how to improve the bound.

\begin{lemma}\label{c:mc:lem:noi} If for a graph $G$ with best known multiterminal
  cut $\bestcut(G)$, vertices $u$ and $v$ belong to different connected
  components of the minimum multiterminal cut $G\backslash \cut$, then
  $\lambda(u,v)+\frac{\sum_{i \in \{1,\dots,t\}\backslash \max_2}
  \lambda(G,t_i,T\backslash\{t_i\})}{4}  \leq |\wgt(G)|$, where $\max_2$ is the
  set of the indices of the largest $2$ values
  $\lambda(G,t_i,T\backslash\{t_i\})$ in the sum. 
\end{lemma}

In order to prove Lemma~\ref{c:mc:lem:noi} we first prove the following useful claim:

  \begin{claim}\label{t:blocks} For any two nodes $u$ and $v$, if $u$ and $v$
    belong to different connected components of $G \backslash \cut(G)$, then
    $\lambda(u,v) \leq \frac{\sum_{i \in \{1,\dots,k\}}\delta (R(t_i))}{4} +
    \frac{\delta(R(u)) + \delta({R(v)})}{4}$, where $\delta$ are the weighted
    node degrees in the quotient graph corresponding to $\cut(G)$ and $R(x)$ is
    the block of a vertex $x$ as defined by the cut $\cut(G)$.
  \end{claim}

  \begin{proof}
    Let $G_R$ be the contracted graph where every block $R(t_i)$ in $G$ is
    contracted into a single vertex and let $|S(u,v)|$ be a minimum $u$-$v$-cut
    in $G_R$. By definition of the minimum cut $\lambda(u,v)$, $\lambda(u,v)
    \leq |S(u,v)|$.
    
    For every vertex $w \in G_R$ that does not represent a block that contains
    either $u$ or $v$, at most $\frac{\text{deg}(w)}{2}$ edges are in
    $|S(u,v)|$. This follows directly from the assumption that $|S(u,v)|$ is
    minimal. If more than $\frac{\text{deg}(w)}{2}$ edges incident to $w$ are in
    $|S(u,v)|$, moving $w$ to the other side of the cut would give a better cut.
    Thus, at most half of the edges incident to $w$ are in $|S(u,v)|$. 
    
    We can not make this argument for the blocks containing $u$ and $v$, as
    potentially all edges incident to their blocks could be in the minimum
    multiterminal cut. Thus, $2 \cdot |S(u,v)| \leq \frac{\sum_{i \in
    \{1,\dots,k\}} \delta(R(t_i))}{2} + \frac{\delta(R(u))}{2} +
    \frac{\delta(R(v))}{2}$. The factor $2$ on the left side is caused by the
    fact that every edge is incident to two blocks. As we do not know the
    multiterminal cut $S$, we need to assume that they could be the blocks with
    the largest cuts $\delta(R(t_i))$. Dividing each side by $2$ finishes the
    proof.
  \end{proof}

  \begin{claim}\label{t:wgt} For any two nodes $u$ and $v$, if $u$ and $v$
    belong to different connected components of $G \backslash \cut(G)$, then
    $\lambda(u,v) + \frac{\sum_{i \in \{1,\dots,k\}} \delta(R(t_i))}{4} \leq
    \wgt$.
  \end{claim}
  \begin{proof}
    Using Claim~\ref{t:blocks} we know that $\lambda(u,v) + \frac{\sum_{i \in
    \{1,\dots,k\}} \delta(R(t_i))}{4} \leq \frac{\sum_{i \in \{1,\dots,k\}}
    \delta(R(t_i))}{2}$. By definition of $\delta$, $\frac{\sum_{i \in
    \{1,\dots,k\}} \delta(R(t_i))}{2} = \wgt(G)$.
  \end{proof}

  We now use Claims~\ref{t:blocks}~and~\ref{t:wgt} to prove Lemma~\ref{c:mc:lem:noi}.
  \begin{proof}
    %Using Claim~\ref{t:blocks} we know that $\lambda(u,v) \leq \frac{\sum_{i
    %\in \{1,\dots,k\}} \delta(t_i)}{4} + \frac{\sum_{\max_2} \delta(t_i)}{4}$
    %for $u$ and $v$ in different blocks. 
    Let vertices $u$ and $v$ be in different blocks. Then \\    
    $\lambda(u,v)+\frac{\sum_{i \in \{1,\dots,t\}\backslash \max_2}
    \lambda(G,t_i,T\backslash\{t_i\})}{4} \leq$\\ 
    $\lambda(u,v)+\frac{\sum_{i \in \{1,\dots,t\}\backslash \max_2}
    \delta(R(t_i))}{4} \leq \\ \frac{\sum_{i \in \{1,\dots,t\}\backslash \max_2}
    \delta(R(t_i))}{2} = \wgt(G)$.

    The first inequality follows from the fact that $\lambda$ is per definition
    the minimal cut separating $t$ from $T\backslash\{t_i\}$ and thus
    $\lambda(G,t_i,T\backslash\{t_i\}) \leq \delta(R(t_i))$.

    Thus, we know that if $\lambda(u,v)+\frac{\sum_{i \in
    \{1,\dots,t\}\backslash \max_2} \lambda(G,t_i,T\backslash\{t_i\})}{4} >
    \wgt(G)$, $u$ and $v$ are in the same block and the edge connecting them can
    be safely contracted.
\end{proof}

We can use Lemma~\ref{c:mc:lem:noi} to contract edges whose high connectivity ensures
that they are not in a minimum multiterminal cut. For any edge $e=(u,v)$, if
$\lambda(u,v)+\frac{\sum_{i \in \{1,\dots,k\}\backslash \max_2}
\lambda(G,t_i,T\backslash\{t_i\})}{4} > |\wgt(G)| > |\bestwgt(G)|$, $u$ and $v$
are guaranteed to be in the same block in $\vopt$. Thus, we can contract them
into a single vertex according to Lemma~\ref{c:mc:lem:cont}. This
condition is denoted as \textttA{HighConnectivity}. 

As it is very expensive to compute the connectivity for every edge, we use the
CAPFOREST algorithm of
Nagamochi~\etal~\cite{nagamochi1992computing,nagamochi1994implementing}~(see
Section~\ref{p:mincut:ss:noi} for a description of the CAPFOREST algorithm)
to compute a connectivity lower bound $\gamma(u,v)$ for each edge $e = (u,v)$ in
$G$ in near-linear time. If the lower bound $\gamma(u,v)$ fulfills
Equation~\ref{eq:cap}, we can use Lemma~\ref{c:mc:lem:noi} to contract $u$ and $v$.

\begin{equation}
  \label{eq:cap} \gamma(u,v) > |\hat{\wgt}| - \frac{\sum_{i \in \{1,\dots,k\}\backslash \max_2} \lambda(G,t_i,T\backslash\{t_i\})}{4}
\end{equation}

\subsubsection{Articulation Points}
\label{ss:ap}

Let $\phi \in V$ be an articulation point in $G$ whose removal disconnects the
graph into multiple connected components. For any of these components that does
not contain any terminals, we show that all vertices in the component can be
contracted into $\phi$.

\begin{lemma} \label{c:mc:lem:ap} For an articulation point $\phi$ whose removal
  disconnects the graph $G$ into multiple connected components $(G_1, \dots,
  G_p)$ and a component $G_i$ with $i \in \{1,\dots,p\}$ that does not contain
  any terminals, no edge in $G_i$ or connecting $G_i$ with $\phi$ can be part of
  $\cut(G)$.
\end{lemma}

\begin{proof}
  Let $e$ be an edge that connects two vertices in $\{V_i \cup \phi\}$. Assume
  $e \in \cut(G)$, \ie $e$ is part of the minimum multiterminal cut of $G$. This
  means that vertices in $\{V_i \cup \phi\}$ are not all in the same block. By
  changing the block affiliation of all vertices in $\{V_i \cup \phi\}$ to
  $\vopt(\phi)$ we can remove all edges connecting vertices in $\{V_i \cup
  \phi\}$ from the multiterminal cut, thus decrease the weight of the
  multiterminal cut by at least $c(e)$. As $\phi$ is an articulation point,
  $G_i$ is only connected to the rest of $G$ through $\phi$ and thus no new
  edges are introduced to the multiterminal cut. This is a contradiction to the
  minimality of $\cut(G)$, thus no edge $e$ that connects two vertices in $\{V_i
  \cup \phi\}$ is in the minimum multiterminal cut $\cut(G)$.
\end{proof}

Using Lemmas~\ref{c:mc:lem:cont}~and~\ref{c:mc:lem:ap} we can contract all
components that contain no terminals into the articulation point $\phi$. All
articulation points of a graph can be found in linear time using an algorithm by
Tarjan and Vishkin~\cite{tarjan1985efficient} based on depth-first search. The
algorithm performs a depth-first search and checks in the backtracking step
whether for a vertex $v$ there exists an alternative path from the parent of $v$
to every of descendant of $v$. If there is no alternative path, $v$ is an
articulation point in $G$. This reduction rule is denoted as \textttA{ArticulationPoints}.

\subsubsection{Equal Neighborhoods}

In many cases, the resulting graph of the reductions contains groups of vertices
that are connected to the same neighbors. If the neighborhood and respective
edge weights of two vertices are equal, we can use
Lemmas~\ref{c:mc:lem:cont}~and~\ref{c:mc:lem:nbrhd} to contract them into a single vertex.

\begin{lemma} \label{c:mc:lem:nbrhd} For two vertices $v_1$ and $v_2$ with $\{N(v_1)
  \backslash v_2\} = \{N(v_2) \backslash v_1\}$ where for all $v \in \{N(v_1)
  \backslash v_2\}$, $c(v_1, v) = c(v_2, v)$, there is at least one minimum
  multiterminal cut where $\vopt(v_1) = \vopt(v_2)$.
\end{lemma}

\begin{proof}
  Let $C$ be a partitioning of the vertices in $G$ with $C(v_1) \not = C(v_2)$,
  let $\zeta$ be the corresponding cut, where $e=(u,v) \in \zeta$, if $C(u) \neq
  C(v)$ and let $cc(v)$ be the total weight of edges in $\zeta$ incident to a
  vertex $v \in V$. W.l.o.g. let $v_2$ be the vertex with $cc(v_2) \geq
  cc(v_1)$. We analyze this in two steps: We assume that when moving $v_2$ to
  $C(v_1)$ that all edges incident to $v_2$ in its old location are removed from
  $\zeta$, which drops the weight of $\zeta$ by $cc(v_2)$ and then all edges
  incident to $v_2$ in its new location are added to $\zeta$, which is exactly
  $cc(v_1)$ by the conditions of the lemma. Thus the weight of $\zeta$ changes
  by $cc(v_1) - cc(v_2) \le 0$. If the edge $e_{12} = (v_1, v_2)$ exists, both
  $cc(v_1)$ and $cc(v_2)$ are furthermore decreased by $c(e_{12})$, as the edge
  connecting them is not a cut edge anymore. As we only moved the block
  affiliation of $v_2$, the only edges newly introduced to $\zeta$ are edges
  incident to $v_2$. Thus, the total weight of the multiterminal cut was not
  increased by moving $v_1$ and $v_2$ into the same block and we showed that for
  each cut $\zeta$, in which $C(v_1) \not = C(v_2)$ there exists a cut of equal
  or better value in which $v_1$ and $v_2$ are in the same block. Thus, there
  exists at least one multiterminal cut where $\vopt(v_1) = \vopt(v_2)$. 
\end{proof}

We detect equal neighborhoods for all vertices with neighborhood size smaller or
equal to a constant $c_{N}$ using two linear time routines. To detect
neighboring vertices $v_1$ and $v_2$ with equal neighborhood, we sort the
neighborhood vertex IDs including edge weights by vertex IDs (excluding the
respective other vertex) for both $v_1$ and $v_2$ and check for equality. To
detect non-neighboring vertices $v_1$ and $v_2$ with equal neighborhood, we
create a hash of the neighborhood sorted by vertex ID for each vertex with
neighborhood size smaller or equal to $c_{N}$. If hashes are equal, we check
whether the condition for contraction is actually fulfilled. As the
neighborhoods to sort only have constant size, they can be sorted in constant
time and thus the procedures can be performed in linear time. We perform both
tests, as the neighborhoods of neighboring vertices contain each other and
therefore do not result in the same hash value; and non-neighboring vertices are
not in each others neighborhood and therefore finding them requires checking the
neighborhood of every neighbor, which results in a large search space. We set
$c_{N} = 5$, as in most cases where we encountered equal neighborhoods they are
in vertices with neighborhood size $\leq 5$. This reduction rule is denoted as \textttA{EqualNeighborhoods}

\subsubsection{Maximum Flow from Non-terminal Vertices}

Let $v$ be an arbitrary vertex in $V \backslash T$, \ie a non-terminal vertex of
$G$. Let $(V_v, V \backslash V_v)$ be the largest minimum isolating cut that
separates $v$ from the set of terminal vertices $T$. Lemma~\ref{c:mc:lem:flow}
shows that there is at least one minimum multiterminal cut $\cut(G)$ so that
$\forall x \in V_v: \vopt(x) = \vopt(v)$ and thus $V_v$ can be contracted into a
single vertex.

\begin{lemma} \label{c:mc:lem:flow} Let $v$ be a vertex in $V \backslash T$. Let
  $(V_v, V \backslash V_v)$ be the largest minimum isolating cut of $v$ and the
  set of terminal vertices $T$ and let $\lambda(G, v, T)$ be the weight of the
  minimum isolating cut $(V_v, V \backslash V_v)$. There exists at least one
  minimum multiterminal cut $\cut(G)$ in which $\forall x \in V_v: \vopt(x) =
  \vopt(v)$.
\end{lemma}

\begin{proof}
  As $(V_v, V \backslash V_v)$ is a minimum isolating cut with the terminal set
  as sinks, we know that no terminal vertex is in $V_v$. Assume that $\cut(G)$
  cuts $V_v$, \ie there is a non empty vertex set $V_C \in V_v$ so that $\forall
  x \in V_C: \vopt(x) \not\in \vopt(v)$. We will show that the existence of such
  a vertex set contradicts the minimality of $\cut(G)$.
  Figure~\ref{c:mc:fig:flow} gives an illustration of the vertex sets defined
  here.

    \begin{figure}
      \centering
      \includegraphics[width=3.3cm]{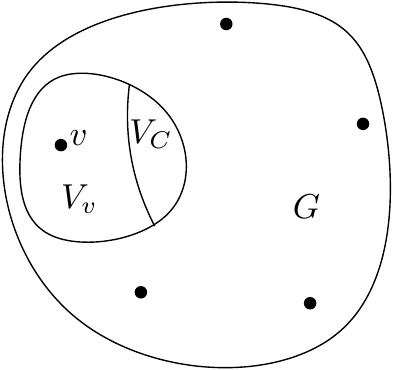}
      \caption{\label{c:mc:fig:flow} Illustration of vertex sets in
      Lemma~\ref{c:mc:lem:flow}.}
    \end{figure}

  Due to the minimality of the minimum isolating cut $(V_v, V \backslash V_v)$,
  we know that $c(V_C, V_v \backslash V_C) \geq c(V_C, V \backslash V_v)$ (i.e.
  the connection of $V_C$ to the rest of $V_v$ is at least as strong as the
  connection of $V_C$ to $(V \backslash V_v)$), as otherwise we could remove
  $V_C$ from $V_v$ and find an isolating cut of smaller size.
  
  We now show that by changing the block affiliation of all vertices in $V_C$ to
  $\vopt(v)$, \ie removing all vertices from the set $V_C$, we can construct a
  multiterminal cut of equal or better cut value. By changing the block
  affiliation of all vertices in $V_C$ to $\vopt(v)$, we remove all edges
  connecting $V_C$ to $(V_v \backslash V_C)$ from $\cut(G)$ and potentially
  more, if there were edges in $\cut(G)$ that connect two vertices both in
  $V_C$. At most, the edges connecting $V_C$ and $(V \backslash V_v)$ are newly
  added to $\cut(G)$. As $c(V_C, V_v \backslash V_C) \geq c(V_C, V \backslash
  V_v)$, the cut value of $\cut(G)$ will be equal or better than previously.
  Thus, there is at least one multiterminal cut in which $V_C$ is empty and
  therefore $\forall x \in V_v: \vopt(x) = \vopt(v)$.
\end{proof}

We can therefore solve a maximum $s$-$T$-flow problem for an arbitrary
non-terminal vertex $s$ and the set of all terminals $T$ and contract the source
side of the largest minimum isolating cut into a single vertex, using
Lemmas~\ref{c:mc:lem:cont}~and~\ref{c:mc:lem:flow}. These flow problems can be
solved embarrassingly parallel, in which every processor solves an independent
maximum $s$-$T$-flow problem for a different non-terminal vertex $v$.

While it is possible to run a flow problem from every vertex in $V$, this is
obviously not feasible as it would entail excessive running time overheads.
Promising vertices to use for maximum flow computations are either high degree
vertices or vertices with a high distance from every terminal. High degree
vertices are promising, as due to their high degree it is more likely that we
can find a minimum isolating cut of weight less than their degree. Vertices that
have a high distance to all terminals are on 'the edge of the graph',
potentially in a subgraph only weakly connected to the rest of the graph.
Running a maximum flow then allows us to contract this subgraph. In every
iteration, we run $5$ flow problems starting from high-distance vertices and $5$
flow problems starting from high-degree vertices. This reduction rule is denoted
as \textttA{NonTerminalFlows}.

\subsubsection{Other Reductions}

We now briefly present other reductions that we tried, but have been
unsuccessful since they are either subsumed by other reductions or have
excessive running time overheads in comparison to how many contractions are
found.

\paragraph*{Bridges.}

A \emph{bridge} is an edge whose removal disconnects a graph $G$ into two blocks
$G_1$ and $G_2$. For every bridge, if one block has no terminals, we can
contract this block into a single vertex, similar to the articulation point
reduction in Section~\ref{ss:ap}. As the two incident vertices of a bridge are
always articulation points, the articulation point reduction already finds these
contractions and finding bridges is not faster than finding articulation. If
both blocks contain terminals, branching on this bridge allows the disconnection
of the problem in one of the subproblems. However, we found that even if bridges
like this exist in the original graph, generally they are already added to the
multiterminal cut by other routines and thus all contractions that the bridge
reduction finds are already found by other reductions.

\paragraph*{Semi-isolated Clique.}

If a graph contains a clique $C$ that has only a weak connection to the rest of
the graph, no minimum multiterminal cut can cut $C$ and we can thus contract it
into a single vertex. We employed the maximal clique search algorithm of
Eppstein~\etal\cite{eppstein2011listing} with aggressive pruning of
cliques that have a strong connection to non-clique vertices. However, as
maximal clique detection is an NP-complete problem~\cite{eppstein2011listing},
even aggressive pruning still entails excessive running time. Also, as the
instances contracted with all reductions usually have increased average degree
and decreased diameter, almost all cliques in them have a large amount of edges
to other vertices and thus there are only few semi-isolated cliques to be found.

\subsection{Branching Tree Search}
\label{c:mc:ss:branch}

If our reductions detailed in Section~\ref{c:mc:ss:kernel} are unable to contract any
edges in $G$, we branch on an edge adjacent to a terminal.
Figure~\ref{c:mc:fig:branch_edge} shows an example in which we chose an edge to
branch on. For each edge, there are two options: either the edge is part of the
minimum multiterminal cut $\cut(G)$ or it is not.
Lemmas~\ref{c:mc:lem:cont}~and~\ref{c:mc:lem:del} show that we can delete an edge that is
in $\cut(G)$ and contract an edge that is not. Therefore we can build two
subproblems, $G/e$ and $G-e$ and add them to the problem queue $\queue$. This
branching scheme for the multiterminal cut problem was introduced by
Chen~\etal\cite{chen2009improved} in their FPT algorithm for the problem.

\begin{figure}[t]
  \centering
  \includegraphics[width=.9\textwidth]{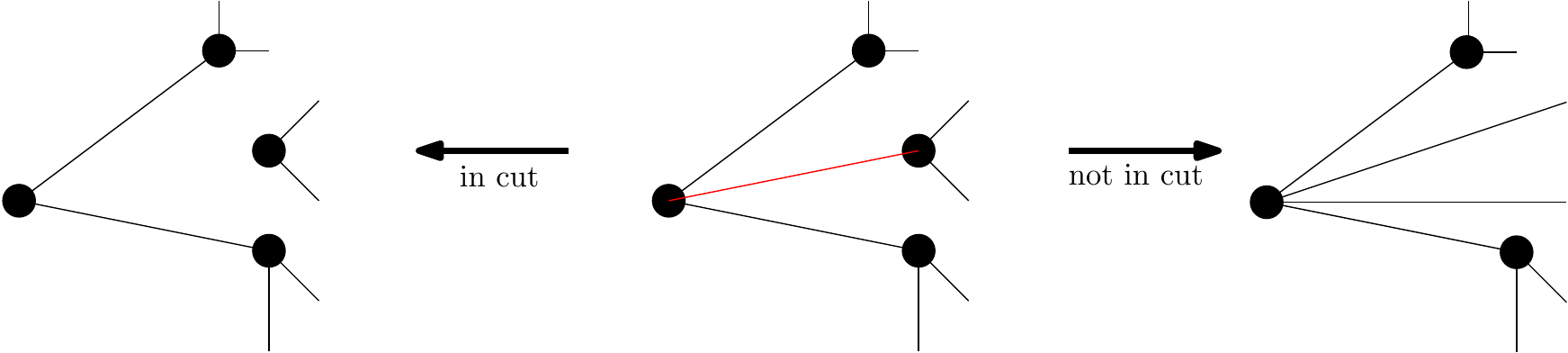}
  \caption{\label{c:mc:fig:branch_edge}Branch on marked edge $e$ in $G$, adjacent to
  a terminal - create two subproblems, (1) $G/e$ and (2) $G-e$.}
\end{figure}

Both of the subproblems will have a higher lower bound and thus, the algorithm
will definitely terminate. For $G-e$, we know that $e$ is adjacent to a terminal
$s$ but not an edge connecting two terminals (otherwise it would have been
deleted). Thus, it is in exactly one minimum s-T-cut
$\lambda(G,s,T\backslash\{s\})$. For the lower bound, we half the value of all
minimum s-T-cuts. Deleting the edge indicates that it is definitely part of the
multiterminal cut. Thus, we increased the lower bound by $c(e) - \frac{c(e)}{2}
= \frac{c(e)}{2}$. 

For $G/e$ we know that $e=(s,v)$ is part of the largest isolating cut of $s$ (as
we contract the largest isolating cut). In $G/e$ terminal $s$ is guaranteed to
have a larger minimum s-T-cut, as otherwise there would be an isolating cut of
equal value containing $v$, which contradicts the maximality of the contracted
isolating cut. Thus $\lambda(G/e,s,T\backslash\{s\}) >
\lambda(G,s,T\backslash\{s\})$ and no other minimum s-T-cut can be decreased by
an edge contraction. Thus, the lower bound of $\wgt(G/e)$ and $\wgt(G-e)$ are
both guaranteed to be higher than the lower bound of $\wgt(G)$.

\subsubsection{Vertex Branching} \label{c:mc:ss:vtxbranch}

When our multiterminal cut algorithm is initialized, it only has a single problem
containing the whole graph $G$. While independent minimum isolating cuts are
computed in parallel, most of the shared-memory parallelism comes
from the embarrassingly parallel solving of different problems on separate
threads. When branching, we select the highest degree vertex that is
adjacent to a terminal and branches on the heaviest edge connecting it to one of
the terminals. The algorithm thus creates only up to two subproblems and is
still not able to use the whole machine. 

We now give a new branching rule that overcomes these limitations by selecting
the highest degree vertex incident to at least one terminal and use it to create
multiple subproblems to allow for faster startup. Let $x$ be the vertex used for
branching, $\{t_1,\dots,t_i\}$ for some $i \geq 1$ be the adjacent terminals of
$x$ and $w_M$ be the weight of the heaviest edge connecting $x$ to a terminal.
We now create up to $i + 1$ subproblems as follows: 

For each terminal $t_j$ with $j \in \{1,\dots,i\}$ with $c(x, t_j) + c(x, V
\backslash T) > w_M$ create a new problem $P_j$ where edge $(x, t_j)$ is
contracted and all other edges connecting $x$ to terminals are deleted. Thus in
problem $P_j$, vertex $x$ belongs to block $\vopt(t_j)$. If $c(x, t_j) + c(x, V
\backslash T) \leq w_M$, \ie the weight sum of the edges connecting $x$ with
$t_j$ and all non-terminal vertices is not heavier than $w_M$, the assignment to
block $\vopt(t_j)$ cannot be optimal and thus we do not need to create the
problem $P_j$, also called \emph{pruning} of the problem. The following
Lemma~\ref{lem:prune} proves the correctness of this pruning step.

\begin{lemma} \label{lem:prune} Let $G = (V,E)$ be a graph, $T \subseteq V$ be
  the set of terminal vertices in $G$, and $x \in V$ be a vertex that is
  adjacent to at least one terminal and for an $i \in \{1,\dots,|T|\}$ be the
  index of the terminal for which $e_i=(x,t_i)$ is the heaviest edge connecting
  $x$ with any terminal. Let $w_M$ be the weight of $e_i$. If there exists a
  terminal $t_j$ adjacent to $x$ with $j \in \{1,\dots, |T|\}$ with $c(x, t_j) +
  c(x, V \backslash T) \geq w_M$, there is at least one minimum multiterminal
  cut $\cut(G)$ so that $\vopt(x) \not = j$, \ie $x$ is not in block $j$. 
  \end{lemma}
  
  \begin{proof}
    If $\vopt(x) = i$, \ie $x$ is in the block of the terminal it has the
    heaviest edge to, the sum of cut edge weights incident to $x$ is $\leq E(x)
    - w_M$, as edge $e_i$ of weight $w_M$ is not a cut edge in that case. If
    $\vopt(x) = j$, \ie $x$ is in the block of terminal $j$, the sum of cut edge
    weights incident to $x$ is $\geq E(x) - (c(x, V \backslash T) + c(x, t_j))$,
    as all edges connecting $x$ with other terminals than $t_j$ are guaranteed
    to be cut edges. As $c(x, t_j) + c(x, V \backslash T) \geq w_M$, even if all
    non-terminal neighbors of $x$ are in block $j$, the weight sum of incident
    cut edges is not lower than when $x$ is placed in block $i$. As the block
    affiliation of $x$ can only affect its incident edges, the cut value of
    every solution that sets $\vopt(x) = j$ would be improved or remain the same
    by setting $\vopt(x) = i$.
  \end{proof}

If $c(x, V \backslash T) > w_M$ and $i < |T|$, we also create problem $P_{i+1}$,
in which all edges connecting $x$ to a terminal are deleted. This problem
represents the assignment of $x$ to a terminal that is not adjacent to it. We
add each subproblem whose lower bound is lower than the currently best found
solution $\bestwgt$ to the problem queue $\queue$. As we create up to $|T|$
subproblems, this allows for significantly faster startup of the algorithm and
allows us to use the whole parallel machine after less time than before.

\subsubsection{Edge Selection}

In Section~\ref{c:mc:ss:branch_impl} we evaluate the following edge selection
strategies: \textttA{HeavyEdge} branches on the heaviest edge incident to a
terminal; \textttA{HeavyVertex} branches on the edge between the heaviest vertex
that is in the neighborhood of a terminal to that terminal; \textttA{Connection}
searches the vertex that is most strongly connected to the set of terminals and
branches on the heaviest edge connecting it to a terminal;
\textttA{NonTerminalWeight} branches on the edge between the vertex that has the
highest weight sum to non-terminal vertices and the terminal it is most strongly
connected with; and \textttA{HeavyGlobal} branches on the heaviest edge in
the~graph.  

\subsubsection{Sub-problem Order}

In Section~\ref{c:mc:ss:queue_impl} we evaluate the following comparators for the
priority queue $\queue$, \ie the order in which we look at the problems. A
straightforward indicator on whether a problem can lead to a low cut is the
current lower and upper bound for the best solution. If a problem has a good
lower bound, it has a large potential for improvement and if it has a good upper
bound there is already a good solution, potentially close to an even better
solution in the neighborhood. Thus, \textttA{LowerBound} orders the problems by
their lower bound and solves the ones with a better lower bound first while
\textttA{UpperBound} first examines problems with a lower bound. In either
comparator, the respective other bound acts as a tie breaker. \textttA{BoundSum}
orders problems by the sum of their upper and lower bound.

\textttA{BiggerDistance} first examines problems in which the distance between
lower and upper bound is very large. The conceptual idea is that those problems
still have many unknowns and thus could be interesting to examine. In contrast
to that, \textttA{LowerDistance} first examines problems with a lower distance of
upper and lower bound, as those branches will likely have fewer subbranches.
Following the same idea, \textttA{MostDeleted} first explores the problem that
has the highest deleted weight. \textttA{SmallerGraph} orders the graphs by the
number of vertices and first examines the smallest graph. As over the course of
the algorithm a terminal might become isolated (as all incident edges were
deleted), not all problems have the same amount of terminals. The isolated
terminals are inactive and thus do not need any more flow computations.
\textttA{FewTerminals} first examines problems with a lower number of active
terminals. As there are many solutions with the same amount of terminals, ties
are broken using \textttA{LowerBound}.

\subsection{Parallel Branch and Reduce}

Our algorithm is shared-memory parallel. As we maintain a queue of problems
which are independent from each other, we can run our algorithm embarassingly
parallel. The shared-memory priority queue of problems is implemented as a
separate queue for each thread to pull from. When a thread adds a problem to the
priority queue, it is added to a random queue with minimum queue size. In order
to exploit data and cache locality, we add problems to the queue of the local
thread if it is one of the queues with minimum size. Additionally, we fix each
thread to a single CPU thread in order to actually use those locality benefits.
In the beginning of the algorithm, there is only a single problem, which would
leave all except for one processors idle, potentially for a long time, as we
have to solve $k$ flow problems on the whole (potentially very large) graph.
Thus, if there are idle processors, we distribute the flow problems over
different threads.

\subsection{Combining Kernelization with ILP}

Multiterminal cut problems are generally solved in practice using integer linear
programs~\cite{nabieva2005whole}. The following ILP formulation is adapted from
our implementation for the graph partitioning problem in
Section~\ref{c:gp:ss:ilp} (without balance constraints) and implemented using
Gurobi 8.1.1. It is functionally equal to \cite{nabieva2005whole}.

%\begin{ceqn}
%  \begin{equation*}
%\hspace*{-4cm}    \min \sum_{\{u,v\} \in E} e_{uv} \cdot c(\{u,v\})
%  \end{equation*}
%  \begin{equation*}
%\hspace*{-2.5cm} \forall \{u,v\} \in E, \forall k: e_{uv} \geq x_{u,k} - x_{v,k}\\
%  \end{equation*}
%  \begin{equation*}
% \hspace*{-2.5cm}\forall \{u,v\} \in E, \forall k: e_{uv} \geq x_{v,k} - x_{u,k}\\
%  \end{equation*}
%  \begin{equation*}
% \hspace*{-4.5cm}   \forall v \in V: \sum_k x_{v,k} = 1\\
%  \end{equation*}
%  \begin{equation*}
%  \hspace*{-4.65cm}  \forall i,j:  x_{t_i,j} = [ i = j ] \\
%  \end{equation*}
%\end{ceqn} 

\begin{ceqn}
  \begin{align}
    \min \sum_{\{u,v\} \in E}& e_{uv} \cdot c(\{u,v\})\\
 \forall \{u,v\} \in E, \forall k&: e_{uv} \geq x_{u,k} - x_{v,k}\\
 \forall \{u,v\} \in E, \forall k&: e_{uv} \geq x_{v,k} - x_{u,k}\\
  %\forall k&: \sum_{v \in V} x_{v,k} \omega(v) \leq L_{\text{max}}\\
    \forall v \in V&: \sum_k x_{v,k} = 1 \\
    \forall i,j \in \{1,\dots,|T|\}&: x_{t_i,j} = [i = j ] 
  \end{align}
\end{ceqn}

Here, $x_{u,k}$ is $1$ iff vertex $u$ is in $V_k$ and $0$ otherwise
and $e_{uv}$ is $1$ iff $(u,v)$ is a cut edge. We use this ILP formulation as a
baseline of comparison. Additionally, we also create a new algorithm that
combines the kernelization of our algorithm with integer linear programming.
Using flow computations and kernelization routines, we are able to significantly
reduce the size of most graphs while still preserving the minimum multiterminal
cut. As the complexity of the ILP depends on the size of the graph and the
complexity of the branch-and-reduce algorithm also depends on the value of the
cut, this is fast on graphs with a high cut value in which the kernelization
routines can reduce the graph to a very small size but with a large cut value.
In the following, our algorithm \textttA{Kernel+ILP} first runs kernelization
until no further reduction is possible and then solves the problem using the
above integer linear programming formulation. We also integrate the ILP
formulation directly into the branch-and-reduce solver as an alternative to a
branching operation. We hereby give the ILP solver a time limit and if it is
unable to find an optimal solution within the time limit, we instead perform a
branch operation. In Section~\ref{ss:exp_ilp} we study which subproblems to
solve with an ILP first.

\subsection{Local Search} \label{c:mc:s:local}

Our algorithm for the multiterminal cut problem prunes problems which cannot
result in a solution which is better than the best solution found so far.
Therefore, even though it is a deterministic algorithm that will output the
optimal result when it terminates, performing greedy optimization on
intermediate solutions allows for more aggressive pruning of problems that
cannot be optimal. Additionally, the algorithm has reductions that depend on the
value of $\bestwgt(G)$ and can thus contract more edges if the cut value
$\bestwgt(G)$ is lower.

For a subproblem $H=(V_H, E_H)$ with solution $\rho$, the original graph $G =
(V_G, E_G)$ and a mapping $\pi: V_G \rightarrow V_H$ that maps each vertex in
$V_G$ to the vertex in $V_H$ that encompasses it, we can transfer the solution
$\rho$ to a solution $\gamma$ of $G$ by setting the block affiliation of every
vertex $v \in V_G$ to $\gamma(v) := \pi(\rho(v))$. The cut value of the solution
$c(\gamma)$ is defined as the sum of weights of the edges crossing block
boundaries, \ie the sum of edge weights where the incident vertices are in
different blocks. Let $\xi_i(V_G)$ be the set of all vertices $v \in V_G$ where
$\gamma(v) = i$.

We introduce the following greedy optimization operators that can transform
$\gamma$ into a better multiterminal cut solution $\gamma_{\text{IMP}}$ with
$c(\gamma_{\text{IMP}}) < c(\gamma)$.

\subsubsection{Kernighan-Lin Local Search}

Kernighan and Lin~\cite{lin1973effective} give a heuristic for the
traveling-salesman problem that has been adapted to many hard optimization
problems~\cite{sandersschulz2013,traff2006direct,xu2005survey,dorigo2006ant},
where each vertex $v \in V_G$ is assigned a gain $g(v) = \max_{i \in
\{i,\dots,|T|\}, i \not = \gamma(v)} \sum c(v, \xi_i(V_G)) - c(v,
\xi_{\gamma(v)}(V_G))$, \ie the improvement in cut value to be gained by moving
$v$ to another block, the best connected other block. We perform runs where we
compute the gain of every vertex that has at least another neighbor in a
different block and move all vertices with non-negative gain. Additionally, if a
vertex $v$ has a negative gain, we store its gain and associated best connected
other block. For any neighbor $u$ of $v$ that also has the same best connected
other block, we check whether $g(w) + g(v) + 2 \cdot c(v, u) > 0$, \ie moving
both $u$ and $v$ at the same time is a positive gain move. If it is, we perform
the move.

\subsubsection{Pairwise Maximum Flow}

For any pair of blocks $1 \leq i < j \leq |T|$ where $c(\xi_i(V_G), \xi_j(V_G))
> 0$, \ie there is at least one edge from block $i$ to block $j$, we can create
a maximum $s$-$t$ flow problem between them: we create a graph $F_{ij}$ that
contains all vertices in $\xi_i(V_G)$ and $\xi_j(V_G)$ and all edges that
connect these vertices. 

Let $H$ be a problem graph created by performing reductions and branching on the
original graph $G$. All vertices that are encompassed in the same vertex in
problem graph $H$ as the terminals $i$ and $j$ are hereby contracted into the
corresponding terminal vertex. We perform a maximum $s$-$t$-flow between the two
terminal vertices and re-assign vertex assignments in $\gamma$ according to the
minimum $s$-$t$-cut between them. As we only model blocks $\xi_i(V_G)$ and
$\xi_j(V_G)$, this does not affect other blocks in $\gamma$. In the first run we
perform a pairwise maximum flow between every pair of blocks $i$ and $j$ where
$c(\xi_i(V_G), \xi_j(V_G)) > 0$ in random order. We continue on all pairs of
blocks where $c(\xi_i(V_G), \xi_j(V_G))$ was changed since the end of the
previous maximum flow iteration between them.

We first perform Kernigham-Lin local search until there is no more improvement,
then pairwise maximum flow until there is no more improvement, followed by
another run of Kernigham-Lin local search. As pairwise maximum flow has
significantly higher running time, we spawn a new thread to perform the
optimization if there is a CPU core that is not currently utilized.

\subsection{Fast Inexact Algorithm} \label{c:mc:s:inexact}  

Our algorithm for the multiterminal cut problem in an exact algorithm, \ie when
it terminates the output is guaranteed to be optimal. As the multiterminal cut
problem is NP-complete~\cite{dahlhaus1994complexity}, it is not feasible to
expect termination in difficult instances of the problem. In fact, in difficult
instances the algorithm often does not terminate with an optimal result but runs
out of time or memory and returns the best result found up to that point. Thus,
it makes sense to relax the optimality constraint and aim to find a high-quality
(but not guaranteed to be optimal) solution faster.

\begin{figure}[t] 
  \centering
  \includegraphics[width=.48\textwidth]{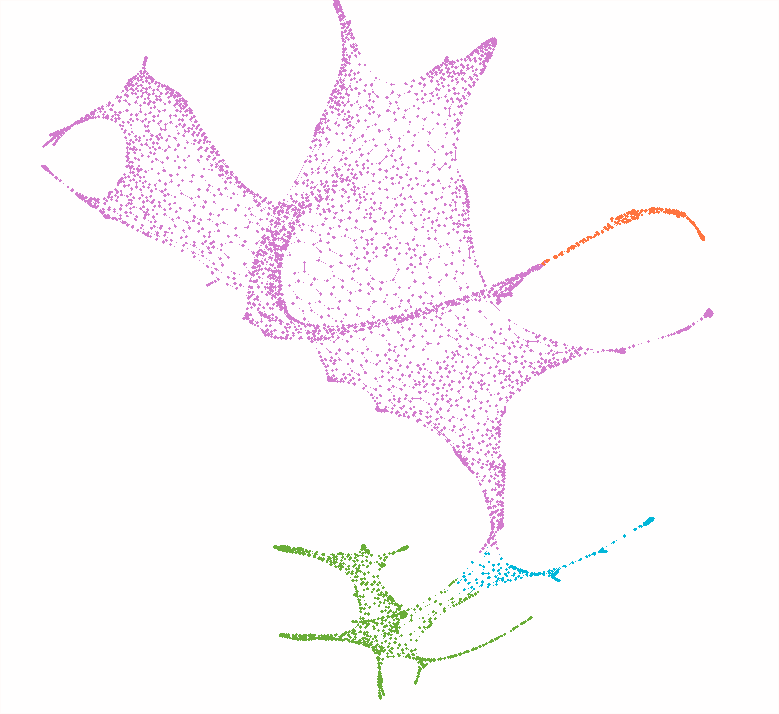}
  \includegraphics[width=.48\textwidth]{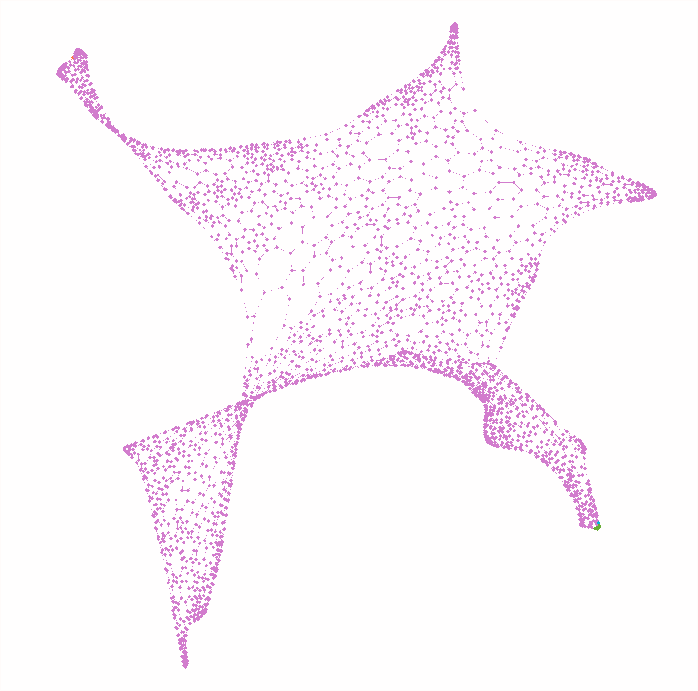}
  \caption{\label{c:mc:fig:uk} Minimum multiterminal cut for graph
  \textttA{uk}~\cite{soper2004combined} and four terminals - on original graph
  (left) and remaining graph at time of first branch operation (right),
  visualized using Gephi-0.9.2~\cite{bastian2009gephi}.}
\end{figure}

A key observation is that in many problems, most, if not all vertices
that are not already contracted into a terminal at the time of the first branch
will be assigned to a few terminals whose weighted degree at that point is
highest. See Figure~\ref{c:mc:fig:uk} for an example with $4$ terminals (selected
with high distance to each other) on graph \textttA{uk} from the Walshaw Graph
Partitioning Archive~\cite{soper2004combined}. As we can see, at the time of the
first branch (right figure), most vertices that are not assigned to the pink
terminal in the optimal solution are already contracted into their respective
terminals. The remainder is mostly assigned to a single terminal. As we can
observe similar behavior in many problems, we propose the following heuristic
speedup operations:

Let $\delta \in (0, 1)$ be a contraction factor and $T_H$ be the set of all
terminals that are not yet isolated in graph $H$. In each branching operation on
an intermediate graph $H$, we delete all edges around the $\lceil \delta \cdot
|T_H| \rceil$ terminals with lowest degree. Additionally, we contract all
vertices adjacent to the highest degree terminal that are not adjacent to any
other terminal into the highest degree terminal. This still allows us to find
all solutions in which no more vertices were added to the lowest degree
terminals and the adjacent vertices are in the same block as the highest degree
terminals.

Additionally, in a branch operation on vertex $v$, we set a maximum branching
factor $\beta$ and only create problems where $v$ is contracted into the $\beta$
adjacent terminals it has the heaviest edges to and one problem in which it is
not contracted into either adjacent terminal. This is based on the fact that all
other edges connecting $v$ to other terminals will be part of the multiterminal
cut and the greedy assumption that it is likely that the optimal solution does
not contain at least one of these heavy edges. By default, we set $\delta = 0.1$
and $\beta = 5$.

\section{Experiments and Results} \label{c:mc:s:experiments}

We now perform an experimental evaluation of the multiterminal cut algorithms
described in this chapter. This is done in the following order: first analyze
the impact of algorithmic components on our branch-and-reduce algorithm in a
non-parallel setting, i.e.~we compare different variants for branching edge
selection, priority queue comparator and the effects of the kernelization
operators. We then report the speedup over ILP formulation on a variety of
graphs. Lastly, we perform experiments on protein-protein interaction networks
and social, map and web graphs to compare the performance of different variants of
our algorithm.

This section describes experiments performed
for~\cite{henzinger2020sharedmemory}~and~\cite{henzinger2020faster},
where~\cite{henzinger2020sharedmemory} introduces our first algorithm for the
multiterminal cut problem and~\cite{henzinger2020faster} enhances this algorithm
by adding more reduction rules, improving the branching rule and including ILP
and local search into the algorithm. The previous sections of this chapter give
the full algorithm as described in both of our works. In the following we will
use the terminology of~\cite{henzinger2020faster}, where the preliminary
algorithm of~\cite{henzinger2020sharedmemory} is denoted as \vcbase{}, the full
algorithm is denoted as \exact{} and the inexact algorithm described in
Section~\ref{c:mc:s:inexact} is denoted as \inexact{}.

\vcbase{} is a shared-memory parallel branch-and-reduce algorithm that uses the
reduction rules \textttA{Low}, \textttA{High}, \textttA{Triangle} and
\textttA{HighConnectivity} to reduce the size of a graph instance and branches on
an edge incident to a terminal when this is not possible anymore. 

The \exact{} and \inexact{} algorithms additionally use the reduction rules
\textttA{ArticulationPoints}, \textttA{EqualNeighborhoods} and
\textttA{NonTerminalFlows}, create multiple subproblems when branching as
described in Section~\ref{c:mc:ss:vtxbranch} and integrate local search and ILP
into the algorithm.

\subsection{Experimental Setup and Methodology}

We implemented the algorithms using \CC-17 and compiled all codes using
g++-7.4.0 with full optimization (\textttA{-O3}). Our experiments are conducted
on three machine types: Machine A is a machine with two Intel Xeon Gold 6130
with 2.1GHz with 16 CPU cores each and $256$ GB RAM in total. Machine B is a
machine with two Intel Xeon E5-2643v4 with $3.4$ GHz with $6$ CPU cores each and
$1.5$ TB RAM in total. Machine C is a machine in the Vienna Scientific Cluster
with two Intel Xeon E5-2650v2 with $2.6$GHz with $8$ CPU cores each and $64$ GB
RAM in total.

We perform five repetitions per instance and report average running~time. In
this section we first describe experimental methodology. Afterwards, we evaluate
different algorithmic choices in our algorithm and then we compare our algorithm
to the state of the art. When we report a mean result we give the geometric mean
as problems differ strongly in result and time.

\begin{table}[t] \centering
  \small
   \caption{\label{t:multicutgraphs}Large Real-world Benchmark Instances. }
   \resizebox*{.49\textwidth}{!}{
   \begin{tabular}{| l | r | r |}
     \hline
     Graph & $n$& $m$\\
     \hline \hline
     \multicolumn{3}{|c|}{Social, Web and Map
     Graphs (1A)}\\
     \hline \hline
     bcsstk30 \cite{soper2004combined} & \numprint{28924} & $1.01M$\\
     ca-2010 \cite{bader2013graph} & $710K$ & $1.74M$\\
     ca-CondMat \cite{davis2011university} & \numprint{23133} &
     \numprint{93439}\\
     cit-HepPh \cite{davis2011university} & \numprint{34546} & $422K$\\
     eu-2005 \cite{BoVWFI} & $862K$ & $16.1M$\\
     higgs-twitter \cite{davis2011university} & $457K$ & $14.9M$\\
     in-2004 \cite{BoVWFI} & $1.38M$ & $13.6M$\\
     ny-2010 \cite{bader2013graph} & $350K$ & $855K$ \\
     uk-2002 \cite{BoVWFI} & $18.5M$ & $261M$ \\
     vibrobox \cite{soper2004combined} & \numprint{12328} & $165K$\\     
     \hline \hline
     
     \multicolumn{3}{|c|}{Social, Web and Map
     Graphs (1B)}\\
     \hline\hline
     598a \cite{soper2004combined} & $111K$ & $742K$ \\
     astro-ph \cite{davis2011university} & \numprint{16706} & $121K$\\
     caidaRouterLevel \cite{davis2011university} & $192K$ & $609K$\\
     citationCiteseer \cite{davis2011university} & $268K$ & $1.16K$\\
     cnr-2000 \cite{davis2011university} & $326K$ & $2.74M$\\
     coAuthorsCiteseer \cite{davis2011university} & $227K$ & $814K$ \\
     cond-mat-2005 \cite{davis2011university} & \numprint{40421} & $176K$\\
     coPapersCiteseer \cite{davis2011university} & $434K$ & $16.0M$\\
     cs4 \cite{soper2004combined} & \numprint{22499} & \numprint{43858} \\
     fe\_body \cite{soper2004combined} & \numprint{45087} & $164K$ \\
     NACA0015 \cite{davis2011university} & $1.04M$ & $3.11M$\\
     venturiLevel3 \cite{davis2011university} & $4.03M$ & $8.05M$ \\
     &&\\
     \hline
     
   \end{tabular}
   }
   \resizebox*{.49\textwidth}{!}{
   \begin{tabular}{| l | r | r |}
    \hline
    Graph & $n$& $m$\\
    \hline\hline
    \multicolumn{3}{|c|}{Protein-protein
    Interaction~\cite{szklarczyk2010string,szklarczyk2018string} (2)}\\
    \hline \hline
    Acidi. ferrivorans & \numprint{3093} & \numprint{5394}\\
    Agaricus bisporus & \numprint{11271} & \numprint{14636} \\
    Candida maltosa & \numprint{5948} & \numprint{19462} \\
    Escherichia coli & \numprint{4127} & \numprint{13488} \\
    Erinaceus europaeus & \numprint{19578} & \numprint{68066} \\
    Homo sapiens & \numprint{19566}& $324K$ \\
    Mesoplasma florum & \numprint{683} & \numprint{2365} \\
    S. cerevisiae & \numprint{6691} & \numprint{69809} \\
    Toxoplasma gondii & \numprint{7988} & \numprint{11779} \\
    Vitis vinifera & \numprint{29697} & \numprint{70206} \\
    \hline \hline
     \multicolumn{3}{|c|}{Map Graphs (3)} \\
     \hline\hline
     ak2010 \cite{bader2013graph} & \numprint{45292} & $109K$\\
     %ca2010 \cite{bader2013graph} & $710K$ & $1.74M$\\
     ct2010 \cite{bader2013graph} & \numprint{67578} & $168K$\\
     de2010 \cite{bader2013graph} & \numprint{24115} & \numprint{58028}\\
     hi2010 \cite{bader2013graph} & \numprint{25016} & \numprint{62063}\\
     luxembourg.osm \cite{davis2011university} & $115K$ & $120K$ \\
     me2010 \cite{bader2013graph} & \numprint{69518} & $168K$\\
     netherlands.osm \cite{davis2011university} & $2.22M$ & $2.44M$ \\
     nh2010 \cite{bader2013graph} & \numprint{48837} & $117K$\\
     nv2010 \cite{bader2013graph} & \numprint{84538} & $208K$\\
     %ny2010 \cite{bader2013graph} & $350K$ & $855K$ \\
     ri2010 \cite{bader2013graph} & \numprint{25181} & \numprint{62875}\\
     sd2010 \cite{bader2013graph} & \numprint{88360} & $205K$\\
     vt2010 \cite{bader2013graph} & \numprint{32580} & \numprint{77799}\\
     \hline
   \end{tabular}
   }
 \end{table}

\subsubsection{Instances}

We use multiple sets of instances to avoid overtuning the branch-and-reduce
algorithm. To analyze the impact of algorithmic components in
Sections~\ref{c:mc:ss:branch_impl} and \ref{c:mc:ss:queue_impl}, we generate random
hyperbolic graphs using the KaGen graph generator~\cite{funke2018communication}.
These graphs have $n = 2^{14} - 2^{18}$ and an average degree of $8$, $16$ and
$32$. For each graph size, we use three generated graphs and compute the
multiterminal cut, each with $k \in \{3,4,5,6,7\}$. We use random hyperbolic
graphs as they have power-law degree distribution and resemble a wide variety of
real-world networks. Additionally, we also use a family of weighted graphs from
the $10^{th}$ DIMACS implementation challenge~\cite{bader2013graph}. These
graphs depict US states, where a vertex depicts a census block and a weighted
edge denotes the length of the border between two blocks. We use the $10$ states
with the fewest census blocks (AK, CT, DE, HI, ME, NH, NV, RI, SD, VT). For each
state, we set the number of terminals $k \in \{3,4,5,6,7\}$. A multiterminal cut
on these graphs depicts the shortest border that respects census blocks and
separates a set of pre-defined blocks (or groups of blocks). Here, we use one
processor and set a timeout of $3$ minutes and a memory limit of $20$GiB.

As the instances generally do not have any terminals, we find random vertices
that have a high distance from each other in the following way: we start with a
random vertex $r$, run a breadth-first search starting at $r$ and select the
vertex $v$ encountered last as first terminal. While the number of terminals is
smaller than desired, we add another terminal by running a breadth-first search
from all current terminals and adding the vertex encountered last to the list of
terminals. We then run a bounded-size breadth-first search around each terminal
to create instances where the minimum multiterminal cut does not have $k - 1$
blocks consisting of just a single vertex each. The parameter $p \in (0,1)$ hereby bounds
the size of the terminal, \ie only up to $\frac{p \cdot n}{k}$ vertices are
added to each terminal. This results in problems in which well separated
clusters of vertices are partitioned and the task consists of finding a
partitioning of the remaining vertices in the boundary regions between already
partitioned blocks. This relates to clustering tasks, in which well separated
clusters are labelled and the task consists of labelling the remaining vertices
inbetween.

When comparing \textttA{Kernel+ILP} with \vcbase{} in
Sections~\ref{c:mc:ss:ppi}~and~\ref{c:mc:ss:social} on large instances, we use all $32$
cores of machine A (for the ILP as well as the branch and reduce framework).
Here, we set a time limit of $1$ hour and a memory limit of $250$GiB. Note that
is a soft limit, in which the algorithm finishes the current operation and exits
afterwards if the time or memory limit is reached. As many of these are very
large instances, most instances in this section are not solved to optimality.

In Section~\ref{c:mc:ss:ppi} we perform experiments on protein-protein
interaction networks (graph family 2) generated from the STRING protein
interaction database~\cite{szklarczyk2010string,szklarczyk2018string} by using
all edges they predict with a high certainty. We use the protein description to
assign functions (block terminal affiliations) to proteins (vertices). We use
the first occurence of a set of pre-defined function classes. For each graph, we
examine problems with the $4,5,6,7,8$ most often occuring functions and with all
(up to $15$, if all occuring in an organism) classes.

\begin{figure}[!t]
  \centering
  \begin{subfigure}{\textwidth}
    \includegraphics[width=\linewidth]{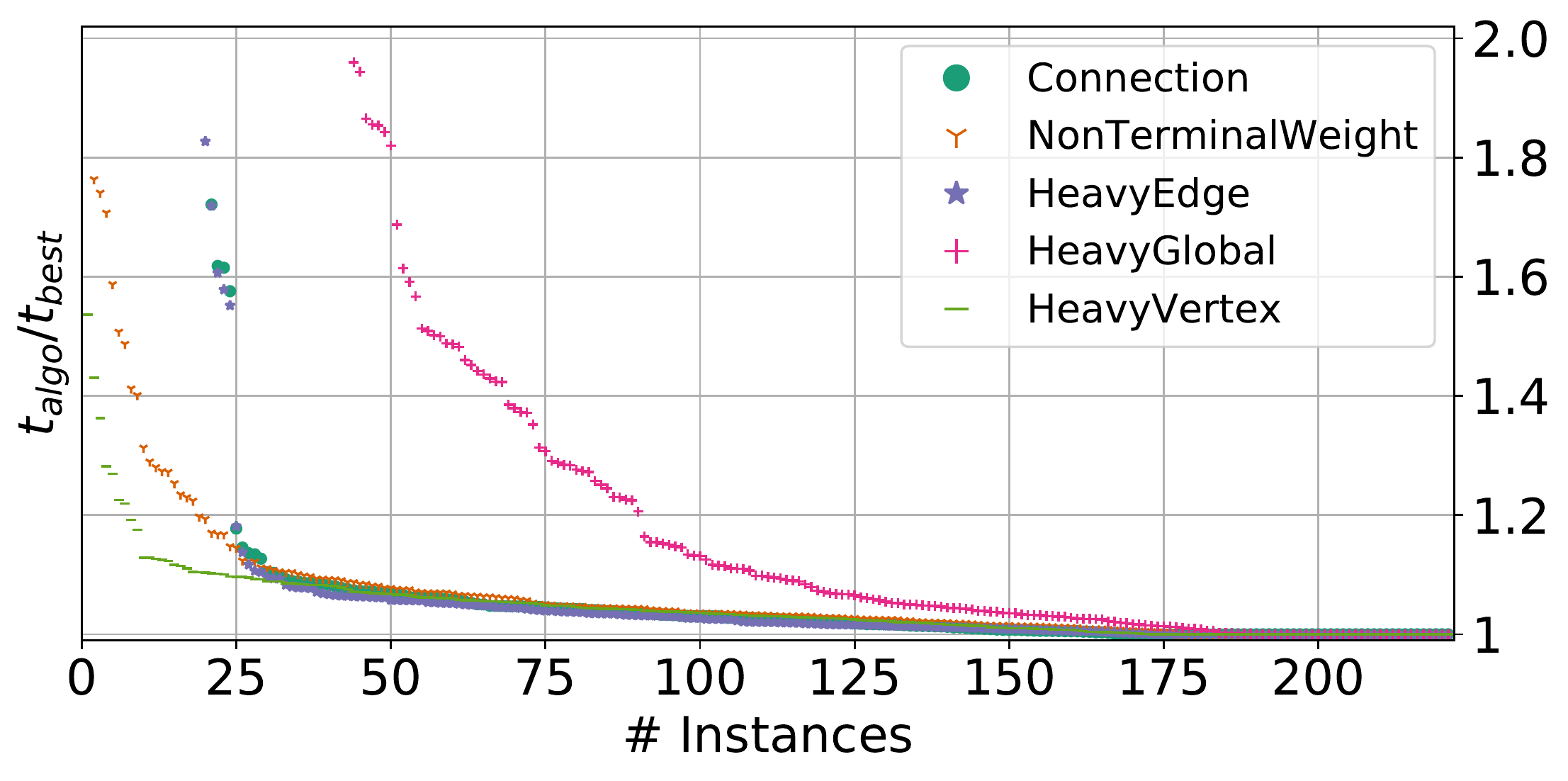}
    \caption{RHG graphs with partition centers as terminals.}
    \label{c:mc:fig:branch2}
  \end{subfigure}
  \begin{subfigure}{\textwidth}
    \includegraphics[width=\linewidth]{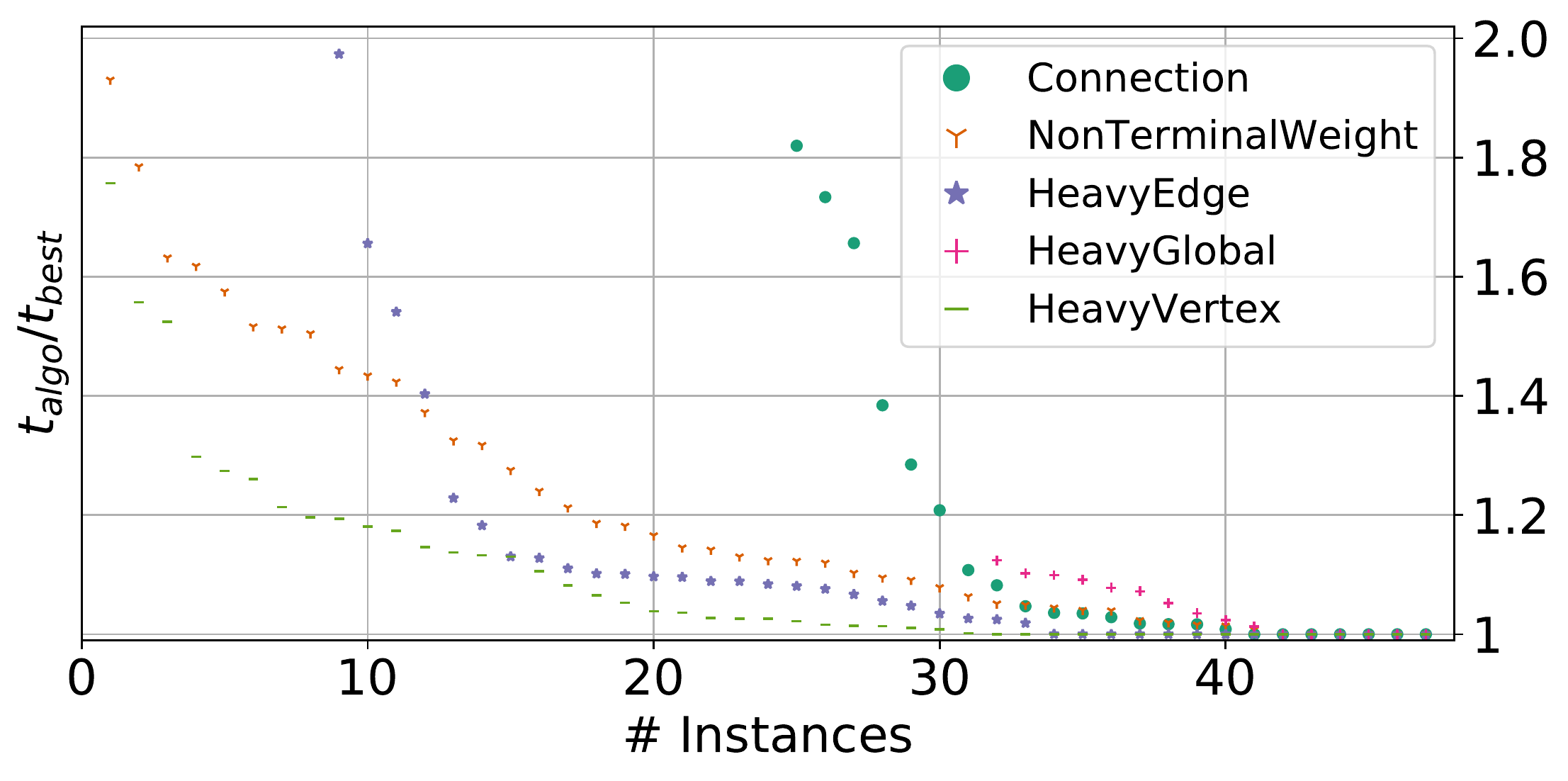}
    \caption{Map graphs with partition centers as terminals.}
    \label{c:mc:fig:branch4}
  \end{subfigure}
  \caption{Performance plots for branching edge selection variants.}
  \label{c:mc:fig:branch}
  \end{figure}

\subsection{Branching Edge Selection}\label{c:mc:ss:branch_impl}

Figure~\ref{c:mc:fig:branch} shows the results for the branching edge selection
rules on machine A. In Subfigure~\ref{c:mc:fig:branch2}, we show performance
plots for RHG graphs and in Subfigure~\ref{c:mc:fig:branch4} we show performance
plots for map graphs. To find terminals, we partition the RHG graphs into $k$
parts and perform a breadth-first search starting in the block boundary. We
define the vertex encountered last as the block center and use it as a terminal.
In this experiment we use the \textttA{BoundSum} comparator and enable
\textttA{Low}, \textttA{High}, \textttA{Triangle} and \textttA{HighConnectivity}
kernelization rules.

As the minimum multiterminal cut of those problems usually turns out to be the
trivial multiterminal cut of $k-1$ blocks of size $1$ and one block that
comprises of the rest of the graph, we instead pick the last $10$ vertices
encountered by the breadth-first search per block and contract them into a
terminal. The minimum multiterminal cut of the resulting graph is usually not
equal to the trivial multiterminal cut.

 In general, we aim to increase the lower bound by a large margin to reduce the
 number of subproblems that need to be checked. When we branch on a heavy edge,
 this increases the lower bound for $G-e$ by a large amount. For $G/e$, the
 lower bound is increased by half the amount of flow that is now added to the
 network. For a vertex that has a large number of edges to non-terminal
 vertices, contracting it into a terminal is expected to increase the flow by a
 large margin. The variant \textttA{HeavyVertex} chooses the edge $e$, for which
 the sum of edge weight and outgoing weights are maximized. It thus outperforms
 all other variants in both experiments. The only variant that is not guaranteed
 to be fixed-parameter tractable is \textttA{HeavyGlobal}, as this variant can
 also contract edges that are not incident to a terminal (and thus do not
 necessarily increase the lower bound). However, most edge contractions happen
 near terminals, so most heavy edges occur near terminals and thus
 \textttA{HeavyGlobal} often performs similar to \textttA{HeavyEdge}. 

 In all following experiments we use \textttA{HeavyVertex}, as it outperforms all
 other variants consistently.

 \begin{figure}[t]
  \centering
  \begin{subfigure}{\textwidth}
    \centering
    \includegraphics[width=\linewidth]{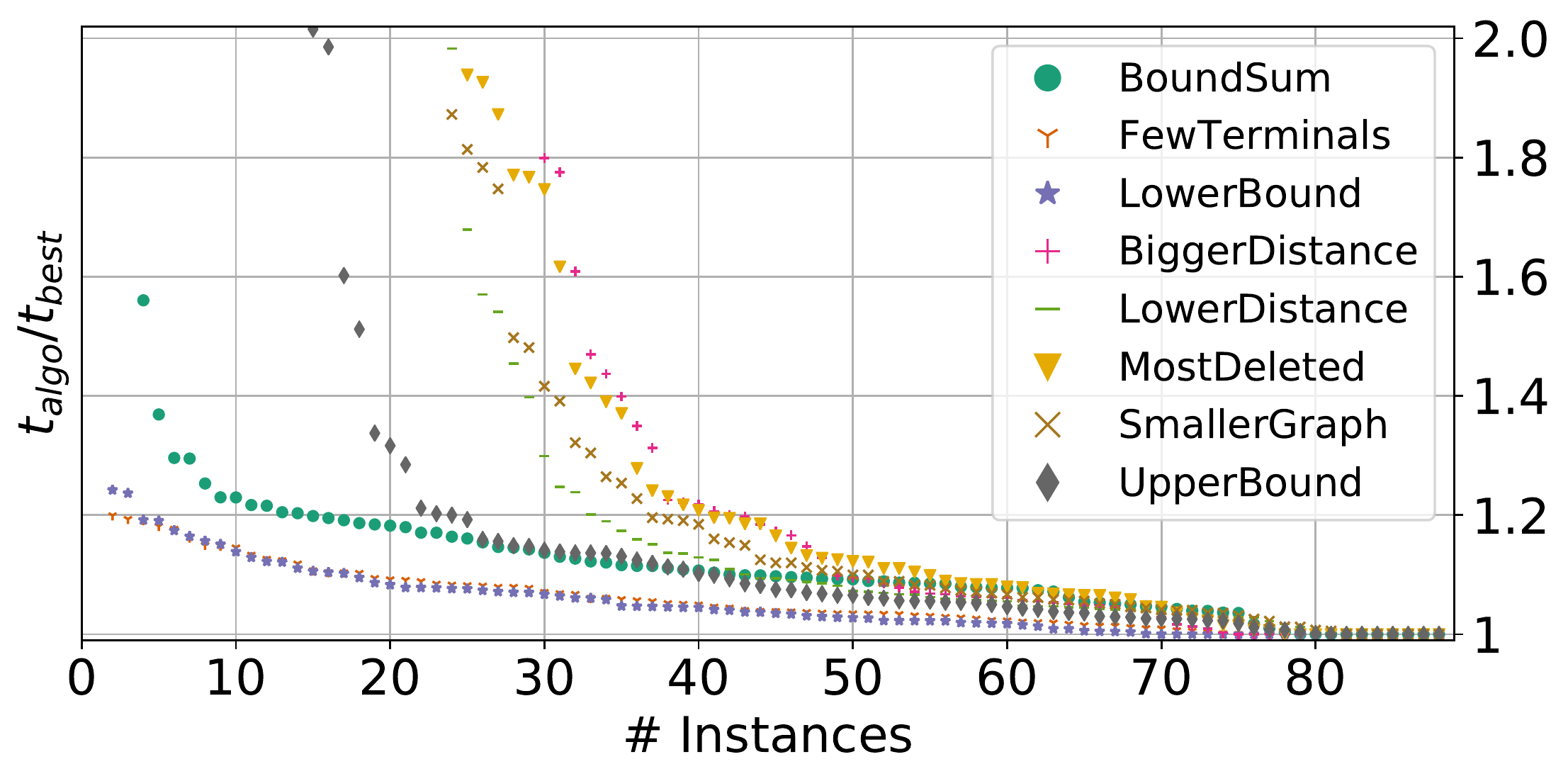}
    \caption{Graphs with $20\%$ of vertices in terminal.}
    \label{c:mc:fig:pq1}
  \end{subfigure}
  \begin{subfigure}{\textwidth}
    \centering
    \includegraphics[width=\linewidth]{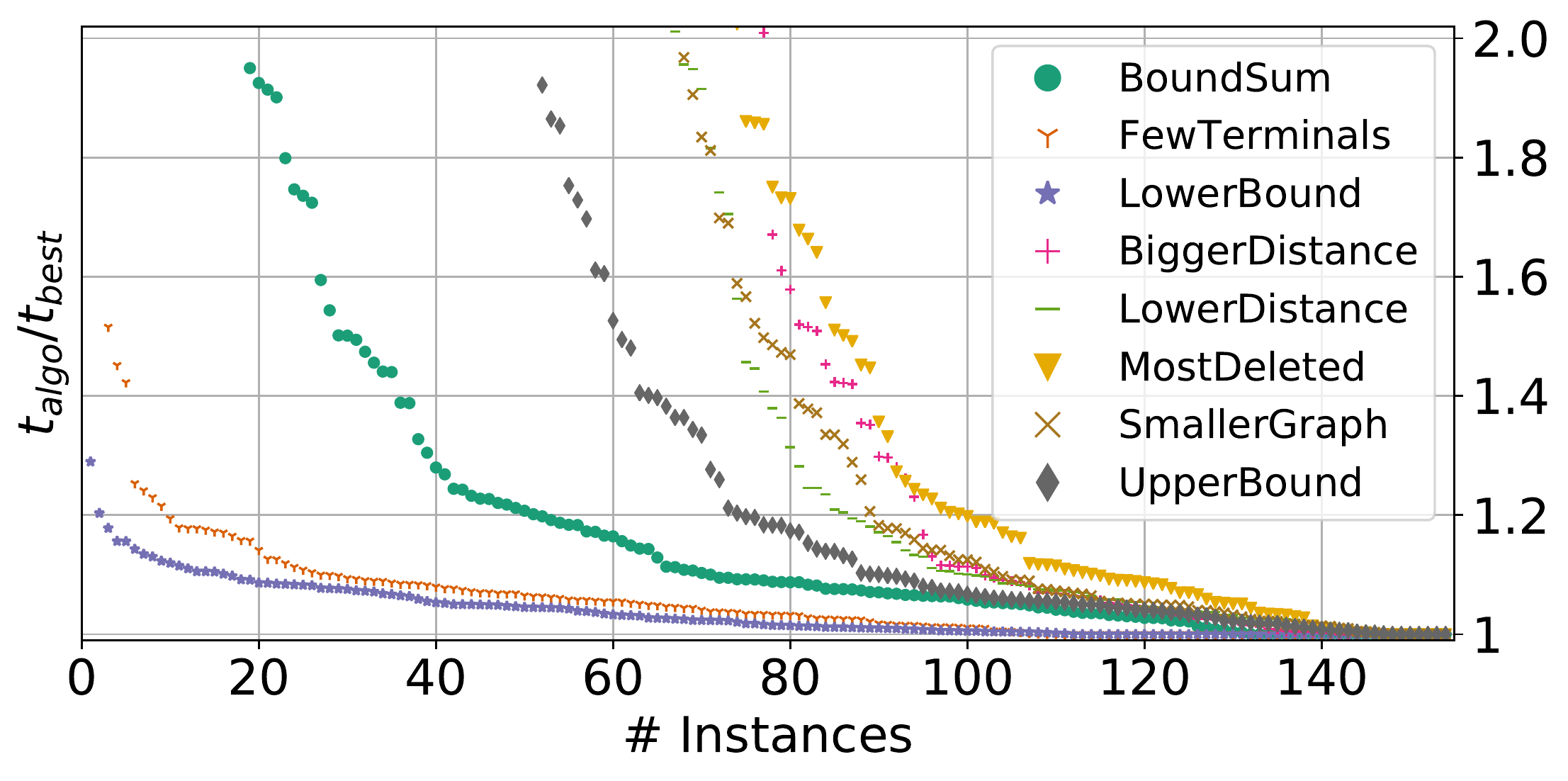}
    \caption{Graphs with $80\%$ of vertices in terminal.}
    \label{c:mc:fig:pq2}
  \end{subfigure}

  \caption{Performance plots for priority queue comparator variants.}
  \label{c:mc:fig:pq}
  \end{figure}

\subsection{Priority Queue Comparator}\label{c:mc:ss:queue_impl}

We now explore the effect of the comparator used in the priority queue $\queue$.
This experiment was performed in machine A with algorithm \vcbase{}. The choice
of comparator decides which problems are highest priority and will be explored
first. We want to first explore the problems and branches which will result in
an improved solution, as this allows us to prune more branches. However, it is
not obvious which criterion correctly identifies problems that might yield
improved solutions, either directly on indirectly. Thus, we perform experiments
on the same set of random hyperbolic and map graphs. 

On the random hyperbolic graphs examined in the previous experiment, the minimum
multiterminal cut is often equal to the sum of all minimum-s-T-cuts excluding
the heaviest. This is the cut that is found in the first iteration. If this is
also the optimal cut, we definitely have to check all subproblems whose lower
bound is lower than this cut. As the priority queue comparator only changes the
order in which we examine those problems, the experimental results using the
same problems as the previous section turned out very inconclusive. However, if
we contract a sizable fraction of each block into its terminal, the minimum
multiterminal cut is usually not equal to the union of s-T-cuts.
Figure~\ref{c:mc:fig:pq1} shows results for $20\%$ of vertices in the terminal
on RHG graphs and Figure~\ref{c:mc:fig:pq2} show results for $80\%$ of vertices
in the terminal. 

\textttA{LowerBound} and \textttA{FewTerminals} are very competitive on most
graphs. This indicates that problems with a low lower bound are very likely to
yield improved results. The next fastest variant is \textttA{BoundSum}, which is
almost competitive with $20\%$ of vertices in the terminal but significantly
slower with $80\%$ of vertices in the terminal. However, \textttA{BoundSum} uses
far less memory, as the lower bound of the newly created problems depends on the
lower bound of the current problem. \textttA{BoundSum} examines many problems for
which the lower bound is close to the currently best known solution. Thus, many
newly created subproblems are immediately discarded when their lower bound is
not lower than the currently best known solution. None of the other variants
have noteworthy performance.

\subsection{Kernelization}
\label{c:mc:ss:kernelizationexp}

\begin{figure}[t!]
  \centering
  \includegraphics[width=\textwidth]{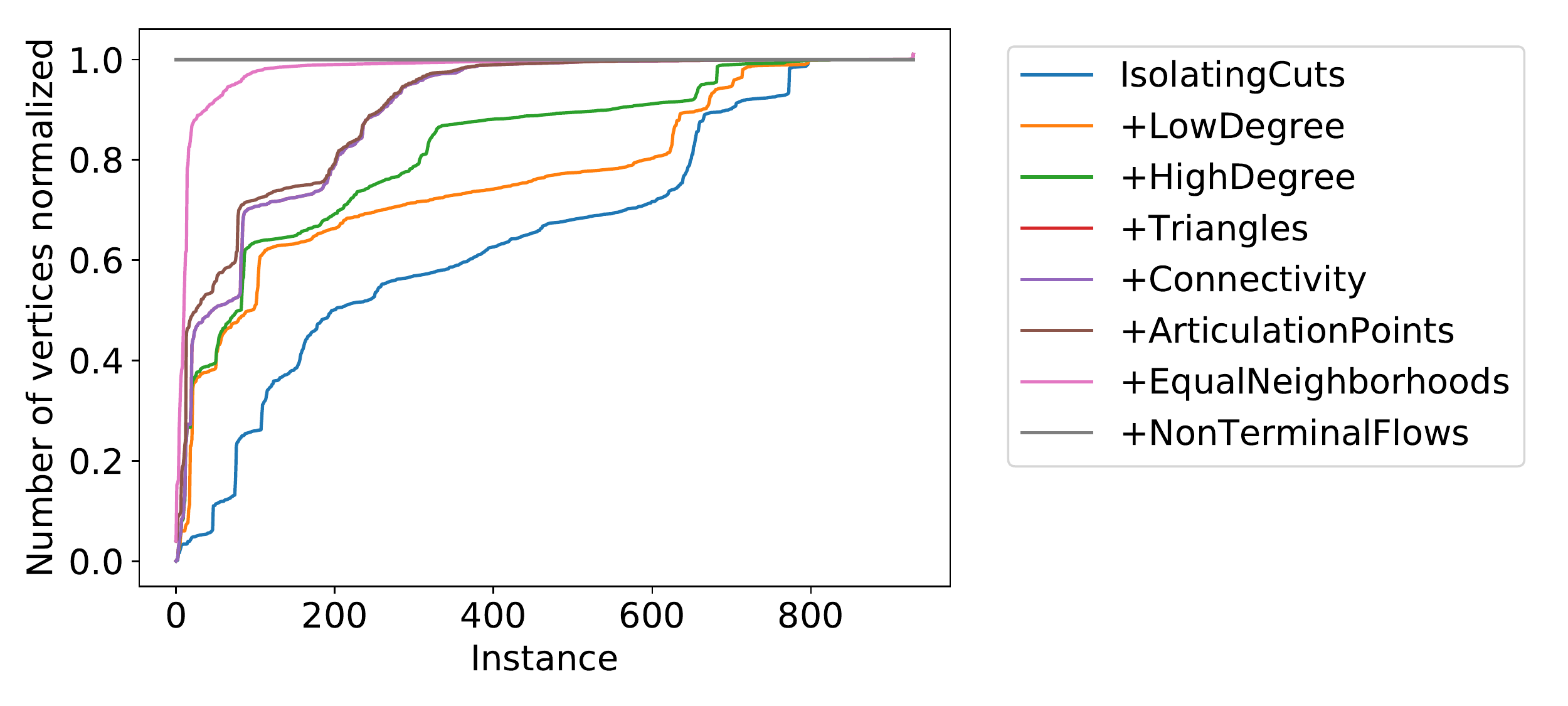}
  \caption{\label{fig:kernel} Number of vertices in graph after reductions are
  finished, normalization by (\# vertices remaining with all reductions / \#
  vertices remaining in variant) and sorted by normalized value.}
\end{figure}

We analyze the impact of the different reductions on the size of the graph at
the time of first branch. For this, we run experiments on all social, web and
map graphs (graph families (1A), (1B) and (3) in Table~\ref{t:multicutgraphs}) with
$k=\{4,8,10\}$ terminals and $10\%$ of all vertices added to the terminals on
machine C. For these instances, we run subsets of all contractions exhaustively
and check how many vertices remain in the graph. Figure~\ref{fig:kernel} gives
results with $8$ different variants, starting with a version that only runs
isolating cuts and adding one reduction family per version. For this, we sorted
the reductions by their impact on the total running time. 

For each instance and variant we normalize by the number of vertices remaining
with all reductions divided by the number of vertices remaining in a given
variant. Thus, a value close to $1$ indicates that this variant already performs
most reductions that the full algorithm does and a value close to $0$ indicates
that the resulting graph is much larger than it is when using the full
algorithm. The effectiveness of a reduction can therefore be read from the area
between a line and the line below it.

We can see that running the local reductions in \vcbase{} are very effective on
almost all instances. In average, \textttA{IsolatingCuts} reduce the number of
vertices by $33\%$, \textttA{LowDegree} reduces the number of vertices in the
remaining graph by $17\%$, \textttA{HighDegree} by $7\%$ and \textttA{Triangles}
by $8\%$. In contrast, \textttA{Connectivity} only has a negligible effect, which
can be explained by the fact that it contracts edges whose connectivity is
larger than a value related to the difference of upper bound to total weight of
deleted edges. As there are almost no deleted edges in the beginning, this value
is very high and almost no edge has high enough connectivity.

Out of the new reductions that are not part of \vcbase{}, all find a significant
amount of contractible edges on the graphs already contracted by the reductions
included therein. In average, \textttA{ArticulationPoints} reduces the number of
vertices on the already contracted graphs by $1.9\%$,
\textttA{EqualNeighborhoods} reduces the number of vertices by $7.8\%$ and
\textttA{NonTerminalFlows} reduces the number of vertices by $2.0\%$. However,
there are some instances in which these reductions reduce the
number of vertices remaining by more than $99\%$.

\begin{figure}[!t]
  \centering
  \includegraphics[width=\linewidth]{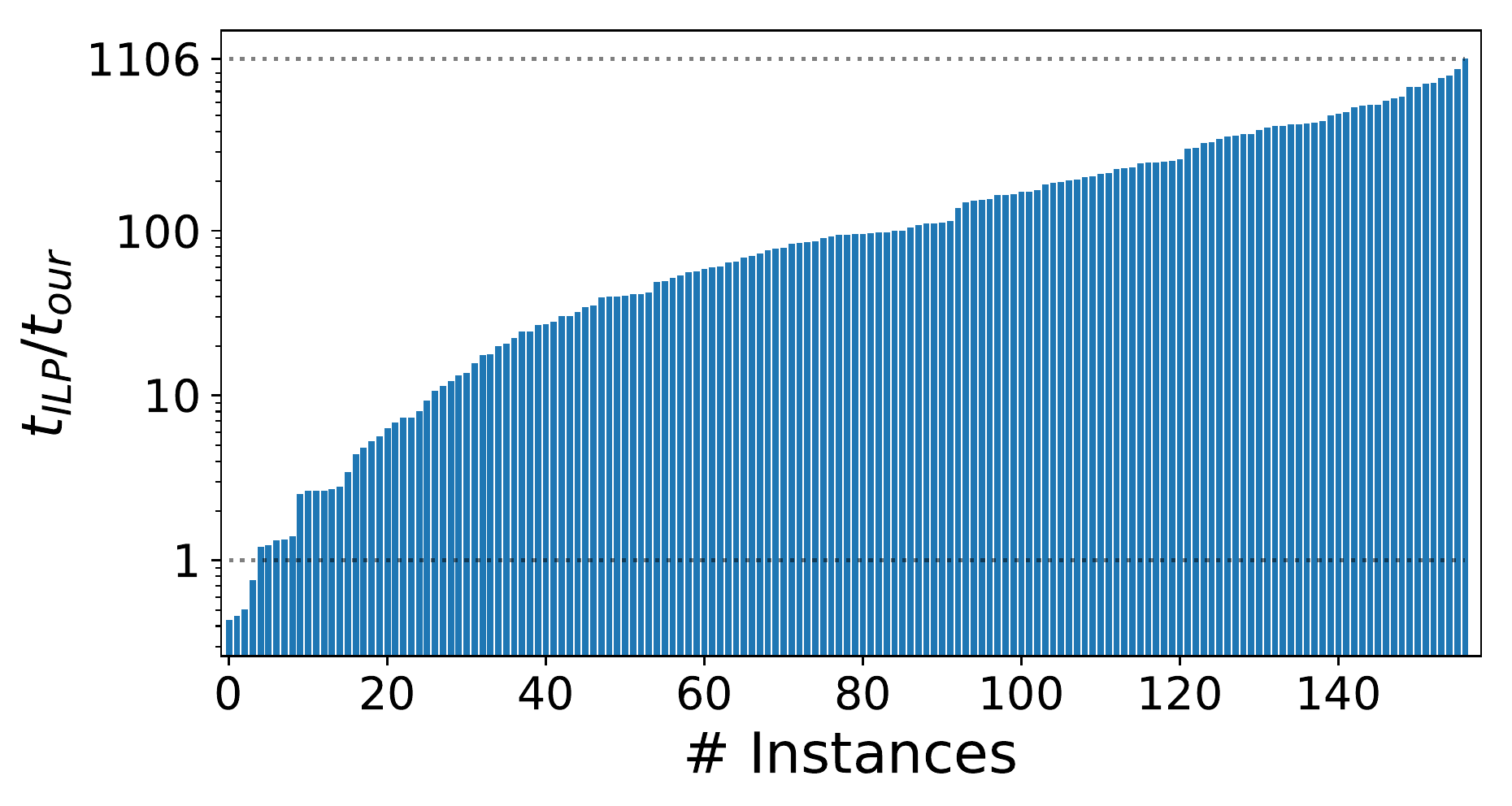}
  \caption{Speedup of opt. branch-and-reduce \vcbase{} to ILP.}
  \label{c:mc:fig:speeduptoilp}
\end{figure}

\begin{figure}[!t]
  \centering
  \includegraphics[width=\linewidth]{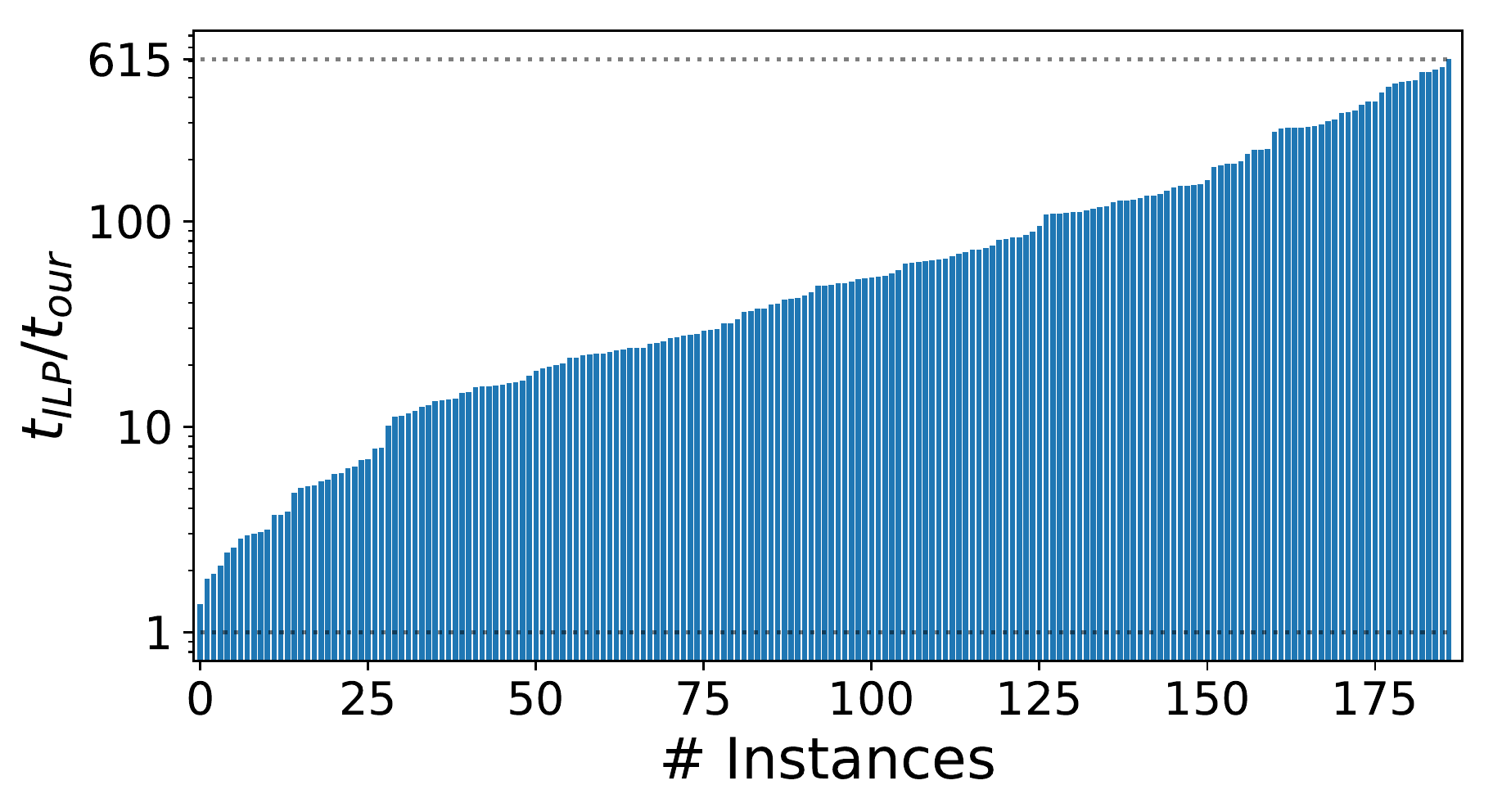}
  \caption{Speedup of \textttA{Kernel+ILP} to ILP.}
  \label{c:mc:fig:nbtoilp}
\end{figure}

\begin{figure}[!t]
  \centering
  \includegraphics[width=\linewidth]{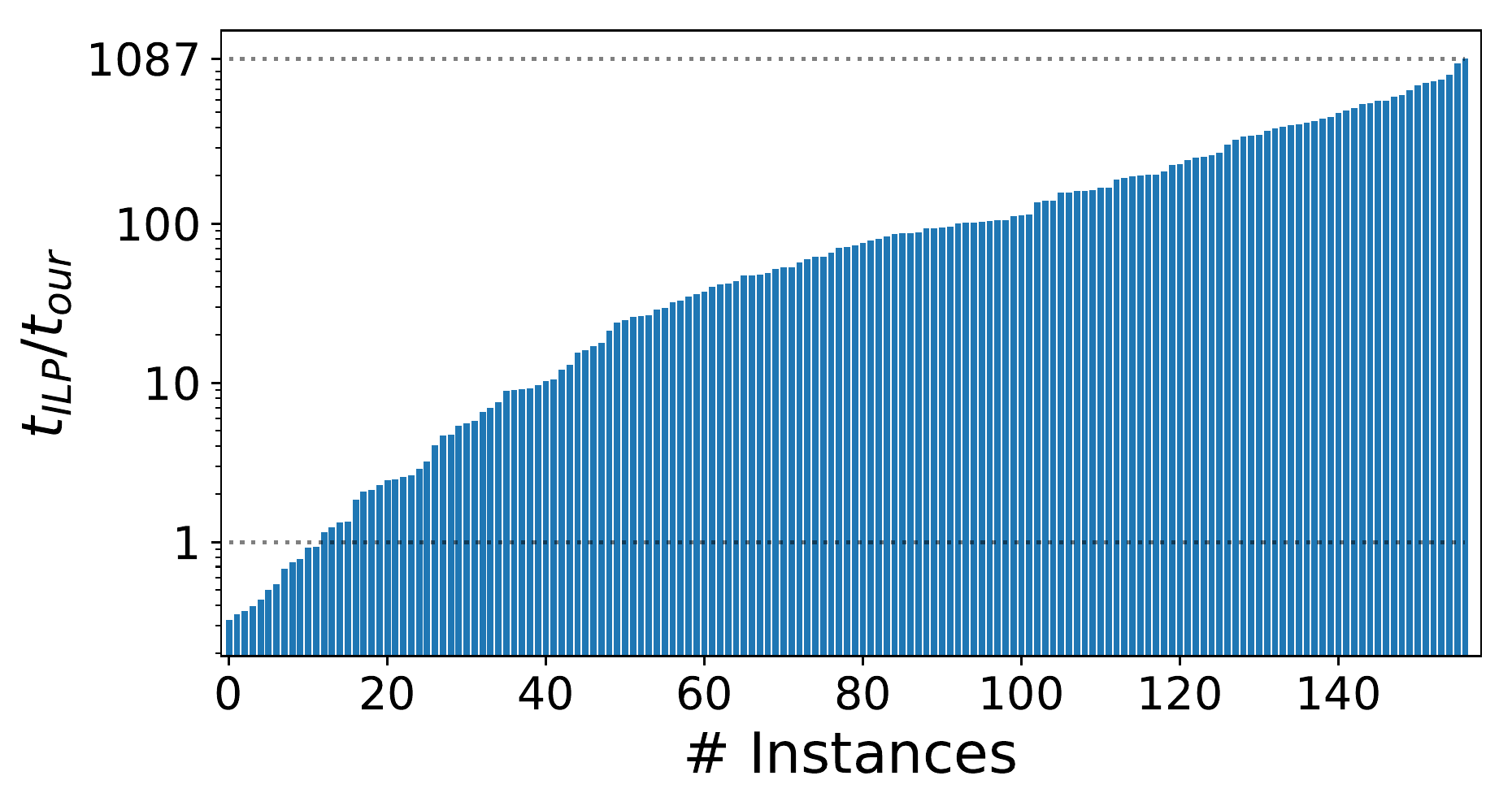}
  \caption{Speedup of avg. branch-and-reduce \vcbase{} to ILP.}
  \label{c:mc:fig:alltoilp}
\end{figure}

\subsection{Comparison between VieCut-MTC and ILP}
\label{c:mc:ss:comparisonilp}

Figure~\ref{c:mc:fig:speeduptoilp} shows the speedup of the engineered \vcbase{}
algorithm, using \textttA{HeavyVertex} edge selection, \textttA{LowerBound}
priority queue comparator and all kernelization rules of \vcbase{} enabled, to
the ILP on all graphs from Sections~\ref{c:mc:ss:branch_impl} and
\ref{c:mc:ss:queue_impl} in which the ILP managed to find the minimum
multiterminal cut within $3$ minutes. The branch-and-reduce algorithm
outperforms the ILP on almost all graphs, often by multiple orders of magnitude.
The ILP only solves $24\%$ of all problems, \vcbase{} solves $61\%$; on the
problems solved by both, the branch-and-reduce algorithm has a mean speedup
factor of $67$, a median speedup factor of $95$ and a maximum speedup factor of
$\numprint{1106}$. The mean speedup factor of the average of all algorithm
configurations compared to ILP is $43$ with a median speedup factor of $71$. The
speedup can be seen in Figure~\ref{c:mc:fig:alltoilp}. Compared to the original
ILP, \textttA{Kernel+ILP} is faster on all instances, has a mean speedup factor
of $44$ and a median speedup factor of $49$, as shown in
Figure~\ref{c:mc:fig:nbtoilp}.

This allows us to solve instances with more than a million vertices, while the
ILP was unable to solve any instance with more than $\numprint{100000}$
vertices. As the basic ILP is unable to solve any large instances, we do not use
it in the following experiments on large graphs.

\subsection{VieCut-MTC on Protein-Protein Interaction Networks}
\label{c:mc:ss:ppi}

Multiterminal cuts can be used for protein function prediction by creating a
terminal for each possible protein function and adding all proteins which have
this function to this
terminal~\cite{karaoz2004whole,nabieva2005whole,vazquez2003global}.
Table~\ref{t:ppiov} shows the results for these graphs. We can see that
\textttA{Kernel+ILP} outperforms \vcbase{} by a large margin on most
graphs. This is the case because the kernelization is able to reduce the size of
the graphs severely. These small problems with high cut values are better suited
for \textttA{Kernel+ILP} than the branch-and-bound variants whose running time is
more correlated with the value of the minimum multiterminal cut. The mean times
are very low as some problems can be solved very quickly and thus drag the mean
of all algorithms down. Due to these results
in~\cite{henzinger2020sharedmemory}, \exact{} integrates the ILP solving into
the branch-and-reduce algorithm and solves some subproblems using an ILP solver.
In the following section we examine which subproblems should be solved using
branching and which should be solved using ILP.

 \begin{table}[t] \centering
  \small
  \caption{Result overview on protein-protein-interaction networks.\label{t:ppiov}}
   \begin{tabular}{| l | r | r | r | r |}
     \hline
     Algorithm & \textttA{K+ILP} & \textttA{BSum} & \textttA{FTerm} &
     \textttA{LBound}\\
     \hline
     best result & \textbf{\numprint{57}} & \numprint{34} & \numprint{26} &
     \numprint{23}\\ 
     terminated & \textbf{\numprint{57}} & \numprint{25} & \numprint{23} &
     \numprint{21}\\
     mean result & \textbf{\numprint{4183}} & \numprint{4210} & \numprint{4218}
     & \numprint{4222}\\
     mean time & \textbf{\numprint{0.21}s} & \numprint{0.33}s & \numprint{0.36}s
     & \numprint{0.40}s\\
     \hline    
   \end{tabular}
 \end{table}

\subsection{Integer Linear Programming}
\label{ss:exp_ilp}

In order to get all a wide variety of ILP problems, we run the \inexact{}
algorithm on all instances in graph families (1A), (1B) and (3) of
Table~\ref{t:multicutgraphs} with $k=10$ terminals and $10\%$ of vertices added to the
terminals on machine C. As \inexact{} removes low-degree terminals and contracts
edges, we have subproblems with very different sizes and numbers of terminals.
In this experiment, whenever the algorithm chooses between branching and ILP on
graph $G$, we select a random integer $r \in (1, $\numprint{200000}$)$. We use
this random integer, as we want to have problems of all different sizes and
using a hard limit would result in many instances just barely below that size
limit. We select $200000$ edges as the maximum, as we did not encounter any
larger instances in which the ILP was solved to optimality in the allotted time.
If $|E| < r$, the problem is solved with ILP, otherwise the algorithm branches
on a vertex incident to a terminal. The timeout is set to $60$ seconds.

\begin{figure}[t!]
  \centering
  \includegraphics[width=\textwidth]{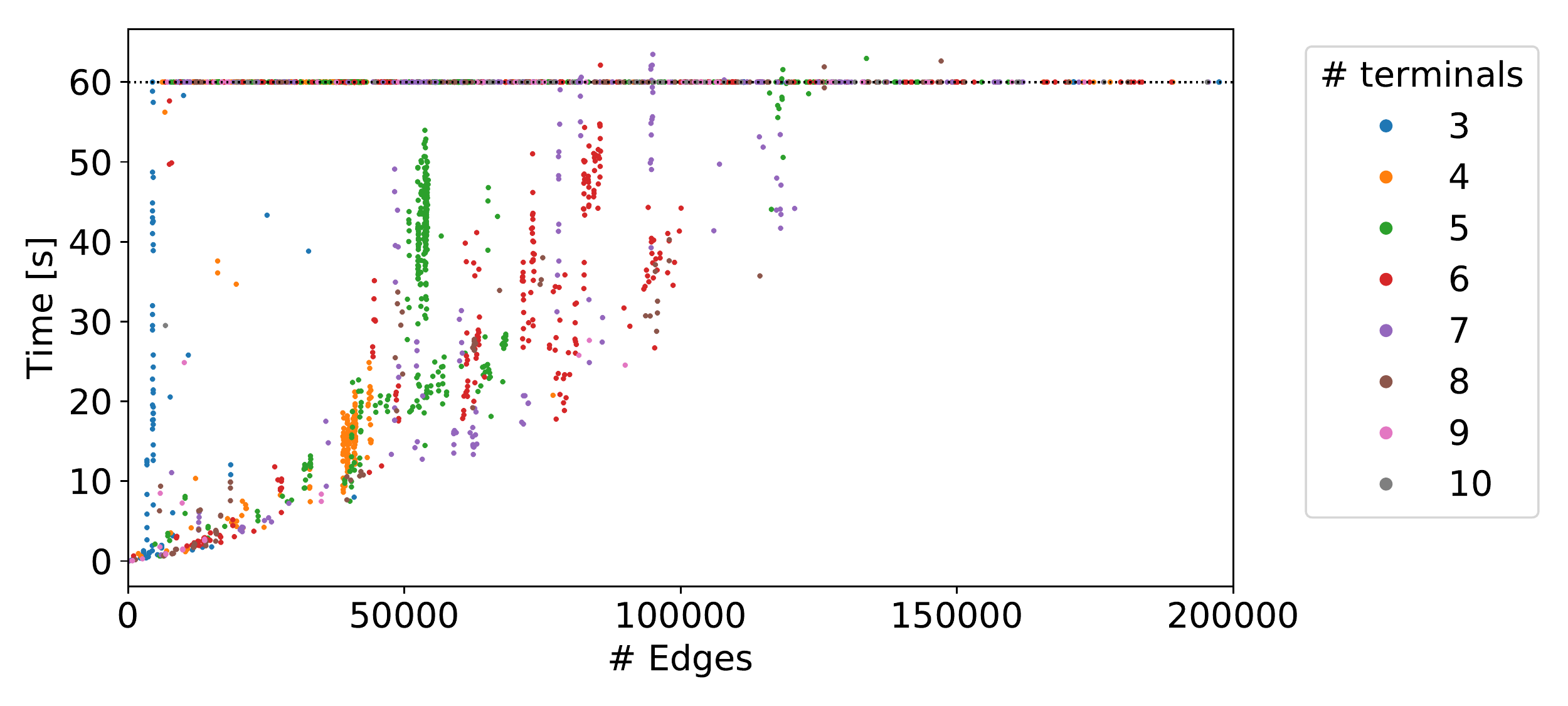}
  \caption{\label{fig:ilp} Running time of ILP subproblems in relation to
  $|E|$.}  
\end{figure}

Figure~\ref{fig:ilp} shows the time needed to solve the ILP problems in relation
to the number of edges in the graph. We can see that there is a strong
correlation between problem size and total running time, but there are still a
large number of outliers that cannot be solved in the allotted time even though
the instances are rather small. In the following, we set the limit to
\numprint{50000} edges and solve all instances with fewer than \numprint{50000}
edges with an integer linear program. If the instance has at least
\numprint{50000} edges, we branch on a vertex incident to a terminal and create
more subproblems.

\subsection{Large Real-World Networks}
\label{c:mc:ss:social}

\begin{figure}[t]
  \centering
  \begin{subfigure}{.49\textwidth}
    \includegraphics[width=\linewidth]{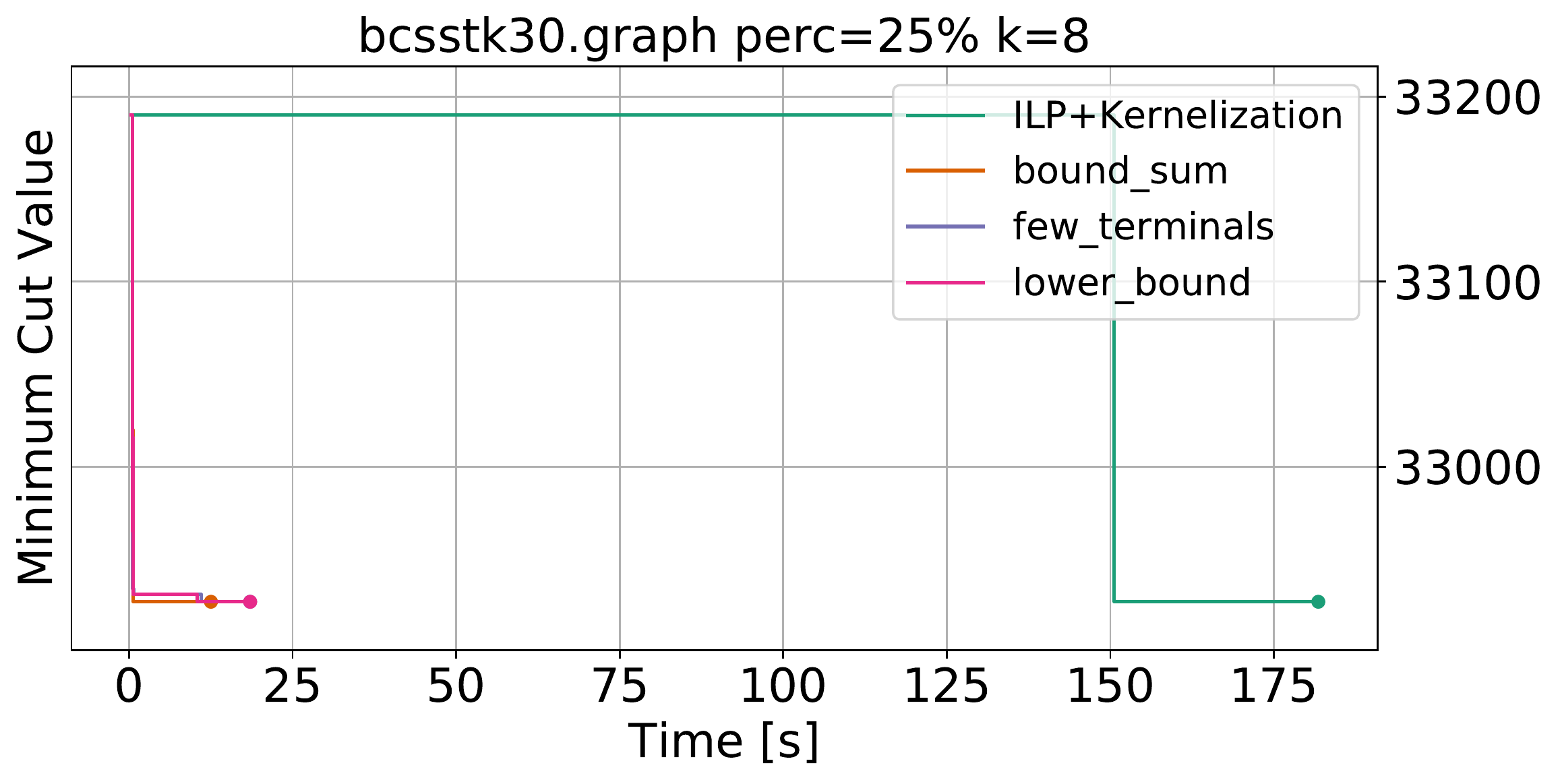}
  \end{subfigure}%
  \begin{subfigure}{.49\textwidth}
    \includegraphics[width=\linewidth]{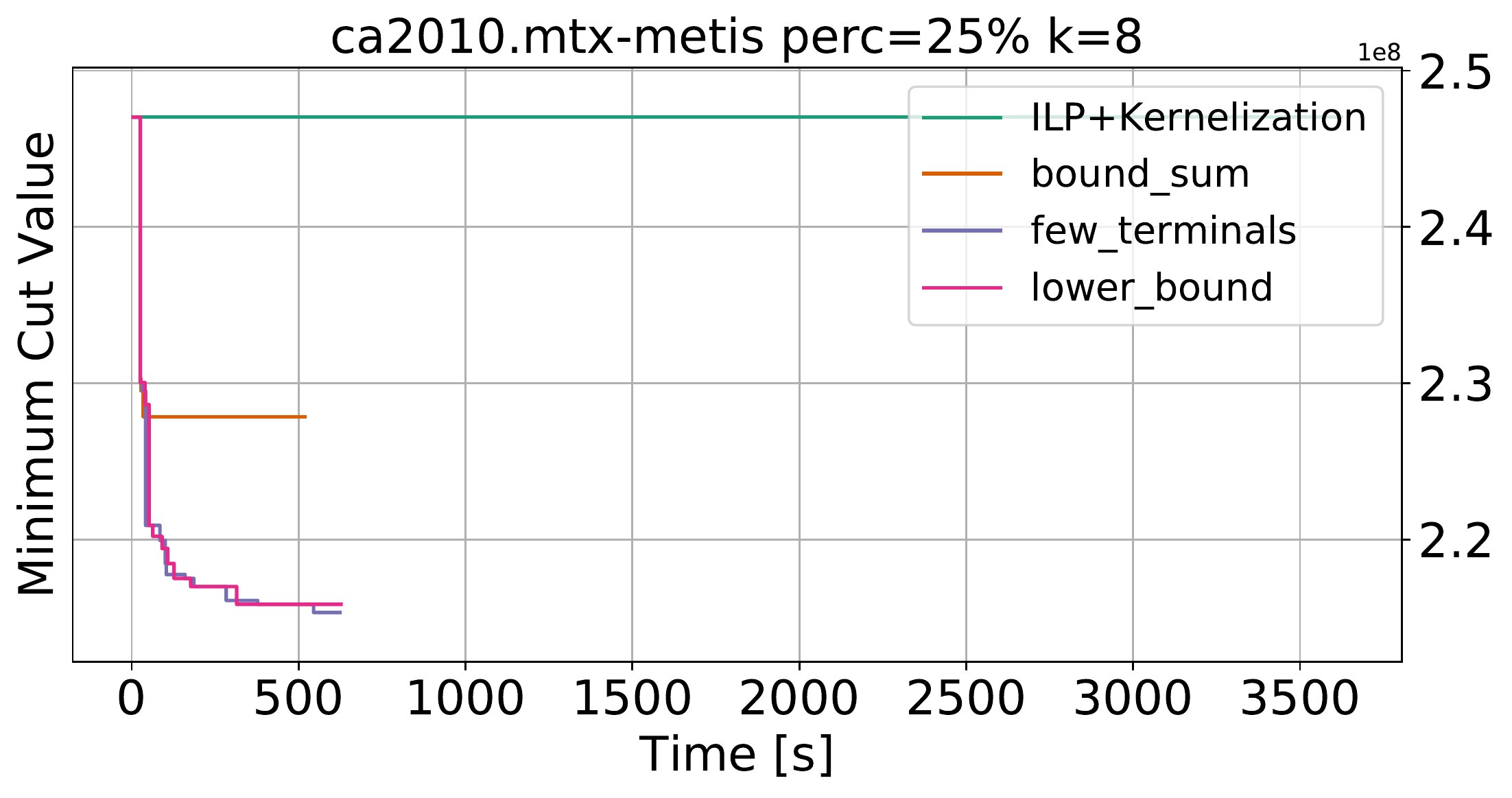}
  \end{subfigure}
  \begin{subfigure}{.49\textwidth}
    \includegraphics[width=\linewidth]{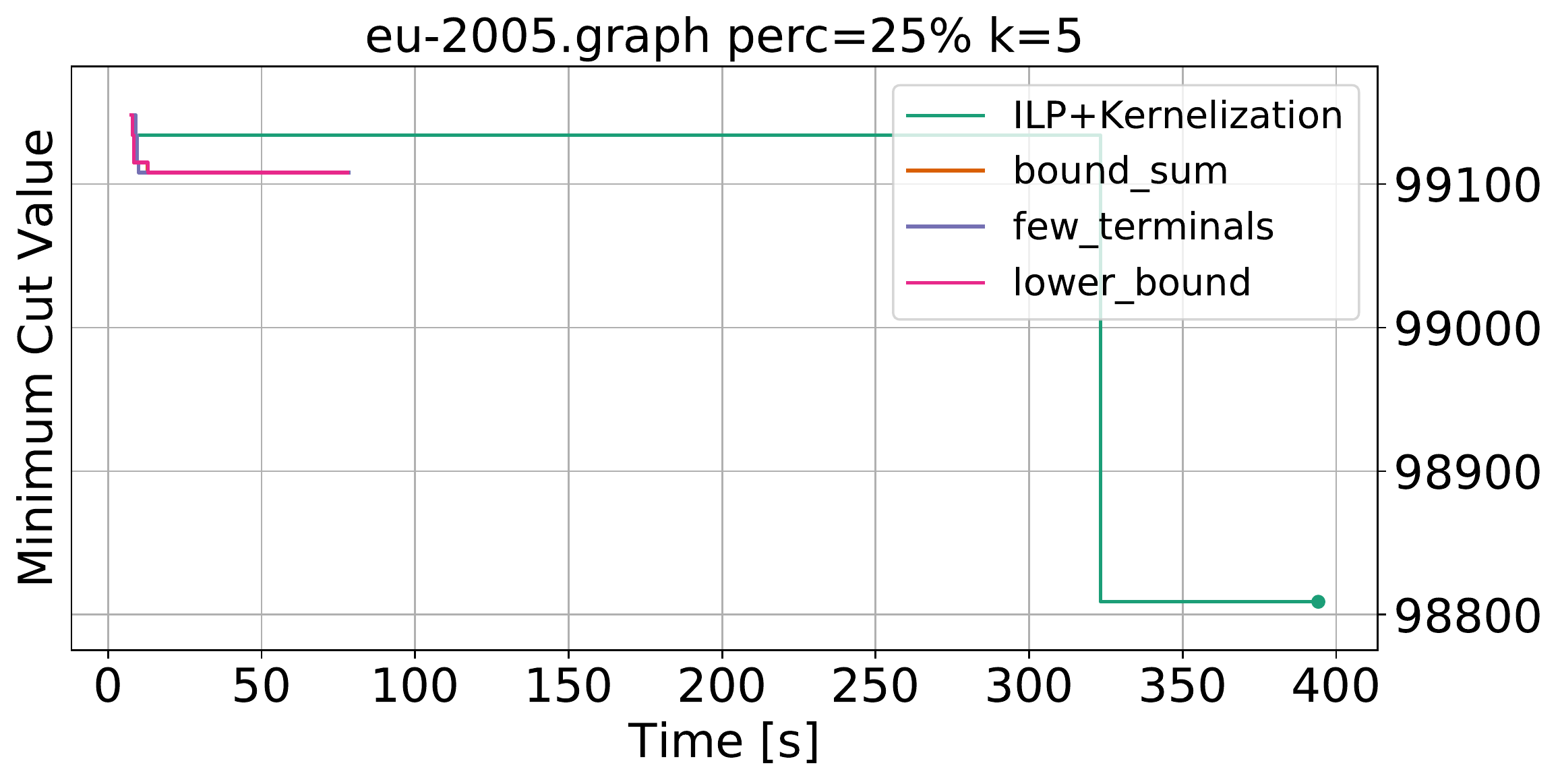}
  \end{subfigure}%
  \begin{subfigure}{.49\textwidth}
    \includegraphics[width=\linewidth]{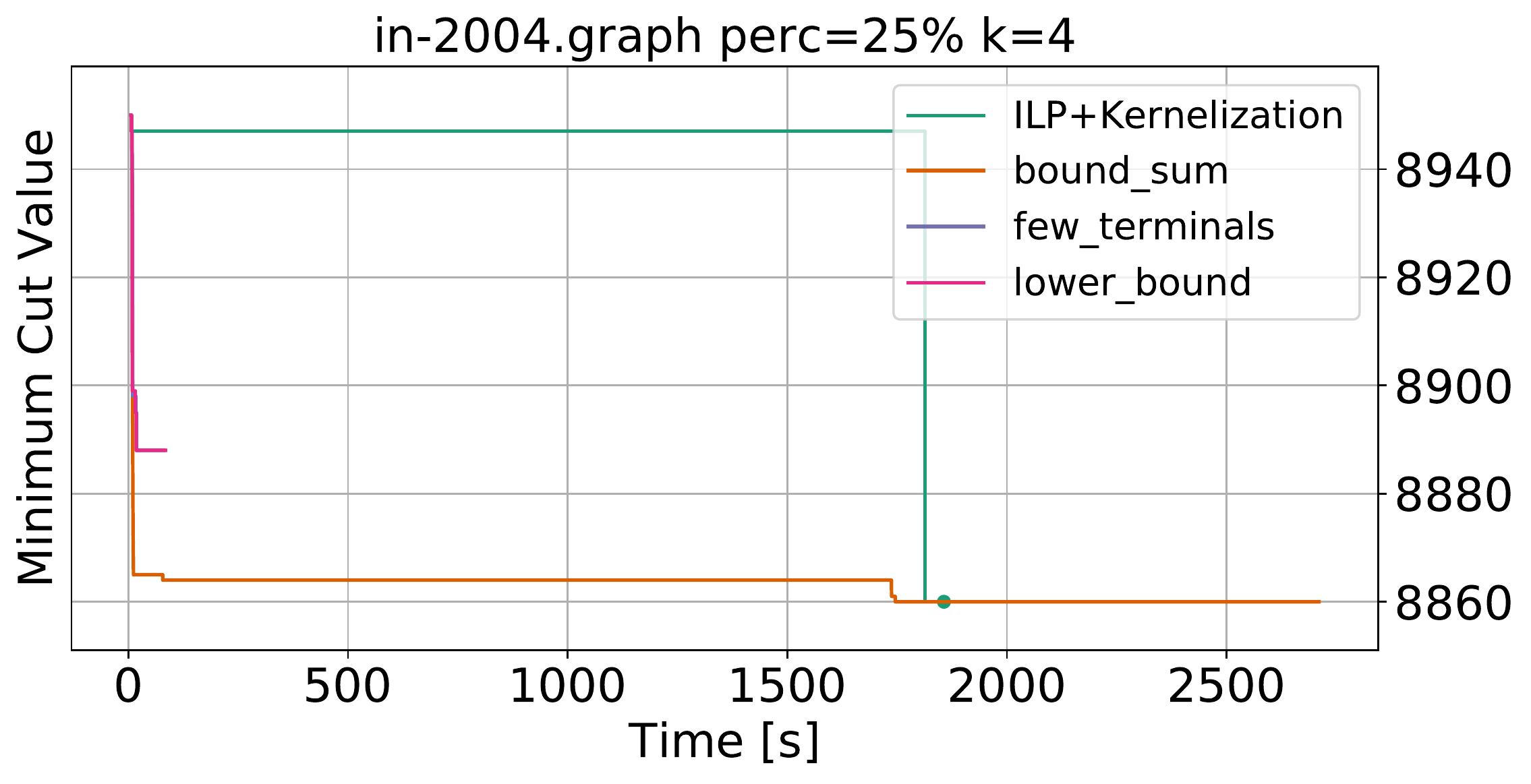}
  \end{subfigure}
  \begin{subfigure}{.49\textwidth}
    \includegraphics[width=\linewidth]{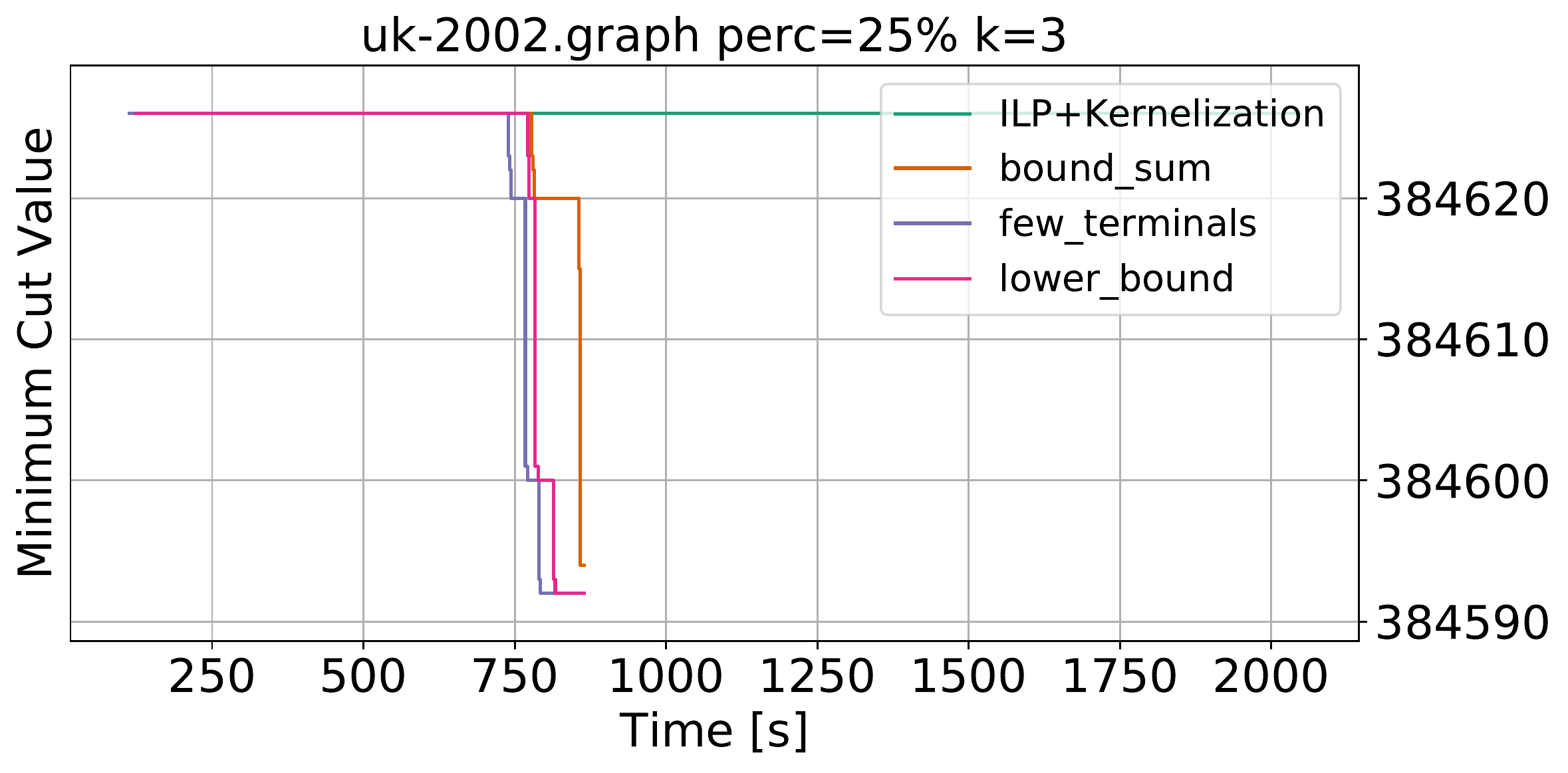}
  \end{subfigure}%
  \begin{subfigure}{.49\textwidth}
    \includegraphics[width=\linewidth]{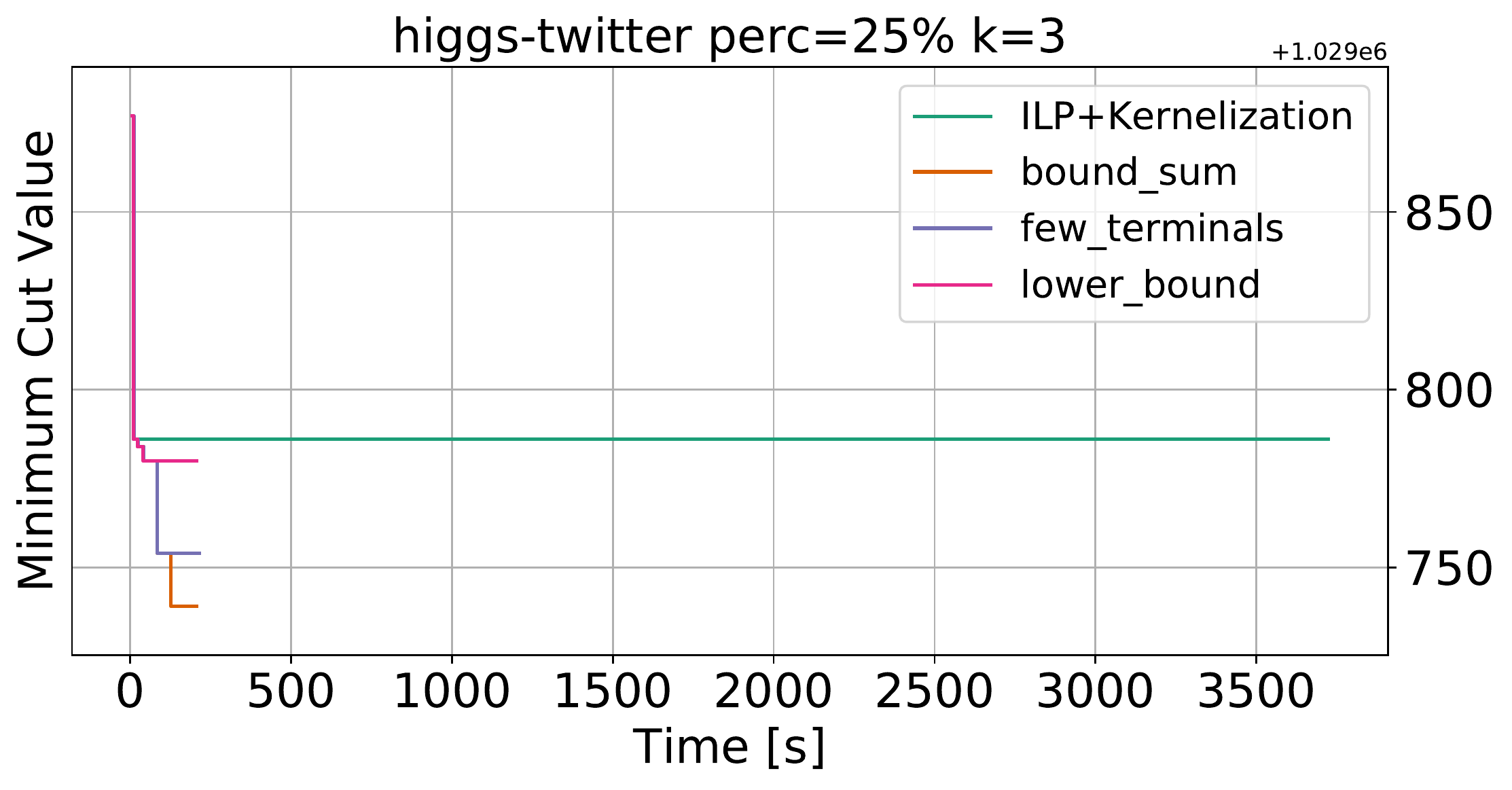}
  \end{subfigure}

  \caption{\label{c:mc:fig:social} Progression of best result over time. Dot at
  end symbolizes that algorithm certifies optimality.}
\end{figure}

\begin{table}[t] \centering
 \small
 \caption{Result overview of Section~\ref{c:mc:ss:social}.\label{t:overview}}
  \begin{tabular}{| l | r | r | r | r |}
    \hline
    Algorithm & \textttA{K+ILP} & \textttA{BSum} & \textttA{FTerm} &
    \textttA{LBound}\\
    \hline
    best result & \numprint{118} & \textbf{\numprint{136}} & \numprint{126} &
    \numprint{125}\\ 
    terminated & \textbf{\numprint{46}} & \numprint{35} & \numprint{33} &
    \numprint{33}\\
    mean result & \numprint{146570} & \textbf{\numprint{145961}} &
    \numprint{146052} & \numprint{146025}\\
    mean time & \numprint{18.69}s & \textbf{\numprint{6.71}s} & \numprint{6.97}s
    & \numprint{6.78}s\\
    \hline    
  \end{tabular}
\end{table}

In this experiment we compare configurations of \vcbase{} with
\textttA{Kernel+ILP}. We use graph family (1A) of Table~\ref{t:multicutgraphs}.
For each graph, we solve the minimum multiterminal cut problem for $k \in
\{3,4,5,8\}$ terminals and $p \in \{10\%, 15\%, 20\%, 25\%\}$ vertices in the
terminal. We hereby use the priority queue configurations \textttA{BoundSum},
\textttA{LowerBound} and \textttA{FewTerminals}. Figure~\ref{c:mc:fig:social}
shows the progression of the best result over time for a set of interesting
problems. Table~\ref{t:overview} gives an overview over the results. For each
variant we show how often it produced the best result over all variants and how
often it terminated with the optimal result. It also gives the mean result and
time for all problems which were solved to optimality by all variants. In both 
Figure~\ref{c:mc:fig:social} and Table~\ref{t:overview} we can see that the
branch and reduce variants find good solutions faster than \textttA{Kernel+ILP}.
However, the variants often run out of memory in some of the largest instances.
In cases where the best multiterminal cut was already found (but not
confirmed to be optimal) by the kernelization, \textttA{Kernel+ILP} managed to
certify optimality more often than the branching variants. Thus it has the
highest amount of terminated results, but reports significantly worse results on
average. \textttA{Kernel+ILP} has about half as much improvements as the best
variant \textttA{BoundSum}. In addition to giving the best results, variant
\textttA{BoundSum} also has the lowest mean time for problems which were solved by
all variants, however the improvement over the other branch-and-reduce variants
is miniscule. The correlation between running time and number of vertices in the
kernel graph is much stronger in \textttA{Kernel+ILP} compared to the branching
variants.

\begin{figure}[t!]
  \includegraphics[width=\linewidth]{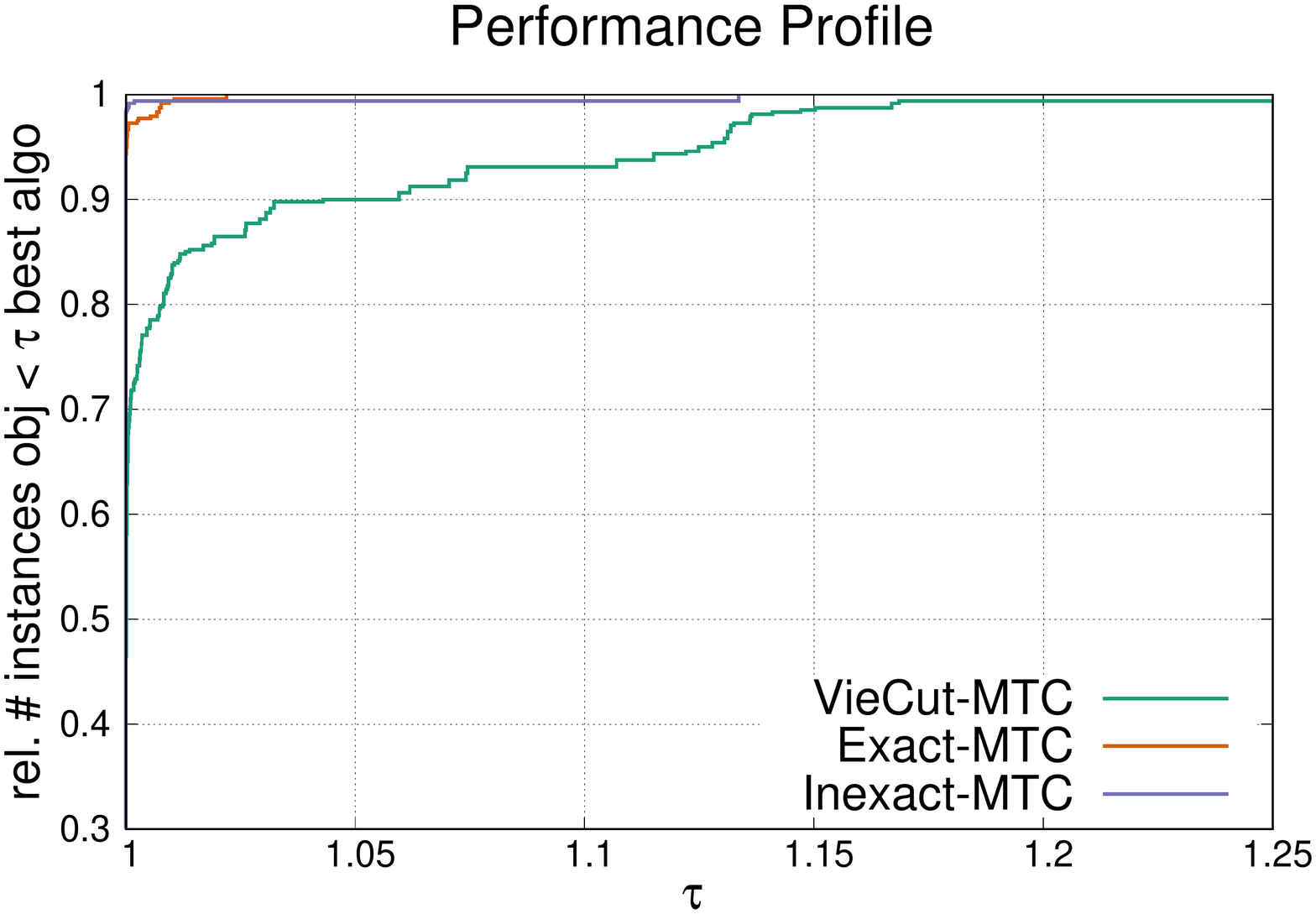}
  \caption{\label{fig:pp53} Performance profile for $k \in
  \{3,4,5,8\}$ and graph family (1A).}
\end{figure}

We use the same instances to compare \exact{} to \vcbase{} (both using
\textttA{BoundSum} as priority queue implementation), using machine C with all
$12$ cores and a time limit to $600$ seconds. Out of $160$ instances, \vcbase{}
terminates with an optimal result in $32$ instances, while \exact{} terminates
with an optimal result in $46$ instances. Of the $114$ instances that were
not solved to optimality by both algorithms, \exact{} gives a better result on
$75$ instances and the same result on all others. The geometric mean of results
given by \exact{} and \inexact{} are both about $1.5\%$ lower than \vcbase{}.
Note that in the first iteration of this experiment, which uses a larger machine
($32$ cores) and has a timeout of $3600$ seconds, \vcbase{} has a geometric mean
of about $0.1\%$ better than \vcbase{} in this comparison. The largest part of
the improvement of \exact{} and \inexact{} over \vcbase{} is gained by the
local search algorithm detailed in Section~\ref{c:mc:s:local}.

Figure~\ref{fig:pp53} shows the performance profile of this experiment. We can
see that both \exact{} and \inexact{} are almost always equal to the best result
on this instance or very close to it. In contrast, \vcbase{} gives noticeably
worse results on about $20\%$ of instances and more than $5\%$ worse results on
$10\%$ of instances. 

\begin{figure}[t!]
    \includegraphics[width=\linewidth]{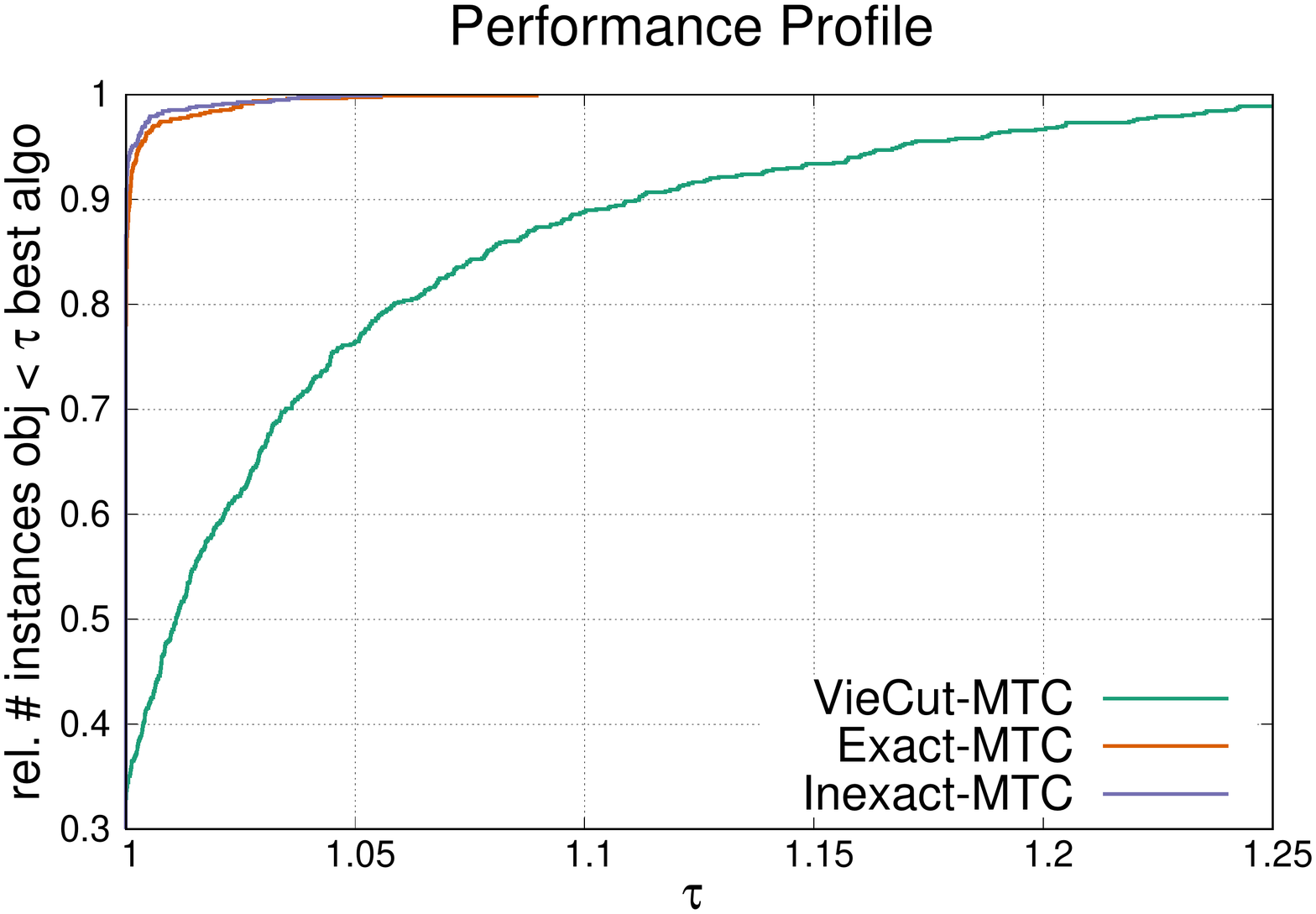}
    \caption{\label{fig:pp54} Performance Profiles for $k \in
    \{4,5,8,10\}$ and graph families (1A), (1B) and (3).}
\end{figure}

\begin{table}[t!] \centering
  \caption{\label{t:probs} Result overview for Section~\ref{c:mc:ss:social}.}
  \begin{tabular}{llrrr}
    \toprule
    \# Terminals & & \vcbase & \exact & \inexact \\
    \midrule
    4 & Best Solution & \numprint{109} & \textbf{\numprint{183}} &
    \numprint{175} \\
    & Mean Solution & \numprint{161799} & \textbf{\numprint{159402}} &
    \numprint{159499} \\
    & Better Exact & \numprint{6} & \textbf{\numprint{94}} & --- \\
    \midrule
    5 & Best Solution & \numprint{81} & \textbf{\numprint{173}} & \numprint{158}
    \\
    & Mean Solution & \numprint{216191} & \textbf{\numprint{210928}} &
    \numprint{211090} \\
    & Better Exact & \numprint{6} & \textbf{\numprint{121}} & --- \\
    \midrule
    8 & Best Solution & \numprint{42} & \numprint{139} & \textbf{\numprint{175}}
    \\
     & Mean Solution & \numprint{346509} & \numprint{331112} &
     \textbf{\numprint{330856}} \\
     & Better Exact & \numprint{2} & \textbf{\numprint{162}} & --- \\
    \midrule
    10 & Best Solution & \numprint{37} & \numprint{129} &
    \textbf{\numprint{173}} \\
     & Mean Solution & \numprint{412138} & \numprint{392561} &
     \textbf{\numprint{391822}} \\
     & Better Exact & \numprint{1} & \textbf{\numprint{165}} & --- \\
    \bottomrule
  \end{tabular}
\end{table}

Additionally, we compare \vcbase{}, \exact{} and \inexact{} on a larger set of
instances, all graphs from Table~\ref{t:multicutgraphs} graph families (1A),
(1B) and (3) with
$k=\{4,5,8,10\}$ terminals and $p = \{10\%,20\%\}$ of vertices added to the
terminal. For each combination of graph, number of terminals and factor of
vertices in terminal, we create three problems with random seeds $s=\{0,1,2\}$.
Thus, we have a total of $816$ problems. We set the time limit per algorithm and
problem to $600$ seconds. We run the experiment on machine B using all $12$ CPU
cores. If the algorithm does not terminate in the allotted time or memory limit,
we report the best intermediate result. Note that is a soft limit, in which the
algorithm finishes the current operation and exits afterwards if the time or
memory limit is reached.

Table~\ref{t:probs} gives an overview of the results. For each algorithm, we
give the number of times, where it gives the best (or shared best) solution over
all algorithms; the geometric mean of the cut value; and for \vcbase{} and
\exact{} the number of instances in which they have a better result than the
respective other. In all instances, in which \vcbase{} and \exact{} terminate
with the optimal result, \inexact{} also gives the optimal result. We can see
that in the problems with $4$ and $5$ terminals, \exact{} slightly outperforms
\inexact{} both in number of best results and mean solution value. In the
problems with $8$ and $10$ terminals, \inexact{} has slightly better results in
average. Thus, disregarding the optimality constraint can allow the algorithm to
give better solutions faster especially in hard problems with a large amount of
terminals. 

\begin{figure}[t!]
  \centering
  \begin{subfigure}{.49\textwidth}
    \includegraphics[width=\linewidth]{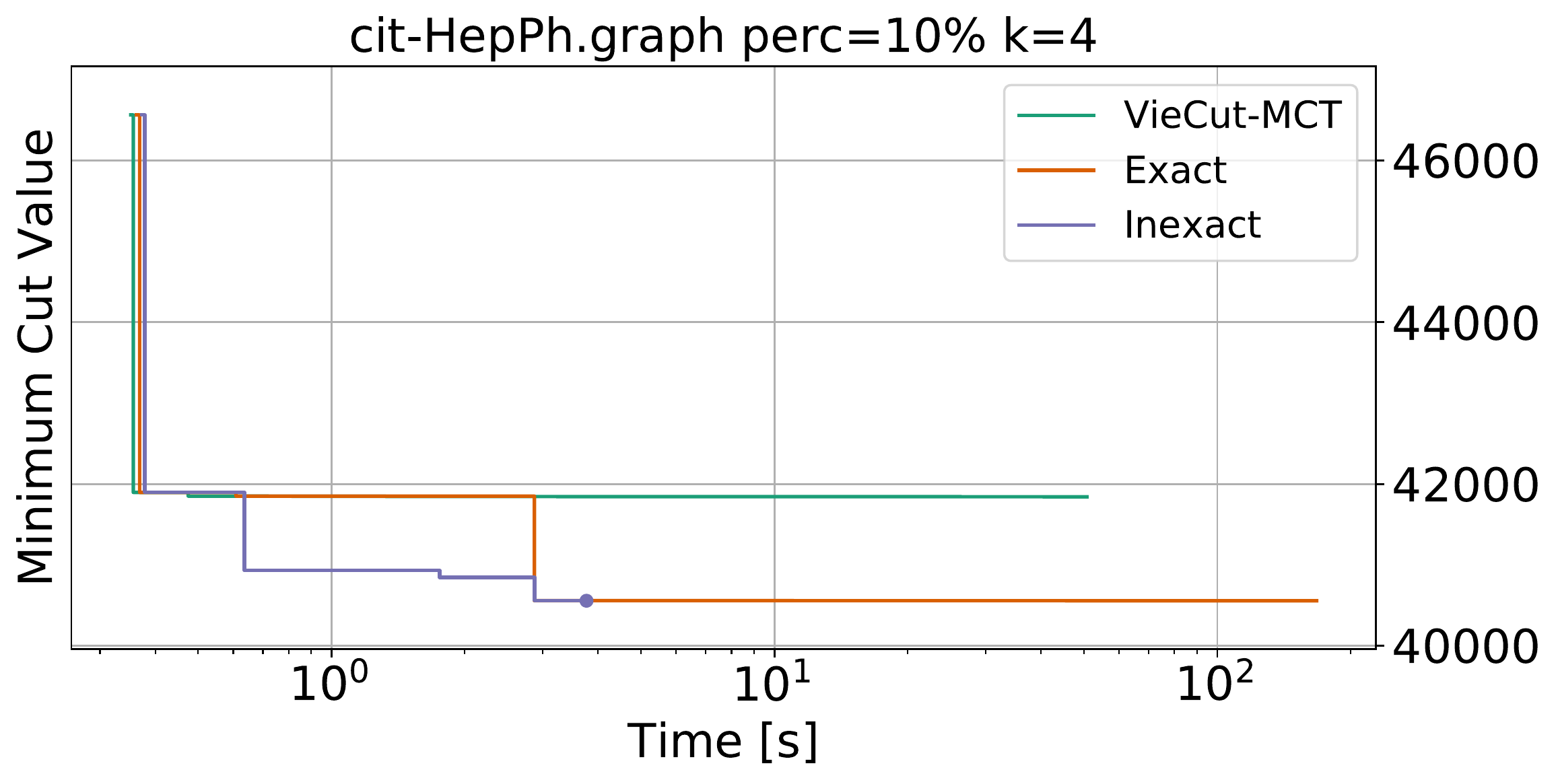}
  \end{subfigure}%
  \begin{subfigure}{.49\textwidth}
    \includegraphics[width=\linewidth]{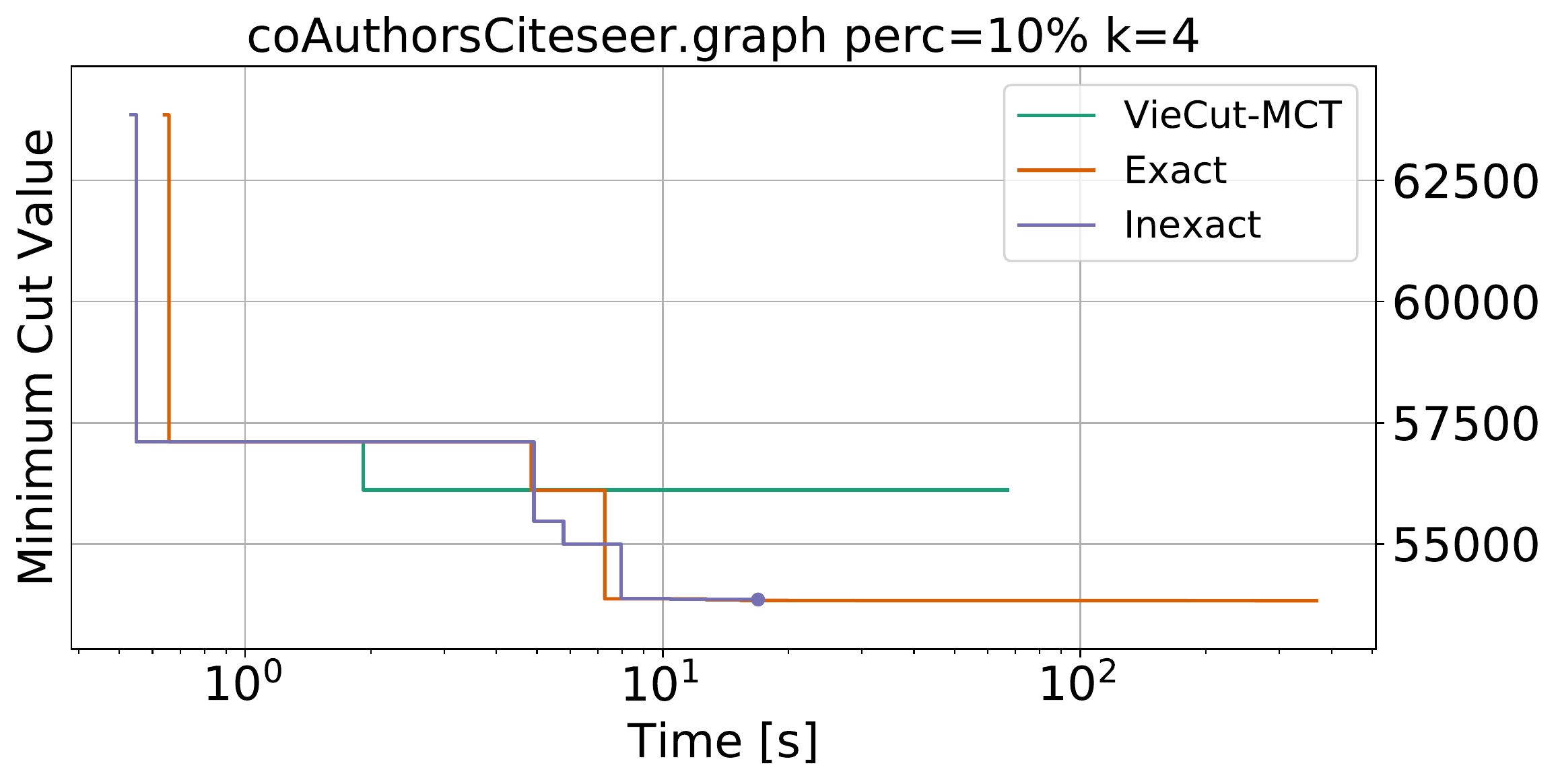}
  \end{subfigure}
  \begin{subfigure}{.49\textwidth}
    \includegraphics[width=\linewidth]{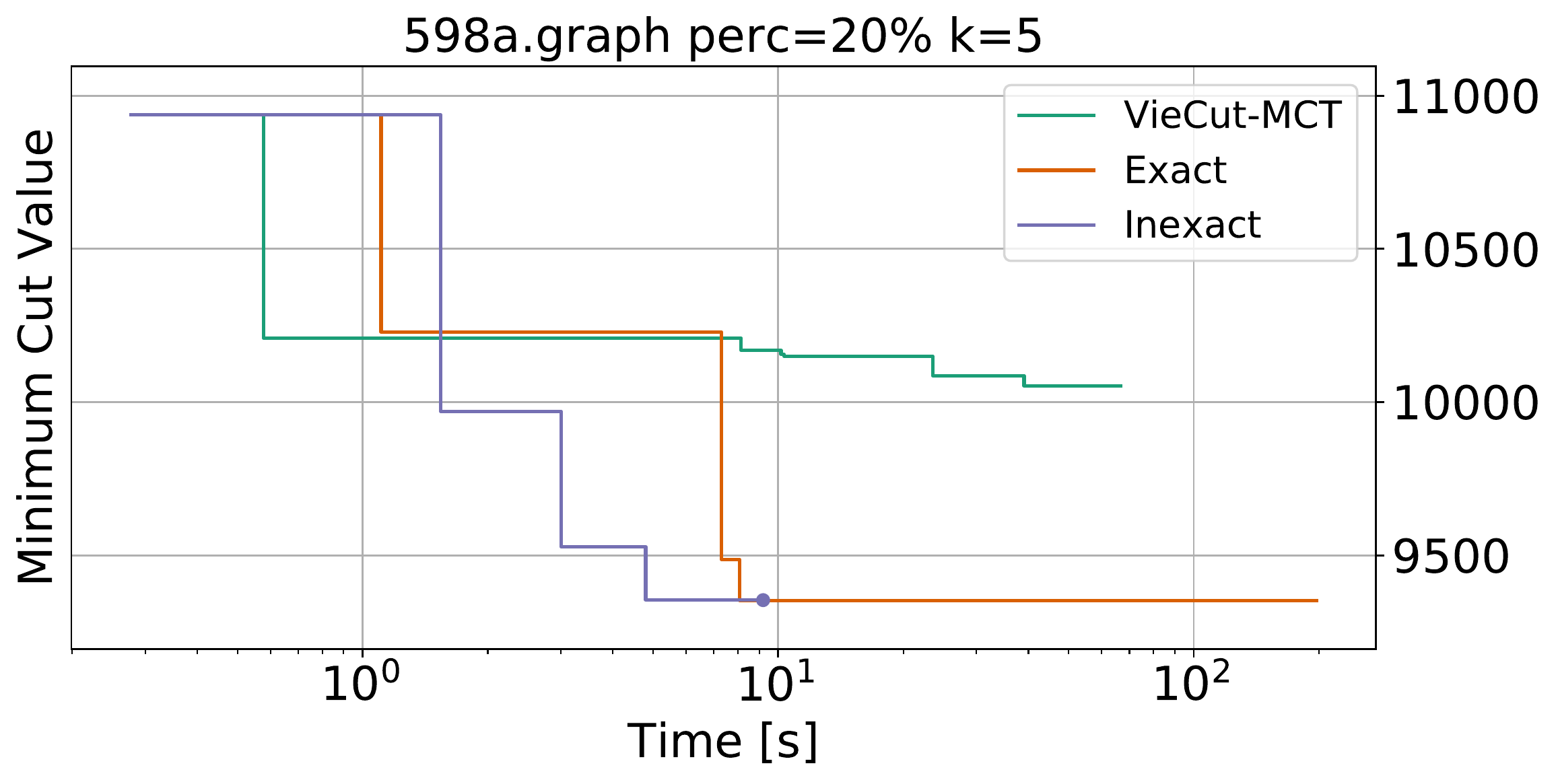}
  \end{subfigure}%
  \begin{subfigure}{.49\textwidth}
    \includegraphics[width=\linewidth]{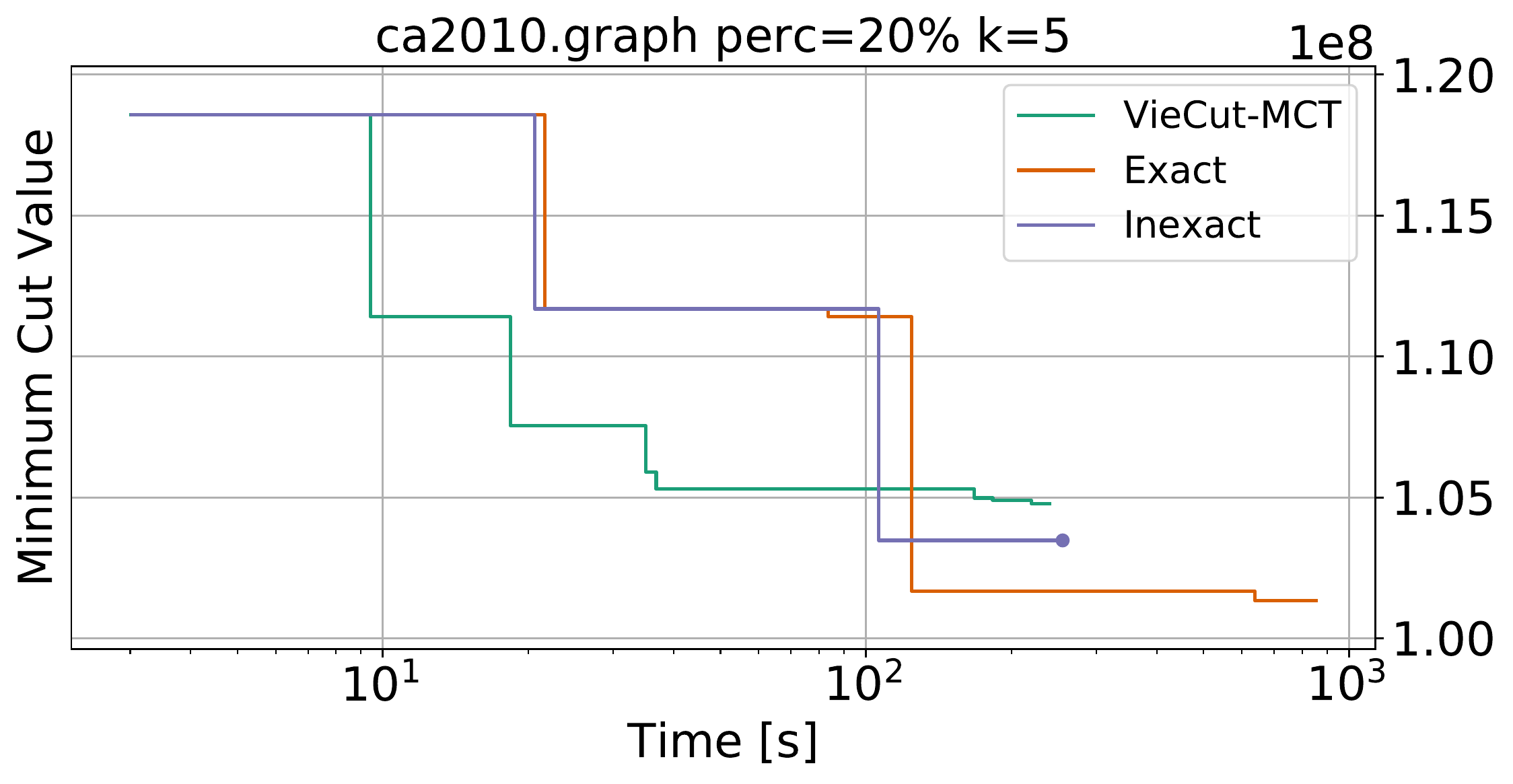}
  \end{subfigure}
  \begin{subfigure}{.49\textwidth}
    \includegraphics[width=\linewidth]{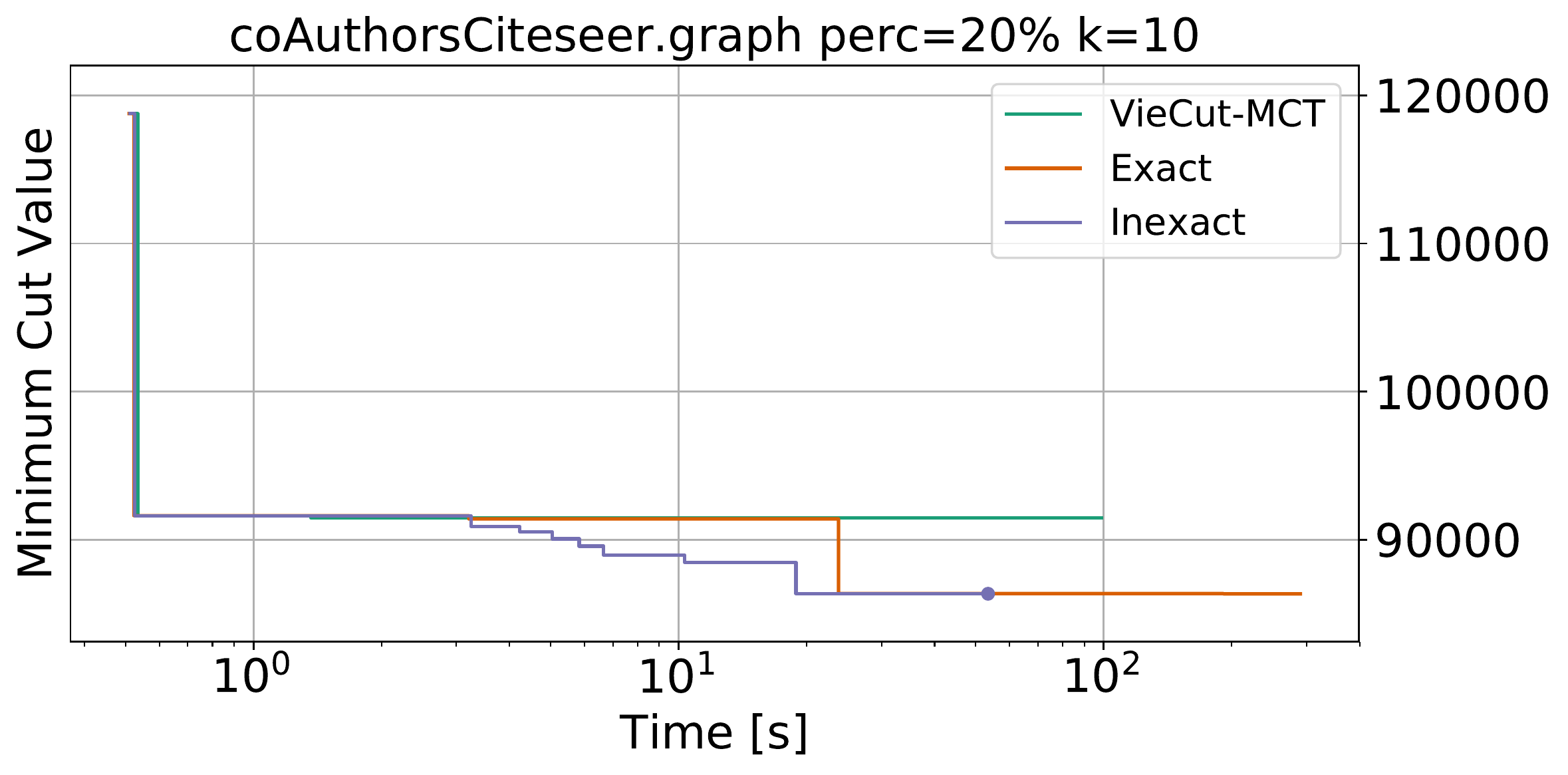}
  \end{subfigure}%
  \begin{subfigure}{.49\textwidth}
    \includegraphics[width=\linewidth]{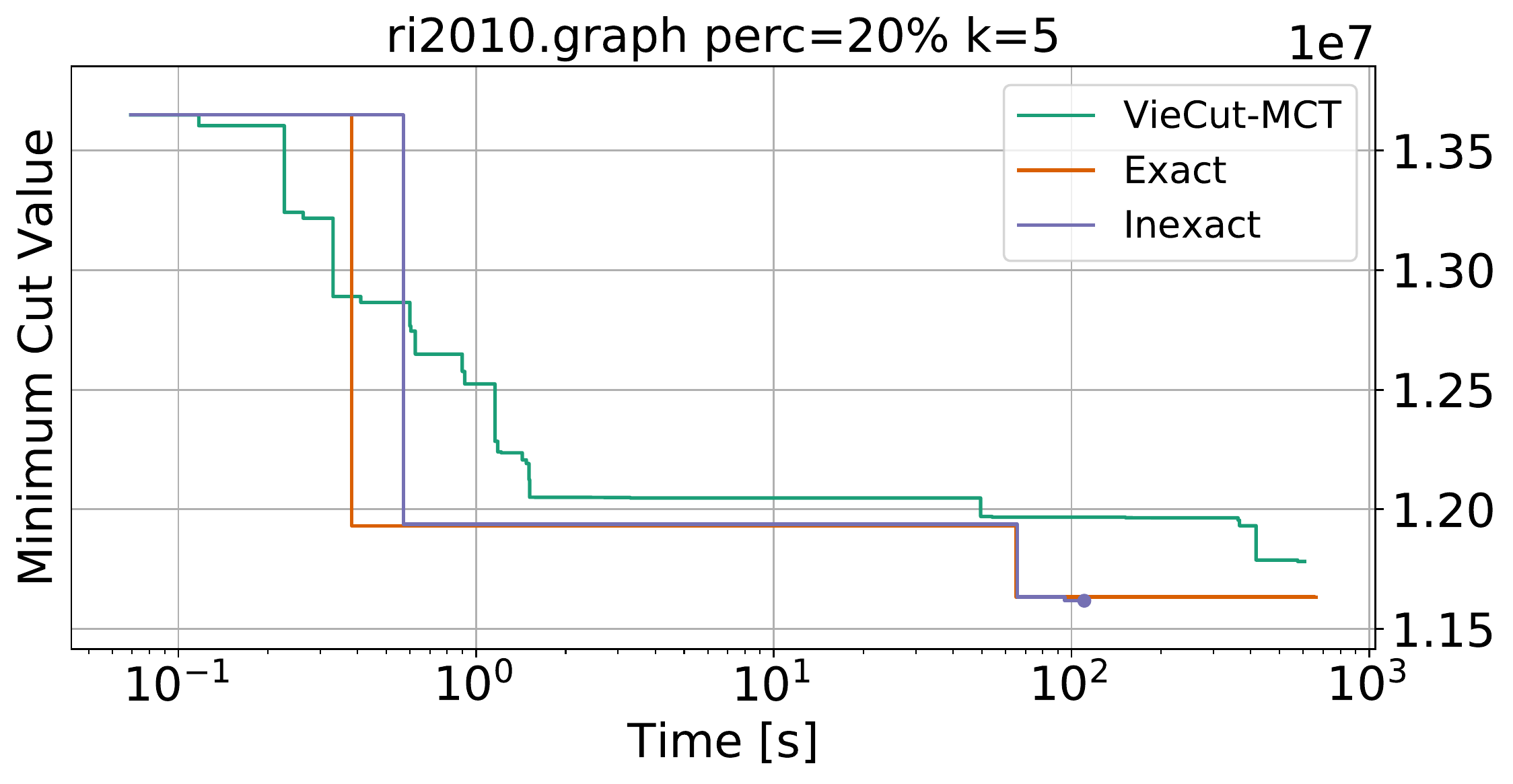}
  \end{subfigure}

  \caption{Progression of best result over time. Dot at end marks termination of algorithm.}
  \label{fig:progress}
\end{figure}

However, both algorithms outperform \vcbase{} on almost all instances where
not all algorithms give the same result. Here, \exact{} gives a better result
than \vcbase{} in $66\%$ of all instances, while \vcbase{} gives the better
result in only $2\%$ of all instances. As most problems do not terminate with an
optimal result, we are unable to say how far the solutions are from the globally
optimal solution. Note that \inexact{} gives an optimal result in all instances
in which all algorithms terminate. Figure~\ref{fig:progress} shows the progress
of the best solution for the algorithms in a set of problems. For both \exact{}
and \inexact{} we can see large improvements to the cut value when the local
search algorithm is finished on the first subproblem. In contrast, \vcbase{} has
more small step-by-step improvements and generally gives worse results.  

Figure~\ref{fig:pp54} shows the performance profile for the instances in this
section. Here we can see that \vcbase{} has significantly worse results on a
large subset of the instances, with more than $10\%$ of instances where the
result is worse by more than $10\%$. Also, on a few instances, the results given
by \exact{} and \inexact{} differ significantly. In general, both of them
outperform \vcbase{} on most instances that are not solved to optimality by every
algorithm.

\section{Conclusion}\label{c:mc:s:conclusion}

In this chapter, we give a fast parallel solver that gives high-quality
solutions for large multiterminal cut problems. We give a set of
highly-effective reduction rules that transform an instance into a smaller
equivalent one. Additionally, we directly integrate an ILP solver into the
algorithm to solve subproblems well suited to be solved using an ILP; and
develop a flow-based local search algorithm to improve a given optimal solution.
These optimizations significantly increase the number of instances that can be
solved to optimality and improve the cut value of multiterminal cuts in
instances that can not be solved to optimality. Additionally, we give an inexact
algorithm for the multiterminal cut problem that aggressively shrinks the graph
instances and is able to outperform the exact algorithm on many of the
hardest instances that are too large to be solved to optimality while still
giving the exact solution for most easier instances. Important future work
consists of improving the scalability of the algorithm by giving a distributed
memory version.

\printbibliography

\backmatter

\thispagestyle{empty}

\end{document}